\documentclass[a4paper,11pt]{article} 
\usepackage[utf8x]{inputenc} 
\usepackage[english]{babel}
\usepackage{amsmath,amsthm,amsfonts,amssymb,mathrsfs, bbm, amsthm,epsfig,epstopdf,titling,array}
\usepackage{latexsym}
\usepackage{slashed}
\usepackage{authblk}
\usepackage{lineno}
\usepackage{eucal}
\usepackage[dvipsnames]{xcolor}
\usepackage{hyperref}
\usepackage{afterpage}
\usepackage{lmodern}
\usepackage[a4paper,left=3cm,right=3cm]{geometry}
\usepackage[compat=1.1.0]{tikz-feynman}
\usepackage[colorinlistoftodos]{todonotes}
\usepackage{mathtools,slashed}
\usepackage[normalem]{ulem}
\usepackage{textcomp}
\usepackage{float}
\usepackage{microtype}

\usepackage{braket}
\usepackage{environ}
\usepackage{verbatim}
\usepackage{dsfont}
\usepackage{tikz}
\usetikzlibrary{shapes}
\usetikzlibrary{positioning}
\usetikzlibrary{patterns}
\usetikzlibrary{decorations.pathmorphing}

\usetikzlibrary{trees}
\usetikzlibrary{calc}
\usetikzlibrary{decorations.pathmorphing}
\usetikzlibrary{decorations.markings}
\usetikzlibrary{shapes}
\usetikzlibrary{fit}

\tikzset{vertex/.style={circle,fill=black,inner sep=2pt},
	ctVertex/.style={diamond,fill=white,draw,inner sep=2.2pt},
	bigvertex/.style={circle,fill=black,inner sep=3pt},
	E/.append style={fill=white,draw},
	probeEP/.style={circle,fill=black,draw,inner sep=2pt,
		prefix after command= {\pgfextra{\tikzset{every pin/.style = {pin edge={decorate,decoration={snake,amplitude=2pt,segment length =4pt}}}}}}
	},
	bareProbeEP/.style={rectangle,fill=black,draw,inner sep=3pt,
		prefix after command= {\pgfextra{\tikzset{every pin/.style = {pin edge={decorate,decoration={snake,amplitude=2pt,segment length =4pt}}}}}}
	},
	nuEP/.style={circle,fill=white,draw, inner sep=2pt},
	linelabel/.style={sloped,above,very near start, inner sep=1pt,execute at begin node=$\scriptstyle,execute at end node=$},
	baseline=(current  bounding  box.center),doubled/.style={double distance= 1pt,line width=1.5pt}
}
\pgfdeclarelayer{background}
\pgfsetlayers{background,main}

\newcommand{\tn}[1]{{\left\vert\kern-0.25ex\left\vert\kern-0.25ex\left\vert #1 
		\right\vert\kern-0.25ex\right\vert\kern-0.25ex\right\vert}}

\definecolor{BeauBlue}{rgb}{0, 0.2, .9}
\definecolor{BeauOrange}{rgb}{.8, .1, 0}
\usepackage{hyperref}
\hypersetup{
	colorlinks = true,
	linkcolor = BeauBlue,
	urlcolor = cyan,
	citecolor = BeauOrange,
}

\def\red#1{{\color{red}#1}}

\numberwithin{equation}{section}

\newtheorem{theorem}{Theorem}[section] % reset theorem numbering for each chapter
\newtheorem{proposition}[theorem]{Proposition}

\newtheorem{corollary}[theorem]{Corollary}
\newtheorem{lemma}[theorem]{Lemma}

\theoremstyle{definition}
\newtheorem{definition}[theorem]{Definition}  % definition numbers are dependent on theorem numbers

\newtheorem{remark}[theorem]{Remark}
\newtheorem{assumption}[theorem]{Assumption}

\DeclareMathOperator{\Tr}{Tr} 

\newcommand{\numberset}{\mathbb}
\newcommand{\N}{\numberset{N}} 
\newcommand{\Z}{\numberset{Z}} 
 
\newcommand{\R}{\numberset{R}} 
\newcommand{\C}{\numberset{C}}

\newcommand{\ul}{\underline}

\newcommand{\D}{\mathbb D_{\beta,L,\mathfrak a}^{\mathrm a}}
\newcommand{\hProp}{\hat{g}^{(\leq N)}_{\omega,\varepsilon,\mathfrak{a}}}
\newcommand{\Prop}{{g}^{(\leq N)}_{\omega,\varepsilon,\mathfrak{a}}}

\newcommand{\La}{\Lambda_{\beta,L,\mathfrak a}}
\newcommand{\Ma}{\mathfrak{a}}
\newcommand{\Dp}{\mathbb D_{\beta,L}^{\mathrm{p}}}
\newcommand{\Den}{\mathcal D_{\omega,\Ma}}
\newcommand{\DenI}{\mathcal D_{\omega_1,\Ma}}

\usepackage{titling}
\thanksmarkseries{arabic}

\title{Large Scale Dynamical Response of Interacting 1$d$ Fermi Systems}
\author[1]{Marcello Porta}
\author[2]{Giuseppe Scola}
\author[1]{Harman Preet Singh}
\affil[1]{Mathematics Area, SISSA, Via Bonomea 265, 34136 Trieste, Italy}
\affil[2]{Department of Mathematics and Computer Science, Universit\`a della Calabria, Ponte P. Bucci - Cubo 30B, 87036 Arcavacata di Rende, Italy}

\date{\today}

\begin{document}

\maketitle

\begin{abstract} We consider the dynamics of a class of weakly interacting, gapless $1d$ fermionic systems, in presence of small external perturbations slowly varying in space and in time. We consider the evolution of the expectation values of the charge density and of the current density, in the thermodynamic limit and for low enough temperatures. We prove the validity and the asymptotic exactness of linear response in the limit of vanishing space-time variation of the perturbation, and we provide the explicit expression of the response of the system. The proof relies on the representation of the real time Duhamel expansion in terms of Euclidean correlation functions, for which we provide sharp estimates using rigorous renormalization group methods. The asymptotic exactness of linear response holds thanks to a cancellation for the scaling limit of the correlations that is reminiscent of bosonization, and which is derived rigorously using emergent chiral Ward identities.
\end{abstract}

\tableofcontents 

\section{Introduction}

\paragraph{Transport in gapped and gapless systems.} We consider a system of interacting fermions on a one-dimensional lattice $\Lambda_{L}$ of length $L$, with periodic boundary conditions, $\Lambda_{L} = \mathbb{Z} / L \mathbb{Z}$. At equilibrium, the Gibbs state of the system is:
\begin{equation}
\rho_{\beta,L} = \frac{e^{-\beta (\mathcal{H} - \mu \mathcal{N})}}{\Tr_{\mathcal{F}_{L}} e^{-\beta (\mathcal{H} - \mu \mathcal{N})}}
\end{equation}
where $\mathcal{H}$ is the many-body Hamiltonian acting on the fermionic Fock space $\mathcal{F}_{L}$ (to be specified below), and $\mathcal{N}$ is the number operator. The parameters $\beta$ and $\mu$ are respectively the inverse temperature and the chemical potential; we will consider the case in which $\beta,L$ will eventually be taken to infinity, and no uniform spectral gap separates $\mu$ from the spectrum of $\mathcal{H}$.

From a physics viewpoint, we are interested in describing the transport properties of $1d$ metallic systems. In order to generate a nontrivial charge transport, we shall evolve the equilibrium state of the system with a time-dependent Hamiltonian $\mathcal{H}(\eta t) = \mathcal{H} + e^{\eta t} \mathcal{P}$:
\begin{equation}\label{eq:dyn}
i\partial_{t} \rho(t) = [ \mathcal{H}(\eta t), \rho(t) ]\;,\qquad t\leq 0
\end{equation}
where $\eta > 0$ is the adiabatic parameter, to be sent to zero after the limit $\beta, L\to \infty$. Given a local observable $\mathcal{O}$, a simple application of Duhamel formula yields the following expression for the variation of the expectation value at first order in the perturbation:
\begin{equation}\label{eq:kubo}
\Tr \mathcal{O} \rho(0) - \Tr \mathcal{O} \rho(-\infty) = -i\int_{-\infty}^{0} dt\, e^{\eta t} \langle [ \mathcal{O}, \tau_{t}(\mathcal{P}) ] \rangle_{\beta,L} + \text{higher orders in $\mathcal{P}$}
\end{equation}
where $\tau_{t}(\mathcal{P}) = e^{i \mathcal{H} t} \mathcal{P} e^{-i\mathcal{H}t}$. The first term in the right-hand side of (\ref{eq:kubo}) is called Kubo formula, and it expresses the linear response of the system. Kubo formula is widely used in condensed matter physics, and it is a nontrivial mathematical challenge to justify its validity: that is, to show that the higher order terms are indeed subleading, uniformly in the system's size and as $\eta \to 0^{+}$.

For gapped many-body systems, that is in cases in which $\mu$ belongs to a spectral gap of $\mathcal{H}$ uniformly in the system's size, Kubo formula has been rigorously justified at zero temperature and in the thermodynamic limit, as a consequence of the many-body adiabatic theorem \cite{BDRF, MT, Teu}, see \cite{HT} for a review and further references. The method applies to general, $d$-dimensional lattice models, for a large class of local many-body Hamiltonians with a spectral gap above the ground state. In cases relevant for the quantum Hall effect, one can actually prove that Kubo formula is exact \cite{BDRFL, Teu2}: all power law corrections to linear response are zero, which contributes to the explanation of the astonishing precision of the observed quantization of the Hall conductivity. A different approach to the validity of Kubo formula, based on a complex deformation from real to imaginary times, and on cluster expansion methods, has been obtained in \cite{GLMP}. This work will play an important role in the present paper.

Much less is known about the validity of linear response for gapless many-body systems. There, Kubo formula is still widely used in applications, however with less control on its rigorous validity. In contrast to the gapped case, no adiabatic theorem is available for the dynamics of the ground state of the system. However, one might still have that Kubo formula holds true, in a physically relevant scaling regime. In the case of a finite system coupled to external leads at thermal equilibrium, the related problem of proving convergence to a non-equilibrium steady state (and proving the validity of Kubo formula) has been studied in the literature. Rigorous results have been proved under suitable spectral assumptions, for non-interacting \cite{AH, AP, Ne} and also interacting \cite{FMU, JOP, MCP, CMP2} systems.

Here we shall focus on extended, closed systems. In the absence of many-body interactions, the work \cite{PS} proved the validity of Kubo formula for $1d$ systems and for the edge modes of $2d$ systems, for time-dependent perturbations that are weak and slowly varying in space and in time. The observable $\mathcal{O}$ is chosen to be the density or the current density operator, which are related by a lattice continuity equation. Restricting for simplicity to the $1d$ case, the time-dependent Hamiltonian is:
\begin{equation}\label{eq:H0}
\mathcal{H}_{0}(\eta t) = \sum_{x,y \in \Lambda_{L}} \sum_{\rho,\rho' \in S_{M}} a^{*}_{x,\rho} H_{\rho\rho'}(x,y) a_{y,\rho'} + e^{\eta t}\sum_{x\in \Lambda_{L}} \sum_{\rho \in S_{M}} \theta \mu(\theta x) a^{*}_{x,\rho} a_{x,\rho}\;,
\end{equation}
where $H(x,y) \equiv H(x-y)$ is finite-ranged, $\mu(x)$ is a smooth, fast decaying function, and we allow for $M$ internal degrees of freedom, whose set of labels is $S_{M} = \{1,\ldots, M\}$. It is assumed that, at the Fermi level, the energy bands of the Hamiltonian $H$ cross the Fermi energy transversally, at the Fermi points $k_{F}^{\omega}$, $\omega = 1,\ldots, N_{f}$.

The scaling of the perturbation is chosen so that its position space $\ell^{1}$ norm is finite. Still, due to the fact that the system is effectively exposed to the perturbation for a time-scale of order $1/\eta$, proving the validity of Kubo formula is a nontrivial task. To achieve this goal, the work \cite{PS} used a mapping of the real-time Duhamel series into an imaginary-time cumulant expansion, proved in \cite{GLMP}. The advantage of this representation is that the oscillatory behavior of real-time correlations translates into (mild) decay properties in imaginary time, which are easier to exploit to control time integrals and spatial sums. In order to prove the validity and asymptotic exactness of Kubo formula, the work \cite{PS} relied on a cancellation for the density-density cumulants of the scaling limit of the model, which describes massless relativistic fermions in $1+1$ dimensions. This cancellation is related to bosonization, and it becomes easier to detect once the time-dependent perturbation theory is formulated in imaginary times. Furthermore, the work \cite{PS} proved convergence of the Duhamel expansion without a smallness assumption on $\mu(\cdot)$: without this improvement, the $m$-th order expansion in the Duhamel series (\ref{eq:kubo}) would be estimated proportionally to $\theta^{m-1} \times \| \hat \mu \|^{m-1}_{\ell^{1}}\times \eta^{2-m}$ (with $\hat \mu(p)$ the Fourier transform of $\mu(x)$), which implies convergence for $\theta / \eta$ bounded and for $\| \hat \mu \|_{\ell_{1}}$ small enough. Instead, the improved analysis of \cite{PS} allows to extract a factor $1/m!$ in the estimate of the correlator, and a factor $\eta^{\gamma}$ with $\gamma>0$ due to the cancellation for the scaling limit. Taking into account these improvements, \cite{PS} allows to prove that the higher order contributions to linear response are summable and vanishingly small as $\eta,\theta \to 0$, with $0< |\theta/\eta| \leq C\left|\log \eta\right|$.

\paragraph{Results of this work.} The goal of the present paper is to extend \cite{PS} to the case of weakly interacting fermions. We focus on one-dimensional fermions, but we believe that, with some extra effort, the analysis could be extended to study edge modes of weakly interacting $2d$ systems. The time-dependent Hamiltonian is:
\begin{equation}
\mathcal{H}(\eta t) = \mathcal{H}_{0}(\eta t) + \lambda \sum_{x,y \in \Lambda_{L}}\sum_{\rho,\rho' \in S_{M}} a^{*}_{x,\rho} a_{x,\rho} w(x-y) a^{*}_{y,\rho'} a_{y,\rho'}\;,
\end{equation}
where $w(\cdot)$ is a finite-ranged two-body potential, and $\mathcal{H}_{0}(\eta t)$ is as in (\ref{eq:H0}). The single-particle Hamiltonian $H$ is taken as in \cite{PS}, with the further requirement that, in the $2$-body scattering, the conservation of Fermi momenta only holds with momenta equal in pairs (elastic scattering); this is generically true for spinless fermions in the absence of special symmetries. An example of $H$ that fits the setting is the case of the lattice Laplacian in $1d$, for spinless fermions, for chemical potential $\mu$ in the spectrum away from band edges.

We consider the adiabatic protocol (\ref{eq:dyn}). The main result of the present paper is the proof of validity of Kubo formula for the expectation of the density and of the current density operator, as $\eta \to 0^{+}$, taken after $\beta, L\to \infty$, as in \cite{PS}. More precisely, we prove that, denoting by $j_{0,x}$ the density operator and by $j_{1,x}$ the current density operator (see Eqs. (\ref{eq:currFou}), (\ref{eq:Jnudef})), for $\theta / \eta$ bounded and for $|\lambda|$ small enough:
\begin{equation}\label{eq:resp}
\Tr j_{\nu,x} \rho(0) - \Tr j_{\nu,x} \rho_{\beta, L} = -i\int_{-\infty}^{0} dt\, e^{\eta t} \langle [ j_{\nu,x}, \tau_{t}(\mathcal{P}) ] \rangle_{\beta,L} + O(\theta \eta^{\gamma}) + O\bigg( \frac{\theta}{\eta^{3} \beta} \bigg)\;,
\end{equation}
where $\gamma>0$ and where all error terms are bounded uniformly in $L$. Moreover, we prove that the Kubo formula appearing on the right-hand side is finite and of order $\theta$ uniformly in $\beta, L$; therefore, Eq. (\ref{eq:resp}) proves the exactness of linear response as $\beta,L \to \infty$ followed by the joint limit $\eta,\theta \to 0$. It is important to emphasize that $|\theta|, \eta, |\lambda|$ are chosen small uniformly in $\beta, L$; however, a technical limitation of the proof is that the range of allowed coupling constants $\lambda$ shrinks to zero as $|\theta| / \eta \to \infty$. The range of parameters we consider includes the physically relevant Euler scaling, in which the perturbation scales in space and time in the same fashion, \emph{i.e.} $\theta\propto\eta$. In the quantum quench setting, this space-time scaling is widely studied in the physics literature, in the context of the (generalized) hydrodynamical approximation, see reviews \cite{Doyrev, CDY}. Our setting can be viewed as an adiabatic counterpart of this regime.

Furthermore, the linear response can be explicitly evaluated, in terms of the renormalized Fermi velocities of the quasi-particles; see Theorem \ref{thm:main} for the explicit expression. For a general number of Fermi points, to the best of our knowledge these expressions are rigorously derived in the present paper for the first time.

\paragraph{Aspects of the proof.} The proof follows the strategy of \cite{PS}, however in a much more difficult context, due to the presence of the many-body interaction. The strategy of \cite{PS} was based on an explicit analysis of the Euclidean correlation functions of the system, performed using the quasi-freeness of the state. Here, the correlation functions of the system are not explicit, since the many-body system is interacting. In order to deal with this issue, we use rigorous renormalization group methods, that allow to express the Euclidean correlation functions of the system at equilibrium in terms of a convergent renormalized expansion. These RG methods have been applied to a wide class of statistical mechanics and condensed matter systems in the last three decades, see \cite{Mabook} for a review. The class of models considered in the present paper belongs to the (multichannel) Luttinger liquid universality class: the Luttinger liquid is a $1+1$ dimensional QFT that describes interacting massless chiral fermions, and it has special integrability features \cite{ML}. The rigorous RG analysis of $1d$ quantum systems belonging to this universality class has been pioneered in \cite{BGPS}, for spinless fermions, and in \cite{BoM} for spinful fermions. These works relied on the exact solvability of the Luttinger model, which represented a limitation of the approach. Such restriction has been removed in \cite{BMdensity, BMchiral}, via a different strategy, based on the combination of Ward identities and Schwinger-Dyson equations, that ultimately allows to control the flow of all RG couplings without relying on the exact solution. The outcome is a construction of a nontrivial RG fixed point, characterized by anomalous scaling exponents. In the applications to condensed matter physics, the method has been used to compute response functions for interacting $1d$ systems, and to prove the validity of the Haldane relations connecting anomalous exponents and transport coefficients, \cite{BFM2, BFM3}. More recently, these techniques have been used to compute edge transport coefficients for interacting $2d$ topological insulators, and to prove the quantization of the edge conductance, \cite{AMP, MPmulti}. 

The technical core of the present paper is a renormalization group analysis of the $m$-point density-density and current-density Euclidean correlation functions for the interacting $1d$ system at equilibrium. Similar analyses have been performed {\it e.g.} in \cite{BFM1}, for the Thirring model, and in \cite{GGM}, for the interacting Ising model. With respect to these works, here we obtain an improved estimate for the Euclidean correlations for the special external momenta compatible with our adiabatic protocol (after the complex deformation to imaginary times, the temporal component of all independent incoming momenta of the Euclidean correlations is equal to $\eta$), which ultimately allows to prove the summability of the Duhamel expansion without any smallness requirement on the perturbation $\mu(\cdot)$. Furthermore, we extract the scaling limit contribution from the general $m$-point correlation, and we prove that it is vanishing. This step was done via an explicit computation in \cite{PS}, and it is here derived via the RG analysis of the scaling limit of the model. A similar cancellation has been derived in the past for the Thirring model or the Luttinger model \cite{BFM1, GMT}. Here we extend the cancellation to the multichannel Luttinger liquid universality class. As a technical note, we only need to establish the cancellation for momentum-space correlations with momenta compatible with the adiabatic protocol, which is easier than working in configuration space as in \cite{BFM1, GMT}. Finally, the explicit form of the susceptibility and of the conductance is determined using lattice Ward identities, following a strategy similar to \cite{MPmulti}.

\paragraph{Perspectives.} As a future perspective, it would be very interesting to connect the results and the methods used in the present paper with the derivation of hydrodynamic models for the non-equilibrium dynamics of integrable systems, a subject on which a lot of progress has been obtained in recent years starting from \cite{BCNF, CDY}, see \cite{Doyrev, Essrev} for reviews. Also, it would be interesting to extend the analysis of the present work to higher dimensional cases, such as interacting graphene and interacting Weyl semimetals; see \cite{GMPcond, GMPweyl} for rigorous results about the universality of linear and quadratic response for these models, respectively. Finally, we believe that the method of the present work could be extended to disordered systems, in presence of weak quasi-periodic potentials; see \cite{CMP} for a recent computation of Kubo formula for edge currents of $2d$ topological insulators in the absence of interactions and in presence of a weak quasi-periodic potential.

\paragraph{Structure of the paper.} The paper is structured as follows. In Section \ref{sec:1d} we define the model, in Fock space, and we state our main result, Theorem \ref{thm:main}. In Section \ref{sec:duha} we review the Wick rotation of the Duhamel series, proved in \cite{GLMP}.  In Section \ref{sec:RGLatticeModel} we discuss the renormalization group analysis of the model, which takes as starting point the Grassmann integral representation of the theory. Until Section \ref{sec:dimest}, the analysis is similar to previous applications of rigorous RG methods \cite{AMP, MPmulti}. Then, in Section \ref{sec:dimest} we discuss the RG analysis of the multi-point correlation functions, and their improved bounds, which allow to prove summability of the Duhamel expansion. These bounds rely on the Gallavotti-Nicolò tree expansion, reviewed in Section \ref{sec:GNtrees}. In Section \ref{Sec:RefMod} we define the reference model, and we prove the Ward identity for the $m$-point density correlations, which allows to prove the vanishing of all correlations with $m\geq 3$. In Section \ref{sec:final} we prove our main result, by comparing the renormalized expansions of lattice and reference model, extracting the cancellation of the reference model, and estimating the remainder taking into account the improvement of Section \ref{sec:dimest}. Also, we compute the susceptibility and the charge conductance, using the explicit form of the scaling limit, and lattice Ward identities to fix the form of the non-universal part. Finally, in Appendix \ref{app:ref} we collect the details of the proof of the vanishing of the density cumulants of the reference model of order higher than two.

\paragraph{Acknowledgments.} H. P. S. thanks F. Caragiulo, S. Fabbri and L. Goller for useful comments. We acknowledge support by the European Research Council through the grant ERC StG MaMBoQ, n. 802901. M. P. acknowledges support from the MUR, PRIN 2022 project MaIQuFi cod. 20223J85K3. Part of this work has been carried out at the University of Z\"urich, which we thank for the hospitality. This work has been carried out under the auspices of the GNFM of INdAM. We thank the anonymous referees for useful comments on a previous version of this manuscript.

\paragraph{Data availability statement.} This manuscript has no associated data.

\paragraph{Conflict of interest.} The authors have no conflict of interest to declare that are relevant to the content of this article.

\section{Lattice fermions}\label{sec:1d}

\subsection{Hamiltonian and Gibbs state}

Let $L \in 2\mathbb{N}+1$. We consider fermions on a one-dimensional lattice 
\[
\Gamma_{L} = \left[-\left\lfloor\frac{L}{2}\right\rfloor,\left\lfloor\frac{L}{2}\right\rfloor\right]\cap \Z
\]
endowed with periodic boundary conditions. We shall allow for internal degrees of freedom, {\it e.g.} the spin, and we shall denote by $S_{M} = \{ 1,\ldots, M \}$ the set of their labels. We shall use the notation $\Lambda_{L} = \Gamma_{L} \times S_{M}$, which we can view as a decorated $1d$ lattice. Given $x\in \Gamma_{L}$ and $\rho \in S_{M}$, we shall denote by ${\bf x} = (x,\rho)$ the corresponding point in $\Lambda_{L}$. Further, we will make use of the following distance on $\Gamma_{L}$:
\begin{equation}
| x- y |_{L} = \min_{n\in \mathbb{Z}} | x - y + n L |\;.
\end{equation}
%
%Since we are interested in the analysis of a class of (weakly) interacting lattice models, in what follows we will introduce the non-interacting Hamiltonian $H$ and the two-body interaction potential $w$.
Let $H$ be a single-particle Hamiltonian on $\Lambda_{L}$. That is, $H$ is a self-adjoint matrix, with entries given by $H({\bf x}; {\bf y}) \equiv H_{\rho\rho'}(x;y)$. We shall suppose that the Hamiltonian is finite-ranged, and we shall assume that the Hamiltonian $H$ is the periodization of an Hamiltonian $H^{\infty}$ defined over $\ell^{2}(\mathbb{Z} \times S_{M})$:
\begin{equation}
H_{\rho\rho'}(x; y) = \sum_{n \in \mathbb{Z}} H^{\infty}_{\rho\rho'}( x + nL; y)\qquad \text{for $x,y\in \Gamma_{L}$.}
\end{equation}
Furthermore, we will assume $H^{\infty}$, and hence $H$,  to be translation-invariant:
\begin{equation}\label{eq:periodiz}
H_{\rho\rho'}^{\infty}( x; y ) = H_{\rho\rho'}^{\infty}( x + z; y + z)\qquad \text{for all $z\in \mathbb{Z}$}.
\end{equation}
A natural example is the lattice Laplacian:
\begin{equation}
(H^{\infty} \psi)(x) = t(\psi_{x-1} + \psi_{x+1} - 2\psi_{x})\;,
\end{equation}
with $t \in \mathbb{R}$ being the hopping parameter.

%Next, we introduce the two-body interaction potential $W$ as a real-valued moltiplication operator on the two-particle Hilbert space, that is determined by a scalar function $w$:
%\begin{equation}
%(W\psi) (\mathbf x, \mathbf y) = W(\mathbf x, \mathbf y)\psi (\mathbf x, \mathbf y) := w(x;y)\psi(\mathbf x, \mathbf y)
%\end{equation}
%for any state $\psi\in \ell^{2}(\Lambda_{L})\wedge\ell^{2}(\Lambda_{L})$. Again, we will suppose $w$ to be translation %invariant
%\begin{equation}
%w(x;y)=w(x+z;y+z)\qquad \forall z\in \mathbb Z/L\mathbb Z,
%\end{equation}
%and we will require it to be finite-range, uniformly in $L$.

Translation-invariance allows to conveniently express the Hamiltonian $H$ in momentum space. To this end, let us introduce the set of quasi-momenta compatible with the periodic boundary conditions:
\begin{equation}\label{eq:BLdef}
B_{L} = \Big\{ k = \frac{2\pi}{L} n \, \Big| \, 0\leq n \leq L-1 \Big\}\;.
\end{equation}
The set $B_{L}$ is the (discretized) Brillouin zone. Given a function $f$ on $\Gamma_{L}$, we define its Bloch transform as:
\begin{equation}\label{eq:fk}
\hat f(k) = \sum_{x\in \Gamma_{L}} e^{-ikx} f(x)\qquad \text{for all $k\in B_{L}$.}
\end{equation}
This relation can be inverted, as:
\begin{equation}
f(x) = \frac{1}{L} \sum_{k\in B_{L}} e^{ikx} \hat f(k)\qquad \text{for all $x\in \Gamma_{L}$.}
\end{equation}
As a function on $\frac{2\pi}{L} \mathbb{Z}$, the function $\hat f(k)$ is periodic with period $2\pi$. Thus, it is natural to introduce the following norm for quasi-momenta:
\begin{equation}
| k |_{\mathbb{T}} = \min_{n\in \mathbb{Z}} | k + 2\pi n |\qquad \text{for all $k\in B_{L}$.}
\end{equation}
Next, we define the Bloch Hamiltonian $\hat H(k)$ as the $M\times M$ Hermitian matrix with entries:
\begin{equation}\label{eq:bloch1d}
\hat H_{\rho\rho'}(k) = \sum_{x\in \Gamma_{L}} e^{-ikx} H_{\rho\rho'}(x; 0)\;.
\end{equation}
%
%\begin{equation}\label{eq:blochInt1d}
%\hat W_{\rho\rho'}(k) = \sum_{x\in \Gamma_{L}} e^{-ikx} w(x; 0)\;.
%\end{equation}
%
Observe that, because of equation (\ref{eq:periodiz}), $\hat H(k)$ is the restriction to $B_{L}$ of $\hat H^{\infty}(k)$, defined on the circle $\mathbb{T}$ of length $2\pi$.

We shall make the following assumptions on the spectrum of the Bloch Hamiltonian and on the chemical potential $\mu \in \mathbb{R}$.

\begin{assumption}[Low energy spectrum]\label{ass:A} There exists $\Delta > 0$ such that the following points are true.
\begin{itemize}
\item[(i)] There exist $N_{f} \in \mathbb{N}$ disjoint sets $I_{\omega} \subset \mathbb T$, labelled by $\omega = 1, \ldots, N_{f}$, and strictly monotone smooth functions $e_{\omega}: I_{\omega} \to \mathbb{R}$ such that:
\begin{equation}
\sigma(H^{\infty}) \cap (\mu - \Delta, \mu + \Delta) = \bigcup_{\omega = 1}^{N_{f}} \mathrm{Ran}(e_{\omega})\;.
\label{eq:spectrum}
\end{equation}
The functions $e_{\omega}$ are the eigenvalue branches of the Bloch Hamiltonian $\hat H^{\infty}$, crossing the Fermi level $\mu$.
\item[(ii)] We introduce the $\omega$-Fermi point $k_{F}^{\omega} \in I_{\omega}$ and the $\omega$-Fermi velocity $v_{\omega}$ as the real numbers satisfying:
\begin{equation}
e_{\omega}(k_{F}^{\omega}) = \mu\;,\qquad v_{\omega} = \partial_{k} e_{\omega}(k_{F}^{\omega})\;. 
\end{equation}
Notice that $v_{\omega} \neq 0$, by the strict monotonicity of $e_{\omega}$. The set of Fermi points $\{k_{F}^{\omega}\}_{\omega = 1,\ldots, N_{f}}$ will be referred to as the (discrete) \emph{Fermi surface} of the system.
\item[(iii)] For any $\omega=1,\dots, N_{f}$, and $k\in I_{\omega}$, $e_{\omega}(k)$ is a non-degenerate eigenvalue of $\hat H^{\infty}(k)$. In particular, the Fermi points $k_{F}^{\omega}$, $\omega = 1, \ldots, N_{f}$ are all different.
\item[(iv)] For any quadruple $(\omega_{1},\omega_{2}, \omega_{3}, \omega_{4})$, we have that
\begin{equation}\label{eq:elasca}
k^{\omega_{1}}_{F} - k^{\omega_{2}}_{F} = k^{\omega_{3}}_{F} - k^{\omega_{4}}_{F}\qquad \mod 2\pi
\end{equation}
only if $\omega_{1}=\omega_{2}$ and $\omega_{3}=\omega_{4}$, or if $\omega_{1}=\omega_{3}$ and $\omega_{2}=\omega_{4}$.
\end{itemize}
\end{assumption}
\begin{remark}
\begin{itemize}
\item[(i)] A paradigmatic example satisfying the assumptions is (minus) the Laplacian on $\mathbb{Z}$, for spinless fermions. The spectrum is described by:
\begin{equation}
\varepsilon(k) = 2t(1-\cos(k))\;,
\end{equation}
with $t>0$ the hopping parameter, so that $N_{f}=2$ for any $\mu\in(0,4t)$.
\item[(ii)] In general, it is a well-known fact that short-ranged one-dimensional lattice models present necessarily an even number $N_{f}$ of Fermi points, with net chirality equal to zero: $\sum_{\omega =1 }^{N_{f}} v_{\omega} / |v_{\omega}| = 0$.
\item[(iii)] Eq. (\ref{eq:elasca}) can be viewed as an elastic scattering assumption for the Fermi momenta of the system. It is generically true for spinless fermions, in the absence of special symmetry properties of the spectrum.
\end{itemize}
\end{remark}
Next, we will introduce the class of many-body interacting systems we will consider in this paper, in the grand-canonical setting. To this end, it is convenient to lift the many particle systems to the fermionic Fock space $\mathcal{F}$,
\begin{equation}
\mathcal{F} = \mathbb{C} \oplus \bigoplus_{n\geq 1} \ell^{2}(\Lambda_{L})^{\wedge n}\;,
\end{equation}
where $\wedge$ is the antisymmetric tensor product. We shall use the notation ${\bf x} = (x, \rho)$ for points in $\Lambda_{L}$, with $x\in \Gamma_{L}$ and $\rho \in S_{M}$, and we shall denote by $a^{*}_{{\bf x}}, a_{{\bf x}}$ the usual fermionic creation and annihilation operators, satisfying the canonical anticommutation relations:
\begin{equation}
\{ a_{{\bf x}}, a^{*}_{{\bf y}} \} = \delta_{{\bf x}, {\bf y}}\;,\qquad \{ a_{{\bf x}}, a_{{\bf y}} \} = \{ a^{*}_{{\bf x}}, a^{*}_{{\bf y}} \} = 0\;.
\end{equation}
Later, it will also be convenient to switch to Fourier space. For any $k \in B_{L}$, we define the Bloch transform of the fermionic operators as:
\begin{equation}\label{eq:foua}
\hat a_{(k, \rho)} = \sum_{x\in \Gamma_{L}} a_{(x, \rho)} e^{-ikx}\;,\qquad \hat a^{*}_{(k, \rho)} = \sum_{x\in \Gamma_{L}} a^{*}_{(x,\rho)} e^{ikx}\;.
\end{equation}
Equations (\ref{eq:foua}) can be inverted as, for $(x,\rho)\in \Lambda_{L}$:
\begin{equation}\label{eq:foua2}
a_{(x,\rho)} = \frac{1}{L} \sum_{k \in B_{L}} e^{ikx} \hat a_{(k, \rho)}\;,\qquad a^{*}_{(x,\rho)} = \frac{1}{L} \sum_{k \in B_{L}} e^{-ikx} \hat a^{*}_{(k, \rho)}\;.
\end{equation}
In terms of the creation and annihilation operators, we define the many-body Hamiltonian $\mathcal H$ acting on $\mathcal{F}$ as:
\begin{equation}\label{eq:H1d}
\mathcal{H} = \sum_{{\bf x}, {\bf y} \in \Lambda_{L}} a^{*}_{{\bf x}} H({\bf x}; {\bf y}) a_{{\bf y}}+\lambda\sum_{{\bf x}, {\bf y} \in \Lambda_{L}}\Big(n_{{\bf x}}-\frac{1}{2}\Big)w({\bf x}; {\bf y})\Big(n_{{\bf y}}-\frac{1}{2}\Big)\;,
\end{equation}
where: $\lambda\in\mathbb R$ is the coupling constant, $n_{{\bf x}} = a^{*}_{{\bf x}} a_{{\bf x}}$ is the density operator; and $w({\bf x}; {\bf y})$ is the two-body potential. As for the single-particle Hamiltonian, we will suppose that $w$ is translation-invariant, compatible with the periodic boundary conditions, and finite-ranged. Futhermore, for simplicity we will assume that $w({\bf x}; {\bf y}) \equiv w(x;y)$. That is:
\begin{equation}
w({\bf x}; {\bf y}) = \frac{1}{L} \sum_{k\in B_{L}} e^{ik(x-y)} \hat w(k)\;,
\end{equation}
where $\hat w(k)$ is the restriction to $B_{L}$ of a real-analytic function on $\mathbb{T}$.
\begin{remark} The factors $-1/2$ in (\ref{eq:H1d}) amount to a shift of the many-body Hamiltonian by a constant term and by a term proportional to the number operator, which can be reabsorbed in the definition of the chemical potential of the interacting Gibbs state. These factors are introduced to simplify the Grassmann representation of the Gibbs state, discussed later.
\end{remark}
Finally, the Gibbs state of the system at inverse temperature $\beta>0$ is:
\begin{equation}
\langle \mathcal{O} \rangle_{\beta,L} = \Tr \mathcal{O} \rho_{\beta, L}\;,\qquad \rho_{\beta, L} = \frac{e^{-\beta (\mathcal{H} - \mu \mathcal{N})}}{\Tr e^{-\beta (\mathcal{H} - \mu \mathcal{N})}}\;,
\end{equation}
with $\mu \in \mathbb{R}$ the chemical potential and $\mathcal{N}$ the number operator. We shall assume that Assumption \ref{ass:A} holds, for our choice of the chemical potential $\mu$. Notice that the Gibbs state depends also on the parameters $\lambda,\mu$, but in order to simplify the notation we have chosen to write explicitly only the dependence on the variables $\beta,L$.

\subsection{Dynamics and linear response}

%The goal of this section will be to define the time-dependent protocol that will generate a non-trivial response in the system, and the mapping of the resulting real-time response into imaginary times. This mapping is useful because, as we will see, the imaginary time response can be studied in a much more efficient way than its real-time counterpart.

%\paragraph{Dynamics.} 

\paragraph{Perturbing the system.} We will be interested in the response properties of the system, after exposing it to a time-dependent and slowly varying perturbation. We shall consider time-dependent Hamiltonians of the form:
\begin{equation}\label{eq:tdep}
\mathcal{H}(\eta t) = \mathcal{H} + e^{\eta t} \mathcal{P}\;,
\end{equation}
for $\eta > 0$ and $t\leq 0$. We will consider external perturbations that couple to the density of the system:
\begin{equation} \label{eq:formPert}
\mathcal{P} =  \theta \sum_{{\bf x} \in \Lambda_{L}} \mu(\theta x) a^{*}_{{\bf x}} a_{{\bf x}}\;,
\end{equation}
where: $\mu(\cdot)$ (see Figure \ref{fig:perturb1dint}) is the restriction to $\Gamma_{L}$ of the periodization of a function $\mu_{\infty}(x)$ in $\mathbb{R}$ such that $\hat \mu_{\infty}(p)$ has compact support:
\begin{equation}\label{eq:permu}
\mu(x) = \sum_{n \in \mathbb{Z}} \mu_{\infty}(x + nL)\;,\qquad \text{for $x\in \Gamma_{L}$.}
\end{equation}
Thus, $\hat \mu(p)$ is $\hat \mu_{\infty}(p)$ restricted to $p\in (2\pi/L) \mathbb{Z}$. The compact support restriction could be replaced by fast decay in momentum space, but we prefer to work in a slightly more restrictive setting to avoid extra technicalities. Observe that the parameter $\theta$ in (\ref{eq:formPert}) defines simultaneously the strength and the space-variation of the perturbation, and will be suitably chosen later on. 
\begin{figure}
    \centering
    \includegraphics[scale=0.65]{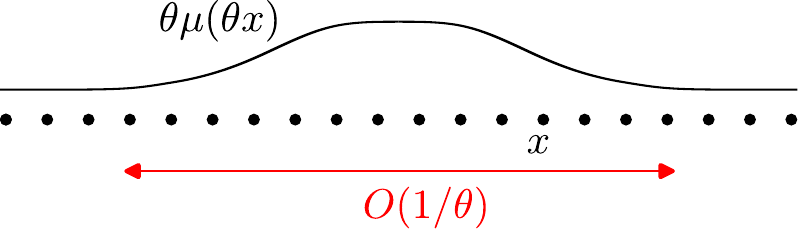}
    \caption{Qualitative representation of the spatial lattice $\Gamma_{L}$ and typical form of the rescaled local potential entering the definition of $\mathcal P$, equation (\ref{eq:formPert}).}
    \label{fig:perturb1dint}
\end{figure}

Let $t,s\leq 0$, and let $\mathcal{U}(t;s)$ be the unitary group generated by $\mathcal{H}(\eta t)$:
\begin{equation}
i\partial_{t} \mathcal{U}(t;s) = \mathcal{H}(\eta t) \mathcal{U}(t;s),\qquad \mathcal{U}(s;s) = \mathbbm{1}\;.
\end{equation}
Let us consider the evolution of the Gibbs state $\rho_{\beta, L}$,
\begin{equation}
\rho(t) = \lim_{T\to +\infty} \mathcal{U}(t;-T) \rho_{\beta, L} \mathcal{U}(t;-T)^{*}.
\end{equation}
We will be interested in the variation of physical observables under this perturbation protocol,
\begin{equation}\label{eq:diffO}
\Tr \mathcal{O} \rho(t) - \Tr \mathcal{O} \rho_{\beta, L}\;.
\end{equation}
We will be interested in the following order of limits: first $L\to \infty$, then $\beta \to \infty$ and then $\eta,\theta \to 0^{+}$. 
Time-dependent perturbation theory is generated by the Duhamel series, that describes the expansion around the unperturbed evolution in the interaction picture. Let $\tau_{t}(\mathcal{A})$ be the Heisenberg evolution of the observable $\mathcal{A}$:
\begin{equation}
\tau_{t}(\mathcal{A}) = e^{i\mathcal{H} t} \mathcal{A} e^{-i\mathcal{H} t}\;.
\end{equation}
We have:
\begin{equation}\label{eq:duhamel}
\begin{split}
\Tr \mathcal{O} \rho(t) - \Tr \mathcal{O} \rho_{\beta, L} &= \sum_{n=1}^{\infty} (-i)^{n} \int_{-\infty \leq s_{n} \leq \ldots \leq s_{1} \leq t} d \underline{s}\, e^{\eta(s_{1} + \ldots + s_{n})} \\&\quad \cdot \langle [ \cdots [[ \tau_{t}(\mathcal{O}), \tau_{s_{1}}(\mathcal{P})], \tau_{s_{2}}(\mathcal{P})] \cdots \tau_{s_{n}}(\mathcal{P}) ] \rangle_{\beta,L}\;.
\end{split}
\end{equation}
It is known that, for $\eta>0$ and for finite $L$, the expansion in the right-hand side is absolutely convergent. The linear response approximation amounts to the truncation of the above expansion to first order in $\mathcal{P}$:
\begin{equation}\label{eq:diffO2}
\Tr \mathcal{O} \rho(t) - \Tr \mathcal{O} \rho_{\beta, L} = -i \int_{-\infty}^{t} ds\, e^{\eta s} \Tr \big[ \tau_{t}(\mathcal{O}), \tau_{s}(\mathcal{P}) \big] \rho_{\beta, L} + \text{h.o.t.}\;.
\end{equation}
Proving that the error term is small uniformly in $\beta$ and $L$ is in general a very difficult problem, in particular for gapless interacting models, as in the present work.
\paragraph{Density and current operators.} We will focus on the dynamics of the expectation values of density and current operators, defined as follows. Let:
\begin{equation}
n_{x} = \sum_{\rho \in S_{M}} n_{(x, \rho)}\equiv \sum_{\rho \in S_{M}} a^{*}_{(x, \rho)}a_{(x,\rho)}\;.
\end{equation}
be the (total) density at $x\in \Gamma_{L}$. In order to introduce the current operator, let us consider the time-variation of the density operator:
\begin{equation}\label{eq:conteq}
\begin{split}
\partial_{t} \tau_{t}( n_{{\bf x}} ) &= i \tau_{t} \big( [ \mathcal{H}, n_{{\bf x}} ]\big) \\
&= \sum_{{\bf y} \in \Lambda_{L}} \tau_{t}( i a^{*}_{{\bf y}} H({\bf y}; {\bf x}) a_{{\bf x}} - i a^{*}_{{\bf x}} H({\bf x}; {\bf y}) a_{{\bf y}}  ) \\
&\equiv \sum_{{\bf y} \in \Lambda_{L}} \tau_{t}( j_{{\bf y}, {\bf x}})\;,
\end{split}
\end{equation}
where $j_{{\bf y}, {\bf x}} = i a^{*}_{{\bf y}} H({\bf y}; {\bf x}) a_{{\bf x}} - i a^{*}_{{\bf x}} H({\bf x}; {\bf y}) a_{{\bf y}}$ is the bond current between ${\bf x}$ and ${\bf y}$. Let us derive the conservation law for $n_{x}$. To this end, in (\ref{eq:conteq}) we sum over the internal label $\rho$ in $\mathbf x=(x,\rho)$, and we introduce the discrete derivative $\text{d}_{x}$ as $(\text{d}_{x}f)(x) = f(x) - f(x-1)$ to obtain the continuity equation
\begin{equation}\label{eq:cons1d}
\partial_{t} \tau_{t}( n_{x} ) = -\text{d}_{x} \tau_{t}(j_{x})\;,
\end{equation}
where we introduced the lattice current $j_{x}$:
\begin{equation}\label{eq:defj1d}
\mathrm d_{x}j_{x}:= \sum_{y\in \Gamma_{L}}\sum_{\rho,\rho'\in S_{M}} j_{(x,\rho),(y,\rho')}\;.
\end{equation}
Eq. (\ref{eq:cons1d}) will have far reaching consequences in our analysis. 

It is convenient to rewrite density and current operators in Fourier space. Let $j_{\nu,x}$, $\nu = 0, 1$, be defined as $j_{0,x}= n_{x}$, $j_{1, x} = j_{x}$. In Fourier space:
\begin{equation}\label{eq:currFou}
j_{\nu,x}=\frac{1}{L}\sum_{p\in B_{L}} e^{ipx} \hat \jmath_{\nu, p}\;,\qquad \hat\jmath_{\nu,p}=\frac{1}{L}\sum_{k\in B_{L}} (a^{*}_{k-p}, \hat J_{\nu}(k,p) a_{k})\;,
\end{equation}
where we used the shorthand notation $(a^{*}, J a)=\sum_{\rho,\rho'} a^{*}_{\rho}J_{\rho\rho'} a_{\rho'}$. An explicit computation gives:
\begin{equation}\label{eq:Jnudef}
\hat J_{0}(k,p) = 1\;, \qquad \hat J_{1}(k,p)=i\frac{\hat H(k)-\hat H(k-p)}{1-e^{-ip}}\;.
\end{equation}
\begin{remark}
If $H$ has range 1, then $j_{1,x}$ takes a particularly simple form:
\begin{equation}
j_{1,x}\equiv j_{x}=\sum_{\rho,\rho'\in S_{M}} j_{(x,\rho),(x+1,\rho')}\;.
\end{equation}
\end{remark}
Finally, we define a smeared version of the density and of the current operators as:
\begin{equation}\label{eq:jmuphi}
j_{\nu}(\mu_{\theta}) = \sum_{x\in \Gamma_{L}} j_{\nu,x} \theta \mu(\theta x)\;.
\end{equation}
\paragraph{Imaginary time evolution and Euclidean correlations.}  For $t\in [0,\beta)$, we define the imaginary time evolution of an observable $\mathcal{O}$ as:
\begin{equation}\label{eq:gammatdef}
\gamma_{t}(\mathcal{O}) := e^{t ( \mathcal{H} - \mu \mathcal{N} )} \mathcal{O} e^{-t(\mathcal{H} - \mu \mathcal{N})}\;.
\end{equation}
Observe that if $[\mathcal{O}, \mathcal{N}] = 0$, the following identity holds:
\begin{equation}
\gamma_{t}(\mathcal{O}) = \tau_{-it}(\mathcal{O})\;.
\end{equation}
The Gibbs state and the imaginary time evolution satisfy the following important identity, which in our finite-dimensional setting is a simple consequence of the cyclicity of the trace:
\begin{equation}\label{eq:KMS}
\langle \gamma_{t}(\mathcal{A}) \gamma_{s}(\mathcal{B})  \rangle_{\beta,L} = \langle \gamma_{s+\beta}(\mathcal{B}) \gamma_{t}(\mathcal{A})  \rangle_{\beta,L}\;.
\end{equation}
Eq. (\ref{eq:KMS}) is called the Kubo-Martin-Schwinger (KMS) identity, and will play an important role in what follows.

Let $0\leq t_{i} < \beta$, for $i=1,2,\ldots, n$, such that $t_{i} \neq t_{j}$ for $i\neq j$. We define the time-ordering of $\gamma_{t_{1}}(a^{\sharp_{1}}_{{\bf x}_{1}}),\ldots, \gamma_{t_{n}}(a^{\sharp_{n}}_{{\bf x}_{n}})$ as:
\begin{equation}\label{eq:time}
{\bf T} \gamma_{t_{1}}(a^{\sharp_{1}}_{{\bf x}_{1}}) \cdots \gamma_{t_{n}}(a^{\sharp_{n}}_{{\bf x}_{n}}) =  (-1)^{\pi} \gamma_{t_{\pi(1)}}(a^{\sharp_{\pi(1)}}_{{\bf x}_{\pi(1)}}) \cdots \gamma_{t_{\pi(n)}}(a^{\sharp_{\pi(n)}}_{{\bf x}_{\pi(n)}})\;,
\end{equation}
where $\pi$ is the permutation needed in order to bring the times in a decreasing order, from the left, with sign $(-1)^{\pi}$. In case two or more times are equal, the definition (\ref{eq:time}) is extended via normal ordering. The above definition extends to all operators on the fermionic Fock space by linearity. In particular, for $\mathcal{O}_{1}, \ldots, \mathcal{O}_{n}$ normal ordered and even in the number of creation and annihilation operators, we have:
\begin{equation}\label{eq:Tord1}
{\bf T} \gamma_{t_{1}}(\mathcal{O}_{1}) \cdots \gamma_{t_{n}}(\mathcal{O}_{n}) = \gamma_{t_{\pi(1)}}(\mathcal{O}_{\pi(1)}) \cdots \gamma_{t_{\pi(n)}}(\mathcal{O}_{\pi(n)})\;.
\end{equation}
We define the time-ordered Euclidean correlation function as:
\begin{equation}\label{eq:Tord2}
\langle {\bf T} \gamma_{t_{1}}(\mathcal{O}_{1}) \cdots \gamma_{t_{n}}(\mathcal{O}_{n}) \rangle_{\beta,L}\;.
\end{equation}
From the definition of fermionic time-ordering, and from the KMS identity, it is not difficult to check that:
\begin{equation}\label{eq:extension}
\langle {\bf T} \gamma_{t_{1}}(\mathcal{O}_{1}) \cdots \gamma_{\beta}(\mathcal{O}_{k}) \cdots \gamma_{t_{n}}(\mathcal{O}_{n}) \rangle_{\beta,L} = (\pm 1) \langle {\bf T} \gamma_{t_{1}}(\mathcal{O}_{1}) \cdots \gamma_{0}(\mathcal{O}_{k}) \cdots \gamma_{t_{n}}(\mathcal{O}_{n}) \rangle_{\beta,L}\;;
\end{equation}
in the special case in which the operators involve an even number of creation and annihilation operators, which is the relevant one for our analysis, the overall sign is $+1$. The property (\ref{eq:extension}) is used to extend in a periodic (sign $+1$) or antiperiodic (sign $-1$) way the correlation functions to all times $t_{i} \in \mathbb{R}$. From now on, we shall always assume that this extension has been taken.

%An important example of Euclidean correlation function is the two-point function,
%
%\begin{equation}\label{eq:2ptfer}
%\langle {\bf T} \gamma_{t}(a_{{\bf x}}) a^{*}_{{\bf y}} \rangle_{\beta,L}\;.
%\end{equation}
%
%By the properties of the fermionic time-ordering, the two point function extends to a $\beta$-antiperiodic function in the imaginary time variables. Let $k_{0} = (2\pi/\beta)( n + 1/2 )$ be the frequencies relevant for $\beta$-antiperiodic functions, also called fermionic Matsubara frequencies, and let $k\in B_{L}$ be a point in the Brillouin zone. We define the Fourier transform of the two-point function as:
%
%\begin{equation}
%\begin{split}
%S_{2}(\underline{k};\lambda)_{\rho\rho'} &:= \int_{0}^{\beta} d t\, e^{-ik_{0} t} \sum_{x\in \Lambda_{L}} e^{-ikx} \langle {\bf T} \gamma_{t}(a_{(x,\rho)}) a^{*}_{(0,\rho')} \rangle_{\beta,L}\;.
%\end{split}
%\end{equation}
%
%For a non-interacting model, the two-point function can be computed explicitly in terms of the Bloch Hamiltonian. We have:
%
%\begin{equation}\label{eq:gfou}
%S_{2}(\underline k;0)\equiv \hat g(\underline{k}) = \frac{1}{ik_{0} + \hat H(k) - \mu}\;.
%\end{equation}
%
%The two-point function completely characterizes the Gibbs state of system at $\lambda=0$: in this case all Euclidean correlation function can be computed via the fermionic Wick's rule. We will employ this property to set up a Renormalisation group analysis of Euclidean correlations at $\lambda\neq 0$.
%
Next, we define the connected time-ordered Euclidean correlation functions, or time-ordered Euclidean cumulants, as:
\begin{equation}\label{eq:Tcumul}
\langle {\bf T} \gamma_{t_{1}}(\mathcal{O}_{1}); \cdots; \gamma_{t_{n}}(\mathcal{O}_{n}) \rangle_{\beta,L} = \frac{\partial^{n}}{\partial \zeta_{1} \cdots \partial \zeta_{n}} \log \Big\{ 1 + \sum_{I \subseteq \{1, 2,\ldots, n\}} \zeta(I) \langle {\bf T} \mathcal{O}(I)  \rangle_{\beta,L} \Big\}\Big|_{\zeta_{i} = 0}
\end{equation}
where $I$ is a non-empty ordered subset of $\{1, 2,\ldots, n\}$, $\zeta(I) = \prod_{i\in I} \zeta_{i}$ and $\mathcal{O}(I) = \prod_{i\in I} \gamma_{t_{i}}(\mathcal{O}_{i})$. For $n=1$, this definition reduces to $\langle {\bf T} \gamma_{t_{1}}(\mathcal{O}_{1}) \rangle \equiv \langle \gamma_{t_{1}}(\mathcal{O}_{1}) \rangle = \langle \mathcal{O}_{1} \rangle$, while for $n=2$ one gets $\langle {\bf T} \gamma_{t_{1}}(\mathcal{O}_{1}); \gamma_{t_{2}}(\mathcal{O}_{2}) \rangle = \langle {\bf T} \gamma_{t_{1}}(\mathcal{O}_{1}) \gamma_{t_{2}}(\mathcal{O}_{2}) \rangle - \langle {\bf T} \gamma_{t_{1}}(\mathcal{O}_{1}) \rangle \langle {\bf T} \gamma_{t_{2}}(\mathcal{O}_{2}) \rangle$. Recall the following relation between correlation functions and connected correlation functions:
\[
\langle {\bf T} \gamma_{t_{1}}(\mathcal{O}_{1}) \cdots \gamma_{t_{n}}(\mathcal{O}_{n}) \rangle_{\beta,L} = \sum_{P} \prod_{J\in P} \langle {\bf T} \gamma_{t_{j_{1}}}(\mathcal{O}_{j_{1}}); \cdots; \gamma_{t_{j_{|J|}}}(\mathcal{O}_{j_{|J|}}) \rangle_{\beta,L}\;,
\]
where $P$ is the set of all partitions of $\{1, 2, \ldots, n\}$ into ordered subsets, and $J$ is an element of the partition $P$, $J = \{ j_{1}, \ldots, j_{|J|} \}$. 

The next result establishes the connection between the real time perturbation theory generated by the Duhamel series and the Euclidean correlation functions, for special choices of the adiabatic parameter $\eta$.
\begin{proposition}[Wick rotation]\label{prop:wick} Let $\eta_{\beta} \in \frac{2\pi}{\beta} \mathbb{N}_{+}$. Let $\mathcal{O}$ and $\mathcal{P}$ be such that $[\mathcal{O}, \mathcal{N}] = [\mathcal{P}, \mathcal{N}] =  0$. Then, the following identity holds true:
\begin{equation}\label{eq:wick}
\begin{split}
&\int_{-\infty \leq s_{n} \leq \ldots \leq s_{1} \leq t} d \underline{s}\, e^{\eta_{\beta}(s_{1} + \ldots + s_{n})} \langle [ \cdots [[ \tau_{t}(\mathcal{O}), \tau_{s_{1}}(\mathcal{P})], \tau_{s_{2}}(\mathcal{P})] \cdots \tau_{s_{n}}(\mathcal{P}) ] \rangle_{\beta,L} \\
&= \frac{(-i)^{n} e^{n \eta_{\beta} t}}{n!} \int_{[0,\beta)^{n}} d\underline{s}\, e^{-i\eta_{\beta}(s_{1} + \ldots + s_{n})}\langle {\bf T} \gamma_{s_{1}}(\mathcal{P}); \gamma_{s_{2}}(\mathcal{P}); \cdots; \gamma_{s_{n}}(\mathcal{P});\mathcal{O}\rangle_{\beta,L}\;.
\end{split}
\end{equation}
\end{proposition}
The identity (\ref{eq:wick}) has been proved in \cite{GLMP}. It is the starting point of the analysis of the present paper: the right-hand side of the identity can be studied with field-theoretic methods. In particular, for gapless cases, like the ones considered in the present work, the right-hand side of (\ref{eq:wick}) can be analyzed via rigorous renormalization group methods. Observe that equation (\ref{eq:wick}) only holds for special values of the adiabatic parameters, corresponding to bosonic Matsubara frequencies. For $\beta$ large enough, an approximation argument based on Lieb-Robinson bounds allows to study the dynamics associated with a general $\eta > 0$ with these special adiabatic parameters, up to subleading error terms (see below).

\subsection{Main result}
For $x\in\Gamma_{L}$, let us define the full response of the density and of the current operators as:
\begin{equation}\label{eq:resp1d}
\chi^{\beta, L}_{\nu}(x;\eta,\theta) := \frac{1}{\theta} \Big( \Tr j_{\nu,x} \rho(0) - \Tr j_{\nu,x} \rho_{\beta, L}\Big)\qquad \nu = 0, 1.
\end{equation}
The quantity $\chi^{\beta,L}_{0}$ is the susceptibility, while the quantity $\chi^{\beta,L}_{1}$ is the conductance. The next theorem determines the structure of these response functions, as $\beta, L \to \infty$, and for $\theta, \eta,\lambda$ small enough. In the following, we shall denote by $\eta_{\beta}$ the best approximation of $\eta$ in $\frac{2\pi}{\beta}\mathbb{N}$, such that $\eta_{\beta} \geq \eta$.
\begin{theorem}[Main result]\label{thm:main} Let $\mathcal{H}$ be of the form (\ref{eq:H1d}), and assume that $H$ satisfies Assumption \ref{ass:A}. Let $\mathcal{H}(\eta t)$ be as in (\ref{eq:tdep})-(\ref{eq:permu}). Let $a(\eta) := \theta / \eta =O(1)$ as $\eta\to 0$.\\

\noindent{\underline{\emph{Validity of linear response.}}}
There exist constants $\gamma, \eta_{0},\bar\lambda > 0$ such that, for all $0<\eta<\eta_{0}$ and for all $|\lambda|<\bar\lambda$, linear response is valid for $\beta,L$ large enough, \emph{i.e.}
\[
\chi^{\beta,L}_{\nu}(x;\eta,\theta)= \chi^{\beta,L,\mathrm{lin}}_{\nu}(x;\eta,\theta) + O(\eta^{\gamma}) + O\bigg(\frac{1}{\eta^{3}\beta}\bigg)\;,
\]
with $\chi^{\beta,L,\mathrm{lin}}_{\nu}$ is the first order term in the Duhamel expansion (\ref{eq:duhamel}), after replacing $\eta$ with $\eta_{\beta}$.\\
%a linear functional of the perturbation $\mu$, encoding the Kubo contribution after Wick rotation.\\

\noindent{\underline{\emph{Structure of linear response.}}} The linear response can be written as follows:
\begin{equation}\label{eq:chilin}
\begin{split}
\chi^{\beta,L,\mathrm{lin}}_{\nu}(x;\eta,\theta) &= \chi^{\mathrm{lin}}_{\nu}(x;\eta,\theta) + O(\beta^{-1})+O(L^{-1})\\
\chi^{\mathrm{lin}}_{\nu}(x;\eta,\theta)&= -\sum_{\omega,\omega'=1}^{N_{f}}\int_{\R} \hat\mu_{\infty}(q)\, e^{i\theta qx}\,\mathcal K^{\nu}_{\omega\omega'}(q;\lambda) \frac{dq}{2\pi} + O(\eta^{\tilde\gamma})\;,
\end{split}
\end{equation}
for some $\tilde\gamma>0$. The matrix-valued response function $\mathcal K^{\nu}$ is given as follows. There exist an analytic function $v:(-\bar\lambda,\bar\lambda)\to\mathrm{GL}(N_{f})$, called the matrix of renormalized Fermi velocities, such that
\[
v_{\omega\omega'}(\lambda)=\delta_{\omega\omega'} v_{\omega}(\lambda)\;,\qquad\qquad v_{\omega}(\lambda) = v_{\omega} + O(\lambda)\;,
\]
and an analytic function $\Lambda:(-\bar\lambda,\bar\lambda)\to\mathrm{Mat}(N_{f},\mathbb R)$ with
\[
\Lambda^{T}(\lambda) = \Lambda(\lambda)\;, \qquad\qquad \Lambda_{\omega\omega'}(\lambda)=\begin{cases}
0 & \omega=\omega'\\
\lambda + O(\lambda^{2}) & \omega\neq\omega'\;,
\end{cases}
\]
such that the matrix $\mathcal K^{\nu}$ is given by
\[
\mathcal K^{\nu}(q;\lambda) = \widetilde{\mathcal K}(q;\lambda) \frac{v^{\nu}(\lambda)}{2\pi|v(\lambda)|}\frac{v(\lambda)q}{-i/a+v(\lambda)q} \mathfrak K^{\nu}(\lambda)\;,
\]
with $v^{\nu}:=\delta_{\nu,0}\mathbbm 1 + \delta_{\nu,1} v$, and $\widetilde{\mathcal K}$ and $\mathfrak K^{\nu}$ defined as
\begin{equation*}
\widetilde{\mathcal K}(q;\lambda) = \left(1-\frac{1}{4\pi|v|}\Lambda\right)\left(1+\frac{1}{4\pi|v|}\frac{i/a+vq}{-i/a+vq}\Lambda\right)^{-1}\;,
\end{equation*}
\begin{equation*}
\begin{split}
\quad\quad
\mathfrak K^{0} \equiv \mathbbm 1_{N_{f}}\;,\qquad\qquad \mathfrak K^{1}:= \mathrm{sgn\,}v\,\frac{1+(4\pi|v|)^{-1}\Lambda}{1-(4\pi|v|)^{-1}\Lambda}\,\mathrm{sgn\,}v\;.
\end{split}
\end{equation*}
\end{theorem}
\begin{remark}
\begin{itemize}
\item[(i)] The above theorem justifies the validity of linear response for interacting one-dimensional systems in the absence of a spectral gap. At zero temperature and as $\eta \to 0^{+}$ the only contribution surviving to the response of the system is the one associated to Kubo formula, which can be explicitly evaluated. 
\item[(ii)] A technical limitation of the proof is that the range of allowed $\lambda$ is not uniform in $a = |\theta| / \eta$. In particular, our current estimates imply that $\bar \lambda$ shrinks to zero as $a\to \infty$, precisely
\begin{equation*}
\bar\lambda(a) \sim (\|\hat\mu_{\infty}\|_{L^{1}}a)^{-\kappa (\|\hat\mu_{\infty}\|_{L^{1}}a)^{4}} \qquad 
\end{equation*}
for any $\kappa>2$.
\item[(iii)] Suppose that $a \to 0$. Then, the above result shows that the response of the system at zero temperature is trivial as $\eta \to 0^{+}$, since: 
\begin{equation}
\chi^{\mathrm{lin}}_{\nu}(x,\eta,\theta) = O(a)\;.
\end{equation}
Although we cannot perform the limit $a\to+\infty$ for the full current response $\chi_{1}$, as stated in point (ii), considering it for just the Kubo term $\chi_{1}^{\mathrm{lin}}$, one obtains a non-trivial cancellation
\begin{equation}
\lim_{a\to\infty}\lim_{\eta\to0}\chi^{\mathrm{lin}}_{1}(x;\eta,a\eta) \equiv \sum_{\omega} \frac{\mathrm{sgn} v_{\omega}}{2\pi}=0\;,
\end{equation}
that is each effective mode contributes to the conductance via its \emph{chirality}, \emph{i.e.} the sign of its Fermi velocity. Such a cancellation is not obvious in the case of finite $a$, since each mode's contribution depends on the \emph{value} of velocity; nevertheless, such cancellation can be established for the case of two Fermi points with opposite velocities also for any finite $a$, as shown in Remark \ref{rmk:2chir} below.
\item[(iv)] Considering non-interacting systems, \emph{i.e.} for $\lambda=0$, we notice that both $\widetilde{\mathcal K}$ and $\mathfrak K^{\nu}$ become the identity matrix, and the Fermi velocities reduce to the free ones, so that
\begin{equation*}
\mathcal K^{\nu}(q;0) = \frac{v^{\nu}(0)}{2\pi|v(0)|} \frac{v(0)q}{-i/a+v(0)q}\;,
\end{equation*}
which is in agreement with the result obtained in \cite{PS}. Thus, the many-body interaction modifies the response in two concurrent ways: by dressing the velocities of the infrared modes labelled by the Fermi points, and by mixing the contribution of each mode via $\widetilde{\mathcal K}$ and $\mathfrak K^{\nu}$.
\item[(v)] For non-interacting systems, a similar result has been obtained in \cite{PS}. The proof of \cite{PS} was based on a careful analysis of the real-time Duhamel expansion after Wick rotation. The Wick rotation of the real-time Duhamel expansion will be the starting point of this work as well; this procedure crucially relies on the exponential switch on of the perturbation (a more general class of functions can also be considered, \cite{GLMP}). The main difference of the present work with \cite{PS} is that now we cannot perform a direct analysis of the Wick-rotated correlators, due to the fact that the theory is interacting. We will instead resort to rigorous renormalization group methods, that will allow to prove essentially optimal estimates on the imaginary-time correlators, and to identify key cancellations at all orders beyond the linear response.
\item[(vi)] At the price of a straightforward approximation argument, carried out in the non-interacting setting in \cite{PS}, the result could be extended to the case in which $\hat \mu_{\infty}(p)$ decays fast enough, rather that assuming that $\hat \mu_{\infty}(p)$ is compactly supported in $p$. We prefer to avoid this extra technical step, which we believe would not add anything substantial to the main result of this work.
\end{itemize}
\end{remark}

\begin{remark} [The case of two chiralities]\label{rmk:2chir}
In the case of two opposite Fermi points with counterpropagating velocities the expression in (\ref{eq:chilin}) drastically simplifies. This case is relevant for instance when $H$ is the nearest-neighbour lattice Laplacian. We will show how the general expression (\ref{eq:chilin}) reproduces known rigorous results, see {\it e.g.} \cite{BMdensity, BFM3}. For simplicity, we shall evaluate it for $a=1$, that is $\theta=\eta$; the expression for $a\neq 1$ will be simply given by substituting $q\to aq$ in the response for $a=1$.

The dispersion relation is $\varepsilon(k)=2(1-\cos k)$, and for any choice of $\mu\in(0,4)$ we have $N_{f}=2$ distinct free Fermi points $k_{F}^{+}=-k_{F}^{-}$ with opposite free Fermi velocities $v_{\pm}=\pm v_{0}$, with $v_{0}=2\sin k_{F}^{+}>0$. Since spatial $\mathbb Z_{2}$-parity is a symmetry of the model, in this case the renormalized parameters needed are just two, namely $v_{\star}(\lambda)=v_{0}+ O(\lambda)$, $\lambda_{\star}(\lambda)=\lambda+O(\lambda^{2})$, so that the matrices $v$ and $\Lambda$ introduced in Theorem \ref{thm:main} become
\[
v = v_{\star}\sigma_{3}\qquad\qquad \Lambda = \lambda_{\star}\sigma_{1}
\]
with $\sigma_{1}$, $\sigma_{3}$ the first and the third Pauli matrices.
For simplicity, let us set
\[
b:=\frac{\lambda_{\star}}{4\pi v_{\star}}\quad\qquad \mathfrak b_{\pm}(q):= \frac{i\pm v_{\star}q}{-i\pm v_{\star}q}\quad\qquad B(q):=b\begin{pmatrix}
0 & \mathfrak b_{+}(q)\\
\mathfrak b_{-}(q) & 0
\end{pmatrix}\;.
\]
Then, the response matrices are given by
\begin{equation} \label{eq:RespMat2Flav}
\mathcal K^{0}(q) = \frac{1}{2\pi v_{\star}}(1-b\sigma_{1})\frac{1}{1+B(q)} \frac{v_{\star}q\sigma_{3}}{-i+v_{\star}q\sigma_{3}}, \qquad \mathcal K^{1}(q) = v_{\star}\mathcal K^{0}(q)\frac{1+b\sigma_{1}}{1-b\sigma_{1}}\sigma_{3}\;.
\end{equation}
Let us compute $\mathcal K^{0}$ first. By using the properties of $\sigma_{3}$, and employing the fact that $\det[1+B(q)]= 1-b^{2}\mathfrak b_{+}(q)\mathfrak b_{-}(q)\equiv1-b^{2}$, we have
\begin{equation*}
\frac{1}{1+B(q)} = \frac{1-B(q)}{1-b^{2}}\qquad \qquad \frac{v_{\star}q\sigma_{3}}{-i+v_{\star}q\sigma_{3}}=\frac{v_{\star}q}{1+v^{2}_{\star}q^{2}}(i\sigma_{3}+v_{\star}q)\;,
\end{equation*}
so that
\begin{eqnarray}\label{eq:matprod}
%\begin{split}
\frac{1}{1+B(q)} \frac{v_{\star}q\sigma_{3}}{-i+v_{\star}q\sigma_{3}}&=&\frac{1}{1-b^{2}}\frac{v_{\star}q}{1+v^{2}_{\star}q^{2}}
\begin{pmatrix} 1 & -b\mathfrak b_{+}(q)\\ -b\mathfrak b_{-}(q) & 1 \end{pmatrix}
\begin{pmatrix} i+v_{\star}q & 0\\ 0 & -i+v_{\star}q \end{pmatrix}\nonumber\\
&=&\frac{1}{1-b^{2}}\frac{v_{\star}q}{1+v^{2}_{\star}q^{2}} \begin{pmatrix}
i+v_{\star}q & -b(i+v_{\star}q) \\
b(i-v_{\star}q) & -i+v_{\star}q
\end{pmatrix}\;.
%\end{split}
\end{eqnarray}
Therefore, inserting (\ref{eq:matprod}) into (\ref{eq:RespMat2Flav}), the susceptibility matrix takes the form
\begin{eqnarray*}
\mathcal K^{0}(q)&=& \frac{1}{2\pi v_{\star}}\frac{v_{\star}q}{1+v^{2}_{\star}q^{2}}\frac{1}{1-b^{2}}\begin{pmatrix} 1 & -b\\ -b & 1 \end{pmatrix}
\begin{pmatrix}
i+v_{\star}q & -b(i+v_{\star}q) \\
b(i-v_{\star}q) & -i+v_{\star}q
\end{pmatrix}\nonumber\\
&=& \frac{1}{2\pi v_{\star}}\frac{v_{\star}q}{1+v^{2}_{\star}q^{2}}\frac{1}{1-b^{2}} \begin{pmatrix}
(1-b^{2})i +(1+b^{2})v_{\star}q & -2bv_{\star}q\\
-2bv_{\star}q & (b^{2}-1)i + (1+b^{2})v_{\star}q
\end{pmatrix}\;;
\end{eqnarray*}
finally, by summing over all the entries of $\mathcal K^{0}$ we obtain the following simple form for the density response function:
\begin{equation}
\sum_{\omega,\omega'=\pm}\mathcal K^{0}_{\omega\omega'}(q) = \frac{1}{\pi v_{\star}}\frac{4\pi v_{\star}-\lambda_{\star}}{4\pi v_{\star}+\lambda_{\star}}\frac{v^{2}_{\star}q^{2}}{1+v^{2}_{\star}q^{2}}\;.
\end{equation}
This expression for the susceptibility has been already derived rigorously, from a linear response ansatz, see \cite{BFM3}, Theorem 1.1, Eq. (1.15). To compute the current response, we proceed as follows. We write
\begin{equation}
\sum_{\omega,\omega'=\pm}\mathcal K^{1}_{\omega\omega'}(q)\equiv (\vec 1, \mathcal K^{1}(q)\vec 1) = v_{\star}\left(\vec 1, \mathcal K^{0}(q)\frac{1+b\sigma_{1}}{1-b\sigma_{1}}\sigma_{3}\vec 1\right)\end{equation}
with $\vec 1$ the vector with all components equal to 1; now, we have the following identity:
\begin{equation*}
\begin{split}
\frac{1+b\sigma_{1}}{1-b\sigma_{1}}\sigma_{3}\vec 1 &= \sigma_{3}\frac{1-b\sigma_{1}}{1+b\sigma_{1}}\vec 1=\sigma_{3}\frac{(1-b\sigma_{1})^{2}}{1-b^{2}}\vec 1=\frac{(1-b)^{2}}{1-b^{2}}\sigma_{3}\vec 1\;,
\end{split}
\end{equation*}
where in the first step we used $\{\sigma_{1},\sigma_{3}\}=0$, and in the last one we exploited $\sigma_{1}\vec 1=\vec 1$.
This identity allows to write
\begin{equation}
(\vec 1, \mathcal K^{1}(q)\vec 1) = v_{\star}\frac{1-b}{1+b} (\vec 1, \mathcal K^{0}(q)\sigma_{3}\vec 1)
\end{equation}
which, up to the prefactor, corresponds to summing the components of the first column of $\mathcal K^{0}$ and subtracting the components of the second one. Therefore:
\begin{equation} \label{eq:K1}
(\vec 1, \mathcal K^{1}(q)\vec 1) = \frac{1}{\pi}\frac{4\pi v_{\star}-\lambda_{\star}}{4\pi v_{\star}+\lambda_{\star}} \frac{iv_{\star}q}{1+v^{2}_{\star}q^{2}}\;.
\end{equation}
Note that, owing to the oddness of the correlator (\ref{eq:K1}), $\lim_{\theta,\eta\to0}\chi^{\mathrm{lin}}_{1}(x;\eta,\theta)\equiv0$ whenever $\mu_{\infty}$ is even, as expected. We observe that this cancellation is not obvious in the regime of $a$ constant, for the case of a general number of Fermi points with arbitrary Fermi velocities.
\end{remark}
\section{Duhamel expansion and Wick rotation}\label{sec:duha}
\subsection{Auxiliary dynamics} 
Recall that $\eta_{\beta}$ is the best approximation of $\eta$ in $\frac{2\pi}{\beta} \mathbb{N}$, such that $\eta_{\beta} \geq \eta$. We define the auxiliary Hamiltonian by replacing the adiabatic parameter $\eta$ with $\eta_{\beta}$:
\begin{equation}
\mathcal{H}_{\beta,\eta}(t) = \mathcal{H} +  e^{\eta_{\beta} t} \theta\sum_{{\bf x} \in \Lambda_{L}} \mu(\theta x) a^{*}_{{\bf x}} a_{{\bf x}}\;.
\end{equation}
Let us denote by $\widetilde{\mathcal{U}}(t;s)$ the two-parameter unitary group generated by $\mathcal{H}_{\beta,\eta}(t)$,
\begin{equation}
i\partial_{t} \widetilde{\mathcal{U}}(t;s) = \mathcal{H}_{\beta,\eta}(t) \widetilde{\mathcal{U}}(t;s)\;,\qquad \widetilde{\mathcal{U}}(s;s) = \mathbbm{1}\;.
\end{equation}
The next result allows to control the error introduced by replacing the dynamics generated by $\mathcal{H}(\eta t)$ with the dynamics generated by $\mathcal{H}_{\beta,\eta}(t)$. 
\begin{proposition}[Approximation by the auxiliary dynamics]\label{prp:LRint} Under the same assumptions of Theorem \ref{thm:main}, it follows that:
\begin{equation}\label{eq:LRprop}
\Big\| \widetilde{\mathcal{U}}(t;-\infty)^{*} \mathcal{O}_{X} \widetilde{\mathcal{U}}(t;-\infty) - \mathcal{U}(t;-\infty)^{*} \mathcal{O}_{X} \mathcal{U}(t;-\infty)   \Big\| \leq C\frac{\theta}{\eta^{3} \beta}
\end{equation}
for some $C>0$ independent of $\theta,\eta,\beta$.
\end{proposition}
\begin{remark} Let us briefly comment on the reason for introducing the auxiliary dynamics. Let $\mathcal{A}$ be such that $[\mathcal{A}, \mathcal{N}] = 0$, and let $\tilde \tau_{t}(\mathcal{A}) = e^{\eta t} \tau_{t}(\mathcal{A})$, which we can view as a (non-unitary) regularization of the original dynamics, damped in the past. Then, if $\eta \in \frac{2\pi}{\beta} \mathbb{N}$, we observe that $\tilde \tau_{t}(\cdot)$ satisfies the KMS identity:
\begin{equation}\label{eq:KMSreg}
\langle \tilde \tau_{t}(\mathcal{A}) \mathcal{B}  \rangle_{\beta,L} = \langle \mathcal{B} \tilde \tau_{t + i\beta}(\mathcal{A}) \rangle_{\beta,L}\;,
\end{equation}
where we used the trivial but crucial identity $e^{\eta t} = e^{\eta (t + i\beta)}$. The KMS identity for the regularized dynamics (\ref{eq:KMSreg}) plays a key role in the mapping of the real time Duhamel series into an imaginary time expansion, which can be efficiently studied \cite{GLMP}; see Proposition \ref{prop:wick} above. Due to the fact that the difference between any $\eta > 0$ and its best approximation $\eta_{\beta} \in \frac{2\pi}{\beta} \mathbb{N}$ vanishes as $\beta \to \infty$, this approach is suitable to study zero or low temperature systems (lower than some $\eta$-dependent value).
\end{remark}
\begin{proof}[Proof of Proposition \ref{prp:LRint}]
The proof is a standard application of Lieb-Robinson bounds, we refer to Proposition 4.1 of \cite{GLMP} for the details.
\end{proof}
\subsection{Duhamel expansion}
Setting $\tilde \rho(t) = \widetilde{\mathcal{U}}(t;-\infty) \rho_{\beta, L} \widetilde{\mathcal{U}}(t;-\infty)^{*}$, by Proposition \ref{prp:LRint} we may write
\begin{equation}\label{eq:approxdynint}
\Tr j_{\nu,x} \rho(t) = \Tr j_{\nu,x} \widetilde\rho(t) + \mathcal{E}^{\beta, L}_{\nu} (x,t;\eta, \theta)\;,
\end{equation}
with the error term $\mathcal E^{\beta,L}_{\nu}$ satisfying the bound
\begin{equation}\label{eq:estLRint}
| \mathcal{E}^{\beta, L}_{\nu} (x,t;\eta, \theta) | \leq C\frac{\theta}{\eta^{3} \beta}\;.
\end{equation}
Consider the main term in the right-hand side of (\ref{eq:approxdynint}). By Proposition \ref{prop:wick}:
\begin{equation}\label{eq:duha1int}
\begin{split}
& \Tr j_{\nu,x} \widetilde\rho(t) -  \Tr j_{\nu,x} \rho_{\beta, L} \\&\quad = \sum_{n\geq 1}\frac{(-1)^{n} e^{n \eta_{\beta} t}}{n!} \int_{[0,\beta)^{n}} d\underline{s}\, e^{-i\eta_{\beta}(s_{1} + \ldots +s_{n})}\langle {\bf T} \gamma_{s_{1}}(\mathcal P); \cdots; \gamma_{s_{n}}(\mathcal P); j_{\nu,x}\rangle_{\beta,L}\;.
 \end{split}
\end{equation}
It is convenient to rewrite the $n$-th order contribution to the expansion in Fourier space. Recalling the definition (\ref{eq:jmuphi}) and the form of the perturbation (\ref{eq:formPert}),
\begin{equation}\label{eq:fouint}
\begin{split}
&\int_{[0,\beta)^{n}} d\underline{s}\, e^{-i\eta_{\beta}(s_{1} + \ldots +s_{n})}\langle {\bf T} \gamma_{s_{1}}(\mathcal P); \cdots; \gamma_{s_{n}}(\mathcal P); j_{\nu,x}\rangle_{\beta,L}\\
&\quad = \int_{[0,\beta)^{n}} d\underline{s}\, e^{-i\eta_{\beta}(s_{1} + \ldots +s_{n})}\langle {\bf T} \gamma_{s_{1}}(j_{0}(\mu_{\theta}));  \cdots; \gamma_{s_{n}}(j_{0}(\mu_{\theta})); j_{\nu,x}\rangle_{\beta,L}\;,
\end{split}
\end{equation}
using that
\begin{equation}\label{eq:currfourierint}
j_{0}(\mu_{\theta}) = \frac{\theta}{L} \sum_{p \in B_{L}} \hat \mu_{\theta}(-p) \hat\jmath_{0,p}\;,\qquad \hat\jmath_{\nu,p} = \sum_{x \in \Gamma_{L}} e^{-ipx} j_{\nu,x}\;,
\end{equation}
with $\hat\mu_{\theta}(p):=\hat\mu(p/\theta)/\theta$, and writing, for $\underline{p} = (p_{0}, p_{1})$ with $p_{1} = p$ and $p_{0} = \eta_{\beta}$,
\begin{equation}\label{eq:densitytimefourierint}
\hat n_{\underline{p}} := \int_{0}^{\beta} ds\, e^{-i\eta_{\beta} s} \gamma_{s}(\hat n_{p})\;,
\end{equation}
we can rewrite the right-hand side of (\ref{eq:fouint}) as:
\begin{equation}\label{eq:fou2int}
\begin{split}
&\int_{[0,\beta)^{n}} d\underline{s}\, e^{-i\eta_{\beta}(s_{1} + \ldots +s_{n})}\langle {\bf T} \gamma_{s_{1}}(j_{0}(\mu_{\theta}));  \cdots; \gamma_{s_{n}}(j_{0}(\mu_{\theta})); j_{\nu,x}\rangle_{\beta,L} \\
&\qquad = \frac{\theta^{n}}{L^{n}} \sum_{\{p_{i}\} \in B_{L}^{n}} \Big[\prod_{i=1}^{n} \hat\mu_{\theta}(-p_{i})\Big] e^{ip_{n+1}x} \frac{1}{\beta L}\langle {\bf T}\, \hat n_{\underline{p}_{1}}\;; \cdots \;; \hat n_{\underline{p}_{n}}\;; \hat\jmath_{\nu,\underline{p}_{n+1}} \rangle_{\beta,L}\;,
\end{split}
\end{equation}
where $\underline{p}_{n+1} = -\underline{p}_{1} - \ldots - \underline{p}_{n}$. In order to derive (\ref{eq:fou2int}), we used the space-time translation invariance of the Gibbs state.

Thus, we are led to the following representation for the Duhamel expansion of the response functions: recalling (\ref{eq:duha1int})-(\ref{eq:fou2int}), we have
\begin{equation}\label{eq:fullrespint}
\begin{split}
\chi^{\beta,L}_{\nu}(x;\eta,\theta) =
-\sum_{m= 2}^\infty& \frac{(-\theta)^{m-2}}{(m-1)!} \frac{1}{L^{m-1}} \sum_{\{p_{i}\} \in B_{L}^{m-1}} \Big[\prod_{j=1}^{m-1} \hat\mu_{\theta}(-p_{j})\Big] e^{ip_{m}x}\\
&\quad \cdot \frac{1}{\beta L}\langle\mathbf T\hat n_{\ul p_{1}};\cdots;\hat n_{\ul p_{m-1}};\hat\jmath_{\nu,\ul p_{m}}\rangle_{\beta,L} + E^{\beta,L}_{\nu}(x;\eta,\theta)\;,
\end{split}
\end{equation}
where: $\underline{p}_{i} := (\eta_{\beta}, p_{i})$ and $\ul p_{m} := -\sum_{i=1}^{m-1} \ul p_{i}$; the error term $E^{\beta,L}_{\nu}$ satisfies, by (\ref{eq:estLRint}):
\begin{equation}\label{eq:errnuint}
| E^{\beta,L}_{\nu}(x;\eta,\theta)|=\frac{|\mathcal E^{\beta,L}_{\nu}(x,0;\eta,\theta)|}{\theta}\leq  \frac{C}{\eta^{3} \beta}\;.
\end{equation}
Starting from identity (\ref{eq:fullrespint}), the remaining part of the proof (and of the paper) will be devoted to the proof of suitable bounds for the Euclidean density-current correlations appearing in the higher order ($m\geq 3$) terms, and to the explicit evaluation of the linear ($m=2$) term.
\begin{remark}
Observe that we have written (\ref{eq:fullrespint}) in terms of the number of points of the correlation functions ($m$), instead of the Duhamel order ($n=m-1$).%: this will be a convenient notation for the next Sections, as these will be devoted to the study of the Euclidean correlations themselves.
\end{remark}

\section{RG analysis of the lattice model}\label{sec:RGLatticeModel}
\subsection{Grassmann integral representation of the model}
In this section, we will introduce a Grassmann  integral formulation that will allow us to rewrite the density Euclidean $m$-point functions as derivatives of a suitable generating functional, that can be analyzed via multiscale methods.

Let $\chi : \mathbb{R}\to [0,1]$ be a smooth, even cut-off function such that:
\begin{equation}\label{eq:cutoff}
\chi(t)=
\begin{cases}
1 & |t|\leq1\\
0 & |t|\geq \frac{3}{2}
\end{cases}
\end{equation}
and, for $N_{0}\in \mathbb N$ set $\chi_{N_{0}}(\cdot)=\chi(2^{-N_{0}}\cdot)$. We define the set of space-time momenta $\mathbb D_{\beta,L,N_{0}}$ as:
\begin{equation}
\mathbb D_{\beta,L,N_{0}} = \left\{\ul k = (k_{0},k)\in \frac{2\pi}{\beta}\Big(\mathbb Z+\frac{1}{2}\Big)\times B_{L}\ \Big|\ \chi_{N_{0}}(k_{0})>0\right\}\;.
\end{equation}
Observe that the set $\mathbb D_{\beta,L,N_{0}}$ is finite. Then, consider the Grassmann $\mathbb C$-algebra generated by the set of anti-commuting symbols $\{\hat\Psi^{\pm}_{\ul k,\rho}\ |\ \ul k\in \mathbb D_{\beta,L,N_{0}}, \rho\in S_{M}\}$. On such Grassmann algebra we define a Gaussian Grassmann measure $P_{N_{0}}[d\Psi]$ as follows. First, we set the Berezin integral $d\Psi$ to be the linear functional which is non-zero only on monomials of maximal degree, and such that its evaluation on such top degree monomials is:
\begin{equation}
\int d\Psi \prod_{\ul k\in\mathbb D_{\beta,L,N_{0}}}\prod_{\rho\in S_{M}} \hat\Psi^{-}_{\ul k,\rho} \hat\Psi^{+}_{\ul k,\rho} = 1.
\end{equation}
Let us define the free propagator as:
\begin{equation}
\hat g_{\rho\rho'}^{(\leq N_{0})}(\ul k) = \left(\frac{\chi_{N_{0}}(k_{0})}{ik_{0}+ \hat H(k)-\mu}\right)_{\rho\rho'}.
\end{equation}
We define the associated Grassmann Gaussian integration $P_{N_{0}}[d\Psi]$ as:
\begin{equation}
\int P_{N_{0}}[d\Psi] (\cdot) = \frac{1}{\mathcal Z^{0}_{\beta,L,N_{0}}}\int d\Psi \exp\Big\{-\frac{1}{\beta L}\sum_{\ul k\in\mathbb D_{\beta,L,N_{0}}} (\hat\Psi^{+}_{\ul k}, \hat g^{(\leq N_{0})}(\ul k)^{-1}\hat\Psi^{-}_{\ul k})\Big\}\ \Big(\cdot\Big)
\end{equation}
with the short-hand $(\phi,A\psi)=\sum_{\rho,\rho'} \phi_{\rho}A_{\rho\rho'}\psi_{\rho'}$ and normalization
\begin{equation}
\mathcal Z^{0}_{\beta,L,N_{0}} = \prod_{\ul k\in\mathbb D_{\beta,L,N_{0}}} [\beta L] \det[\hat g^{(\leq N_{0})}(\ul k)].
\end{equation}
Given this definition for $P_{N_{0}}[d\Psi]$, it is easy to check that, for a given monomial $\prod_{i=1}^{n} \hat\Psi_{\ul k_{i}, \rho_{i}}^{\epsilon_{i}}$ its action yields zero if $|\{i : \epsilon_{i}=+\}|\neq |\{i : \epsilon_{i}=-\}|$; otherwise,
\begin{equation}\label{eq:wickGra}
\int P_{N_{0}}[d\Psi] \prod_{i=1}^{m}\hat \Psi^{-}_{\ul k_{i},\rho_{i}} \hat\Psi^{+}_{\ul p_{i}, \rho'_{i}} = \det[K(\ul k_{i},\rho_{i};\ul p_{j}, \rho'_{j})]_{i,j=1}^{m}
\end{equation}
with $K(\ul k, \rho;\ul p, \rho')=\beta L\delta_{\ul k,\ul p}\,\hat g_{\rho\rho'}^{(\leq N_{0})}(\ul k)$. Eq. (\ref{eq:wickGra}) is the Wick's rule, for Grassmann Gaussian fields.

Next, let us introduce the real-space fields as:
\begin{equation}
\Psi^{\pm}_{\ul x,\rho} = \frac{1}{\beta L}\sum_{\ul k\in \mathbb D_{\beta,L,N_{0}}} e^{\mp i\ul k\cdot\underline x} \hat \Psi^{\pm}_{\ul k,\rho}\;.
\end{equation}
One can check that:
\begin{equation}
\int P_{N_{0}}[d\Psi] \Psi^{-}_{\ul{x},\rho}\Psi^{+}_{\ul{y},\rho'} = g_{\rho\rho'}^{(\leq N_{0})}(\ul {x}; \ul{y}),
\end{equation}
where the real-space propagator is:
\begin{equation}
g_{\rho\rho'}^{(\leq N_{0})}(\ul {x}; \ul{y}) = \frac{1}{\beta L}\sum_{\ul k\in \mathbb D_{\beta,L,N_{0}}} e^{ i\ul k\cdot(\underline x-\underline y)} \hat g_{\rho\rho'}^{(\leq N_{0})}(\ul k).
\end{equation}
It is well-known that, as $\underline{x} \neq \underline{y}$, the $N_{0}\to \infty$ limit of the real-space two-point function  converges to the Euclidean two-point function of the non-interacting model,
\begin{equation}
\lim_{N_{0}\to \infty} g_{\rho\rho'}^{(\leq N_{0})}(\ul {x}; \ul{y}) = \langle {\bf T} \gamma_{x_{0}} (a_{x,\rho}) \gamma_{y_{0}} (a^{*}_{y,\rho'}) \rangle_{\beta,L}^{0} \equiv  g_{\rho\rho'}(\ul {x}; \ul{y})\;.
\end{equation}
Let us now define the interacting Grassmann field theory. The Grassmann interaction is:
\begin{equation} \label{eq:quarticInt}
V(\Psi):=\lambda \int_{\mathbb T^{2}_{\beta}} dx_{0}dy_{0} \sum_{\mathbf x,\mathbf y\in\Lambda_{L}} n_{\ul{\mathbf x}} \delta(x_{0}-y_{0})w(x,y) n_{\ul{\mathbf y}}
\end{equation}
with $\ul{\mathbf x}=(\ul x,\rho)$ and with $n_{\ul{\mathbf x}}=\Psi^{+}_{\ul{\mathbf x}}\Psi^{-}_{\ul{\mathbf x}}$. It should be noted that the factors $-1/2$ are absent in this Grassmann representation; this is related to the fact that:
\begin{equation}
\lim_{N_{0}\to \infty} g_{\rho\rho'}^{(\leq N_{0})}((x_{0},x); (x_{0},y)) = g_{\rho\rho'}((x_{0},x); (x_{0},y)) - \frac{\delta_{{\mathbf x},{\mathbf y}}}{2}\;.
\end{equation}
We can rewrite expression (\ref{eq:quarticInt}) in momentum space as:
\begin{equation}
V(\Psi):=\frac{1}{\beta L} \sum_{\ul p\in \mathbb M_{\beta,L,N_{0}}} \lambda\hat w(p)\,\hat n_{-\ul p} \hat n_{\ul p}\;,
\end{equation}
where $\mathbb M_{\beta,L,N_{0}}:=\{\ul p=(p_{0},p)\in \frac{2\pi}{\beta}\mathbb Z\times B_{L}\ |\ \chi_{N_{0}}(p_{0})>0\}$ and:
\begin{equation}
\hat n_{\ul p} = \int_{0}^{\beta} dx_{0} \sum_{\mathbf x\in\Lambda_{L}} e^{-i \ul p\cdot \ul x}\, n_{\ul{\mathbf x}} = \frac{1}{\beta L} \sum_{\ul k\in \mathbb D_{\beta,L,N_{0}}} (\hat \Psi^{+}_{\ul k-\ul p}, \hat J_{0}(k,p) \hat \Psi^{-}_{\ul k})\;.
\end{equation}
At last, we introduce the partition function of this Grassmann model as:
\begin{equation} \label{eq:partitionfunction}
\mathcal Z_{\beta,L,N_{0}}:= \int P_{N_{0}}[d\Psi]\, e^{-V(\Psi)}\;.
\end{equation}
We are now ready to define the generating functional of the Grassmann field theory. Let us start by introducing the source terms. Using the short-hand notation:
\begin{equation}
\int_{\beta,L} d\ul x := \int_{\mathbb T_{\beta}} dx_{0}\sum_{x\in\Gamma_{L}},
\end{equation}
and considering two Grassmann fields $\phi^{\pm}_{\ul x,\rho}$ and a real valued vector field $A^{\nu}_{\ul x}$, we set
\begin{equation} \label{eq:sources}
B(\Psi; \phi) := \int_{\beta,L} d\ul x\, (\phi^{+}_{\ul x},\Psi^{-}_{\ul x}) + (\Psi^{+}_{\ul x}, \phi^{-}_{\ul x})\;,\qquad 
C(\Psi; A) := \sum_{\nu = 0,1}\int_{\beta,L} d\ul x\, A^{\nu}_{\ul x}\, j_{\nu,\ul x}\;,
\end{equation}
with $j_{\nu,\ul x}$ the Grassmann counterpart of the lattice 2-current defined in (\ref{eq:currFou}):
\begin{equation}\label{eq:currGrassFou}
j_{\nu,x}=\frac{1}{\beta L}\sum_{\ul p\in \mathbb M_{\beta,L,N_{0}}} e^{i\ul p\cdot\ul x} \hat \jmath_{\nu, \ul p}\qquad \hat\jmath_{\nu,\ul p}=\frac{1}{\beta L}\sum_{\ul k\in \mathbb D_{\beta,L,N_{0}}} (\hat \Psi^{+}_{\ul k-\ul p}, \hat J_{\nu}(k,p) \hat \Psi^{-}_{\ul k})\;.
\end{equation}
In order to make sure that the field $\phi^{\pm}$ also generates a finite dimensional algebra, we will suppose that it is supported for finitely many Fourier modes.

The generating functional of the correlation functions $\mathcal{W}_{\beta,L,N_{0}}(\phi,A)$ is then given by:
\begin{equation} \label{eq:genfcn}
\mathcal{W}_{\beta,L,N_{0}}(\phi,A) = \log \int P_{ N_{0}} [d\Psi] e^{-V(\Psi) + B(\Psi; \phi) + C(\Psi; A)}\;.
\end{equation}
A corollary of the construction that we are going to describe is that, for $\|A\|_{\infty}, |\lambda|$ small enough, the argument of the logarithm is non-zero, and hence the generating functional is well defined at finite $\beta, L, N_{0}$.

It is well-known that correlation functions of the Gibbs state of the model in Fock space at non-coinciding space-time points can be represented as functional derivatives of the generating functional (\ref{eq:genfcn}), in the limit $N_{0}\to\infty$; see \emph{e.g.}  \cite{GMP} for a discussion on this point. In particular, we are interested in the following type of momentum-space correlation functions:
\begin{equation}\label{eq:equivss}
\langle\mathbf T\hat n_{\ul p_{1}};\,\cdots; \hat n_{\ul p_{m-1}}; \hat \jmath_{\nu,\ul p_{m}} \rangle_{\beta,L} = \lim_{N_{0}\to\infty} (\beta L)^{m}\left.\frac{\partial^{m}\mathcal{W}_{\beta,L,N_{0}}(\phi,A)}{\partial \hat A^{0}_{-\ul p_{1}}\cdots\,\partial \hat A^{0}_{-\ul p_{m-1}}\partial \hat A^{\nu}_{-\ul p_{m}}}\right|_{A=\phi=0}\;.
\end{equation}
The identity (\ref{eq:equivss}) will allow us to use the properties of Grassmann Gaussian fields in order to investigate the Euclidean correlation functions of our model. Equation (\ref{eq:equivss}) holds in the analyticity domain of the Grassmann theory, which is contained in the analyticity domain of the Fock space model.
\subsection{Integration of the ultraviolet degrees of freedom}
As in all applications of renormalization group to condensed matter problems, the starting point is the integration of the ultraviolet degrees of freedom. After this step, which will be briefly described in this section (referring to other works for the technical details), the problem will be reduced to the analysis of a suitable, effective infrared model. 

%The Grassmann field theory can be studied using multiscale analysis and renormalization group. To do so, a preliminary step is to introduce a notion of `integration of bulk degrees of freedom', that allows to recast the problem into the study of an emergent $1+1$ dimensional QFT, describing the low energy excitation (equivalently, the large distance behaviour) of our lattice model. The separation between \emph{ultraviolet} (UV) and \emph{infrared} (IR) degrees of freedom has been discussed in detail in \cite{AMP}. Here we shall outline the main ideas, referring to \cite{AMP}, Section 5.2, for further details.

%The main advantage of the Grassmann integral formulation with respect to the original Fock space formulation is the addition principle of Grassmann variables. Suppose that $\psi^{\pm}_{\mathbf x}$ is a Grassmann field with propagator $g = g^{(1)} + g^{(2)}$. Let $\psi^{(1)\pm}$, $\psi^{(2)\pm}$ be two independent Grassmann fields, with propagators $g^{(1)}$, $g^{(2)}$. Let $\mathbb{E}$ be the Grassmann Gaussian average with respect to the field $\psi$, and let $\mathbb{E}_{i}$ be the Gaussian expectation with respect to the field $\psi^{(i)}$, $i=1,2$. Then:
%
%\begin{equation}
%\mathbb{E} f(\psi) = \mathbb{E}_{2} \mathbb{E}_{1} f(\psi^{(1)} + \psi^{(2)}) \equiv \mathbb{E}_{2} f^{(2)}(\psi^{(2)})\;,
%\end{equation}
%
%where $f^{(2)}(\psi^{(2)}) = \mathbb{E}_{1} f(\psi^{(1)} + \psi^{(2)})$. This simple identity is the starting point for the multiscale analysis of the Grassmann field theory.

Let $\chi(\cdot)$ be a smooth cutoff function as in (\ref{eq:cutoff}). We decompose $\hat g^{(\leq N_{0})}$ into the sum of an ultraviolet (UV) and infrared (IR) part as:
\begin{equation}\label{eq:gsplit}
\hat g^{(\leq N_{0})}(\underline{k}) = \hat g^{\text{UV}}(\underline{k}) + \hat g^{\text{IR}}(\underline{k})\;,\qquad 
\hat g^{\text{UV}}(\underline{k}) = \frac{\chi_{[0,N_{0}]}(k_{0})}{ik_{0}+ \hat H(k)-\mu}
\end{equation}
where $\chi_{[0,N_{0}]}(k_{0}) = \chi_{N_{0}}(k_{0}) - \chi(k_{0})$. By the addition principle of Grassmann Gaussian variables, we split the Grassmann field as:
\begin{equation}
\hat \Psi^{\pm}_{\underline{k}} = \hat \Psi^{\text{UV},\pm}_{\underline{k}} + \hat \Psi^{\text{IR}}_{\underline{k}}\;,
\end{equation} 
where $\hat \Psi^{\text{UV},\pm}_{\underline{k}}$, $\hat \Psi^{\text{IR}}_{\underline{k}}$ are independent Grassmann fields, with propagators give by, respectively, $\hat g^{\text{UV}}(\underline{k})$ and $\hat g^{\text{IR}}(\underline{k})$. Thanks to the fact that the cutoff function $\chi_{[0,N_{0}]}(k_{0})$ is vanishing for $|k_{0}| \leq 1$, the propagator $\hat g^{\text{UV}}(\underline{k})$ is non-singular, uniformly in $\beta, L$. Let:
\begin{equation}
V(\Psi;A,\phi) = -V(\Psi) + B(\Psi; \phi) + C(\Psi; A)\;,
\end{equation}
recall (\ref{eq:genfcn}). We then rewrite the generating functional of correlations as:
\begin{equation}
e^{\mathcal{W}_{\beta,L,N_{0}}(\phi, A)} = \int P_{\text{IR}}[d\Psi^{\text{IR}}] \int P_{\text{UV}}[d\Psi^{\text{UV}}] e^{V(\Psi^{\text{IR}} + \Psi^{\text{UV}};A,\phi)}\;,
\end{equation}
and we proceed with the integration of the ultraviolet field. A standard multiscale analysis for the ultraviolet degrees of freedom yields, see {\it e.g.} \cite{GMP}:
\begin{equation}
\int P_{\text{UV}}[d\Psi^{\text{UV}}] e^{V(\Psi^{\text{IR}} + \Psi^{\text{UV}};A,\phi)} = e^{\mathcal{W}^{\text{IR}}_{\beta,L,N_{0}}(\phi, A) + V^{\text{IR}}(\Psi^{\text{IR}}; \phi,A)}\;,
\end{equation}
where the argument of the exponential have the form:
\begin{equation}
\label{eq:kernelexpIR}
 V^{\text{IR}}(\Psi^{\text{IR}}; \phi,A) = \sum_{\Gamma=(\Gamma_{\Psi}, \Gamma_{\phi}, \Gamma_{A})} \int_{\beta,L} d\ul X d\ul Y d\ul Z \,\Psi^{\text{IR}}_{\Gamma_{\Psi}}(\ul X) \phi_{\Gamma_{\phi}}(\ul Y) A_{\Gamma_{A}}(\ul Z) W^{\text{IR}}_{\Gamma} (\ul X, \ul Y, \ul Z)\;,
\end{equation}
and:
\begin{equation}\label{eq:kerngen}
\mathcal{W}^{\text{IR}}_{\beta,L,N_{0}}(\phi, A) = \sum_{\Gamma=(\Gamma_{\phi}, \Gamma_{A})} \int_{\beta,L} d\ul Y d\ul Z \, \phi_{\Gamma_{\phi}}(\ul Y) A_{\Gamma_{A}}(\ul Z) W^{\text{IR}}_{\Gamma} (\ul Y, \ul Z)\;;
\end{equation}
let us explain the meaning of the various objects entering these expressions.
\begin{enumerate}
\item The labels $\Gamma_{\Psi}, \Gamma_{\phi}, \Gamma_{A}$ denote the collection of all labels of the fields, integrated and external. That is, if $n$ is the order of the $\Psi$-monomial, $\Gamma_{\Psi} = \{ (\rho_{1}, \varepsilon_{1}), \ldots, (\rho_{n}, \varepsilon_{n}) \}$, and:
\begin{equation}
\Psi_{\Gamma_{\Psi}}(\underline{X}) = \prod_{i=1}^{n} \Psi^{\varepsilon_{i}}_{\underline{x}_{i}, \rho_{i}}
\end{equation}
with the understanding that $\underline{X} = (\underline{x}_{1},\ldots, \underline{x}_{n})$. If $r$ is the order of the $\phi$-monomial, $\Gamma_{\phi} = \{ (\rho_{1}, \varepsilon_{1}), \ldots, (\rho_{r}, \varepsilon_{r}) \}$, and:
\begin{equation}
\phi_{\Gamma_{\phi}}(\underline{Y}) = \prod_{i=1}^{r} \phi^{\varepsilon_{i}}_{\underline{y}_{i}, \rho_{i}}\;,
\end{equation}
with $\underline{Y} = (\underline{y}_{1}, \ldots, \underline{y}_{r})$. If $m$ is the order of the $A$-monomial, $\Gamma_{A} = (\nu_{1}, \ldots, \nu_{m})$, and:
\begin{equation}
A_{\Gamma_{A}}(\underline{Z}) = \prod_{i=1}^{m} A^{\nu_{i}}_{\underline{z}_{i}}\;,
\end{equation}
with $\underline{Z} = (\underline{z}_{1}, \ldots, \underline{z}_{m})$.
\item The integral runs over all the space-time labels of the monomials determined by $\Gamma$.
\item The functions $W^{\text{IR}}_{\Gamma} (\ul X, \ul Y, \ul Z)$, with the understanding that $W^{\text{IR}}_{(\emptyset, \Gamma_{\phi}, \Gamma_{A})} (\ul X, \ul Y, \ul Z) \equiv W^{\text{IR}}_{(\Gamma_{\phi}, \Gamma_{A})} (\ul Y, \ul Z)$, also called kernels, are the outcome of the integration of the ultraviolet field. They are analytic in $\lambda$ for $|\lambda|$ small enough, they are invariant under simultaneous equal translation of all space-time coordinates, and they satisfy the following bound:
\begin{equation} \label{eq:locboundIR}
\int_{\beta,L} D\ul Q \prod_{i< j} \|\ul q_{i}-\ul q_{j}\|_{\beta,L}^{n_{ij}} | W^{\text{IR}}_{\Gamma} (\ul Q)| \leq \beta L C_{\Gamma;\{n_{ij}\}}\quad \forall n_{ij}\in\mathbb N,
 \end{equation}
 with $\underline{Q} = (\underline{X}, \underline{Y}, \underline{Z}) = (\underline{q}_{1}, \ldots, \underline{q}_{n+m+r})$, where $\|\cdot\|_{\beta,L}$ denotes the periodic space-time norm:
\begin{equation}
\| \underline{q} \|_{\beta,L} = \inf_{(n_{0}, n_{1}) \in \mathbb{Z}^{2}} \| \underline{q} + (n_{0}\beta, n_{1} L) \|\;.
\end{equation}
\end{enumerate}
The kernels $W^{\text{IR}}_{\Gamma} (\ul Q)$ play the role of new effective interactions for the infrared field. The analyticity of the effective interactions is a nontrivial result, and it follows from the application of the Brydges-Battle-Federbush (BBF) formula, for fermionic truncated expectations. This formula provides a convergent representation of fermionic perturbation theory; see {\it e.g.} \cite{GM, GMR} for a review. Observe that the expressions (\ref{eq:kernelexpIR}), (\ref{eq:kerngen}) actually involve finitely many monomials $\Psi^{\text{IR}}$ and $\phi$, due to the fact that these Grassmann fields involve finitely many Fourier modes, and hence generate a finite-dimensional Grassmann algebra. Instead, the sum over the $A$ monomials can be proved to be convergent, for $\|A\|_{\infty}$ small enough. Thus, both (\ref{eq:kernelexpIR}), (\ref{eq:kerngen}) are well defined.

We are left with the integration of the field $\Psi^{\text{IR}\pm}$. Its propagator is:
\begin{equation}
\hat g^{\text{IR}}(\underline{k}) = \frac{\chi(k_{0})}{ik_{0} + \hat H(k) - \mu}\;,
\end{equation}
which is now singular as $\beta,L\to \infty$, for $k_{0} = 0$ and for $k$ equal to one of the Fermi points $k_{F}^{\omega}$. Let $k_{F}^{\omega}(\lambda)$ be points in $B_{L}$ such that $|k_{F}^{\omega} - k_{F}^{\omega}(\lambda)| \leq C|\lambda|$, to be fixed later on; they play the role of renormalized Fermi points. It is convenient to rewrite the infrared propagator as:
\begin{equation}\label{eq:12}
\hat g^{\text{IR}}(\underline{k}) = \hat g^{\text{IR},1}(\underline{k}) + \hat g^{\text{IR},2}(\underline{k})
\end{equation}
where, for $\underline{k}_{F}^{\omega}(\lambda) = (0, k_{F}^{\omega}(\lambda))$:
\begin{equation}
\hat g^{\text{IR},1}(\underline{k}) = \sum_{\omega = 1}^{N_{f}} \frac{\chi_{0}(\|\underline{k} - \underline{k}_{F}^{\omega}(\lambda)\|_{\omega})}{ik_{0} + \hat H(k) - \mu}\;,
\end{equation}
with: $\chi_{0}(\cdot) = \chi(\delta (\cdot))$ for $\delta>0$ small enough; $\|\underline{q}\|_{\omega}^{2} = |q_{0}|^{2} + v_{\omega}^{2} |q_{1}|_{\mathbb{T}}^{2}$; and $\hat g^{\text{IR},2}(\underline{k})$ is a smooth, bounded function of $\underline{k} = (k_{0}, k_{1})$. The number $\delta$ is chosen so that the cutoffs $\chi_{0}(\|\underline{k} - \underline{k}_{F}^{\omega}(\lambda)\|_{\omega})$ have disjoint supports, and so that $\chi_{0}(\|\underline{k} - \underline{k}_{F}^{\omega}(\lambda)\|_{\omega}) \neq 0$ implies $\chi(k_{0}) = 1$. The advantage of the rewriting (\ref{eq:12}) is that the cutoffs entering the massless propagator  $\hat g^{\text{IR},1}(\underline{k})$ are symmetric in $k_{0}$ and in $v_{\omega}[k_{1} - k_{F}^{\omega}(\lambda)]$.

We use the addition principle to write:
\begin{equation}
\hat \Psi^{\text{IR},\pm}_{\underline{k}} = \hat \Psi^{\text{IR}1,\pm}_{\underline{k}} + \hat \Psi^{\text{IR}2,\pm}_{\underline{k}}\;,
\end{equation}
and to integrate the field $\hat \Psi^{\text{IR}2,\pm}_{\underline{k}}$ via fermionic cluster expansion, thus producing a new effective action for the field $\hat \Psi^{\text{IR}1,\pm}_{\underline{k}}$, determined by kernels $W^{\text{IR}1}_{\Gamma} (\cdot)$ satisfying the same bounds as $W^{\text{IR}}_{\Gamma} (\cdot)$, with different numerical constants.

Let us now focus on the massless propagator $\hat g^{\text{IR},1}(\underline{k})$. Thanks to Assumption \ref{ass:A}, we know that the spectrum in a window of size $\Delta$ around $\mu$ is characterised by the local band functions $\{e_{\omega}:I_{\omega}\to\mathbb R\}_{\omega=1,\dots, N_{f}}$, so that the IR part may be rewritten to be
\begin{equation} \label{eq:IRgsplit}
\hat g^{\text{IR},1}(\underline{k}) = \sum_{\omega=1}^{N_{f}} \hat g^{\text{IR}}_{\omega}(\ul k) P_{\omega}(k)\;,
\end{equation}
with $P_{\omega}(k)$ the rank-one projector onto the eigenvector associated with the eigenvalue $e_{\omega}(k)$, namely
\begin{equation}
P_{\omega}(k):= |\xi^{\omega}(k) \rangle \langle \xi^{\omega}(k)|\qquad\qquad\hat H(k)\xi^{\omega}(k)=e_{\omega}(k)\xi^{\omega}(k)\;,
\end{equation}
and where the propagators $\hat g^{\text{IR}}_{\omega}(\underline{k})$ are given by:
\begin{equation}\label{eq:go}
\hat g^{\text{IR}}_{\omega}(\ul k) = \frac{\chi_{0}(\|\underline{k} - \underline{k}_{F}^{\omega}(\lambda)\|_{\omega})}{ik_{0} + e_{\omega}(k)-\mu}\;.
\end{equation}
By the addition principle, we can represent the field $\hat \Psi^{\text{IR}1,\pm}_{\underline{k}}$ as:
\begin{equation}\label{eq:deca}
\hat \Psi^{\text{IR}1,\pm}_{\underline{k},\rho} = \sum_{\omega = 1}^{N_{f}} \hat \Psi^{\text{IR}1,\pm}_{\underline{k},\omega,\rho}
\end{equation}
where the field $\hat \Psi^{\text{IR}1,\pm}_{\underline{k},\omega,\rho}$ has propagator given by $\hat g^{\text{IR}}_{\omega}(\ul k) P_{\omega,\rho,\rho'}(k)$. This Gaussian Grassmann field can be conveniently represented as, for all $\underline{k}$ in the support of $\chi_{0}(\|\underline{k} - \underline{k}_{F}^{\omega}\|_{\omega})$:
\begin{equation}\label{eq:decb}
\hat \Psi^{\text{IR}1,-}_{\underline{k},\omega,\rho} = \psi^{\text{IR},-}_{\underline{k}, \omega} \xi^{\omega}_{\rho}(k)\;,
\end{equation}
where the new Grassmann field $\psi^{\text{IR},\pm}_{\underline{k}, \omega}$ has propagator given by $\hat g^{\text{IR}}_{\omega}(\ul k)$ in (\ref{eq:go}). The fields $\{\psi^{\text{IR},\pm}_{\underline{k}, \omega}\}$ are quasi-particle fields, and they characterize the low-energy excitations of the system.

Let us now represent the effective action in terms of these new fields. Let:
\begin{equation}
\check{\xi}^{\omega}_{\rho}(x) = \frac{1}{L} \sum_{k\in B_{L}} e^{ikx} \chi\big( |k - k_{F}^{\omega}(\lambda)|_{\mathbb{T}}/2\big) \hat \xi_{\rho}^{\omega}(k)\;.
\end{equation}
Observe that, from (\ref{eq:deca}), (\ref{eq:decb}):
\begin{equation}\label{eq:conv}
\Psi^{\text{IR}1,-}_{\ul x,\rho} = \sum_{\omega} (\check\xi^{\omega}_{\rho}*\psi^{\text{IR},-}_{\omega})_{\ul x}\;,\quad \Psi^{\text{IR}1,+}_{\ul x,\rho} =\sum_{\omega} (\overline{\check\xi^{\omega}_{\rho}}*\psi^{\text{IR},+}_{\omega})_{\ul x}
\end{equation}
with $*$ denoting spatial convolution:
\begin{equation}\label{eq:checkxi}
(\check\xi^{\omega}_{\rho}*\psi^{\text{IR},-}_{\omega})_{(x_{0},x)} = \sum_{z\in\Gamma_{L}} \check\xi^{\omega}_{\rho}(x-z) \psi^{\text{IR},-}_{(x_{0},z),\omega}\;.
\end{equation}
Since the field $\psi^{\text{IR},\pm}_{\underline{k}, \omega}$ has the same support properties of $\hat g^{\text{IR}}_{\omega}(\ul k)$, in Eq. (\ref{eq:decb}) we can freely multiply the Bloch function $\xi^{\omega}_{\rho}(k)$ by a cutoff function which is equal to $1$ in the support of the field; this yields (\ref{eq:conv}) with $\check\xi^{\omega}_{\rho}$ as in (\ref{eq:checkxi}). A standard integration by parts argument, combined with the smoothness of the cutoff function and of $\hat \xi^{\omega}(k)$ around the $\omega$-Fermi point, shows that:
\begin{equation}
| \check \xi_{\rho}(x) | \leq \frac{C_{n}}{1 + |x|_{L}^{n}}\qquad \text{for all $n\in \mathbb{N}$.}
\end{equation}
Thus, we obtained the following rewriting of the generating functional of the correlations:
\begin{equation}
\begin{split}
e^{\mathcal{W}_{\beta,L,N_{0}}(\phi, A)} &= e^{\mathcal{W}^{\text{IR}1}_{\beta,L,N_{0}}(\phi, A)} \int P_{\text{IR}1}[d\Psi^{\text{IR}1}] e^{V^{\text{IR}1}(\Psi^{\text{IR}1}; \phi,A)} \\
&= e^{\mathcal{W}^{\text{IR}1}_{\beta,L,N_{0}}(\phi, A)} \int P_{\text{IR}}[d\psi^{\text{IR}}] e^{V^{\text{IR}1}(\psi^{\text{IR}}*\xi; \phi,A)}\;,
\end{split}
\end{equation}
where in the last step we used the representation (\ref{eq:conv}) of the field $\Psi^{\text{IR}1}$. Let $\{\nu_{\omega}\}$ be real numbers, such that $|\nu_{\omega}| \leq C|\lambda|$, to be determined later on. We further add and subtract a quadratic term in the effective interaction, up to a redefinition of the Gaussian integration:
\begin{equation}\label{eq:resum}
\begin{split}
&e^{\mathcal{W}^{\text{IR}1}_{\beta,L,N_{0}}(\phi, A)} \int P_{\text{IR}}[d\psi^{\text{IR}}] e^{V^{\text{IR}1}(\psi^{\text{IR}}*\xi; \phi,A)} \\
&= e^{\mathcal{W}^{\text{IR}1}_{\beta,L,N_{0}}(\phi, A)} \int P_{\text{IR}}[d\psi^{\text{IR}}] e^{\sum_{\omega} \nu_{\omega} \int_{\beta,L} d\underline{x}\, \psi^{\text{IR},+}_{\underline{x},\omega} \psi^{\text{IR},-}_{\underline{x},\omega}} e^{V^{\text{IR}1}(\psi^{\text{IR}}*\xi; \phi,A) - \sum_{\omega} \nu_{\omega} \int_{\beta,L} d\underline{x}\, \psi^{\text{IR},+}_{\underline{x},\omega} \psi^{\text{IR},-}_{\underline{x},\omega}} \\
&\equiv e^{\mathcal{W}^{(0)}_{\beta,L,N_{0}}(\phi, A)}  \int P_{0}[d\psi^{(\leq 0)}] e^{V^{(0)}(\psi^{(\leq 0)}; \phi,A)}\;,
\end{split}
\end{equation}
where, setting $\underline{k}' = \underline{k} - \underline{k}_{F}^{\omega}(\lambda)$,  the new field $\psi^{(\leq 0)\pm}_{\underline{k}',\omega}$ has propagator given by:
\begin{equation}\label{eq:newg}
\hat g_{\omega}^{(\leq 0)}(\underline{k}') = \frac{\chi_{0}(\|\underline{k}'\|_{\omega})}{ik_{0} + e_{\omega}(k' + k_{F}^{\omega}(\lambda)) - \mu + \nu_{\omega} \chi_{0}(\|\underline{k}'\|_{\omega}) }\;.
\end{equation}
The variable $\underline{k}'$ expresses the relative shift of the Euclidean quasi-momentum with respect to the Fermi points. Thus, the configuration-space counterpart of the new field is:
\begin{equation}
e^{\mp i \underline{k}^{\omega}_{F}(\lambda) \cdot \underline{x}} \psi^{(\leq 0)\pm}_{\underline{x}, \omega}\qquad \text{with}\quad \psi^{(\leq 0)\pm}_{\underline{x}, \omega} = \frac{1}{\beta L} \sum_{\underline{k}'} e^{\mp i \underline{k}' \cdot \underline{x}} \psi^{(\leq 0)\pm}_{\underline{k}',\omega}\;.
\end{equation} 
From (\ref{eq:newg}), observe that the new Gaussian integration takes into account the first exponential factor in (\ref{eq:resum}). The change of normalization of the Grassmann integration is taken into account by $\mathcal{W}^{(0)}_{\beta,L,N_{0}}(\phi, A)$, while the new effective interaction $V^{(0)}(\psi^{(\leq 0)}; \phi,A)$ takes into account all the remaining contributions appearing at the exponent. It has the form:
\begin{equation}
V^{(0)}(\psi^{(\leq 0)}; \phi,A) = \sum_{\Gamma=(\Gamma_{\psi}, \Gamma_{\phi}, \Gamma_{A})} \int_{\beta,L} d\ul X d\ul Y d\ul Z \,\psi^{(\leq 0)}_{\Gamma_{\psi}}(\ul X) \phi_{\Gamma_{\phi}}(\ul Y) A_{\Gamma_{A}}(\ul Z) W^{(0)}_{\Gamma} (\ul X, \ul Y, \ul Z)\;,
\end{equation}
where now $\Gamma_{\psi} = \{(\omega_{1}, \varepsilon_{1}), \ldots, (\omega_{n}, \varepsilon_{n})\}$.

Thanks to the decay properties of $\check{\xi}$, and to the properties of the kernels of the field $\Psi^{\text{IR}1}$, the new kernels are: analytic in $\lambda$ for $|\lambda|$ small enough; invariant under translation of all space-time arguments; they satisfy the decay estimate:
\begin{equation}\label{eq:bdkernels}
\int_{\beta,L} D\ul Q \prod_{i< j} \|\ul q_{i}-\ul q_{j}\|_{\beta,L}^{n_{ij}} | W^{(0)}_{\Gamma} (\ul Q)| \leq \beta L \widetilde{C}_{\Gamma;\{n_{ij}\}}\quad \forall n_{ij}\in\mathbb N,
\end{equation}
which is the analogue of (\ref{eq:locboundIR}), with $C_{\{n_{ij}\}} < \widetilde{C}_{\{n_{ij}\}}\leq 2C_{\{n_{ij}\}}$ for $|\lambda|$ small enough. Observe that, if $\Gamma_{\psi} \neq \emptyset$, the bound is proportional to $|\lambda|$; actually, it is $|\lambda|^{|\Gamma_{\psi}|/4}$ if $|\Gamma_{\psi}| \geq 4$. All the bounds are uniform in the ultraviolet cutoff $N_{0}$; from now on, we will not carry on the dependence on $N_{0}$ (the limit $N_{0}\to\infty$ will be eventually taken at the end).

Finally, let us comment on the necessity for ``adding and subtracting'' the quadratic terms proportional to $\nu_{\omega}$. These quadratic terms will play the role of counterterms for the renormalization group iteration, discussed in the next sections. In particular, they will be chosen to control the flow of the relevant terms, which are the quadratic terms in the fermionic fields. The shifted Fermi points $k^{\omega}_{F}(\lambda)$ will be chosen in order to take into account the shift  of the singularity of (\ref{eq:newg}), due to the presence of $\nu_{\omega}$ as the denominator. 

More precisely, the pair $(k_{F}^{\omega}(\lambda), \nu_{\omega})$ will satisfy the condition, as $\beta,L\to \infty$:
\begin{equation}\label{eq:cond}
e_{\omega}(k_{F}^{\omega}(\lambda)) - \mu + \nu_{\omega} = 0\;,
\end{equation}
and with $|k_{F}^{\omega}(\lambda) - k_{F}^{\omega}(0)|\leq C|\lambda|$, $|\nu_{\omega}| \leq C|\lambda|$. For $L$ finite, equation (\ref{eq:cond}) may not have a solution, due to the condition $k_{F}^{\omega}(\lambda) \in B_{L}$. To accommodate this, we will instead require that the condition (\ref{eq:cond}) holds with $0$ on the right-hand side replaced by $O(1/L)$.

To conclude, we obtained the following rewriting for the generating functional of the correlations:
\begin{equation} \label{eq:RGstart}
\mathcal W_{\beta,L}(\phi,A) = \mathcal W^{(0)}_{\beta,L}(\phi,A) + \log \int P_{0}[d\psi^{(\leq 0)}] e^{V^{(0)}(\psi^{(\leq 0)}; \phi,A)}\;,
\end{equation}
where $\mathcal W^{(0)}$, $V^{(0)}$ have good analyticity and locality properties (kernels that decay faster than any power in Euclidean space). Renormalization group allows to further integrate the massless field $\psi^{(\leq 0)}$, via a renormalized, convergent expansion. This will be the content of the next section.
\subsection{Renormalization group analysis}\label{sec:sketchRG}
In this section we discuss the multiscale integration of the massless fields $\psi^{(\leq 0)}_{\omega}$. This is done via a renormalization group strategy, that has been carried out in a number of earlier references; see {\it e.g.} \cite{AMP, MPmulti} for recent, closely related works, and references therein. In this section we will review the construction. This will be needed to set-up the rest of the paper, since the method will be used to derive sharp estimates on the multi-point current-current correlators, which will play a key role in the asymptotic validity of linear response. 

The starting point is the outcome of the integration of the ultraviolet and high-frequency degrees of freedom, Eq. (\ref{eq:RGstart}). Due to the singularity of the covariance of the field at the Fermi points, it turns out that the perturbation theory obtained expanding the theory around the Gaussian case is divergent at all orders, in the limit $\beta,L\to \infty$. In order to obtain a meaningful expansion, we will perform a multiscale analysis, by decomposing the field into a sum of single scale fields. The single scale fields will be integrated in a progressive way, starting from the high energy scales until the lowest energy scales. The covariance and the effective action at a given scale will be defined inductively, via the combination of suitable {\it localization} and {\it renormalization} operations. This will ultimately allow to exploit {\it cancellations} in the naive expansion in a systematic way, and to prove analyticity of the correlation functions for $|\lambda| < \lambda_{0}$ for some $\lambda_{0}$ independent of $\beta, L$.
\begin{remark}
In the rest of the paper, we will use the following short-hand notations:
\begin{equation}
\int_{\beta,L} d\ul k\, (\cdots) := \frac{1}{\beta L}\sum_{\ul k\in \mathbb D_{\beta,L}} (\cdots)\;, \qquad \int_{\beta,L} d\ul p\, (\cdots) := \frac{1}{\beta L} \sum_{\ul p\in \mathbb M_{\beta,L}} (\cdots)
\end{equation}
with
\begin{equation}
\mathbb D_{\beta,L} = \Big\{\ul k\ |\ k_{0}\in\frac{2\pi}{\beta}\Big(\mathbb Z + \frac{1}{2}\Big),\ k\in B_{L}\Big\}\;,\quad \mathbb M_{\beta,L} = \Big\{\ul p\ |\ p_{0}\in\frac{2\pi}{\beta}\mathbb Z,\ p\in B_{L}\Big\}.
\end{equation}
For $\beta,L$ finite, the summands that we will consider will be supported on a finite number of momenta.
\end{remark}
\subsubsection{Setting up the multiscale integration}
Let $h_{\beta} := \lfloor \log_{2} (\pi/\delta\beta) \rfloor$, and let $h \in \mathbb{Z}_{-}$, $0\geq h \geq h_{\beta}$. Suppose that the generating functional can be written as:
\begin{equation}\label{eq:leqh}
\mathcal{W}_{\beta,L}(\phi,A) = \mathcal{W}^{( h)}_{\beta,L}(\phi,A) + \log \int P_{ h}[d\psi^{(\leq h)}] e^{V^{( h)}(\sqrt{Z_{h}}\psi^{(\leq h)}; \phi,A)}
\end{equation}
with the following meaning of the various objects entering (\ref{eq:leqh}).
\begin{itemize}
\item[1.] The notation $\sqrt{Z_{h}}\psi^{(\leq h)}$ is a short-hand for $\sqrt{Z_{h,\omega}}\psi^{(\leq h)}_{\omega}$, meaning that each field $\psi^{(\leq h)}_{\omega}$ appearing in $V^{(h)}$ is multiplied by a factor $\sqrt{Z_{h,\omega}}$.
\item[2.] The Grassmann Gaussian measure $P_{ h}(d\psi^{(\leq h)})$ has covariance given by:
\begin{equation}
\begin{split}
\int P_{ h}(d\psi^{(\leq h)}) \psi^{(\leq h)-}_{\underline{x}, \omega} \psi^{(\leq h)+}_{\underline{y}, \omega'} &= \delta_{\omega\omega'} \int_{\beta,L} d\underline{k}'\, e^{i \underline{k}'\cdot (\underline{x} - \underline{y})} \hat g^{(\leq h)}_{\omega}(\underline{k}')\\
\hat g^{(\leq h)}_{\omega}(\underline{k}') &= \frac{1}{Z_{h,\omega}} \frac{\chi_{h,\omega}(\underline{k}')}{ik_{0} + v_{h,\omega} k'}(1 + r_{h,\omega}(\underline{k}'))\;,
\end{split}
\end{equation}
where
\begin{equation}
\chi_{h,\omega}(\ul k') = \chi\Big(\sqrt{k_{0}^{2} + v_{h,\omega}^{2} {|k'|_{\mathbb T}^{2}}}/(2^{h}\delta)\Big)\;;
\end{equation}
$\hat g^{(\leq h)}_{\omega}(\underline{k}')$ is the {\it renormalized propagator} on scale $h$. The quantities $r_{h,\omega}$, $Z_{h,\omega}$, $v_{h,\omega}$ are analytic functions of $\lambda$ for $|\lambda| < \lambda_{0}$ such that, for $a>0$, and for $\underline{k}'$ in the support of $\chi_{h,\omega}(\ul k')$:
\begin{equation}
\label{eq:rcc}
\big| \text{d}_{k_{0}}^{n_{0}} \text{d}_{k_{1}}^{n_{1}} r_{h,\omega}(\underline{k}') \big|\leq C2^{h(a - n_{0} - n_{1})}\;,\qquad  \big|v_{h,\omega} - v_{\omega}\big|\leq C|\lambda|\;,\qquad \Bigg| \frac{Z_{h+1,\omega}}{Z_{h,\omega}} \Bigg| \leq e^{c |\lambda|}\;.
\end{equation}
Furthermore, $Z_{h,\omega}$, $v_{h,\omega}$ are real. The constants $Z_{h,\omega}$ are called the wavefunction renormalizations, and while the constants $v_{h,\omega}$ are called the effective Fermi velocities, on scale $h$.

\item[3.] The effective interaction on scale $h$ takes the form (suppressing the $(\leq h)$ label on the fermionic fields):
\begin{equation} \label{eq:Vh}
V^{(h)}(\psi; \phi,A) = \sum_{\Gamma=(\Gamma_{\psi},\Gamma_{\phi},\Gamma_{A})} \int_{\beta, L} d \underline{X}\,\psi_{\Gamma_{\psi}}(\underline{X}_{\psi}) \phi_{\Gamma_{\phi}}({\ul X_{\phi}}) A_{\Gamma_{A}}({\ul X_{A}})  W_{\Gamma}^{(h)}(\ul X)
\end{equation}
with translation-invariant kernels $W^{(h)}_{\Gamma}$ analytic for $|\lambda| < \bar\lambda$; $\mathcal W^{(h)}_{\beta,L}$ admits an analogous expansion, with the constraint $\Gamma_{\psi}=\varnothing$. The kernels satisfy the following estimate, for $P_{\sharp} = |\Gamma_{\sharp}|$:
\begin{equation}\label{eq:boundeffpot}
\int_{\beta,L} d\underline{X}\, \prod_{i< j} \|\ul x_{i}-\ul x_{j}\|_{\beta,L}^{n_{ij}} \Big| W_{\Gamma}^{(h)}(\ul X)\Big| \leq \beta L C_{\{n_{ij}\}} 2^{h (2 - (1/2)P_{\psi} - P_{A} - (3/2) P_{\phi} - c|\lambda| P_{A})}\;.
\end{equation}
\end{itemize}
Our goal will be to determine a recursion relation for the kernels, as functions of the scale parameter $h$, that allows in particular to check the estimate (\ref{eq:boundeffpot}) on all scales, until $h_{\beta}$. To conclude the section, we remark that all inductive assumptions are true on scale $h=0$, as a consequence of the cluster expansion that has been used to integrate the ultraviolet degrees of freedom.
\subsubsection{Localization and renormalization}\label{sec:loc}
In this section we define a localization operator, that allows to extract from the effective interaction on scale $h$ the relevant and marginal contributions; these are terms that a priori might expand under the renormalization group iteration, and that will be studied separately. Let us define:
\begin{equation} \label{eq:Vhpieces}
V^{(h)}_{n,m,l}(\psi; \phi, A) = \sum_{\substack{\Gamma\\|\Gamma_{\psi}|=n,|\Gamma_{\phi}|=m,|\Gamma_{A}|=l}} \int_{\beta, L} d \underline{X}\,\psi_{\Gamma_{\psi}}(\underline{X}_{\psi}) \phi_{\Gamma_{\phi}}({\ul X_{\phi}}) A_{\Gamma_{A}}({\ul X_{A}})  W_{\Gamma}^{(h)}(\ul X)\;.
\end{equation}
We will be interested in the following contributions to the effective interaction:
\begin{equation}
\begin{split}
V^{(h)}_{2,0,0}(\psi; \phi, A) &= \sum_{\omega}\int_{\beta,L} d\ul k'\, \hat\psi^{+}_{\ul k',\omega}\hat\psi^{-}_{\ul k',\omega} \widehat W^{(h)}_{2,0,0;\omega}(\ul k')\\
V^{(h)}_{4,0,0}(\psi; \phi, A) &= \sum_{\ul \omega}\int_{\beta,L} d\ul k'_{1}d\ul k'_{2}d\ul k'_{3}\, \hat\psi^{+}_{\ul k'_{1},\omega_{1}}\hat\psi^{-}_{\ul k'_{2},\omega_{2}}\hat\psi^{+}_{\ul k'_{3},\omega_{3}}\hat\psi^{-}_{\ul k'_{1}+\ul k'_{3}-\ul k'_{2},\omega_{4}} \widehat W^{(h)}_{4,0,0;\ul \omega}(\ul k'_{1},\ul k'_{2},\ul k'_{3})\\
V^{(h)}_{1,1,0}(\psi; \phi, A) &= \sum_{\omega,\rho} \int_{\beta,L} d\ul k' \Big(\hat \phi^{+}_{\ul k' + \underline{k}_{F}^{\omega},\rho} \hat \psi^{-}_{\ul k',\omega} \widehat W^{(h)}_{1,1,0;\rho,\omega}(\underline{k'}) +  \hat \psi^{+}_{\ul k',\omega} \hat \phi^{-}_{\ul k' + \underline{k}_{F}^{\omega},\rho} \widehat W^{(h)}_{1,1,0;\omega,\rho}(\underline{k'})\Big)\\
V^{(h)}_{2,0,1}(\psi; \phi, A) &= \sum_{\ul \omega ,\nu} \int_{\beta,L} d\ul k' d\ul p\,  \hat\psi^{+}_{\ul k'+\ul p,\omega_{1}}\hat\psi^{-}_{\ul k',\omega_{2}}\hat A^{\nu}_{-\ul p} \widehat W^{(h)}_{2,0,1;\ul\omega,\nu}(\ul k',\ul p)\;,
\end{split}
\end{equation}
for suitable kernels. Let $\ul 0^{\pm}_{\beta} := (\pm \pi/\beta,0)$. We set:
\begin{equation}
F(\ul 0_{\beta}):= \frac{1}{2}\sum_{\sigma=\pm} F(\ul 0^{\sigma}_{\beta})\;.
\end{equation}
\begin{definition} \label{def:loc}
The localization operator $\mathfrak L$ is a linear operator that acts as follows on $V^{(h)}_{n,m,l}$:
\begin{equation}
\begin{split}
\mathfrak L V^{(h)}_{n,m,l}(\psi,\phi,A) &= 0 \qquad\text{if}\qquad (n,m,l)\neq (2,0,0),(4,0,0), (2,0,1) \\
\mathfrak L V^{(h)}_{n,m,l}(\psi,\phi,A) &= V^{(h)}_{n,m,l} \qquad\text{if}\qquad (n,m,l) = (1,1,0)\;,
\end{split}
\end{equation}
and for the remaining cases, its action is defined by the following identities for the momentum space kernels:
\begin{equation} \label{eq:kernloc}
\begin{split}
\mathfrak L \widehat W^{(h)}_{2,0,0;\omega}\,(\ul k') &:=\widehat W^{(h)}_{2,0,0;\omega}(\ul 0_{\beta}) + \ul k'\cdot \ul{\mathrm d}\widehat W^{(h)}_{2,0,0;\omega}(\ul 0_{\beta})\\
\mathfrak L \widehat W^{(h)}_{4,0,0;\ul\omega}(\underline{k}'_{1}, \underline{k}'_{2}, \underline{k}'_{3}) & := \widehat W^{(h)}_{4,0,0;\ul \omega}(\ul 0_{\beta},\ul 0_{\beta},\ul 0_{\beta})\\
\mathfrak L \widehat W^{(h)}_{2,0,1;\ul\omega,\nu}(\underline{k}',\underline{p}) &:= \widehat W^{(h)}_{2,0,1;\ul\omega,\nu}(\ul 0_{\beta},\ul k_{F}^{\omega_{1}}(\lambda) - \ul k_{F}^{\omega_{2}}(\lambda))\;,
\end{split}
\end{equation}
with $\text{d}_{0}, \text{d}_{1}$ the discrete derivatives. Observe that, by momentum conservation and by assumption (\ref{eq:elasca}):
\begin{equation} \label{eq:bos}
%\begin{split}
%\widehat W^{(h)}_{2,0,1;\ul\omega,\nu}(\ul 0_{\beta},\ul 0) &= 0\qquad \text{unless $\omega_{1} = \omega_{2}$,}\\
\widehat W^{(h)}_{4,0,0;\ul\omega}(\ul 0_{\beta},\ul 0_{\beta},\ul 0_{\beta}) =0 \qquad \text{unless $\omega_{1}=\omega_{2},\,  \omega_{3}=\omega_{4}$}\ \text{or}\ \omega_{1}=\omega_{4},\omega_{2}=\omega_{3}\;.
%\end{split}
\end{equation}
Furthermore, one may check that $ \widehat W^{(h)}_{\omega_{1}\omega_{2}\omega_{3}\omega_{4}}(\ul k_{1},\ul k_{2},\ul k_{3}) = - \widehat W^{(h)}_{\omega_{3}\omega_{2}\omega_{1}\omega_{4}}(\ul k_{3},\ul k_{2},\ul k_{1}) $, which implies
\begin{equation} \label{eq:crossingsym}
\widehat W^{(h)}_{4,0,0;\omega\omega'\omega'\omega}(\ul 0_{\beta},\ul 0_{\beta},\ul 0_{\beta}) = - \widehat W^{(h)}_{4,0,0;\omega'\omega'\omega\omega}(\ul 0_{\beta},\ul 0_{\beta},\ul 0_{\beta})
\end{equation}
so that the case $\omega_{1}=\omega_{4}$, $\omega_{2}=\omega_{3}$ in (\ref{eq:bos}) can be reduced to the case $\omega_{1}=\omega_{2}$, $\omega_{3}=\omega_{4}$.
Next, keeping (\ref{eq:kernloc})-(\ref{eq:crossingsym}) in mind, let:
\begin{equation}
\begin{split}
z_{h,0,\omega}&:=-i\mathrm d_{0}\widehat W^{(h)}_{2,0,0;\omega}(\ul 0_{\beta})\qquad \qquad z_{h,1,\omega}:=\mathrm d_{1}\widehat W^{(h)}_{2,0,0;\omega}(\ul 0_{\beta})\\
\tilde \nu_{h,\omega} &:= 2^{-h}\widehat W^{(h)}_{2,0,0;\omega}(\ul 0_{\beta})\qquad\qquad \tilde \lambda_{h,\omega\omega'} := 2\widehat W^{(h)}_{4,0,0;\omega\omega\omega'\omega'}(\ul 0_{\beta}, \ul 0_{\beta}, \ul 0_{\beta})\;,
\end{split}
\end{equation}
so that
\begin{equation}
\begin{split}
\mathfrak L V^{(h)}_{2,0,0}(\psi; \phi, A) &= \sum_{\omega}\int_{\beta,L} d\ul k'\, [2^{h}\tilde \nu_{h,\omega} + i z_{h,0,\omega} k_{0} + z_{h,1,\omega} k'_{1}]\hat\psi^{+}_{\ul k',\omega}\hat\psi^{-}_{\ul k',\omega}\\
\mathfrak L V^{(h)}_{4,0,0}(\psi; \phi, A) &= \sum_{\omega,\omega'}\int_{\beta,L}d\ul p\, \tilde \lambda_{h,\omega\omega'}\hat n_{-\ul p,\omega}\hat n_{\ul p,\omega'}\\
\mathfrak L V^{(h)}_{2,0,1}(\psi; \phi, A) &= \sum_{\omega,\omega',\nu} \int_{\beta,L} d\ul p d\ul k'\,  \widehat W^{(h)}_{2,0,1;\omega\omega',\nu}(\ul 0_{\beta}, \ul k_{F}^{\omega}(\lambda) - \ul k_{F}^{\omega'}(\lambda)) \hat A^{\nu}_{-\ul p} \hat \psi^{+}_{\underline{k}' + \underline{p},\omega} \hat \psi^{-}_{\underline{k}',\omega'}\;.
\end{split}
\end{equation}
One may check that $\mathfrak L$ is idempotent, \emph{i.e.} $\mathfrak L^{2}=\mathfrak L$; defining the \emph{renormalization} operator $\mathfrak R := 1 - \mathfrak L$, one immediately has that $\mathfrak R^{2} =\mathfrak R$ as well.
\end{definition}
\begin{remark}
The constants $\tilde\nu_{h,\omega}, \tilde\lambda_{h,\omega\omega'}, z_{h,0,\omega}, z_{h,1,\omega}$ are real numbers, since the functional integral defining the partition function $\mathcal Z_{\beta,L,N_{0}}$ in (\ref{eq:partitionfunction}) is invariant under the transformation
\begin{equation}
\Psi^{-}\to \Psi^{+}\qquad \Psi^{+}\to -\Psi^{-}\qquad c\to \bar c\qquad k_{0}\to -k_{0}\;,
\end{equation}
where $c$ is a generic constant entering in the action. It should be noted that this symmetry can be implemented for any finite $N_{0}$, since the cut-off function $\chi_{N_{0}}$ is even in $k_{0}$.
\end{remark}
Next, we absorb the $\ul k$-dependent part of $\mathfrak L V^{(h)}_{2,0,0}$ in the covariance of the Gaussian measure $P_{h}$, obtaining:
\begin{equation} \label{eq:tildeV2}
P_{h}[d \psi^{(\leq h)}] e^{V^{(h)}(\sqrt{Z_{h}}\psi^{(\leq h)};\phi,A)} = e^{\beta L t_{h}}\widetilde{P}_{h}[d\psi^{(\leq h)}]e^{\widetilde{V}^{(h)}(\sqrt{Z_{h}}\psi^{(\leq h)};\phi,A)}\;;
\end{equation}
here $t_{h}$ is a normalization constant, due to the change of covariance of the Grassmann integration. The new covariance is given by
\begin{equation}
\tilde{g}^{(\leq h)}_{\omega}(\underline{k}') = \frac{1}{\widetilde{Z}_{h-1,\omega}(\underline{k}')} \frac{\chi_{h,\omega}(\underline{k}')}{ik_{0} + \widetilde{v}_{h-1,\omega}(\underline{k}') k'}[1 + \tilde r_{h,\omega}(\underline{k}')]\;,
\end{equation}
for a suitable new error term $\tilde r_{h,\omega}(\underline{k}')$, and where:
\begin{equation}\label{eq:Zvh-1}
\begin{split}
\widetilde{Z}_{h-1,\omega}(\underline{k}') &= Z_{h,\omega}[1+ z_{h,0,\omega} \chi_{h,\omega}(\underline{k}')]\\
\widetilde{Z}_{h-1,\omega}(\underline{k}') \widetilde{v}_{h-1,\omega}(\underline{k}') &= Z_{h,\omega}[v_{h,\omega} + z_{h,1,\omega} \chi_{h,\omega}(\underline{k}')]\;.
\end{split}
\end{equation}
The effective potential $\widetilde V^{(h)}$ has the form
\begin{equation} \label{eq:scalehstructuretilde}
\widetilde V^{(h)}(\psi; \phi,A) = \mathfrak{L} \widetilde V^{(h)} (\psi) + B^{(h)}(\psi;\phi) + C^{(h)}(\psi; A) + \mathfrak{R} V^{(h)} (\psi; \phi,A)\;,
\end{equation}
with:
\begin{equation}\nonumber
\begin{split}
\mathfrak{L} \widetilde{V}^{(h)}(\psi) &= \sum_{\omega}2^{h}\tilde \nu_{h,\omega}\int_{\beta,L}d\ul k'\,  \hat\psi^{+}_{\ul k',\omega}\hat\psi^{-}_{\ul k',\omega} + \mathfrak LV^{(h)}_{4,0,0}(\psi; 0,0)\\
B^{(h)}(\psi;\phi) &= \sum_{\omega,\rho} \int_{\beta,L} d\ul k' \Big(\hat \phi^{+}_{\ul k' + \underline{k}_{F}^{\omega}(\lambda),\rho} \hat \psi^{-}_{\ul k',\omega} \widehat W^{(h)}_{1,1,0;\rho,\omega}(\underline{k'}) +  \hat \psi^{+}_{\ul k',\omega} \hat \phi^{-}_{\ul k' + \underline{k}_{F}^{\omega}(\lambda),\rho} \widehat W^{(h)}_{1,1,0;\omega,\rho}(\underline{k'})\Big)\\
C^{(h)}(\psi; A) &=  \sum_{\omega,\omega',\nu} \int_{\beta,L} d\ul k' d\ul p\,  \hat\psi^{+}_{\ul k'+\ul p,\omega}\hat\psi^{-}_{\ul k',\omega'}\hat A^{\nu}_{-\ul p} \widehat W^{(h)}_{2,0,1;\omega\omega',\nu}(\ul 0_{\beta},\ul k_{F}^{\omega}(\lambda) - \ul k_{F}^{\omega'}(\lambda))\;.
\end{split}
\end{equation}
Finally, we set $Z_{h-1,\omega}:=\widetilde Z_{h-1,\omega}(\ul 0_{\beta})$ and we rescale the fields:
\begin{equation}
e^{\beta L t_{h}}\widetilde{P}_{h}[d\psi^{(\leq h)}]e^{\widetilde{V}^{(h)}(\sqrt{Z_{h}}\psi^{(\leq h)};\phi,A)}  = e^{\beta L t_{h}}\widetilde{P}_{h}[d\psi^{(\leq h)}]e^{\widehat{V}^{(h)}(\sqrt{Z_{h-1}}\psi^{(\leq h)};\phi,A)}\;,
\end{equation}
where $\widehat V^{(h)}$ can be represented as follows:
\begin{equation} \label{eq:scalehstructurehat}
\widehat V^{(h)}(\psi; \phi,A) = \mathfrak L\widehat V^{(h)} (\psi) + \widehat B^{(h)}(\psi;\phi) + \widehat C^{(h)}(\psi; A) + \mathfrak R \widehat V^{(h)} (\psi; \phi,A)\;,
\end{equation}
where
\begin{equation}
\begin{split}
\mathfrak L\widehat V^{(h)} (\psi) &= \sum_{\omega}2^{h} \nu_{h,\omega}\int_{\beta,L}d\ul k'\,  \hat\psi^{+}_{\ul k',\omega}\hat\psi^{-}_{\ul k',\omega} + \sum_{\omega,\omega'} \lambda_{h,\omega\omega'} \int_{\beta,L}d\ul p\, \hat n_{-\ul p,\omega} \hat n_{\ul p,\omega'}\\
\widehat B^{(h)}(\psi;\phi) &= \sum_{\omega} \int_{\beta,L} d\ul k' \Big[ (\hat\phi^{+}_{\ul k' + \ul k_{F}^{\omega}(\lambda)}, \mathsf Q^{+}_{h,\omega}(\ul k'))\hat\psi^{-}_{\ul k', \omega} + \hat\psi^{+}_{\ul k', \omega} (\mathsf Q^{-}_{h,\omega}(\ul k'), \hat\phi^{-}_{\ul k' + \ul k_{F}^{\omega}(\lambda)})\Big]\\
\widehat C^{(h)}(\psi; A) &= \sum_{\omega,\omega',\nu} \int_{\beta,L} d\ul p d\ul k'\, \hat A^{\nu}_{-\ul p} \mathsf{Z}^{\nu}_{h,\omega\omega'} \hat\psi^{+}_{\ul k'+\ul p,\omega}\hat\psi^{-}_{\ul k',\omega'}\;,
\end{split}
\end{equation}
with effective couplings
\begin{equation}\label{eq:scalehcouplings}
\nu_{h,\omega} := \frac{Z_{h,\omega}}{Z_{h-1,\omega}} \tilde \nu_{h,\omega}\qquad \qquad \lambda_{h,\omega\omega'} := \frac{Z_{h,\omega}Z_{h,\omega'}}{Z_{h-1,\omega}Z_{h-1,\omega'}}\tilde \lambda_{h,\omega\omega'}\\
\end{equation}
and
\begin{equation}
\begin{split}
\mathsf Q^{+}_{h,\omega,\rho}(\ul k') &= \sqrt{\frac{Z_{h,\omega}}{Z_{h-1,\omega}}} \widehat W^{(h)}_{1,1,0;\rho,\omega}(\ul k')\;,\qquad \mathsf Q^{-}_{h,\omega,\rho}(\ul k') = \sqrt{\frac{Z_{h,\omega}}{Z_{h-1,\omega}}} \widehat W^{(h)}_{1,1,0;\omega,\rho}(\ul k') \\
\mathsf{Z}^{\nu}_{h,\omega\omega'} &= \frac{\sqrt{Z_{h,\omega} Z_{h,\omega'}} }{\sqrt{Z_{h-1,\omega} Z_{h-1,\omega'}} } \widehat W^{(h)}_{2,0,1;\omega\omega',\nu}(\ul 0_{\beta},\ul k_{F}^{\omega}(\lambda) - \ul k_{F}^{\omega'}(\lambda))\;.
\end{split}
\end{equation}
In the following, we shall use the notation:
\begin{equation}
\mathsf{Z}^{\nu}_{h,\omega} \equiv \mathsf{Z}^{\nu}_{h,\omega\omega}\;.
\end{equation}
\begin{remark}
Notice that $\mathfrak L \widehat V^{(h)}_{4,0,0}$ can be written in real space as
\begin{equation}
\sum_{\omega,\omega'} \lambda_{h,\omega\omega'} \int_{\beta,L} d\ul x\, n_{\ul x,\omega} n_{\ul x,\omega'}\;.
\end{equation}
Thus, without loss of generality, we can suppose that:
\begin{equation}
\lambda_{h,\omega\omega} = 0\;,\qquad \lambda_{h,\omega\omega'} = \lambda_{h,\omega'\omega}\;.
\end{equation}
The first identity follows thanks to the fact that the square of a Grassmann variable is zero.
\end{remark}
\subsubsection{Single-scale integration}
We now perform the single-scale integration that will yield $\mathcal W^{(h-1)}, V^{(h-1)}$. In order to do so, we split the propagator as:
\begin{equation}\label{eq:propsh}
\begin{split}
\tilde g^{(\leq h)}_{\omega} &= \hat g^{(\leq h-1)}_{\omega} + g^{(h)}_{\omega}\\
g^{(h)}_{\omega}(\ul k') &= \frac{1}{\widetilde{Z}_{h-1,\omega}(\underline{k}')} \frac{f_{h,\omega}(\underline{k}')}{ik_{0} + \widetilde{v}_{h-1,\omega}(\underline{k}') k'}[1 + r_{h-1,\omega}(\underline{k}')]\\
\hat g^{(\leq h-1)}_{\omega}(\ul k') &= \frac{1}{Z_{h-1,\omega}} \frac{\chi_{h-1,\omega}(\underline{k}')}{ik_{0} + v_{h-1,\omega} k'}[1 + r_{h-1,\omega}(\underline{k}')]\;,
\end{split}
\end{equation}
with $f_{h,\omega}(\underline{k}') = \chi_{h,\omega}(\underline{k}')-\chi_{h-1,\omega}(\ul k')$. In the last identity, we used the fact that, for $\underline{k}'$ in the support of $\chi_{h-1,\omega}(\ul k')$:
\begin{equation}
\widetilde{Z}_{h-1,\omega}(\underline{k}') = \widetilde{Z}_{h-1,\omega}(\underline{0}_{\beta}) \equiv Z_{h-1,\omega}\;,\qquad \widetilde{v}_{h-1,\omega}(\underline{k}') = \widetilde{v}_{h-1,\omega}(\underline{0}_{\beta}) \equiv v_{h-1,\omega}\;,
\end{equation}
recall equation (\ref{eq:Zvh-1}). Denoting with $\widetilde P_{h}(d\psi^{(h)})$ the Grassmann Gaussian integration associated with $g^{(h)}$, we use the addition principle to write:
\begin{equation}
\begin{split}
\int &\widetilde P_{h}[d\psi^{(\leq h)}] e^{\widehat V^{(h)}\big(\sqrt{Z_{h-1}}\psi^{(\leq h)};\phi,A\big)}\\
&\quad=\int P_{h-1}[d\psi^{(\leq h-1)}] \int \widetilde P_{h}[d\psi^{(h)}] e^{\widehat V^{(h)}\big(\sqrt{Z_{h-1}}(\psi^{(\leq h-1)}+\psi^{(h)});\phi,A\big)}\;,
\end{split}
\end{equation}
where $P_{h-1}[d\psi^{(\leq h-1)}]$ is the Grassmann Gaussian integration associated with $\hat g^{(\leq h-1)}_{\omega}(\ul k')$. This identity allows to define the effective potential $V^{(h-1)}$ and the effective generating functional $\mathcal W^{(h-1)}$ as:
\begin{equation} \label{eq:RecursionEffPot}
\begin{split}
&\mathcal W^{(h-1)}_{\beta,L}(\phi,A) + V^{(h-1)}\Big(\sqrt{Z_{h-1}}\psi^{(\leq h-1)};\phi,A\Big)\\
&\quad := \mathcal W^{(h)}_{\beta,L}(\phi,A) + \beta L t_{h}+\log\int \widetilde P_{h}[d\psi^{(h)}] e^{\widehat V^{(h)}\big(\sqrt{Z_{h-1}}(\psi^{(\leq h-1)}+\psi^{(h)});\phi,A\big)}\;.
\end{split}
\end{equation}
As reviewed in Section \ref{sec:GNtrees}, equation (\ref{eq:RecursionEffPot}) is the starting point for the Gallavotti-Nicol\`o tree expansion for the effective potential, which will allow to control the flow of the effective potential. A key input to prove the convergence of the tree expansion is the control of the flow of the running coupling constants, discussed in the next section.
\subsection{Flow of the effective couplings} \label{sec:flow}
The above iteration gives rise to a discrete dynamical system for the running coupling parameters, which has the following form:
\begin{eqnarray}\label{eq:betaflow0}
\frac{Z_{h-1,\omega}}{Z_{h,\omega}} &=& 1 + z_{h,0,\omega} =: 1 + \beta^{z}_{h,\omega} \nonumber\\
v_{h-1,\omega} &=& \frac{Z_{h,\omega}}{Z_{h-1,\omega}} (v_{h,\omega} + z_{h,1,\omega}) =: v_{h,\omega} + \beta^{v}_{h,\omega}\nonumber\\
2^{h} \nu_{h,\omega} &=& \frac{Z_{h,\omega}}{Z_{h-1,\omega}} \widehat{W}^{(h)}_{2;\omega}(\underline{0}_{\beta}) =: 2^{h+1} \nu_{h+1,\omega} + 2^{h+1}\beta^{\nu}_{h+1,\omega}\nonumber\\
\lambda_{h,\omega\omega'} &=& \frac{Z_{h,\omega}Z_{h,\omega'}}{ Z_{h-1,\omega}Z_{h-1,\omega'}} 2\widehat{W}^{(h)}_{4;\omega\omega\omega'\omega'}(\underline{0}_{\beta},\underline{0}_{\beta},\underline{0}_{\beta}) =: \lambda_{h+1,\omega\omega'} + \beta^{\lambda}_{h+1,\omega\omega'}\nonumber\\
\mathsf Z^{\nu}_{h,\omega\omega'} &=& \mathsf Z^{\nu}_{h+1,\omega\omega'} + \beta^{\mathsf z,\nu}_{h+1,\omega\omega'}\;.
\end{eqnarray}
with initial data given by
\begin{equation}
\begin{split}
Z_{0,\omega} &= 1\;,\qquad v_{0,\omega} = v_{\omega} + O(\lambda)\;,\qquad \nu_{0,\omega} = \nu_{\omega} + O(\lambda)\;,\\ \lambda_{0,\omega\omega'} &= -\lambda + O(\lambda^{2})\;,\qquad |\mathsf Z^{\nu}_{0,\omega\omega'}|\leq C
\end{split}
\end{equation}
where the higher-order corrections in $\lambda$ are due to the integration of the ultraviolet degrees of freedom. The vector field:
\begin{eqnarray}
\beta_{h} &=& \beta_{h}\big( \{ Z_{k}, v_{k}, \nu_{k}, \lambda_{k}, \mathsf Z^{\nu}_{k} \}_{k=h+1}^{0} \big)  \nonumber\\
&=& ( \beta^{z}_{h,\omega}, \beta^{v}_{h,\omega}, \beta^{\nu}_{h,\omega}, \beta^{\lambda}_{h,\ul \omega}, \beta^{\mathsf z,\nu}_{h,\omega})_{\ul \omega}
\end{eqnarray}
defined by (\ref{eq:betaflow0}) is called the {\it beta function} of the theory, and it drives the flow of the couplings. The couplings $\{Z_{k}, v_{k}, \lambda_{k}, \mathsf Z^{\nu}_{k}\}$ are marginal in the RG sense, while $\{\nu_{k}\}$ are relevant. Precise estimates on this object are needed in order to safely iterate the equations in (\ref{eq:betaflow0}) in $h$. The control of the running coupling constants is the content of the following proposition. The result has been obtained for lattice fermions with two Fermi points in \cite{BMWI, BMdensity, BMchiral}, and it has been generalized to multiple Fermi points in \cite{MPmulti}.
\begin{proposition}[Bounds for the running coupling constants]\label{prp:flow} There exists $\bar\lambda>0$ such that for $|\lambda| < \bar \lambda$ the following is true. There exists a choice of real parameters $\{\nu_{\omega}\}$, with $\nu_{\omega}\equiv \nu_{\omega}(\lambda)$, $|\nu_{\omega}| \leq C|\lambda|$, such that:
\begin{equation} \label{eq:BdsEffCoup}
\begin{split}
\left|\frac{Z_{h-1,\omega}}{Z_{h,\omega}}\right| &\leq e^{c|\lambda|}\;,\qquad |v_{h,\omega}-v_{\omega}|\leq C|\lambda|\;,\qquad |\lambda_{h,\omega\omega'}-\lambda|\leq C|\lambda|^{2}\\
|\mathsf Z^{\nu}_{h,\omega}| &\leq C|\mathsf Z^{\nu}_{1,\omega}|\;,\qquad \left|\frac{ \mathsf Z^{\nu}_{h-1,\omega\omega'}}{\mathsf Z^{\nu}_{h,\omega\omega'}}\right| \leq e^{c|\lambda|}\;,\qquad |\nu_{h,\omega}| \leq C |\lambda| 2^{\theta h}.
\end{split}
\end{equation}
for some positive constants $c,C,\theta$ independent of $\beta, L$, for all $\omega,\omega'=1,\dots, N_{f}$, $\nu=0,1$, and for all $h\geq h_{\beta}:=\lfloor\log(\pi/\beta)\rfloor$.
\end{proposition}
\begin{remark}
It should be noted that the independence of these bounds from $\beta$ allows to control the dynamical system also in the zero temperature limit $h_{\beta}\to-\infty$.
\end{remark}

Let us outline the proof of Proposition \ref{prp:flow}, referring to \cite{BMWI, BMdensity, BMchiral, MPmulti} for more details. The flows of $\lambda_{h,\omega\omega'}, v_{h,\omega}, \mathsf Z^{\nu}_{h,\omega}$ are controlled via a rigorous comparison with a {\it reference model}, reviewed in Section \ref{Sec:RefMod}.  The analogous running coupling constants for the reference model satisfy the bounds stated in Proposition \ref{prp:flow}. This is a consequence of the anomalous, chiral Ward identities for the reference model, as discussed in Appendix \ref{app:RGref}. In order to prove the same result for the lattice model, we show that the beta function of the running coupling constants in (\ref{eq:betaflow0}) can be written as the beta function of the running coupling constants of the reference model, for a certain choice of bare parameters of the reference model, plus an error term, that admits a dimensional gain $2^{\theta h}$. Up to new error terms of order $2^{\theta h}$, the beta functions for these running coupling constants of the reference model turns out to be asymptotically vanishing as $h\ll -1$ \cite{BMWI, MPmulti}, and this yields the stated bounds in Proposition \ref{prp:flow}. The control of the flows of $\lambda_{h,\omega\omega'}, v_{h,\omega}$ allows then to prove the estimates for $Z_{h-1,\omega}, \mathsf Z^{\nu}_{h,\omega\omega'}$, by iteration of the RG recursion, using that, as a consequence of the convergence of the tree expansion:
\begin{equation}
|\beta^{z}_{h,\omega}|\leq C |\lambda_{\geq h+1}|\;,\qquad |\beta^{\mathsf z,\nu}_{h+1,\omega\omega'}|\leq C |\lambda_{\geq h+1}|
\end{equation}
with $\lambda_{\geq h} = \max_{k\geq h} \max_{\omega,\omega'} |\lambda_{k,\omega\omega'}|$. Finally, the bound for $\nu_{h,\omega}$ follows from a suitable choice of $\nu_{\omega}$, via a standard fixed point argument. This concludes the outline of the proof.

Finally, proceeding as in {\it e.g.} Section 3.4 of \cite{GiuM}, it is also possible to prove that the external $\phi^{\pm}$ source term kernels $\mathsf Q^{\pm}_{h,\omega}$ satisfy the recursion relation:
\begin{eqnarray} \label{eq:Q+}
\mathsf Q^{+}_{h,\omega}(\underline{k}') &=& \mathsf Q^{+}_{h+1,\omega}(\underline{k}') - \widehat{W}_{2,0,0; \omega\omega}^{(h)}(\underline{k}') \sum_{h'=h+1}^{1} \hat g^{(h')}_{\omega}(\underline{k}') \mathsf Q^{+}_{h',\omega}(\underline{k}')\nonumber\\
\mathsf Q^{-}_{h,\omega}(\underline{k}') &=& \overline{\mathsf Q^{+}_{h,\omega}(\underline{k}')}\;.
\end{eqnarray}
In particular, for $\ul k'\in\mathrm{supp}\, \hat g^{(h)}_{\omega}$, equation (\ref{eq:Q+}) reduces to
\begin{equation} \label{eq:Qh}
\mathsf Q^{+}_{h,\omega}(\ul k') = \mathsf Q^{+}_{1,\omega}(\ul k')[1-\widehat W^{(h)}_{2,0,0;\omega\omega}(\ul k')\hat g^{(h)}_{\omega}(\ul k')]\;,
\end{equation}
which, thanks to the estimates (\ref{eq:BdsEffCoup}), and the bounds for $\widehat W^{(h)}_{2,0,0;\omega\omega}$ which will be derived in Section \ref{sec:GNtrees}, implies that 
\begin{equation}
\mathsf Q^{+}_{h,\omega}=\mathsf Q^{+}_{1,\omega}(1+O(\lambda)) = \xi^{\omega}(1 + O(\lambda))\;.
\end{equation}
%
%The proof of this result uses the structure of the scaling limit of the system, which describes relativistic, chiral $1+1$ dimensional massless fermions. The model describing this system, called the reference model, will be reviewed later on. The boundedness of the flow of the running coupling constants is a deep consequence of the vanishing of the beta function for the reference model, which can be viewed as an instance of the integrability of the scaling limit.
\begin{remark} The wave function renormalization $Z_{h,\omega}$ has an interaction-dependent power law behavior. It follows that:
\begin{equation}\label{eq:Zano}
Z_{h,\omega} \sim 2^{\xi_{\omega} h}\;,\qquad \text{with $\xi_{\omega} = O(\lambda^{2})$.}
\end{equation}
The anomalous exponent is given by a convergent series in the renormalized couplings, and the first order term is vanishing. In the case of two chiralities, it is known that $\xi_{\omega}<0$, which implies that the wave function renormalization diverges. We expect the same to hold true in general; this fact will not be needed in the present work. The power law behavior (\ref{eq:Zano}) changes the scaling of the two-point function, and introduces interaction dependent anomalous exponents. This is the reason for the rescaling of the fields in the multiscale analysis. 

A priori, the vertex renormalizations $\mathsf Z^{\nu}_{h,\omega\omega'}$ might also have a divergent power-law behavior. The fact that $\mathsf Z^{\nu}_{h,\omega}$ stays bounded is non-trivial, and follows from the emergent (anomalous) chiral gauge symmetry of the model \cite{BMdensity}.
\end{remark}

\subsection{Tree expansion for kernels} \label{sec:GNtrees}
In this section, we recall the \emph{Gallavotti-Nicolò} (GN) {\it tree} representation of the kernels $W^{(h)}_{\Gamma}$ contributing to the effective potential $V^{(h)}$ and of the generating functional $\mathcal W^{(h)}$. Although this is by now a standard construction (see \cite{GM} for a review), we prefer to recall some of its basic elements because they will be useful in deriving the improved estimates in Section \ref{sec:dimest}. We will discuss how to obtain a tree expansion for the effective potential in terms of trees that extend up to scale $0$, which takes as input the effective potential generated after the integration of the ultraviolet degrees of freedom; this is the infrared construction, and it is the technically most demanding one. The effective potential on scale $0$ is in turn the outcome of an ultraviolet multiscale analysis, which is much easier to carry out and will be recalled in a second moment.

Employing the fact that the logarithm of a Grassmann Gaussian expectation may be written in terms of cumulants, we may transform equation (\ref{eq:RecursionEffPot}), upon replacing $h-1$ with $h$, by as
\begin{equation} \label{eq:ET}
\begin{split}
\mathcal{W}^{( h)}_{\beta,L}(\phi,A) &+ V^{(h)}(\sqrt{Z_{h}}\psi^{(\leq h)}; \phi,A)\\
&=\mathcal{W}^{(h+1)}_{\beta,L}(\phi,A) + \beta L t_{h+1} + \sum_{n\geq 0} \frac{(-1)^{n}}{n!}\mathbb{E}_{h+1} [\underbrace{\widehat{V}^{(h+1)}\,; \cdots \,; \widehat{V}^{(h+1)}}_{\text{$n$ times}}]
\end{split}
\end{equation}
where $\mathbb{E}_{h}[\,\cdot\,;\cdots;\,\cdot\,]$ denotes the {\it truncated expectation} (or cumulant) with respect to the single-scale Gaussian Grassmann measure $\widetilde{P}_{h}$. Equation (\ref{eq:ET}) can be iterated over all scales $h+1,\, h+2,\, \ldots$, until $h=0$ and the result can be graphically expressed as a sum over GN trees.

Let us first introduce some definitions.
\begin{enumerate}
\item An {\it unlabelled} tree is a connected, acyclic\footnote{That is, it has no loops.} finite graph connecting a point $r$, called the root, with an ordered set of $n\geq 1$ points, the endpoints of the tree, so that $r$ is not a branching point. The number $n$ will be called the {\it order} of the unlabelled tree; in particular, the root $r$ is not counted as a vertex. We will denote by $V(\tau)$ the set of vertices of the unlabelled tree $\tau$.
The unlabelled trees have a natural partial ordering $\prec$, that is: orienting the tree from the root, we set $v\prec v'$ if there is an oriented path from $r$ to $v'$ that passes through $v$.

We denote by $V_{f}(\tau)\subset V(\tau)$ the set of endpoints, also called \emph{leaves} or \emph{final vertices}, of $\tau$. For each vertex $v$, we will denote by $s_{v}$ its branching number: in particular, vertices with $s>1$ will be called the {\it non-trivial vertices}. Notice that $s_{v}=0$ if $v\in V_{f}(\tau)$.

Two unlabelled trees are identified if they can be superposed by a continuous deformation, which is compatible with the partial ordering of the non-trivial vertices. It is then easy to see that the number of unlabelled trees with $n$ endpoints and only non-trivial vertices is bounded by $4^{n}$.

We shall also consider {\it labelled trees} (or just trees, in the following), see Fig. \ref{fig:GN}; they are defined by associating suitable labels to the vertices of unlabelled trees, as explained in the following points.
\item Given a labelled tree $\tau$, to each vertex $v$ of the labelled tree $\tau$ we associate a scale label $h_{v}$ in $\{h+1, \ldots , 0\}$, as in Fig. \ref{fig:GN}. The scale of the root $r$ is set to be $h$. Note that if $v_{1}\prec v_{2}$, then $h_{v_{1}} < h_{v_{2}}$.
\item There is only one vertex immediately following the root, which will be denoted by $v_{0}$ and cannot be an endpoint. Its scale is $h_{v_{0}} = h+1$.
\item With each endpoint $v$ on scale $h_{v} = 0$ we associate one of the monomials contributing to $\widehat{V}^{(0)}(\sqrt{Z_{-1}}\psi^{(\leq 0)};\phi,A)$. Instead, with each leaf $v$ on scale $h_{v}<0$ we associate one of the monomials contributing to $\mathfrak{L} \widehat{V}^{(h_{v})}(\sqrt{Z_{h_{v} -1}} \psi^{(\leq h_{v})};\phi,A)$. In particular, a leaf $v\in V_{f}(\tau)$ will be called:
\begin{itemize}
\item[i.] of type $\lambda$ if it is associated with $\mathfrak L \widehat V^{(h_{v})}_{4,0,0}$;
\item[ii.] of type $\nu$ if it is associated with a $\nu_{h_{v}}$ coupling constant;
\item[iii.] of type $\phi$ if it is associated with a $\widehat B^{(h_{v})}$ term;
\item[iv.] of type $A$ if it is associated with a $\widehat C^{(h_{v})}$ term.
\end{itemize}
\item We introduce a {\it field label} $f$ to distinguish the field and external source variables appearing in the monomials associated with the endpoints. The set of field labels associated with the endpoint $v$ will be called $I_{v} = \{ f_{1}, \ldots, f_{|I_{v}|} \}$ and it will admit a partitioning $I_{v}=I^{\psi}_{v}\cup I^{\phi}_{v}\cup I^{A}_{v}$.

If $f$ labels a field $\psi$, we denote by $\underline{x}(f)$ its space-time position, by $\omega(f)$ its quasi-particle label, and by $\epsilon(f)$ its particle-hole label. We have an analogous labelling for the external sources $\phi,A$ associated to the endpoint $v$, so that the monomial associated to it is given by
\begin{equation}
\prod_{f\in I^{\psi}_{v}} \psi^{(\leq h_{v}),\epsilon(f)}_{\ul x(f),\omega(f)}\prod_{f\in I^{\phi}_{v}}\phi^{\epsilon(f)}_{\ul x(f),\rho(f)} \prod_{f\in I^{A}_{v}} A^{\nu(f)}_{\ul x(f)}\;.
\end{equation}
If $v$ is not an endpoint, we shall call $I_{v}$ the set of field labels associated with the endpoints following the vertex $v$, namely
\begin{equation}
I_{v} := \coprod_{\substack{v'\in V_{f}(\tau)\\ v\prec v'}} I_{v'}\;.
\end{equation}
\end{enumerate}

	\begin{figure}[H]
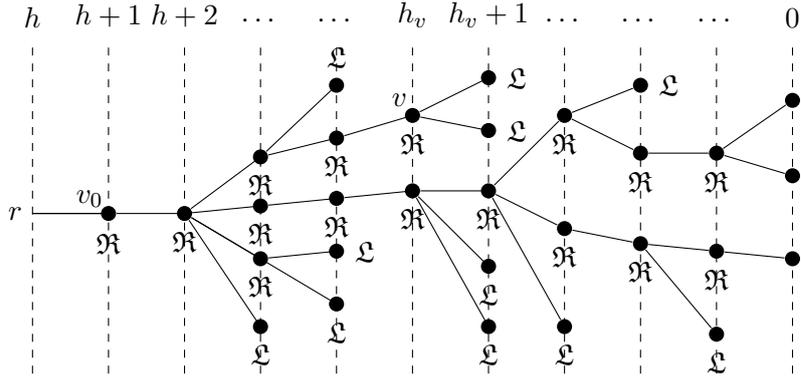

		\begin{center}
			\label{fig:GN}
			\begin{tabular}{rcl}
				&\tikz[baseline=-3pt]{ %ramo principale
				\coordinate (v001) at (-2,0) node[left] at (-2,0){$r$};
				\draw (v001) -- ++ (1,0);
					\draw(v001)--(-1,0)node[vertex](v00){};
					\draw(-0.8,0.2)node[label=left:{$v_0$}]{};
					\draw(-1,0)node[label=below:{$\mathfrak{R}$}]{};
					\draw(v00)--(0,0)node[vertex,label=below:{$\mathfrak{R}$}](v0){}--(1,0.1)node[vertex,label=below:{$\mathfrak{R}$}](v01){};
					\draw(v01)--(2,0.2)node[vertex,label=below:{$\mathfrak{R}$}](v02){};
					\draw(v02)--(3,0.3)node[vertex,label=below:{$\mathfrak{R}$}](v1){};
					\draw(v1)--(4,0.3)node[vertex,label=below:{$\mathfrak{R}$}](S){};
					%% prima scala >h_v+1
					\draw(5,1.3)node[vertex,label=below:{$\mathfrak{R}$}](S1){};
					\draw(5,-0.2)node[vertex,label=below:{$\mathfrak{R}$}](S2){};
					\draw(5,-1.5)node[vertex,label=below:{$\mathfrak{L}$}](S3){};
					\draw(S)--(S1);
					\draw(S)--(S2);
					\draw(S)--(S3);
					%%---------seconda scala >h_v+1
					\draw(6,1.7)node[vertex,label=right:{$\mathfrak{L}$}](T1){};
					\draw(6,0.8)node[vertex,label=below:{$\mathfrak{R}$}](T){};
					\draw(S1)--(T1);
					\draw(S1)--(T);
					\draw(6,-0.4)node[vertex,label=below:{$\mathfrak{R}$}](P){};
					\draw(S2)--(P);
					%%----terza scala >h_v+1
					\draw(7,0.8)node[vertex,label=below:{$\mathfrak{R}$}](B1){};
					\draw(T)--(B1);
					\draw(7,-0.5)node[vertex,label=below:{$\mathfrak{R}$}](D){};
					\draw(7,-1.6)node[vertex,label=below:{$\mathfrak{L}$}](D1){};
					\draw(P)--(D);
					\draw(P)--(D1);
					%% Scala 0
					\draw(8,1.5)node[vertex](G){};
					\draw(8,0.5)node[vertex](G1){};
					\draw(B1)--(G);
					\draw(B1)--(G1);
					\draw(8,-0.6)node[vertex](H){};
					\draw(D)--(H);
					%%--------verticali
					\draw[dashed](-2,2.2)node[label=above:{$h$}]{}--(-2,-2.2);
					\draw[dashed](-1,2.2)node[label=above:{$h+1$}]{}--(-1,-2.2);
					\draw[dashed](0,2.2)node[label=above:{$h+2$}]{}--(0,-2.2);
					\draw[dashed](1,2.2)node[label=above:{$\cdots$}]{}--(1,-2.2);
					\draw[dashed](2,2.2)node[label=above:{$\cdots$}]{}--(2,-2.2);
					\draw[dashed](3,2.2)node[label=above:{$h_v$}]{}--(3,-2.2);
					\draw[dashed](4,2.2)node[label=above:{$h_{v}+1$}]{}--(4,-2.2);
					\draw[dashed](5,2.2)node[label=above:{$\cdots$}]{}--(5,-2.2);
					\draw[dashed](6,2.2)node[label=above:{$\cdots$}]{}--(6,-2.2);
					\draw[dashed](7,2.2)node[label=above:{$\cdots$}]{}--(7,-2.2);
					\draw[dashed](8,2.2)node[label=above:{$0$}]{}--(8,-2.2);
					%%foglie e articulation intermedi
					\draw(v0)--(1,-0.6)node[vertex,label=below:{$\mathfrak{R}$}](v){};
					\draw(v)--(2,-0.5)node[vertex,label=right:{$\mathfrak{L}$}](v126){};
					\draw(v0)--(1,-1.5)node[vertex,label=below:{$\mathfrak{L}$}]{};
					\draw(v0)--(2,-1.2)node[vertex,label=below:{$\mathfrak{L}$}]{};
					\draw(v0)--(1,0.75)node[vertex,label=below:{$\mathfrak{R}$}](w){};
					\draw(w)--(2,1.7)node[vertex,label=above:{$\mathfrak{L}$}]{};
					\draw(w)--(2,1)node[vertex,label=below:{$\mathfrak{R}$}](w1){};
					\draw(3.2,1.5)node[label=left:{$v$}](w2){};
					\draw(w1)--(3,1.3)node[vertex,label=below:{$\mathfrak{R}$}](w2){};
					\draw(w2)--(4,1.8)node[vertex,label=right:{$\mathfrak{L}$}]{};
					\draw(w2)--(4,1.1)node[vertex,label=right:{$\mathfrak{L}$}]{};
					\draw(v1)--(4,-0.7)node[vertex,label=below:{$\mathfrak{L}$}]{};
					\draw(v1)--(4,-1.5)node[vertex,label=below:{$\mathfrak{L}$}]{};
				}
			\end{tabular}
		\end{center}
		\caption{Example of a Gallavotti-Nicolò tree: each endpoint is associated to a monomial contributing to $\mathfrak L\widehat V$, while each internal vertex carries an action of $\mathfrak R$.}
	\end{figure}
We shall denote as $\mathcal T^{\,m}_{h,n}$ the set of GN rooted, scale-labelled trees with root scale $h$, $n$ endpoints of type in $\{\lambda,\nu\}$, and $m$ leaves of type in $\{\phi,A\}$. Then, we may rewrite equation (\ref{eq:ET}) as a sum over GN trees, that is
\begin{equation} \label{eq:expGNtree}
\mathcal{W}^{( h)}_{\beta,L}(\phi,A) + V^{(h)}(\sqrt{Z_{h}}\psi^{(\leq h)}; \phi,A)= \mathcal W^{(h+1)}_{\beta,L}(\phi,A) +\sum_{n,m=0}^{\infty} \sum_{\tau\in \mathcal{T}^{\,m}_{h,n}}  V^{(h)}[\tau]\;;
\end{equation}
where the contribution subordinated to $\tau$ is defined iterating on the vertices: if $v_{1},\dots,v_{s_{v_{0}}}$ are the successors of $v_{0}$, and if we denote with $\tau_{v_{i}}$ the (unique) sub-tree with root $v_{0}$ and first vertex $v_{i}$, $i=1,\dots,s_{v_{0}}$, we have
\begin{equation} \label{eq:iter2}
V^{(h)}[\tau] = \frac{(-1)^{s_{v_{0}}+1}}{s_{v_{0}}!} \mathbb{E}_{(h+1)} \big[ \overline{V}^{(h+1)}[\tau_{v_{1}}]\,;\cdots\,; \overline{V}^{(h+1)}[ \tau_{v_{s_{v_{0}}}}] \big]\;,
\end{equation}
with
\begin{itemize}
\item[i)] $\overline{V}^{(h+1)}[\tau_{v_{i}}] = \mathfrak{L} \widehat{V}^{(h +1)}(\sqrt{Z_{h}} \psi^{(\leq h + 1)};\phi,A)$ if $\tau_{v_{i}}$ is trivial\footnote{A tree $\tau$ is said to be trivial if it contains only one edge (and hence only one leaf).} and $h+1 <0$;
\item[ii)] $\overline{V}^{(h+1)}[\tau_{v_{i}}] = \widehat{V}^{(0)}(\sqrt{Z_{-1}} \psi^{(\leq 0)})$ if $\tau_{v_{i}}$ is trivial and $h+1 =0$;
\item[iii)] $\overline{V}^{(h+1)}[\tau_{v_{i}}] = \mathfrak{R} \widehat{V}^{(h+1)}(\sqrt{Z_{h}} \psi^{(\leq h+1)})$ if $\tau_{v_{i}}$ is not trivial.
\end{itemize}
Using this inductive definition, the right-hand side of Eq. (\ref{eq:iter2}) can be further expanded, and in order to describe the resulting expansion we introduce the last ingredient to the labellings of our GN tree $\tau$.
\begin{itemize}
\item[6.] To any vertex $v\in V(\tau)$, we associate a subset $\Gamma^{\psi}_{v}\subseteq I^{\psi}_{v}$ of fields \emph{external} to the cluster $v$, with the following constraints:
\begin{itemize}
\item[i.] if $v\in V_{f}(\tau)$, then $\Gamma^{\psi}_{v}\equiv I^{\psi}_{v}$;
\item[ii.] if $v_{1},\dots,v_{s_{v}}$ are the vertices immediately following $v$ on $\tau$, then
\[
\Gamma^{\psi}_{v}\subseteq\bigcup_{i=1}^{s_{v}}\Gamma^{\psi}_{v_{i}}\;.
\]
\end{itemize}
We will denote by $Q_{v_{i}}:=\Gamma^{\psi}_{v}\cap \Gamma^{\psi}_{v_{i}}$ the set of fields external to $v$ subordinated on $v_{i}$; the union of the subsets $\Gamma^{\psi}_{v_{i}}\setminus Q_{v_{i}}$ will constitute the set of fields \emph{internal} to $v$, and each one of these subsets has to be non-empty if $s_{v}>1$.\\
For external sources, we simply set $\Gamma^{\phi}_{v}=I^{\phi}_{v}$ and $\Gamma^{A}_{v}=I^{A}_{v}$.
\end{itemize}
We shall denote with $\mathfrak L_{\tau}$ the set of all possible choices of external fields subordinated to the tree $\tau$, and we shall indicate its elements, \emph{i.e.} given labellings, with
\[
\mathbf \Gamma =(\Gamma_{v})_{v}\equiv (\Gamma^{\psi}_{v},\Gamma^{\phi}_{v},\Gamma^{A}_{v})_{v}\;.
\]
\begin{remark}
In order to better explain the outcome of the iteration of (\ref{eq:iter2}) by employment of the external fields labellings $\mathbf\Gamma$, let us suppose for a moment that $\mathfrak R=1$: we shall take its action into account in a second moment.
\end{remark}
Given an internal vertex $v\notin V_{f}(\tau)$, let us call $\tau_{v}$ the (unique) sub-tree having $v$ as first vertex, so that the scale of the root of $\tau_{v}$ is $h_{v}-1$; we set
\begin{equation} \label{eq:labelexp}
\begin{split}
V^{(h_{v}-1)}[\tau_{v}]&=\sum_{\mathbf \Gamma_{v}\in \mathcal L_{\tau_{v}}} V^{(h_{v}-1)}_{\mathbf \Gamma_{v}}[\tau_{v}]\\
V^{(h_{v}-1)}_{\mathbf \Gamma_{v}}[\tau_{v}] &= \int_{\beta,L} d\ul X_{\Gamma_{v}} \widetilde\psi^{(\leq h_{v}-1)}_{\Gamma^{\psi}_{v}} \phi_{\Gamma^{\phi}_{v}} A_{\Gamma^{A}_{v}}\, W^{(h_{v}-1)}_{\mathbf \Gamma_{v}}[\ul X_{\Gamma_{v}};\tau_{v}]
\end{split}
\end{equation}
where, as usual, $d\ul X_{\Gamma_{v}}=\prod_{f\in\Gamma_{v}} d\ul x(f)$ and
\begin{equation}\label{eq:iterkern}
\widetilde\psi^{(\leq h)}_{\Gamma} := \prod_{f\in\Gamma} \sqrt{Z_{h,\omega(f)}} \psi^{(\leq h),\epsilon(f)}_{\ul x(f),\omega(f)}\;,\quad \phi_{\Gamma} := \prod_{f\in \Gamma} \phi^{\epsilon(f)}_{\ul x(f),\rho(f)}\;,\quad A_{\Gamma} := \prod_{f\in \Gamma} A^{\nu(f)}_{\ul x(f)}\;.
\end{equation}
\begin{remark}
Strictly speaking, it should be noted that  the kernels $W^{(h_{v}-1)}_{\mathbf\Gamma_{v}}[\ul X_{\Gamma_{v}};\tau_{v}]$ that we have defined depend on all the coordinates $\ul X_{I_{v}}$, not just on the coordinates $\ul X_{\Gamma_{v}}$; with a slight abuse of notation, we explicit only the latter ones because we want to highlight the ones associated to fields external to $v$.
\end{remark}
We shall denote by $\mathbf \Gamma_{v}$ a labelling for the sub-tree $\tau_{v}$; hence $\mathbf \Gamma_{v_{0}}\equiv \mathbf \Gamma$. The recursion (\ref{eq:iter2}), together with the representation (\ref{eq:labelexp}), implies the following recursion for the kernels $W^{(h_{v}-1)}_{\mathbf \Gamma}$:
\begin{equation}
\begin{split}
W^{(h_{v}-1)}_{\mathbf \Gamma_{v}}[\ul X_{\Gamma_{v}};\tau_{v}] = \frac{(-1)^{s_{v}+1}}{s_{v}!} \Big[\prod_{f\in\Gamma^{\psi}_{v}}\sqrt{\frac{Z_{h_{v},\omega(f)}}{Z_{h_{v}-1,\omega(f)}}}\Big] \Big[\prod_{i=1}^{s_{v}} W^{(h_{v_{i}}-1)}_{\mathbf\Gamma_{v_{i}}}[\ul X_{\Gamma_{v_{i}}};\tau_{v_{i}}]\Big]\\
\cdot \mathbb E_{(h_{v})}\Big[\widetilde\psi^{(h_{v})}_{\Gamma^{\psi}_{v_{1}}\setminus Q_{v_{1}}};\cdots; \widetilde\psi^{(h_{v})}_{\Gamma^{\psi}_{v_{s_{v}}}\setminus Q_{v_{s_{v}}}}\Big]
\end{split}
\end{equation}
where $h_{v_{i}}=h_{v}+1$ by construction, and $\mathbf \Gamma_{v_{i}}$ is the restriction of $\mathbf\Gamma_{v}$ to the sub-tree $\tau_{v_{i}}$. This iteration stops when $v$ is a leaf, in which case we will have that $W^{(h_{v}-1)}_{\mathbf\Gamma_{v}}\equiv W^{(h_{v}-1)}_{\Gamma_{v}}$ is equal to one the kernels of the monomials contributing to $\mathfrak L \widehat V^{(h_{v})}(\sqrt{Z_{h_{v}-1}}\psi^{(\leq h_{v})})$; in momentum space, we get
\begin{equation}
\widehat W^{(h_{v}-1)}_{\mathbf \Gamma_{v}}[\ul K_{\Gamma_{v}};\tau_{v}]=
\begin{cases}
\lambda_{h_{v},\omega\omega'} & v\ \lambda\mathrm{-type},\ h_{v}<0\\
\nu_{h_{v},\omega} & v\ \nu\mathrm{-type},\ h_{v}<0\\
\mathsf Q^{\pm}_{h_{v},\omega}(k') & v\ \phi^{\pm}\mathrm{-type},\ h_{v}<0\\
\mathsf Z^{\mu}_{h_{v},\omega\omega'} & v\ A^{\mu}\mathrm{-type},\ h_{v}<0\\
\widehat W^{(0)}_{\Gamma_{v}}[\ul K_{\Gamma_{v}}] & h_{v}=0\;.
\end{cases}
\end{equation}
where as usual $k'=k-k^{\omega}_{F}(\lambda)$, and where $W^{(0)}_{\Gamma}$ indicates a kernel appearing in the monomial appearing in $\widehat V^{(0)}(\sqrt{Z_{-1}}\psi^{(\leq0)};\phi,A)$ parameterized by $\Gamma$.

Therefore, noticing that $\tau_{v_{0}}=\tau$, and writing the overall monomial associated to a given choice of $(\tau,\mathbf \Gamma)$, namely
\begin{equation}\label{eq:VGammatau}
V^{(h)}_{\mathbf \Gamma}[\tau] = \int_{\beta,L} d\ul X_{\Gamma_{v_{0}}} \widetilde\psi^{(\leq h)}_{\Gamma^{\psi}_{v_{0}}} \phi_{\Gamma^{\phi}_{v_{0}}} A_{\Gamma^{A}_{v_{0}}}\, W^{(h)}_{\mathbf \Gamma_{v_{0}}}[\ul X_{\Gamma_{v_{0}}};\tau]\;,
\end{equation}
we get
\begin{equation} \label{eq:kernexptree}
W^{(h)}_{\mathbf \Gamma_{v}}[\ul X_{\Gamma_{v_{0}}};\tau] = (-1)^{n+m}L[\tau]\prod_{v\notin V_{f}(\tau)}\frac{1}{s_{v}!}\sqrt{\frac{Z_{h_{v}}(\Gamma^{\psi}_{v})}{Z_{h_{v}-1}(\Gamma^{\psi}_{v})}}
\mathbb E_{(h_{v})}\Big[\widetilde\psi^{(h_{v})}_{\Gamma^{\psi}_{v_{1}}\setminus Q_{v_{1}}};\cdots; \widetilde\psi^{(h_{v})}_{\Gamma^{\psi}_{v_{s_{v}}}\setminus Q_{v_{s_{v}}}}\Big]
\end{equation}
with the short-hand $Z_{h}(\Gamma):=\prod_{f\in\Gamma}Z_{h,\omega(f)}$ and with $L[\tau]$ taking into account the contribution of the leaves, namely
\begin{equation}
L[\tau]=\prod_{v\in V_{f}(\tau)} W^{(h_{v}-1)}_{\Gamma_{v}}[\ul X_{v};\tau_{v}]\;.
\end{equation}
Comparing (\ref{eq:Vh}) with (\ref{eq:expGNtree}) and (\ref{eq:VGammatau}), we have thus obtained the following GN representation for kernels $W^{(h)}_{\Gamma}$ appearing in $\mathcal W^{(h)}_{\beta,L}$ and $V^{(h)}$:
\begin{equation}
W^{(h)}_{\Gamma}(\ul X) = \sum_{n=1}^{\infty}\sum_{\tau\in\mathcal T^{m}_{h,n}}\sum_{\substack{\mathbf\Gamma\in\mathcal L_{\tau}\\ \Gamma_{v_{0}}=\Gamma}} \int_{\beta,L} d\ul X_{I_{v_{0}}\setminus\Gamma_{v_{0}}} W^{(h)}_{\mathbf\Gamma}[\ul X;\tau]\;.
\end{equation}
This is useful because it allows to prove bounds analogous to (\ref{eq:locboundIR}) for the kernels $W^{(h)}_{\Gamma}$, with an explicit $h$-dependence which is uniform in $\beta,L$, as we now illustrate.
\begin{remark}\label{rmk:BBF}
The \emph{Battle-Brydges-Federbush (BBF) formula}, see \cite{Les, Bry} and references therein, and \cite{GM} for a review, allows to express a connected expectation between clusters as a sum over \emph{spanning} trees over these. Given clusters $\{\Gamma_{i}\}_{i=1}^{s}$ such that $\sum_{i}|\Gamma_{i}|=2N$, we have that
\begin{equation}\label{eq:BBF}
\mathbb{E} [\Psi_{\Gamma_{1}};\cdots;\Psi_{\Gamma_{s}}] = \sum_{T} \sigma_{T}\Big[\prod_{\ell\in T} g_{\ell}\Big] \int d\mu_{T}(\mathbf t)\, \det[G_{T}(\mathbf t)]
\end{equation}
where:
\begin{itemize}
\item the sum is over so-called anchored (or spanning) trees $T$, \textit{i.e.} collections of edges in $E(\bigcup_{i}\Gamma_{i})$ that become a tree if each $\Gamma_{i}$ is collapsed to a point $i$; furthermore, $\sigma_{T}\in{\pm 1}$ is a sign associated to $T$: it will not play any role in the analysis;
\item the product is over the oriented edges $\ell$ composing the anchored tree $T$, the factor being the propagator $g$ evaluated on it;
\item the integration is over parameters $\mathbf t =(t_{ij})_{i,j=1}^{s}$, the probability measure $\mu_{T}$ is supported on a set of $\bf t$ such that $t_{ij}=\vec u_{i}\cdot \vec u_{j}$, for some family of unit-norm vectors $(\vec u_{i})_{i=1}^{s}\subset\mathbb R^{s}$;
\item $G_{T}(\mathbf t)\in M_{N-s+1}(\mathbb C)$ is a matrix whose entries are given by
\[
(G_{T}(\mathbf t))_{(i,a)(j,b)}=t_{ij}\, g(\ul x(i,a) - \ul x(j,b))
\]
where $i,j$ parameterize the clusters and $a$ (resp. $b$) a half-edge in $i$ (resp. $j$), such that the oriented edge $((i,a),(j,b))$ does not lie in $T$.
\end{itemize}
\end{remark}
Employing the BBF formula for the expectations taken at each internal vertex $v$ of the GN tree $\tau$, we may rewrite equation (\ref{eq:kernexptree}) as
\begin{equation}
\begin{split}
&|W^{(h)}_{\bf \Gamma}[\ul X_{\Gamma_{v_{0}}};\tau]| \\
&\quad = \left|L[\tau]\prod_{v\notin V_{f}(\tau)}\frac{1}{s_{v}!}\sqrt{\frac{Z_{h_{v}}(\Gamma^{\psi}_{v})}{Z_{h_{v}-1}(\Gamma^{\psi}_{v})}} \sum_{T_{v}} \Big[\prod_{\ell\in T_{v}} Z_{h_{v}-1, \omega_{\ell}}g^{(h_{v})}_{\ell}\Big] \int d\mu_{T_{v}}(\mathbf t)\, \det[G^{(h_{v})}_{T_{v}}(\mathbf t)]\right|\;.
\end{split}
\end{equation}
Now, recall the bounds on the effective couplings listed in Proposition \ref{prp:flow}:
\begin{equation} \label{eq:ass}
\left|\frac{Z_{h,\omega}}{Z_{h-1,\omega}}\right|\leq e^{c|\lambda|}\;,\qquad |v_{h-1,\omega} - v_{\omega}|\leq C|\lambda|\;,\qquad |\lambda_{h,\omega\omega'}| \leq C|\lambda|\;,\qquad |\nu_{h, \omega}|\leq C|\lambda|\;.\nonumber
\end{equation}
These bounds can be used to prove, following {\it e.g.} \cite{BM}, Section 3.14:
\begin{eqnarray}
\label{eq:bdGg}
\| \det G^{(h_{v})}_{T_{v}} ({\bf t}) \|_{\infty} &\leq& C^{\sum_{i=1}^{s_{v}} \frac{|\Gamma^{\psi}_{v_{i}}|}{2} - \frac{|\Gamma^{\psi}_{v}|}{2} - (s_{v} - 1) } 2^{h_{v}\big[\sum_{i=1}^{s_{v}} \frac{|\Gamma^{\psi}_{v_{i}}|}{2} - \frac{|\Gamma^{\psi}_{v}|}{2} - (s_{v} - 1)\big]}\nonumber\\
| Z_{h-1, \omega} \partial_{0}^{n_0} \mathrm d_{1}^{n_1} g_{\omega}^{(h)}(\underline{x} - \underline{y}) | &\leq& \frac{ C_{n+ n_{0} + n_{1}} 2^{h(1 + n_0 + n_1)} }{1 + (2^{h} \| \underline{x} - \underline{y} \|_{\beta, L} )^{n}}\;,\qquad \forall n\in \mathbb{N}\;.
\end{eqnarray}
The first bound follows from the Gram-Hadamard inequality for determinants (using that the single-scale propagator admits a Gram representation, Eq. (3.97) of \cite{BM}), while the second follows from the smoothness and  support properties of the single-scale propagator in momentum space: in particular, it implies $\|g^{(h)}\|_{1}\leq c\, 2^{-h}$. Let $V^{od}_{f}(\tau) \subseteq V_{f}(\tau)$ be the set of endpoints of the tree associated with $\mathsf Z^{\mu_{v}}_{h_{v},\omega_{v}\omega'_{v}}$ with $\omega_{v} \neq \omega_{v}'$. Using the fact that
\[
\sum_{T_{v}} 1 \leq s_{v}!\,C^{\sum_{i=1}^{s_{v}} \frac{|\Gamma^{\psi}_{v_{i}}|}{2} - \frac{|\Gamma^{\psi}_{v}|}{2} - (s_{v} - 1)}\;,
\]
and defining the intensive $L^{1}$-norm
\begin{equation}
\|W^{(h)}_{\mathbf\Gamma}[\tau]\|_{1}:= \frac{1}{\beta L} \int_{\beta,L} d\ul X_{I_{v_{0}}} |W^{(h)}_{\mathbf\Gamma}[\ul X_{\Gamma_{v_{0}}};\tau]|\;,
\end{equation}
the bounds (\ref{eq:bdGg}) allow to prove (we refer to {\it e.g.} \cite{GM} for the details) that
\begin{equation}
\begin{split}
\|W^{(h)}_{\mathbf\Gamma}[\tau]\|_{1} &\leq C^{n+m} |\lambda|^{n} \prod_{v \in V^{od}_{f}(\tau)} | \mathsf Z^{\mu_{v}}_{h_{v},\omega_{v}\omega'_{v}} | \\
&\quad \cdot  2^{h\big(2-\frac{|\Gamma^{\psi}_{v_{0}}|+3|\Gamma^{\phi}_{v_{0}}|}{2}-|\Gamma^{A}_{v_{0}}|\big)} \prod_{v\notin V_{f}(\tau)} 2^{(h_{v}-h_{v'})\big(2-\frac{|\Gamma^{\psi}_{v}|+3|\Gamma^{\phi}_{v}|}{2}-|\Gamma^{A}_{v}| + c|\lambda|\big)}\;,
\end{split}
\end{equation}
where $v'$ is the vertex immediately preceding $v$ on the tree $\tau$, and hence $h_{v} - h_{v'} >0$ (in fact, $h_{v} - h_{v'} = 1$).
This bound however {\it does not} imply summability in the scale labels of the nontrivial vertices for all ${\bf \Gamma}$, since the factor multiplying $h_{v}-h_{v'}$ may be positive; this happens when
\begin{equation}
(|\Gamma^{\psi}_{v}|, |\Gamma^{\phi}_{v}|, |\Gamma^{A}_{v}|)=
\begin{cases}
(2,0,0)\\
(4,0,0)\\
(2,0,1)\;.
\end{cases}
\end{equation}
If $(v,\Gamma_{v})$ is such that the first case is true, then it is called a {\it relevant} cluster; for the second or third case, it is called {\it marginal}; otherwise it is called {\it irrelevant}.
\begin{remark}
Another potentially dangerous case would be $(|\Gamma^{\psi}_{v}|, |\Gamma^{\phi}_{v}|, |\Gamma^{A}_{v}|)=(1,1,0)$, but this cannot occur in an internal vertex, since it contributes to the local part, and hence it contributes only to the leaves. Moreover, cases with $|\Gamma^{\psi}_{v}|=0$ also cannot occur in internal vertices other than $v_{0}$.
\end{remark}
We take care of this apparent non-summability by reintroducing the action of the operator $\mathfrak R$ in each single-scale expectation. Equation (\ref{eq:labelexp}) becomes
\begin{equation} \label{eq:labelexpren}
V^{(h_{v}-1)}_{\mathbf \Gamma_{v}}[\tau_{v}] = \int_{\beta,L} d\ul X_{\Gamma_{v}} \widetilde\psi^{(\leq h_{v}-1)}_{\Gamma^{\psi}_{v}} \phi_{\Gamma^{\phi}_{v}} A_{\Gamma^{A}_{v}}\, \mathfrak RW^{(h_{v}-1)}_{\mathbf \Gamma_{v}}[\ul X_{\Gamma_{v}};\tau_{v}]
\end{equation}
for all internal vertices $v$. Then, proceeding as in Section 3.14 of \cite{BM}, one may show that $W^{(h)}_{\mathbf\Gamma}[\tau]$ admits the improved bound
\begin{equation} \label{eq:treebd}
\|W^{(h)}_{\mathbf\Gamma}[\tau]\|_{1}\leq C^{n+m} |\lambda|^{n} \Big[\prod_{v \in V^{od}_{f}(\tau)} |\mathsf Z^{\mu_{v}}_{h_{v},\omega_{v}\omega'_{v}} |\Big] 2^{h D_{v_{0}}} \prod_{v\notin V_{f}(\tau)} 2^{(h_{v}-h_{v'})D_{v}}
\end{equation}
with $D_{v}$ being the (renormalized) \emph{scaling dimension} of cluster $(v,\Gamma_{v})$, defined as
\begin{equation}\label{eq:dvir}
D_{v}:= 2 - \frac{|\Gamma^{\psi}_{v}|+3|\Gamma^{\phi}_{v}|}{2} - |\Gamma^{A}_{v}| - z_{v} + c|\lambda|
\end{equation}
with the improvement $z_{v}$ given by
\[
z_{v}=\begin{cases}
2 & \text{if}\ (|\Gamma_{v}^{\psi}|, |\Gamma_{v}^{\phi}|, |\Gamma_{v}^{A}|)=(2,0,0)\\
1 & \text{if}\ (|\Gamma_{v}^{\psi}|, |\Gamma_{v}^{\phi}|, |\Gamma_{v}^{A}|)=(4,0,0)\\
1 & \text{if}\ (|\Gamma_{v}^{\psi}|, |\Gamma_{v}^{\phi}|, |\Gamma_{v}^{A}|)=(2,0,1)\\
0 & \text{otherwise}\;.
\end{cases}
\]
This ensures that, for $|\lambda|$ small enough, $D_{v}<0$ for all internal vertices $v$, so that summation over $(n,\tau,\mathbf\Gamma)$ yields (see \cite{GM} for details, specifically Lemma A3 and Section A6.1)
\begin{equation}
\|W^{(h)}_{\Gamma}\|_{1}\leq C^{m} 2^{h\big(2-\frac{|\Gamma^{\psi}|+3|\Gamma^{\phi}|}{2}-|\Gamma^{A}| - c|\lambda| |\Gamma^{A}| \big)}\sum_{n\geq0} \tilde C^{n}|\lambda|^{n}
\end{equation}
for some constants $C,\tilde C$ uniform in $\beta,L$, and with $m=|\Gamma^{\phi}|+|\Gamma^{A}|$.
In a similar way, one can prove that:
\begin{equation} \label{eq:bdWh3}
\begin{split}
&\frac{1}{\beta L} \int_{\beta, L} d \underline{X}\, \prod_{i,j\in\Gamma}\| \underline{x}_{i} - \underline{x}_{j} \|_{\beta,L}^{n_{ij}} | W^{(h)}_{\Gamma}(\underline{X})| \\&\qquad \leq C_{\{n_{ij}\}}^{m} 2^{h\big(2-\frac{|\Gamma^{\psi}|+3|\Gamma^{\phi}|}{2}-|\Gamma^{A}| - c|\lambda| |\Gamma^{A}| - \sum_{i,j} n_{ij}\big)}\sum_{n\geq0} \tilde C^{n}|\lambda|^{n}\;.
\end{split}
\end{equation}
Analyticity in $\lambda$ follows from analyticity of all the effective couplings on scales $\geq h$, and from the convergence of the GN tree expansion. Finally, the bound (\ref{eq:bdWh3}), combined with equation (\ref{eq:Qh}), also allows to prove that $\mathsf Q^{+}_{h, \omega}(\underline{k}') = \mathsf Q^{+}_{1,\omega}(\underline{k}')(1 + O(\lambda))$, as anticipated at the end of Section \ref{sec:flow}.
We conclude this section with two important properties of GN trees, both of which will be employed in the next sections.
\begin{remark}{(On the external dimension).} The term $-c|\lambda| |\Gamma^{A}|$ in the estimate (\ref{eq:bdWh3}) takes into account the (pessimistic) situation in which all $A$-endpoints are associated with off-diagonal vertex renormalizations, which diverge with an anomalous exponent; recall (\ref{eq:BdsEffCoup}). In the application that follows we will be interested in bounding momentum-space correlations where the choice of the external momenta will imply that at most one $A$-endpoint is off-diagonal.
\end{remark}
\begin{remark}{{(The short memory property).}}\label{rem:sm}
Thanks to the presence of the $z_{v}$ factors in equation (\ref{eq:treebd}), each branch between two non-trivial vertices $v'\prec v$ comes with a factor $2^{(h_{v} - h_{v'})D_{v}}$, where $D_{v}\leq D = -1 + c|\lambda|$. This implies that {\em long trees are exponentially suppressed}: if one restricts the sum $\sum_{\tau\in \mathcal{T}_{h, n}}$ to trees having at least one vertex on scale $k>h$, then the final bound is improved by a factor $2^{a(h - k)}$, for $a \in (0, |D|)$. Clearly, the constants appearing will depend on the choice of $a$. This property is usually referred to as the {\em short memory property of GN trees}.
\end{remark}
\begin{remark}{{(The continuity property).}}\label{rem:cont} Consider a tree $\tau \in \mathcal{T}^{m}_{h, n}$ and denote its value with $\mathrm{Val}(\tau)$. Suppose that $|\mathrm{Val}(\tau)|\leq B_{D}(\tau)$, where $B_{D}(\tau)$ is the dimensional bound of the tree, obtained as in (\ref{eq:treebd}), with $D_{v}\leq D=-1 + c|\lambda|$ for all $v$. Consider a propagator $\hat g^{(h_{v})}_{\omega}$ arising in the truncated expectation associated to a given vertex $v\in \tau$. Suppose that $\hat g^{(h_{v})}_{\omega}(\underline{k}')$ is replaced by $\hat g^{(h_{v})}_{\omega} + \delta \hat g^{(h_{v})}_{\omega}$, with $|\delta \hat g^{(h_{v})}_{\omega}(\underline{k}')| \leq C2^{-h_{v}} 2^{a h_{v}}$ for some $1/2>a >0$, and admitting a Gram representation. Let $\mathrm{Val}(\tau) + \delta \mathrm{Val}(\tau)$ be the new value of the tree. The contribution $\delta \mathrm{Val}(\tau)$ can be bounded as follows. Let $\gamma_{r\to v}$ be the path on $\tau$ that connects the root $r$ to $v$. Then, writing $2^{a h_{v}} = 2^{a h} \prod_{v\in \gamma_{r\to v}} 2^{a}$ and using that $D_{v} \leq D$, we have:
\begin{equation}
\label{eq:dV}
|\delta\mathrm{Val}(\tau)|\leq 2^{a h} B_{D + a}(\tau)\;.
\end{equation}
in a similar fashion, suppose that $Z_{h_{v},\omega}/ Z_{h_{v} - 1,\omega}$ is replaced by $Z_{h_{v}, \omega}/ Z_{h_{v} - 1, \omega} + \delta z_{h_{v}, 0,\omega}$, with $|\delta z_{h_{v}, 0,\omega}|\leq C2^{a h_{v}}$, or suppose that the effective coupling $\lambda_{h_{v}, \omega\omega'}$ is replaced by $\lambda_{h_{v}, \omega\omega'} + \delta \lambda_{h_{v}, \omega\omega'}$, with $|\delta\lambda_{h_{v}, \omega\omega'}|\leq C 2^{a h_{v}}$. Then, the new value is $\mathrm{Val}(\tau) + \delta \mathrm{Val}(\tau)$, with $\delta \mathrm{Val}(\tau)$ again satisfying the bound (\ref{eq:dV}). We shall refer to this property as the {\em continuity property of GN trees}.
\end{remark}

Finally, the representation of the effective potential can be made even more explicit by using the tree expansion for the effective potential on scale $0$. This tree expansion is similar to the one described above, but much simpler (the ultraviolet regime is superrenormalizable). It is completely standard and largely model independent, see {\it e.g.} \cite{GMPcond, GMP}. Since we will use this more explicit representation later on, let us briefly describe it. 

In the ultraviolet regime, the single-scale propagator $\hat g^{(h)}$ is supported for momenta such that $2^{h-1} \leq |k_{0}| \leq 2^{h+1}$, $h=0,1,\ldots, N_{0}$. In configuration space, it satisfies the estimate:
\begin{equation}
\| g^{(h)}(\underline{x} - \underline{y}) \| \leq \frac{C_{m}}{1 + 2^{h m}|x_{0} - y_{0}|^{m}_{\beta} + |x-y|_{L}^{m}}\;.
\end{equation}
The single-scale propagator is not renormalized; in particular, there is no rescaling of the fields, in contrast to the infrared regime. The trees have endpoints on scale $\leq N_{0}$. The endpoints on scale $N_{0}$ correspond to the bare quartic interaction or to the bare source terms. The endpoints on scale $h<N_{0}$ have $|\Gamma_{v}| = |\Gamma^{\psi}_{v}| = 2$; they can actually be ``unfolded'' in trivial trees with root on scale $h$ and endpoint on scale $N_{0}$. The value of the corresponding kernel is the tadpole graph, evaluated with the propagator $g^{[h,N_{0}]}:=\sum_{h'=h}^{N_{0}} g^{(h')}$; it is finite uniformly in $h\geq 0$ and $N_{0}$ and of order $\lambda$. Finally, the scaling dimension associated with a vertex $v$ that is not an endpoint is:
\begin{equation}\label{eq:dvuv}
D_{v} = (1 - n_{v})
\end{equation}
where $n_{v}$ is the number of endpoints following the vertex $v$ on the tree. This can be written as $n_{v} = n_{v}^{\psi} + n_{v}^{A}$, where $n_{v}^{A} = |\Gamma_{v}^{A}|$ and $n_{v}^{\psi}$ takes into account the other two types of endpoints (quadratic or quartic in the fermions, and independent of $A$). Observe that by construction $D_{v} \leq -1$: the case $n_{v} = 1$ and $|\Gamma^{A}_{v}| = 0$ cannot take place, since the trivial subtrees are taken into account in the definition of endpoints. With these ingredients, it is clear that the estimate (\ref{eq:treebd}) can be replaced by a similar bound, for trees that are now extended all the way until the scale $N_{0}$, and with scaling dimension equal to (\ref{eq:dvuv}) if the vertex is on scale $h_{v} \geq 0$, or equal to (\ref{eq:dvir}) if the vertex is on scale $h_{v} < 0$. In the following, we shall denote by $\widetilde{\mathcal T}_{h,n}^{m}$ the set of trees that extend in the ultraviolet regimes, with root on scale $h$, with $n$ endpoints independent of $A$, and with $m$ external source terms.

\subsection{The correlation functions} \label{sec:dimest}
Thanks to the previous section, we may write the GN expansion for the generating functional $\mathcal W_{\beta,L}$. Indeed, we obtain
\begin{equation}
\begin{split}
\mathcal W_{\beta,L}(\phi,A) &= \mathcal W^{\text{(u.v.)}}_{\beta,L}(\phi,A) + \sum_{h=h_{\beta}}^{0} \sum_{n,m\geq 0}\sum_{\tau\in\mathcal T^{m}_{h,n}} \sum_{\substack{\mathbf\Gamma\in\mathcal L_{\tau}\\ \Gamma^{\psi}_{v_{0}}=\varnothing}} \mathcal W^{(h)}_{\mathbf\Gamma}[\phi,A;\tau]\\
\mathcal W^{(h)}_{\mathbf\Gamma}[\phi,A;\tau] &= \int_{\beta,L} d\ul X_{I_{v_{0}}}\, \phi_{\Gamma^{\phi}_{v_{0}}} A_{\Gamma^{A}_{v_{0}}} W^{(h)}_{\mathbf\Gamma}[\ul X_{\Gamma_{v_{0}}};\tau]\;,
\end{split}
\end{equation}
where $\mathcal W^{\text{(u.v.)}}_{\beta,L}(\phi,A)$ collects the contribution of trees with root on positive scales (purely ultraviolet contributions). In particular, since
\begin{equation}
\langle\mathbf T\hat n_{\ul p_{1}};\,\cdots; \hat n_{\ul p_{m-1}}; \hat \jmath_{\mu,\ul p_{m}} \rangle_{\beta,L} = (\beta L)^{m}\left.\frac{\partial^{m}\mathcal{W}_{\beta,L}(\phi,A)}{\partial \hat A^{0}_{-\ul p_{1}}\cdots\,\partial \hat A^{0}_{-\ul p_{m-1}}\partial \hat A^{\mu}_{-\ul p_{m}}}\right|_{A=\phi=0}\;,
\end{equation}
we have:
\begin{equation}
\langle \mathbf T\hat n_{\ul p_{1}};\cdots; \hat \jmath_{\mu,\ul p_{m}}\rangle_{\beta,L} = \langle \mathbf T\hat n_{\ul p_{1}};\cdots; \hat \jmath_{\mu,\ul p_{m}}\rangle_{\beta,L}^{(\text{u.v.})} + \langle \mathbf T\hat n_{\ul p_{1}};\cdots; \hat \jmath_{\mu,\ul p_{m}}\rangle_{\beta,L}^{(\text{i.r.})}
\end{equation}
where the first term takes into account the purely ultraviolet contributions, and the second term all the other trees:
\begin{equation}\label{eq:GNexp}
\langle \mathbf T\hat n_{\ul p_{1}};\cdots; \hat \jmath_{\mu,\ul p_{m}}\rangle_{\beta,L}^{(\text{i.r.})} = \sum_{h=h_{\beta}}^{0}\sum_{n=0}^{\infty}\sum_{\tau\in \mathcal T_{h,n}^{m}} \sum_{\mathbf\Gamma\in \mathcal L_{\tau}} \widehat W^{(h)}_{\bf \Gamma}[\tau]\;,
\end{equation}
where:
\begin{equation} \label{eq:fourierkernel}
\widehat W^{(h)}_{\bf \Gamma}[\tau] := \int_{\beta,L} d\ul X_{I_{v_{0}}}\, W^{(h)}_{\mathbf\Gamma}[\ul X_{\Gamma_{v_{0}}};\tau] \prod_{\substack{v\in V_{f}(\tau)\\ v\ A\text{-type}}} e^{-i\ul p_{v}\cdot \ul x_{v}}\;.
\end{equation}
In equation (\ref{eq:GNexp}), and until the end of this section, the set $\mathcal T_{h,n}^{m}$ denotes
\[
\mathcal T_{h,n}^{m}=\{\tau\ |\ \text{$\tau$ has: root scale}\ h;\ n\ \lambda, \nu\text{-leaves,}\ 0\ \phi\text{-leaves,}\ m-1\ A^{0}\text{-leaves,}\ 1\ A^{\mu}\text{-leaf}\}
\]
and the set $\mathcal L_{\tau}$ parameterizes labellings of $\tau$ with no fermionic fields $\psi$ at the first vertex $v_{0}\in V(\tau)$ and no external fields $\phi$, namely
\[
\mathcal L_{\tau} := \{\mathbf\Gamma\ \text{labelling of}\ \tau\ |\ \Gamma_{v_{0}}\equiv\Gamma_{v_{0}}^{A}\}\;.
\]
Notice that the labelling $\mathbf \Gamma$ is such that an $A$-type leaf $v$, having label $-\ul p_{v}$, carries a factor $e^{-i\ul p_{v}\cdot \ul x_{v}}$ ($\ul x_{v}$ being the integration coordinate associated to $v$), which allows to see (\ref{eq:fourierkernel}) as a Fourier transform.
\begin{remark}
Recall that $\ul p_{i} = (\eta_{\beta}, p_{i})$ for $i=1,\dots,m-1$, and that $\sum_{i=1}^{m}\ul p_{i}=0$, hence $\|\ul p_{j}\|\geq \eta\geq 2^{h_{\eta}}$ for all $j=1,\dots, m$. Here $h_{\eta}:=\left\lfloor\log_{2}\eta_{\beta}\right\rfloor$, but for $\beta$ large one may think of it as essentially equal to $\left\lfloor\log_{2}\eta\right\rfloor$. In what follows, we shall use this identification freely.
\end{remark}
\subsubsection{Bound for the correlation functions}
We wish to find an $L^{\infty}$-estimate for $\widehat W^{(h)}_{\bf \Gamma}[\tau]$ in terms of $\eta$, summable in $n, h, \tau$ and $\mathbf\Gamma$.
\begin{proposition}[Bound for density correlators] \label{prop:CorrDimBound}
Let $m\geq 3$, and let $\ul p_{i}=(\eta_{\beta}, p_{i})$ for all $i=1,\dots, m-1$, with $|p_{i}| \leq B|\theta|$. Then, there exists positive constants $C, c, \varepsilon$ such that, for $\beta,L$ large enough, for any $\mu\in\{0,1\}$:
\begin{equation}\label{eq:estcorr}
\frac{1}{\beta L}|\langle \mathbf T\hat n_{\ul p_{1}};\cdots; \hat n_{\ul p_{m-1}}; \hat\jmath_{\mu,\ul p_{m}}\rangle_{\beta,L}|\leq m!\, C^{m} \eta^{2-m} (1 + \mathbbm 1_{m\geq m(\theta)} |\theta|^{\varepsilon/2 + c|\lambda|} \eta^{-c|\lambda|})\;,
\end{equation}
with $m(\theta) = K |\theta|^{-1}$ and $K = \kappa \min_{\omega,\omega':\, \omega\neq \omega'} |k_{F}^{\omega} - k_{F}^{\omega'}|$ for some $\kappa > 0$.
\end{proposition}
\begin{remark}
For $|\theta| \leq K \eta$, as assumed in Theorem \ref{thm:main}, and choosing $\eta$ small enough, the last term in (\ref{eq:estcorr}) is subleading, and the bound reads:
\begin{equation}
\frac{1}{\beta L}|\langle \mathbf T\hat n_{\ul p_{1}};\cdots; \hat n_{\ul p_{m-1}}; \hat\jmath_{\mu,\ul p_{m}}\rangle_{\beta,L}|\leq m!\, C^{m} \eta^{2-m}\;.
\end{equation}
\end{remark}
The proof of the proposition will use the following lemma.
\begin{lemma}[Existence of a lower bound for the scale of a cluster with few external fermions] \label{lemma:boundedscale}
Let $\tau$ be a tree with $j>0$ A-type leaves carrying momenta $\|\ul p_{i}\|\geq\eta_{\beta}$ and such that $\left|\sum_{i=1}^{j} p_{i,0}\right|\geq \eta_{\beta}$. Consider a vertex $v$ between the root and the last vertex with $|\Gamma^{A}_{v}|=j$. Then, if $|\Gamma^{\psi}_{v}|=2,4$, we have that $h_{v}\geq h_{\eta} - 4$.
\end{lemma}
\begin{proof}
Suppose $|\Gamma_{v}^{\psi}|=2$, then by momentum conservation the monomial associated to $v$ has the general form
\[
\widehat W_{\omega_{v},\omega'_{v}}^{(h_{v})}[\ul k,\ul p_{1},\dots, \ul p_{j};\tau_{v}]\, \hat\psi^{(\leq h_{v}),+}_{\ul k-\ul k^{\omega_{v}}_{F}-\ul p_{\text{tot}},\omega_{v}}\hat\psi^{(\leq h_{v}),-}_{\ul k-\ul k^{\omega'_{v}}_{F},\omega'_{v}}
\]
with $\tau_{v}$ being the sub-tree with first vertex $v$, and where $\ul p_{\text{tot}} = \sum_{i = 1}^{j}\ul p_{i}$. By construction, we have that 
\[
\|\ul k-\ul k^{\omega'_{v}}_{F}\|,\|\ul k-\ul k^{\omega_{v}}_{F}-\ul p_{\text{tot}}\|\leq 2^{h_{v}+1}\;,
\]
and using $\|\ul p_{\text{tot}}-\ul k^{\omega_{v}}_{F}+\ul k^{\omega'_{v}}_{F}\|\geq |p_{\text{tot},0}|\geq 2^{h_{\eta}}$ we obtain
\[
2^{h_{\eta}}-2^{h_{v}+1}\leq\|\ul k\|\leq 2^{h_{v}+1}\qquad \Longrightarrow\qquad h_{v}\geq h_{\eta}-2.
\]
If $|\Gamma_{v}^{\psi}|=4$, we have
\[
 \widehat W_{\ul\omega_{v}}^{(h_{v})}[\ul k,\ul k',\ul q,\ul p_{1},\dots, \ul p_{j};\tau_{v}]\, \hat\psi^{(\leq h_{v}),+}_{\ul k-\ul k^{\omega_{v}}_{F}-\ul p_{\text{tot}},\omega_{v}}\hat\psi^{(\leq h_{v}),-}_{\ul k-\ul k^{\omega'_{v}}_{F}+\ul q,\omega'_{v}}\hat\psi^{(\leq h_{v}),+}_{\ul k'-\ul k^{\omega''_{v}}_{F},\omega''_{v}}\hat\psi^{(\leq h_{v}),-}_{\ul k'-\ul k^{\omega'''_{v}}_{F}-\ul q,\omega'''_{v}}
\]
and reasoning in an analogous way we obtain $\|\ul q\|\geq 2^{h_{\eta}}- 2^{h_{v}+2}$; the computation above implies then that $h_{v}\geq h_{\eta} - 4$.
\end{proof}
\begin{remark}\label{rem:od} The argument of the above proof has the following corollary. Let $j=2$ and suppose that $p_{\text{tot},0} = -m \eta_{\beta}$. Then, the scale $h_{v}$ of the corresponding quadratic monomial is such that $h_{v} \geq \lfloor \log_{2}  m \eta_{\beta} \rfloor - 2$. This provides a lower bound on the scale of the endpoint corresponding to the external field $\hat A^{\mu}_{-\ul p_{m}}$.
\end{remark}
\begin{proof}[Proof of Proposition \ref{prop:CorrDimBound}] Recall the expansion in equation (\ref{eq:GNexp}). To begin, it is easy to check that the purely ultraviolet contributions is bounded in $\eta$: from the ultraviolet tree expansion, we have:
\begin{equation}\label{eq:uvcorr}
\frac{1}{\beta L}|\langle \mathbf T\hat n_{\ul p_{1}};\cdots; \hat \jmath_{\mu,\ul p_{m}}\rangle_{\beta,L}^{(\text{u.v.})}|\leq m!\,C^{m}\;.
\end{equation}
Consider now the second term in (\ref{eq:GNexp}). Since by assumption $|p_{i}| \leq B|\theta|$ for $i=1,\ldots,m-1$, for $|\theta|$ small enough the marginal couplings associated with the fields $\hat A_{\underline{p}_{i}}^{\nu_{i}}$, $i=1,\ldots,m-1$, correspond to the vertex renormalizations $\mathsf Z^{\mu_{v}}_{h_{v},\omega_{v}\omega_{v}}$, that is with equal fermionic chiralities, by momentum conservation. These running coupling constants have a bounded flow, recall Eq. (\ref{eq:BdsEffCoup}). Observe that this does not apply to $\hat A_{-\underline{p}_{m}}^{\nu_{m}}$: for large enough $m$, this external field might multiply fermions with different chiralities. The corresponding vertex renormalization might have an unbounded flow, (\ref{eq:BdsEffCoup}); this is ultimately the reason for the last term in the right-hand side of (\ref{eq:estcorr}).

Let us write:
\begin{equation}\label{eq:dod}
\langle \mathbf T\hat n_{\ul p_{1}};\cdots; \hat \jmath_{\mu,\ul p_{m}}\rangle_{\beta,L}^{(\text{i.r.})} = \langle \mathbf T\hat n_{\ul p_{1}};\cdots; \hat \jmath_{\mu,\ul p_{m}}\rangle_{\beta,L}^{(\text{i.r.}), (\text{d})} + \langle \mathbf T\hat n_{\ul p_{1}};\cdots; \hat \jmath_{\mu,\ul p_{m}}\rangle_{\beta,L}^{(\text{i.r.}), (\text{od})}
\end{equation}
where the first term takes into account GN trees such that all $A$-endpoints on scale $<0$ multiply multiply fermions with the same chirality (``diagonal endpoints''), while the second term collects the GN trees where $\hat A_{-\underline{p}_{m}}^{\nu_{m}}$ multiplies fermions with different chiralities (``off-diagonal endpoints''). Let us start by considering the first term.
% we have:
%\begin{equation}\label{eq:trexp}
%\langle \mathbf T\hat n_{\ul p_{1}};\cdots; \hat \jmath_{\mu,\ul p_{m}}\rangle_{\beta,L}=\sum_{h=h_{\beta}}^{0}\sum_{n=0}^{\infty}\sum_{\tau\in \mathcal T_{h,n}^{m}} \sum_{\mathbf\Gamma\in \mathcal L_{\tau}} \widehat W^{(h)}_{\bf \Gamma}[\tau]\;.
%\end{equation}

 As a first simple step, we employ equation (\ref{eq:treebd}) to see that
\begin{equation} \label{eq:bdfourier}
\frac{1}{\beta L}|\widehat W^{(h)}_{\bf \Gamma}[\tau]| \leq 2^{h(2-m)} |L[\tau]| \prod_{v\notin V_{f}(\tau)} 2^{(h_{v}-h_{v'})D_{v}},
\end{equation}
where we used $|\Gamma^{A}_{v_{0}}|=m$, and where $L[\tau]$ denotes the contribution of the leaves, namely

\begin{equation}
L[\tau]:= \prod_{v\in V_{f}(\tau)} \text{val}(v)\qquad\qquad \text{val}(v)= \begin{cases} \lambda_{h_{v},\omega\omega'} & v\ \lambda\mathrm{-type},\ h_{v}<0\\
\nu_{h_{v},\omega} & v\ \nu\mathrm{-type},\ h_{v}<0\\
\mathsf Z^{\mu}_{h_{v},\omega} & v\ A^{\mu}\mathrm{-type},\ h_{v}<0\\
\| W^{(0)}_{\Gamma_{v}} \|_{1} & h_{v}=0\;.
\end{cases}
\end{equation}
%\begin{equation} \label{eq:leafvalue}
%L[\tau]:= \prod_{v\in V_{f}(\tau)} \text{val}(v)\qquad\qquad \text{val}(v)=\begin{cases}
%\lambda_{h_{v},\omega_{v}\omega'_{v}} &  \text{if}\  v\ \text{of}\ \lambda\text{-type}\\
%\nu_{h_{v},\omega_{v}}  &  \text{if}\  v\ \text{of}\ \nu\text{-type}\\
%\mathsf Z^{\mu}_{h_{v},\omega_{v}} e^{-i \ul p_{v}\cdot \ul x_{v}} &  \text{if}\  v\ \text{of}\ A^{\mu}\text{-type}\;.
%\end{cases}
%\end{equation}
Therefore, since $D_{v}< 0$ for all $v$, we obtain the thesis easily as long as the root scale $h$ is greater or equal to $h_{\eta}-4$, in which case (\ref{eq:bdfourier}), combined with Proposition \ref{prp:flow}, immediately yields
\begin{equation} \label{eq:bigscales}
\frac{1}{\beta L}\sum_{h=h_{\eta}-4}^{0}\sum_{n=0}^{\infty}\sum_{\tau\in \mathcal T_{h,n}^{m}} \sum_{\mathbf\Gamma\in \mathcal L_{\tau}} |\widehat W^{(h)}_{\bf \Gamma}[\tau]|\leq m!\, C^{m} \eta^{2-m} \sum_{h=h_{\eta}-4}^{0} 2^{(2-m)(h-h_{\eta}+4)}\;.
\end{equation}
Let us now consider the case $h\leq h_{\eta}-5$. Let us call $v^{*}$ the first vertex, encountered starting from the root, such that all external field leaves are in the sub-tree $\tau_{v^{*}}$ whose first vertex is $v^{*}$, and let us denote by $h^{*}$ its scale label; then, for all $v$ in the path $\gamma_{v_{0}\to v^{*}}$ going from $v_{0}$ to $v^{*}$, the number of $A$-fields is constant, so that
\[
D_{v} = 2-m - \frac{|\Gamma^{\psi}_{v}|}{2} + c|\lambda|
\]
which implies
\begin{equation}\label{eq:vstar}
\frac{1}{\beta L}|\widehat W^{(h)}_{\bf \Gamma}[\tau]| \leq 2^{h^{*}(2-m)} |L[\tau]|  \prod_{v\in\gamma_{v_{0}\to v^{*}}} 2^{-(h_{v}-h_{v'})(|\Gamma^{\psi}_{v}|/2 - c|\lambda|)} \prod_{v\notin\gamma_{v_{0}\to v^{*}}} 2^{(h_{v}-h_{v'})D_{v}}\;.
\end{equation}
If $h^{*}\geq h_{\eta}-5$, we do not need to obtain anything better than equation (\ref{eq:vstar}), since $|\Gamma_{v}^{\psi}|>0$ for all vertices $v\in \gamma_{v_{0}\to v^{*}}$ ensures summability over the labellings of $\tau$.%, and long trees are suppressed by the short memory property of GN trees.

Let us now consider the case $h^{*} < h_{\eta}-5$. In this case, we need to examine carefully the structure of the GN trees contributing to (\ref{eq:GNexp}). We notice that each $A$-type leaf $v$ with label $(\mu_{v},\omega_{v},-\ul p_{v})$ attached to the tree on scale $<0$ is associated with a monomial (written in momentum space)
\[
\mathsf Z^{\mu_{v}}_{h_{v},\omega_{v}}\int_{\beta,L} d\ul k'\,\hat\psi^{(\leq h_{v}),+}_{\ul k',\omega_{v}}\hat\psi^{(\leq h_{v}),-}_{\ul k'+\ul p_{v},\omega_{v}}\;,
\]
which implies
\begin{equation} \label{eq:leafscale}
2^{h_{v}+2}\geq \|\ul p_{v}\| \geq 2^{h_{\eta}}\quad\Longrightarrow\quad h_{v}\geq h_{\eta}-2\;;
\end{equation}
hence, $A$-type leaves attach at a scale (roughly) at least $h_{\eta}$: at least one of the $\psi$ fields appearing in the labelling of each $A$-type leaf needs to be contracted at a scale greater or equal to (roughly) $h_{\eta}$.

All in all, we may summarize these findings as follows: as graphically represented in Figure \ref{fig:tree}, there exists a collection of $1\leq s\leq m$ disjoint sub-trees with root $v^{*}$ containing all of the $A$-type leaves; if we call them $\{\tau_{i}\}_{i=1}^{s}$, we have
\[
\Gamma^{A}[\tau_{i}]\neq\varnothing\ \ \forall i\;,\qquad\qquad \bigcup_{i=1}^{s} \Gamma^{A}[\tau_{i}] = \Gamma^{A}_{v^{*}} = \Gamma^{A}_{v_{0}}\;.
\]
Since $h^{*}< h_{\eta}-5$, and since all $A$-type endpoints attach at scale at least $h_{\eta}-2$, each sub-tree $\tau_{i}$ will have $1\leq \ell_{i}\leq |\Gamma^{A}[\tau_{i}]|$ internal vertices $\{v^{**}_{i,j}\}_{j=1}^{\ell_{i}}\subseteq V(\tau_{i})$, such that:
\begin{itemize}
\item[1.] the associated subtrees contain all of the $A$-type leaves in $\tau_{i}$, that is
\[
\bigcup_{j=1}^{\ell_{i}} \Gamma^{A}_{v^{**}_{i,j}}=\Gamma^{A}[\tau_{i}]\;;
\]
\item[2.] $h^{**}_{i,j}:=h_{v^{**}_{i,j}}= h_{\eta}-5$.
\end{itemize}
Noticing that condition 1 implies
\[
\bigcup_{i=1}^{s} \bigcup_{j=1}^{\ell_{i}} \Gamma^{A}_{v^{**}_{i,j}} = \Gamma^{A}_{v^{*}}\;,
\]
we may number these vertices simply as $v^{**}_{j}$ with $j=1,\dots, \ell$, and $s\leq\ell\leq m$.
\begin{figure}[h]
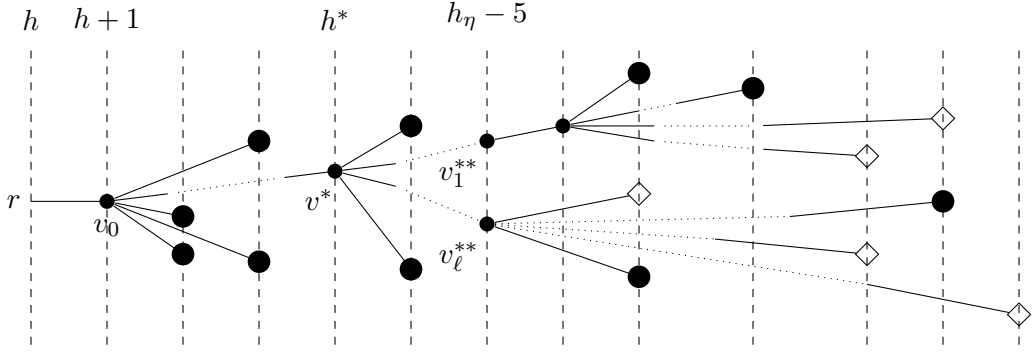

	\begin{center}
		\begin{tabular}{rcl}
			&\tikz[baseline=-3pt]{ %ramo principale
				%\draw(-1,0)--node[vertex,label=left:{$r$}](v00){}--(0,0)node[](v0){};
				\coordinate (A) at (-1,0) node[left] at (-1,0){$r$};
				\draw (A) -- ++ (1,0);
				\draw(0,0)node[vertex,label=below:{$v_0$}](v0){}--(0.8,0.1)node[](v01){};
				\draw[dotted](v01)--(2.2,0.3)node[](v02){};
				\draw(v02)--(3,0.4)node[vertex,label=below:{$v^*\;\;\;\;$}](v1){};
				\draw(v1)--(3.8,0.5)node[](v11){};
				\draw(v1)--(3.8,0.2)node[](v111){};
					\draw[dotted](v11)--(5,0.8)node[vertex,label=below:{$v^{**}_{1}\;\;\;\;\;\;\;$}](v12){};
				\draw(v12)--(6,1)node[vertex](v3){};
				\draw[dotted](v111)--(5,-0.3)node[vertex,label=below:{$v^{**}_{\ell}\;\;\;\;\;\;\;$}](v121){};
					\draw(v121)--(7,-1)node[bigvertex](v4){};
				%%foglie intermedie
				\draw(v0)--(1,-0.2)node[bigvertex]{};
				\draw(v0)--(1,-0.7)node[bigvertex]{};
				\draw(v0)--(2,-0.8)node[bigvertex]{};
				\draw(v0)--(2,0.8)node[bigvertex]{};
				\draw(v1)--(4,-0.9)node[bigvertex]{};
				\draw(v1)--(4,1)node[bigvertex]{};
				%% foglie finali-
				\draw(v3)--(7,1.7)node[bigvertex]{};
				\draw(v3)--(7.2,1)node[](v72){};
				\draw[dotted](v72)--(8.5,1)node[](v92){};
				\draw(v92)--(11,1.1)node[ctVertex]{};
				\draw(v3)--(7.2,.8)node[](v71){};
				\draw[dotted](v71)--(8.5,0.7)node[](v91){};
				\draw(v91)--(10,.6)node[ctVertex]{};
				\draw(8.5,1.5)node[bigvertex]{}--(7.5,1.3){};
				\draw[dotted](7.5,1.3)--(7,1.2){};
				\draw(7,1.2)--(v3){};
				\draw[dotted](v121)--(9,-.2);
				\draw(9,-.2)--(11,0)node[bigvertex]{};
				\draw[dotted](v121)--(8,-.5);
				\draw(8,-.5)--(10,-.7)node[ctVertex]{};
				\draw(10,-1.1)--(12,-1.5)node[ctVertex]{};
				\draw[dotted](v121)--(10,-1.1);
				\draw(v121)--(7,.1)node[ctVertex,label=right:{}]{}; 
					%%linee verticali
				\draw[dashed](-1,2)node[label=above:{$h$}]{}--(-1,-2);
				\draw[dashed](0,2)node[label=above:{$h+1$}]{}--(0,-2);
				\draw[dashed](1,2)--(1,-2);
				\draw[dashed](2,2)--(2,-2);
				\draw[dashed](3,2)node[label=above:{$h^{*}$}]{}--(3,-2);
				\draw[dashed](4,2)--(4,-2);
				\draw[dashed](5,2)node[label=above:{$h_{\eta}-5$}]{}--(5,-2);
				\draw[dashed](6,2)node[label=above:{}]{}--(6,-2);
				\draw[dashed](7,2)--(7,-2);
				\draw[dashed](12,2)--(12,-2);
				\draw[dashed](11,2)--(11,-2);
				\draw[dashed](10,2)--(10,-2);
				\draw[dashed](8.5,2)--(8.5,-2);
			}
		\end{tabular}
	\end{center}
	\caption{A representation of the structure of $\tau$: white diamonds denote $A$-type leaves, and big circles indicate sub-trees involving only $\lambda$ or $\nu$ type leaves.}
	\label{fig:tree}
\end{figure}	

Setting $\gamma_{v^{*}\to v^{**}_{j}}$ to be the path from $v^{*}$ to $v^{**}_{j}$, $\gamma_{*\to**}$ to be the union of all such paths, and by setting $\gamma_{0\to **}:=\gamma_{v_{0}\to v^{*}}\cup\gamma_{*\to**}$,  we are led to the following bound for the second product appearing in the right-hand side of equation (\ref{eq:vstar}):
\begin{equation} \label{eq:telescopic}
\begin{split}
\prod_{v\notin\gamma_{v_{0}\to v^{*}}} 2^{(h_{v}-h_{v'})D_{v}}&= \prod_{v\in\gamma_{*\to **}} 2^{(h_{v}-h_{v'})D_{v}} \prod_{v\notin\gamma_{0\to **}} 2^{(h_{v}-h_{v'})D_{v}}\\
& = \prod_{j=1}^{\ell} 2^{(h^{*}-h^{**}_{j})|\Gamma^{A}_{v^{**}_{j}}|}\prod_{v\in\gamma_{*\to **}} 2^{(h_{v}-h_{v'})\big(2-\frac{|\Gamma_{v}^{\psi}|}{2}-z_{v} + c|\lambda|\big)}\prod_{v\notin\gamma_{0\to **}} 2^{(h_{v}-h_{v'})D_{v}}
\end{split}
\end{equation}
where in the second passage we just used telescopic sums to factor out only the external field contribution to $D_{v}$. Keeping in mind that $h^{**}_{j}=h_{\eta}-5$ for all $j=1,\dots,\ell$, and inserting (\ref{eq:telescopic}) in (\ref{eq:vstar}), we obtain
\begin{equation}\label{eq:vstarstar}
\begin{split}
\frac{1}{\beta L}|\widehat W^{(h)}_{\bf \Gamma}[\tau]| \leq  |L[\tau]|\prod_{v\in\gamma_{v_{0}\to v^{*}}}& 2^{-(h_{v}-h_{v'})(|\Gamma^{\psi}_{v}|/2 - c|\lambda|)} \prod_{v\notin\gamma_{0\to **}} 2^{(h_{v}-h_{v'})D_{v}}\,\,\times\\
\times\,\,& 2^{2h^{*}-m(h_{\eta}-5)} \prod_{v\in\gamma_{*\to **}} 2^{(h_{v}-h_{v'})\big(2-\frac{|\Gamma_{v}^{\psi}|}{2}-z_{v} + c|\lambda|\big)}.
\end{split}
\end{equation}
Whereas the first line of (\ref{eq:vstarstar}) is summable over the scale labels, we see that the product over $v\in\gamma_{*\to**}$ is possibly not summable in $h_{v}-h_{v'}$ whenever $|\Gamma_{v}^{\psi}|=2,4$, since
\begin{equation}\label{eq:star}
2^{(h_{v}-h_{v'})\big(2-\frac{|\Gamma_{v}^{\psi}|}{2}-z_{v}\big)} =\begin{cases}
2^{h_{v}-h_{v'}}&\text{if}\ |\Gamma^{\psi}_{v}|=2, |\Gamma_{v}^{A}|>1\\
1 & \text{if}\ |\Gamma^{\psi}_{v}|=2, |\Gamma_{v}^{A}|=1\ \text{or}\ |\Gamma_{v}^{\psi}|=4,
\end{cases}\ \ \geq 1;
\end{equation}
this is obviously due to us extracting the contribution relative to the $A$-type external fields. In order to prove summability in the scale jumps $h_{v}-h_{v'}$, we would like to show that these instances cannot occur, \emph{i.e.} that
\[
\{v\in\gamma_{*\to**}\ |\ |\Gamma^{\psi}_{v}|\leq 4\}=\varnothing\;.
\]
We proceed by \emph{reductio ad absurdum}: suppose there exists $\bar v\in\gamma_{*\to**}$ such that $|\Gamma^{\psi}_{\bar v}|\leq 4$, then Lemma \ref{lemma:boundedscale} would imply $h_{\bar v}\geq h_{\eta} - 4$, but this contradicts $h_{v}\leq h_{\eta} - 5$ for all $v\in\gamma_{*\to**}$. Therefore, we have that
\begin{equation}\label{eq:scalejumps}
\prod_{v\in\gamma_{*\to **}} 2^{(h_{v}-h_{v'})\big(2-\frac{|\Gamma_{v}^{\psi}|}{2}-z_{v}+c|\lambda|\big)} = \prod_{v\in\gamma_{*\to **}} 2^{(h_{v}-h_{v'})\big(2-\frac{|\Gamma_{v}^{\psi}|}{2}-z_{v}+c|\lambda|\big)}\mathbbm 1_{|\Gamma^{\psi}_{v}|\geq 6}\;.
\end{equation}
Inserting this estimate (\ref{eq:scalejumps}) into (\ref{eq:vstarstar}), we are led to
\begin{equation}\label{eq:vstarstarstar}
\begin{split}
\frac{1}{\beta L}|\widehat W^{(h)}_{\bf \Gamma}[\tau]| \leq  2^{5m}|L[\tau]|\prod_{v\in\gamma_{v_{0}\to v^{*}}}& 2^{-(h_{v}-h_{v'})(|\Gamma^{\psi}_{v}|/2 - c|\lambda|)} \prod_{v\notin\gamma_{0\to **}} 2^{(h_{v}-h_{v'})D_{v}}\,\,\times\\
\times\,\,&  2^{2h^{*}-m h_{\eta}} \prod_{v\in\gamma_{*\to **}} \Big[2^{(h_{v}-h_{v'})\big(2-\frac{|\Gamma_{v}^{\psi}|}{2}-z_{v}+c|\lambda|\big)}\mathbbm 1_{|\Gamma^{\psi}_{v}|\geq 6}\Big]\;.%\\
%\leq 2^{5m} |L[\tau]| \prod_{v\in\gamma_{v_{0}\to v^{*}}}& 2^{-(h_{v}-h_{v'})|\Gamma^{\psi}_{v}|/2} \prod_{v\notin\gamma_{0\to **}} 2^{(h_{v}-h_{v'})D_{v}}\,\,\times\\
%\times\,\,& 2^{(2-m)h_{\eta}} \prod_{v\in\gamma_{*\to **}} \Big[2^{(h_{v}-h_{v'})\big(2-\frac{|\Gamma_{v}^{\psi}|}{2}-z_{v}+c|\lambda|\big)}\mathbbm 1_{|\Gamma^{\psi}_{v}|\geq 6}\Big]\;;
\end{split}
\end{equation}
The first factor in the second line is used to sum over $h^{*}\leq h_{\eta} - 5< h_{\eta}$, and it allows to obtain the correct dimensional scaling in $\eta$. Furthermore, summability over $\mathbf{\Gamma}\in\mathcal L_{\tau}$ and $\tau\in \mathcal T_{h,n}^{m}$ is ensured, since all summands involve now factors lesser than 1. In particular, since $A$-leaves are constrained to attach at scale at least $h_{\eta}-2$, choosing any $\varepsilon\in(0,1-c|\lambda|)$, we may extract a gain from the minimum depth of $\tau$ (which amounts to $h_{\eta}-2-h$) by employing the short-memory property of GN trees; we have
\begin{equation} \label{eq:shortmem}
\frac{1}{\beta L}\sum_{\tau\in \mathcal T_{h,n}^{m}}' \sum_{\mathbf\Gamma\in \mathcal L_{\tau}} |\widehat W^{(h)}_{\bf \Gamma}[\tau]| \leq m!\, C_{\varepsilon}^{m+n}   |\lambda|^{n}\, \eta^{2-m}\, 2^{\varepsilon(h-h_{\eta}+2)}\;,
\end{equation}
where the prime denotes the constraint that all $A$-endpoints on scale $<0$ multiply fermions with the same chirality. Summing (\ref{eq:shortmem}) over $n$ and combining the resulting bound with (\ref{eq:bigscales}), for $\lambda$ small enough one gets the following estimate for the correlators, for $m\geq 3$:
\begin{equation}\label{eq:dcorr}
\begin{split}
&\frac{1}{\beta L}|\langle \mathbf T\hat n_{\ul p_{1}};\cdots; \hat \jmath_{\mu,\ul p_{m}}\rangle_{\beta,L}^{(\text{i.r.}), (\text{d})}| \\
&\qquad \leq m!\, \eta^{2-m}\left[C_{\varepsilon}^{m}\sum_{h=h_{\beta}}^{h_{\eta}-5} 2^{\varepsilon(h-h_{\eta}+2)} + \tilde C^{m}\sum_{h=h_{\eta}-4}^{0} 2^{(2-m)(h-h_{\eta}+4)}\right]\\
&\qquad \leq m!\, C^{m} \eta^{2-m}\;,
\end{split}
\end{equation}
for $C=\mathrm{max}\{C_{\varepsilon},\tilde C\}$. 

Consider now the second term in (\ref{eq:dod}). Recall that this contribution is given by trees with one off-diagonal vertex renormalization, corresponding to the external field $\hat A^{\nu_{m}}_{-\underline{p}_{m}}$. The associated running coupling constant satisfies the bound:
\begin{equation}\label{eq:odrcc}
|\mathsf Z^{\mu_{v}}_{h_{v},\omega_{v},\omega'_{v}}| \leq C2^{-c|\lambda| h_{v}}\;.
\end{equation}
By Remark \ref{rem:od}, the scale $h_{v}$ satisfies $h_{v} \geq \lfloor \log_{2}  m \eta \rfloor - 2 =: h(m,\eta)$; this scale constraint will ultimately allow to to show that the sum of all these trees gives a subleading contribution. 

We distinguish three different cases for the scale $h$ of the root: the first is $h\geq h(m,\eta)$; the second is $h(m,\eta) > h \geq h_{\eta} - 5$; the last is $h\leq h_{\eta} - 5$. Consider the first case. Denoting by a double prime the condition that there is exactly one off-diagonal endpoint:
\begin{equation}
\frac{1}{\beta L}\sum_{\tau\in \mathcal T_{h,n}^{m}}'' \sum_{\mathbf\Gamma\in \mathcal L_{\tau}} |\widehat W^{(h)}_{\bf \Gamma}[\tau]| \leq m!\, C^{m+n}   |\lambda|^{n}\, 2^{h(2-m - c|\lambda|)}\qquad \text{for $h\geq h(m,\eta)$.}
\end{equation}
This bound is a direct consequence of (\ref{eq:vstar}), taking into account the bound (\ref{eq:odrcc}) for the off-diagonal running coupling constant. Observe that, for $m\geq 3$:
\begin{equation}
\sum_{h\geq h(m,\eta)}^{0} 2^{h(2-m - c|\lambda|)}\leq C 2^{h(m,\eta)(2-m - c|\lambda|)} \leq K^{m} m^{2-m-c|\lambda|} \eta^{2-m - c|\lambda|}\;,
\end{equation}
where we used that $h(m,\eta) \geq \log_{2} m + \log_{2} \eta - c$. In particular, since $2-m\leq -1$:
\begin{equation}\label{eq:dod1}
\frac{1}{\beta L} \sum_{h\geq h(m,\eta)}^{0}\sum_{\tau\in \mathcal T_{h,n}^{m}}'' \sum_{\mathbf\Gamma\in \mathcal L_{\tau}} |\widehat W^{(h)}_{\bf \Gamma}[\tau]| \leq m!\, K^{m+n}   |\lambda|^{n} \bigg(\frac{1}{m}\bigg)^{1 - c|\lambda|} \eta^{2-m - c|\lambda|}\;.
\end{equation}
Next, consider the second case. Here we proceed as in the previous case, with the only addition of using the short-memory property of the GN trees, Remark \ref{rem:sm}, to take into account the fact that the off-diagonal endpoint is on scale $\geq h(m,\eta)$, which might be much larger than the root scale $h$. We have, for $h(m,\eta) > h \geq h_{\eta} - 5$:
\begin{equation}
\frac{1}{\beta L}\sum_{\tau\in \mathcal T_{h,n}^{m}}'' \sum_{\mathbf\Gamma\in \mathcal L_{\tau}} |\widehat W^{(h)}_{\bf \Gamma}[\tau]| \leq m!\, C_{\varepsilon}^{m+n}   |\lambda|^{n}\, 2^{h(2-m)} 2^{\varepsilon(h - h(m,\eta))} 2^{-c|\lambda| h(m,\eta)}\;.
\end{equation}
Observe that, for $m\geq 3$:
\begin{equation}
\begin{split}
\sum_{h = h_{\eta}}^{h(m,\eta)} 2^{h(2-m)} 2^{\varepsilon (h - h(m,\eta))} &= 2^{h_{\eta}(2-m)} \sum_{h = h_{\eta}}^{h(m,\eta)} 2^{(h - h_{\eta})(2-m)} 2^{\varepsilon (h - h(m,\eta))} \\
&\leq C^{m} \eta^{2-m} 2^{(\varepsilon/2) (h_{\eta} - h(m,\eta))} \\
&\leq K^{m} \eta^{2-m} \bigg(\frac{1}{m}\bigg)^{\varepsilon/2}\;.
\end{split}
\end{equation}
In particular,
\begin{equation}\label{eq:dod2}
\frac{1}{\beta L} \sum_{h = h_{\eta}-5}^{h(m,\eta)} \sum_{\tau\in \mathcal T_{h,n}^{m}}'' \sum_{\mathbf\Gamma\in \mathcal L_{\tau}} |\widehat W^{(h)}_{\bf \Gamma}[\tau]| \leq m!\, K^{m+n} |\lambda|^{n} \eta^{2-m} \bigg(\frac{1}{m}\bigg)^{\varepsilon/2} \bigg( \frac{1}{\eta m} \bigg)^{c|\lambda|}\;.
\end{equation}
For the third case, we proceed as in (\ref{eq:shortmem}), extracting the short memory factor introduced by the off-diagonal endpoint. We have:
\begin{equation}
\frac{1}{\beta L}\sum_{\tau\in \mathcal T_{h,n}^{m}}'' \sum_{\mathbf\Gamma\in \mathcal L_{\tau}} |\widehat W^{(h)}_{\bf \Gamma}[\tau]| \leq m!\, C_{\varepsilon}^{m+n}   |\lambda|^{n}\, \eta^{2-m}\, 2^{\varepsilon(h-h(m,\eta))} 2^{-c|\lambda| h(m,\eta)}\qquad \text{for $h < h_{\eta} - 5$.}
\end{equation}
Observe that:
\begin{equation}
\sum_{h\leq h_{\eta}} 2^{\varepsilon (h - h(m,\eta))} \leq C 2^{\varepsilon(h_{\eta} - h(m,\eta))} \leq C \bigg( \frac{1}{m} \bigg)^{\varepsilon}\;.
\end{equation}
Therefore, we get:
\begin{equation}\label{eq:dod3}
\frac{1}{\beta L} \sum_{h < h_{\eta}-5} \sum_{\tau\in \mathcal T_{h,n}^{m}}'' \sum_{\mathbf\Gamma\in \mathcal L_{\tau}} |\widehat W^{(h)}_{\bf \Gamma}[\tau]| \leq m!\, K^{m+n}   |\lambda|^{n}\, \eta^{2-m} \bigg( \frac{1}{m} \bigg)^{\varepsilon} \bigg( \frac{1}{\eta m} \bigg)^{c|\lambda|}\;.
\end{equation}
The estimate for the second contribution in (\ref{eq:dod}) is given by the sum of the estimates (\ref{eq:dod1}), (\ref{eq:dod2}), (\ref{eq:dod3}). Next, we observe that, in order for the off-diagonal vertex to be nonzero, we need to have that the spatial component of the momentum $\underline{p}_{m}$ satisfies $|p_{m}| \geq c\min_{\omega,\omega': \omega\neq \omega'} | k_{F}^{\omega} - k_{F}^{\omega'}|$, which is only possible for $m\geq K |\theta|^{-1} \equiv m(\theta)$, since $|p_{i}| \leq B|\theta|$ for all $i=1,\ldots, m-1$. We then have, for $m\geq m(\theta)$:
\begin{equation}
\frac{1}{\beta L}| \langle \mathbf T\hat n_{\ul p_{1}};\cdots; \hat \jmath_{\mu,\ul p_{m}}\rangle_{\beta,L}^{(\text{i.r.}), (\text{od})} | %m!\, C^{m} \eta^{2-m} \Big( \Big(\frac{1}{m}\Big)^{m} \eta^{-c|\lambda|} + \Big(\frac{1}{m}\Big)^{\varepsilon/2} \Big( \frac{1}{\eta m} \Big)^{c|\lambda|} \Big) \\
\leq m!\, C^{m} \eta^{2-m} |\theta|^{\varepsilon/2 + c|\lambda|} \eta^{-c|\lambda|}\;.
%&\leq m!\, C^{m} \eta^{2-m} \eta^{\varepsilon/2} |\lambda|^{-\alpha(\varepsilon/2 + c|\lambda|)}\;.
\end{equation}
This bound, combined with (\ref{eq:uvcorr}) and (\ref{eq:dcorr}), concludes the proof of (\ref{eq:estcorr}).
\end{proof}
%
%\begin{remark}
%Notice that this Proposition ensures that we may obtain the scaling of the free case also in the (weakly) interacting one, which is hard to determine a priori.
%\end{remark}
%
Proposition \ref{prop:CorrDimBound} provides bounds uniform in $\beta,L$ large, and with the correct $\eta$-dependence, however they are not strong enough to ensure summability of the original (Wick-rotated) Duhamel series, Eq. (\ref{eq:fullrespint}). Summability could be ensured by assuming that $\|\hat \mu\|_{1}$ is small enough. Remarkably, this is not needed; in the next section, we will improve the $m$-dependence of these bounds, in a way that will allow to prove summability of the Duhamel series (\ref{eq:fullrespint}) without any extra assumption on the strength of the external perturbation.
\subsubsection{Improved estimates on the Euclidean density-current correlations}
Keeping in mind the GN representation (\ref{eq:GNexp}), in order to get bounds which are summable in $m$, we will consider separately the cases with few and many (with respect to $m$) interactions. Indeed, setting
\begin{equation}\label{eq:perturbtrees}
\widehat W_{m;n}(\ul p_{1},\dots,\ul p_{m})= \sum_{h=h_{\beta}}^{0} \sum_{\tau\in \widetilde{\mathcal T}_{h,n}^{m}} \sum_{\mathbf\Gamma\in \mathcal L_{\tau}} \widehat W^{(h)}_{\bf \Gamma}[\tau]\;,
\end{equation}
where $\widetilde{\mathcal T}_{h,n}^{m}$ is the set of trees that extend in the ultraviolet regime, recall the end of Section \ref{sec:GNtrees}, we rewrite Eq. (\ref{eq:GNexp}) as:
\begin{equation}\label{eq:GNsplit}
\begin{split}
\langle \mathbf T\hat n_{\ul p_{1}};\cdots; \hat \jmath_{\mu,\ul p_{m}}\rangle_{\beta,L} &= \langle \mathbf T\hat n_{\ul p_{1}};\cdots; \hat \jmath_{\mu,\ul p_{m}}\rangle_{\beta,L}^{(\text{u.v.})} + \sum_{n=0}^{n_{*}} \widehat W_{m;n}(\ul p_{1},\dots,\ul p_{m})\\&\quad + \sum_{n=n_{*}+1}^{\infty} \widehat W_{m;n}(\ul p_{1},\dots,\ul p_{m})\;,
\end{split}
\end{equation}
with $n_{*}\equiv n_{*}(m):=\lfloor\alpha m\rfloor$, $\alpha\in(0,1)$ constant to be chosen appropriately.
Let us start by considering the first sum in (\ref{eq:GNsplit}). We will show that this term admits an estimate with an improved $m$-dependence, with respect to the dimensional bounds in Proposition \ref{prop:CorrDimBound}. 
%Given a perturbative order $n<n_{*}(m)$, our goal is to provide a bound for
%\begin{equation} \label{eq:perturbtrees}
%\widehat W_{m;n}(\ul p_{1},\dots,\ul p_{m})= \sum_{h=h_{\beta}}^{0} \sum_{\tau\in \mathcal T_{h,n}^{m}} %\sum_{\mathbf\Gamma\in \mathcal L_{\tau}} \widehat W^{(h)}_{\bf \Gamma}[\tau]
%\end{equation}

Fixing a root scale $h$, a GN tree $\tau\in \widetilde{\mathcal T}^{m}_{h,n}$ and a labelling $\mathbf \Gamma\in \mathcal L_{\tau}$, and recalling (\ref{eq:kernexptree})-(\ref{eq:labelexpren}), we can write:
\begin{equation}\label{eq:WT}
\widehat W^{(h)}_{\mathbf \Gamma}[\tau] = \sum_{{\bf T} = \bigcup_{v\notin V_{f}(\tau)} T_{v}} \widehat W^{(h)}_{\mathbf \Gamma}[\tau; {\bf T}]
\end{equation}
where $\widehat W^{(h)}_{\mathbf \Gamma}[\tau; {\bf T}]$ is constructed via the tree expansion described in Section \ref{sec:GNtrees}, keeping fixed the spanning trees $T_{v}$'s. Recall that each $T_{v}$ arises from the use of the BBF formula, Eq. (\ref{eq:BBF}). In the estimate for $ \widehat W^{(h)}_{\mathbf \Gamma}[\tau; {\bf T}]$, the propagators associated with the edges of ${\bf T}$ are bounded in $L^{1}$ norm; in what follows, we shall discuss a modification of this part of the estimate of $\widehat W^{(h)}_{\mathbf \Gamma}[\tau; {\bf T}]$, that will allow us to extract a summable factor in $m$.  

To this end, let us analyse the structure of $\mathbf T$ more in detail. Since it is given by the union of the anchored trees $(T_{v})_{v\notin V_{f}(\tau)}$, it follows that ${\bf T}$ is itself a tree spanning all endpoints of $\tau$, and $|\mathbf T|=m+n-1$. Since $n\leq n_{*}(m)< m$ by assumption, we see that any spanning tree $\mathbf T$ will have a certain number of \emph{chains} contained in it: by chain we mean the graph that is formed by only connecting $A$-vertices. If $n=0$, the only admissible graphs are (closed) chains. In general, if $n\neq 0$, there exists a set of chains $\mathbf C_{\mathbf T}\subseteq \mathbf T$ with $N_{c}$ connected components such that
\begin{equation}
\mathbf T = \mathbf C_{\mathbf T}\cup (\mathbf T\setminus\mathbf C_{\mathbf T})\qquad \qquad \mathbf C_{\mathbf T}:= \coprod_{j=1}^{N_{c}} \mathbf C_{j},
\end{equation}
and where each chain $\mathbf C_{j}$ is obtained by considering the maximum number of consecutive edges between only $A$-type leaves along $\mathbf T$ (see Fig. \ref{fig:chains} for an example). Thus, the number of propagators in a chain with $m_{j}$ $A$-vertices is $|\mathbf C_{j}| = m_{j}-1$. It should be noted that the number of chains is bounded as $1\leq N_{c}\leq \max\{2n,1\}$.

\begin{figure}[ht]
\centering
\begin{tikzpicture}
\begin{feynman}
\vertex (b);
\node at (b) [square dot];
\vertex (b1) at ($(b)-(0,0.4)$) {\( \ul x_{1}^{j}\)};
%\vertex [left=of b] (a) {\( \ul x_{0}^{j} \)};
\vertex [right=of b] (c);
\node at (c) [square dot];
\vertex (c1) at ($(c)-(0,0.4)$) {\( \ul x_{2}^{j}\)};
\vertex [right= of c](d);
\vertex (z) at ($(d)+(0.5,0)$);
\draw[black,fill=black] (z) circle (.05 cm);
%\node at (z) [dot];
\vertex (z1) at ($(z)+(0.35,0)$);
\draw[black,fill=black] (z1) circle (.05 cm);
\vertex (z2) at ($(z1)+(0.35,0)$);
\draw[black,fill=black] (z2) circle (.05 cm);
\vertex (r) at ($(z2)+(0.5,0)$);
\vertex [right=of r] (s);
\node at (s) [square dot];
\vertex (s1) at ($(s)-(0,0.4)$) {\;\;\( \ul x_{m_{j}}^{j}\)};
\vertex [above=of s] (t);
\vertex [above=of b] (e);
\vertex [above=of c] (f);
\diagram*[large]{
%(a) -- [fermion, edge label'=\( h_{1}\)] (b);
(b) -- [fermion, edge label'=\( h_{1}\)] (c);
(c) -- [fermion, edge label'=\( h_{2}\)] (d);
(r) -- [fermion, edge label'=\( h_{m_{j}-1}\)] (s);
(t) -- [purple,boson, edge label'=\(\ul p_{m_{j}}^{j}\)] (s);
(e) -- [purple,boson, edge label'=\( \ul p_{1}^{j} \)] (b);
(f) -- [purple, boson, edge label'=\( \ul p_{2}^{j}\)] (c);
};
\end{feynman}
\end{tikzpicture}
\caption{An example of a chain $\mathbf C_{j}$ with $m_{j}$ $A$-leaves: external momenta injections are denoted by wiggly lines, and the scale of the $i$-th edge $\ell_{i}$ on the chain is indicated with $h_{i}$. We have indicated also the position variables associated to each density leaf.}
\label{fig:chains}
\end{figure}
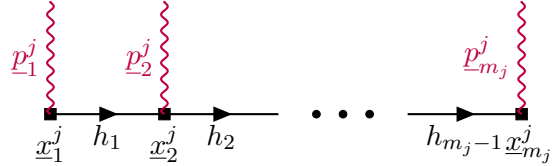

% As it can be seen from Figure \ref{fig:chains}, each chain $\mathbf C_{j}$ will contain an edge arriving at the first density leaf or an edge exiting the last one\footnote{It may contain them both, but this latter case will not hold for all chains (and it will not be necessary for our analysis).}, so that its cardinality $|\mathbf C_{j}|$ is equal to the number of $A$-leaves $m_{j}$ connected by it.

Notice that the number of chains $N_{c}$ and their lengths $m_{j}-1$ depend on the triple $(\tau,\mathbf\Gamma, \mathbf T)$, with the following additional constraint: by construction we have that
\[
|\mathbf C_{\mathbf T}|=\sum_{j=1}^{N_{c}}|\mathbf C_{j}|\equiv\sum_{j=1}^{N_{c}} (m_{j}-1) = m - N_{c},
\]
and therefore $|\mathbf T\setminus \mathbf C_{\mathbf T}|= n -1 + N_{c}$.
Each $\mathbf C_{j}$ may be parameterized (see Fig. \ref{fig:chains}) by the ordered set of leaves connected by it, and these in turn are identified by the ordered external momenta $m_{j}$-tuple $(\ul p_{1}^{j},\dots, \ul p_{m_{j}}^{j})$ and directed edges $(\ell_{i})_{i=1}^{m_{j}-1}$, with $\ell_{i}\equiv(\ul x_{i}^{j},\ul x_{i+1}^{j})$.

Therefore, the value of a chain may be written (up to the action of $\mathfrak R$) as
\begin{equation} \label{eq:chainvalue}
\prod_{\ell\in \mathbf C_{j}} g^{(h_{\ell})}_{\ell} \equiv \prod_{i=1}^{m_{j}-1} g^{(h_{i})}(\ul x_{i}^{j}-\ul x_{i+1}^{j})
\end{equation}
with $h_{i}\equiv h_{\ell_{i}}$. As we will show below, for the special choice of external momenta that we are considering, in which the temporal part of all but one external momenta are the same and equal to $\eta$, the presence of ``long chains'' will introduce a gain in the final estimate which is \emph{factorial} in their length. In turn, this will imply an improvement in the bounds for the kernels $\widehat W_{m;n}$ with $n<n_{*}(m)$, providing overall summability in $m$.
\begin{lemma}[Bounding chains] \label{lemma:factorial}
Let $r\in \N$ and let $\mathbf C$ be a chain as described above, with $r$ oriented edges labelled as $\ell_{1},\dots, \ell_{r}$, and scale labels $h_{1}, \ldots, h_{r}$. Let us indicate with $h^{\mathrm{min}}:=\min_{i} h_{i}$ the lowest scale; then, for any given choice of scales $\{h_{i}\}_{i=1}^{r}$ compatible with the values of the external momenta, and for any subset $G\subseteq \mathbf C$ containing at least one line with scale $h^{\mathrm{min}}$, we have
\begin{equation}\label{eq:bdfact}
\prod_{\ell\in G} 2^{-h_{\ell}} \leq \left(\frac{C}{\eta}\right)^{|G|-1} \frac{2^{-h^{\mathrm{min}}}}{(|G|-1)!}\;.
\end{equation}
In particular, by picking $G=\mathbf C$ we obtain
\[
\prod_{i=1}^{r} 2^{-h_{i}} \leq \left(\frac{C}{\eta}\right)^{r-1} \frac{2^{-h^{\mathrm{min}}}}{(r-1)!}\;.
\]
\end{lemma}
\begin{proof}
We start by considering the value of the chain in momentum space. We start by assuming that all propagators in the chain are on scale $\leq 0$. We will comment at the end on the case of chains extending in the ultraviolet regime.
%For the sake of simplicity, we will assume $\mathbf C$ to be of the first kind, but the same proof strategy applies to second kind ones as well. 
Recalling (\ref{eq:chainvalue}), the value of the chain $\mathbf C$ (up to the action of $\mathfrak{R}$) is given by
\begin{equation} \label{eq:chainval}
\prod_{\ell\in \mathbf C} g^{(h_{\ell})}_{\ell} \equiv \prod_{s=1}^{r} g_{\omega_{s}}^{(h_{s})}(\ul x_{s}-\ul x_{s+1})
\end{equation}
with $h_{s}\equiv h_{\ell_{s}}$. Observe that the chain is attached to an incoming and an outgoing line; taking into account these lines, the value of the extended chain is:
\begin{equation}\label{eq:inout}
g^{(h_{in})}_{\ell_{in}}(\ul x_{in} - \ul x_{1}) \Big[ \prod_{\ell\in \mathbf C} g^{(h_{\ell})}_{\ell} \Big] g^{(h_{out})}_{\ell_{out}}(\ul x_{r+1} - \ul x_{out})\;.
\end{equation}
Observe that the coordinates $\ul x_{1},\dots, \ul x_{r+1}$, which are the coordinates of the $A$-leaves attached to the chain, only appear in (\ref{eq:inout}), and nowhere else in the evaluation of the GN tree. Let us consider (\ref{eq:inout}), multiplied by the phase factors $e^{-i\underline{p}_{i}\cdot \underline{x}_{i}}$ attached to the external fields, $i=1,\ldots, r+1$. Using the Fourier representation for the single-scale propagator
\begin{equation}
g_{\omega}^{(h)}(\ul x-\ul x')=\int_{\beta,L} d\ul k'\, e^{i\ul k'\cdot(\ul x-\ul x')} \hat g_{\omega}^{(h)}(\ul k')\;,
\end{equation}
and introducing the short-hand
\[
\ul p_{\leq s} := \sum_{s'\leq s} \ul p_{s'}\qquad \forall s\in\{1,\dots, r+1\},
\]
the Fourier transform of (\ref{eq:inout}) is:
\begin{equation}\label{eq:chain}
\begin{split}
&\int_{\beta,L} d\ul x_{1}\cdots d\ul x_{r+1}\, e^{-i\sum_{s=1}^{r+1} \ul p_{s}\cdot \ul x_{s}} g^{(h_{in})}_{\ell_{in}}(\ul x_{in} - \ul x_{1}) \Big[ \prod_{\ell\in \mathbf C} g^{(h_{\ell})}_{\ell} \Big] g^{(h_{out})}_{\ell_{out}}(\ul x_{r+1} - \ul x_{out}) \\
&= e^{- i\ul p_{\leq r+1}\cdot \ul x_{out}}\int_{\beta,L} d\ul k'\, e^{i \ul k'\cdot (\ul x_{in}-\ul x_{out})} \hat g^{(h_{in})}_{\omega_{in}}(\underline{k}') \Big[\prod_{s=1}^{r} \hat g^{(h_{s})}_{\omega_{s}}(\ul k'+\ul p_{\leq s})\Big] \hat g^{(h_{out})}_{\omega_{out}}(\underline{k}' + \underline{p}_{\leq r+1})\;.
\end{split}
\end{equation}
Observe that, by momentum conservation, the momenta of the propagators on the chain are completely determined by the ingoing momentum $\ul k'$ and the external momenta of the source terms. By the support properties of the single-scale propagators $\hat g^{(h_{i})}$ we see that
\begin{equation} \label{eq:scalebound}
\|\ul k' + \ul p_{\leq s}\|\geq 2^{h_{s}-1}\qquad\Longrightarrow\qquad 2^{-h_{s}}\leq \frac{1}{2\|\ul k' + \ul p_{\leq s}\|}\;.
\end{equation}
In what follows, we shall assume that $p_{i,0}=\eta$ for all $i=1,\dots, r$. Later we will discuss the adaptation to the case in which one of the momenta among $\ul p_{1},\dots,\ul p_{r}$ corresponds to $\ul p_{m}$, \emph{i.e.} it has temporal component equal to $-(m-1)\eta$.
%hence, to prove the thesis we will provide a suitable upper bound for the right-hand side.

Notice that the first propagator in the chain is supported for the following values of $k_{0}$, $|k_{0} + \eta| \leq 2^{h_{1}+1}$. We will partition this interval as follows. Let us define:
\begin{equation}\label{eq:As}
\mathcal A := \coprod_{s=1}^{\bar n} \mathcal A_{s}\qquad\qquad \mathcal A_{s}:= \left\{k_{0}\in\mathbb M_{\beta}\ |\ |k_{0}+s\eta|<\eta/2\right\}
\end{equation}
with\footnote{Since the width of the intervals is $\eta$, the number of disjoint intervals needed to cover the region $|k_{0} + \eta| \leq 2^{h_{1}+1}$ is less than $\lfloor2^{h_{1}+1-h_{\eta}}\rfloor$.} $\bar n=\min\{r, \lfloor2^{h_{1}+1-h_{\eta}}\rfloor\}$. Then, for all $k_{0}\in\mathcal A_{s}$ and for all $l\neq s$, we have:
\begin{equation} \label{eq:lowbd}
\|\ul k' + \ul p_{\leq l}\| \geq |k_{0}+l\eta|\geq \frac{|l-s|\eta}{2}.
\end{equation}
Let us number the lines appearing in $G$ with labels $l_{\alpha}$, $\alpha=1,\dots,|G|$. Consider the estimate for the chain in momentum space, in the support of $\mathcal{A}_{s}$:
\begin{equation}
1_{\mathcal A_{s}} \prod_{\alpha=1}^{|G|} \frac{1}{\|\ul k' + \ul p_{\leq l_{\alpha}}\|}\;,
\end{equation}
for $\mathcal{A}_{s}$ as in (\ref{eq:As}). If $s=l_{\bar \alpha}$ for some $\bar \alpha = 1,\ldots, |G|$, applying (\ref{eq:lowbd}) to all the other lines with indices $l_{\alpha}$ with $\alpha\neq\bar\alpha$, we obtain:
\begin{equation} \label{eq:fact1}
1_{\mathcal A_{s}}\prod_{\alpha=1}^{|G|} \frac{1}{\|\ul k' + \ul p_{\leq l_{\alpha}}\|}\leq \left(\frac{2}{\eta}\right)^{|G|-1} \frac{2^{-h_{l_{\bar\alpha}}}}{(\bar\alpha-1)!(|G|-\bar\alpha)!}\leq \left(\frac{4}{\eta}\right)^{|G|-1} \frac{2^{-h^{\mathrm{min}}}}{(|G|-1)!};
\end{equation}
in the first step we used that $|l_{\alpha}-l_{\bar\alpha}|\geq |\alpha-\bar\alpha|$, and in the last step we used the binomial theorem combined with the fact that $h_{l_{\bar\alpha}}\geq h^{\mathrm{min}}$. Instead, if $s\neq l_{\alpha}$ for all $\alpha = 1, \ldots, |G|$, then $\|\underline{k}' + \underline{p}_{\leq l_{\alpha}} \| \geq \eta/2$ for $k_{0}$ in $\mathcal{A}_{s}$. We write:
\begin{equation}
1_{\mathcal A_{s}}\prod_{\alpha=1}^{|G|} \frac{1}{\|\ul k' + \ul p_{\leq l_{\alpha}}\|} = 1_{\mathcal A_{s}} \Bigg[\prod_{\substack{\alpha=1 \\ l_\alpha \leq s}}^{|G|} \frac{1}{\|\ul k' + \ul p_{\leq l_{\alpha}}\|}\Bigg] \Bigg[\prod_{\substack{\alpha=1 \\ l_\alpha > s}}^{|G|} \frac{1}{\|\ul k' + \ul p_{\leq l_{\alpha}}\|}\Bigg]\;. 
\end{equation}
In the first term, we bound the denominator from below as $|k_{0} + l_{\alpha} \eta| \geq \eta (s - l_{\alpha} - 1/2)$, while in the second term we bound the denominator from below as $|k_{0} + l_{\alpha} \eta| \geq \eta (l_{\alpha} - s - 1/2)$. Consider the first product, and let us denote by $\bar \alpha$ the largest $\alpha$ in the product. We have: $(s - l_{\alpha} - 1/2) = (s - l_{\bar \alpha} - 1/2) + (l_{\bar \alpha} - l_{\alpha}) \geq 1/2 + |\alpha - \bar \alpha|$. Concerning the second product, we have $(l_{\alpha} - s - 1/2) \geq 1/2 + (\alpha - \bar \alpha)$. Therefore, using these bounds for all $\alpha\neq\bar\alpha$, and estimating the term with $\bar\alpha$ as
\begin{equation}
\|\ul k' +\ul p_{\leq l_{\bar\alpha}}\| \geq 2^{h_{l_{\bar\alpha}}-1} \geq 2^{h^{\mathrm{min}}-1}\;,
\end{equation}
we get:
\begin{equation}
\begin{split}
1_{\mathcal A_{s}}\prod_{\alpha=1}^{|G|} \frac{1}{\|\ul k' + \ul p_{\leq l_{\alpha}}\|} &\leq \left(\frac{C}{\eta}\right)^{|G|-1} \frac{2^{-h^{\mathrm{min}}}}{(\bar \alpha-1)!( |G| - \bar \alpha)!} \\
&\leq \left(\frac{K}{\eta}\right)^{|G|-1} \frac{2^{-h^{\mathrm{min}}}}{(|G|-1)!}\;.
\end{split}
\end{equation}

Next, let us consider the case in which $k_{0}\in \mathcal A^{c}$. In this case, we have:
\begin{equation}  \label{eq:facgain}
\|\ul k' + \ul p_{\leq l}\| \geq |k_{0}+l\eta|\geq \begin{cases}
 (l - 1/2)\eta & \text{if $k_0\geq -\eta/2$}\\
(r+1/2 - l)\eta & \text{if $k_0\leq-\eta r -\eta/2$,}
\end{cases}
\end{equation}
where the second case holds only if $\bar n= r$, \emph{i.e.} if $r <\lfloor2^{h_{1}+1-h_{\eta}}\rfloor$. Let us consider the first case. We apply the bound (\ref{eq:facgain}), together with $l_{\alpha} - 1/2 \geq \alpha - 1/2 \geq \alpha /2$, to all the lines with $\alpha \neq 1$ and we estimate the line with $\alpha = 1$ by $2^{-h^{\mathrm{min}}}$. We have:
\begin{equation} \label{eq:fact2}
1_{\mathcal A^{c}}\prod_{\alpha=1}^{|G|} \frac{1}{\|\ul k' + \ul p_{\leq l_{\alpha}}\|} \leq \left(\frac{2}{\eta}\right)^{|G|-1} \frac{2^{-h^{\mathrm{min}}}}{(|G|-1)!}\;.
\end{equation}
Instead, in the second case we bound the propagator with $\alpha = |G|$ by $2^{-h^{\mathrm{min}}}$ and we estimate the product over all the other lines in $G$ using the second of (\ref{eq:facgain}). We obtain: 
\begin{equation}
1_{\mathcal A^{c}} \prod_{\alpha=1}^{|G|} \frac{1}{\|\ul k' + \ul p_{\leq l_{\alpha}}\|} \leq 2^{-h^{\mathrm{min}}} \prod_{\alpha=1}^{|G|-1} \frac{1}{\eta(r+1/2 - l_{\alpha})} \leq  \left(\frac{2}{\eta}\right)^{|G|-1} \frac{2^{-h^{\mathrm{min}}}}{(|G|-1)!}\;.
\end{equation}
Combining the bounds (\ref{eq:fact1})-(\ref{eq:fact2}) together with equation (\ref{eq:scalebound}) only for lines in $G$, we are then led to
\begin{equation}
\prod_{\alpha=1}^{|G|} 2^{-h_{l_{\alpha}}} \leq \left(\frac{C}{\eta}\right)^{|G|-1} \frac{2^{-h^{\mathrm{min}}}}{(|G|-1)!}
\end{equation}
which is the claim, equation (\ref{eq:bdfact}).

Let us now briefly discuss the case in which one of the $A$-lines carries a momentum $-(m-1)\eta$. Here, we divide the chain $G$ in two parts, called $G_{1}$ and $G_{2}$, $G_{1}$ being on the left of the $A$-line carrying momentum $-(m-1)\eta$ and $G_{2}$ on the right. We extract the factorial gain separately for each part of the chain. We obtain:
\begin{equation}
\prod_{\alpha=1}^{|G|} 2^{-h_{l_{\alpha}}} \leq \left(\frac{C}{\eta}\right)^{|G|-1}\frac{2^{-h^{\text{min}}}}{(|G_{1}|-1)!(|G_{2}| - 1)!} \leq \left(\frac{K}{\eta}\right)^{|G|-1} \frac{2^{-h^{\text{min}}}}{(|G|-1)!}\;.
\end{equation}
To conclude, let us comment on the case of chains extending in the ultraviolet regime. Here, the chirality of the propagators is not defined, so in the above argument the single-scale propagator $g_{\omega_{s}}^{(h_{s})}(\underline{k}')$ is replaced by $g^{(h_{s})}(\underline{k}' + \underline{k}_{F}^{\omega})$ with $h_{s} > 0$. The analysis applies unchanged, using $2^{-h_{s}} \leq \| k_{0} + (\underline p_{\leq s})_{0} \|$. This concludes the proof of Lemma \ref{lemma:factorial}.
\end{proof}
Next, we will use Lemma \ref{lemma:factorial} to obtain an improved bound for the first term in (\ref{eq:GNsplit}).
%Applying the Lemma that we have just proven to each chain $\mathbf C_{j}$, via the replacements $\mathbf C\rightarrow \mathbf C_{j}$ and $r\rightarrow m_{j}$, one is led to the following Proposition.
\begin{proposition}[Improved estimate for low perturbative orders] \label{prop:SumBoundLow}
For any $m\geq 16$,
\begin{equation}
\frac{1}{\beta L}|\widehat W_{m;0}(\ul p_{1},\dots,\ul p_{m})| \leq m!\, C^{m}\, \frac{1}{\lfloor m/4-4\rfloor!}\,\eta^{2-m} (1 + \mathbbm 1_{m\geq m(\theta)} |\theta|^{\varepsilon/2 + c|\lambda|} \eta^{-c|\lambda|})\;,
\end{equation}
and for any $\alpha\in(0,1/34)$, and any $1\leq n\leq n_{*}\equiv \lfloor\alpha m\rfloor$
\begin{equation}
\begin{split}
&\frac{1}{\beta L}|\widehat W_{m;n}(\ul p_{1},\dots,\ul p_{m})| \\
&\quad \leq m!\, C^{m+n} \left|\lambda\right|^{n}\, \left(\frac{2n}{m-32n}\right)^{m(1/4 - 8\alpha)}\, \eta^{2-m}(1 + \mathbbm 1_{m\geq m(\theta)} |\theta|^{\varepsilon/2 + c|\lambda|} \eta^{-c|\lambda|})\;.
\end{split}
\end{equation}
\end{proposition}
\begin{proof} The starting point is the identity (\ref{eq:WT}). Let us focus on the propagators that belong to the chains $\mathbf C_{j}$.  One would be tempted to employ Lemma \ref{lemma:factorial} with $G=\mathbf C_{j}$ to extract a full factorial gain along the chain, namely
\begin{equation} \label{eq:chainL1naive}
\prod_{\ell\in \mathbf C_{j}} \|g^{(h_{\ell})}_{\ell}\|_{L^{1}}\leq C^{m_{j}} \prod_{i=1}^{m_{j}-1} 2^{-h_{i}} \leq \left(\frac{K}{\eta}\right)^{m_{j}-2} \frac{2^{-h^{\mathrm{min}}_{j}}}{(m_{j}-2)!}\;
\end{equation}
with $h^{\mathrm{min}}_{j}$ the minimum scale on $\mathbf C_{j}$; however, as it will become clear later in the proof, this will not work, due to the fact that the factor $(1/\eta)^{m_{j}-2}2^{-h^{\mathrm{min}}_{j}}$ might be much bigger than the factor $\prod_{i=1}^{m_{j}-1} 2^{-h_{i}}$. Estimating in this way the propagators that belong to ${\bf T}$ would not allow to obtain an estimate for the corresponding GN tree that is summable over the internal scale labels. 

In order to deal with this issue, we shall apply the previous reasoning to a proper subset $G^{j}\subset \mathbf C_{j}$ of \emph{good} lines, defined as follows.
\begin{definition}[Good and bad lines] \label{def:badlines}
Let $\widehat W^{(h)}_{\mathbf \Gamma}[\tau; {\bf T}]$ as in (\ref{eq:WT}), and suppose that $\mathbf C_{j}$ is a chain contained in ${\bf T}$. We define the set $B^{j}_{v}$ of {\it bad lines} associated to a non-trivial, internal vertex $v\in V(\tau)$ as follows. Let $\Gamma^{A,j}_{v}\subseteq \Gamma^{A}_{v}$ be the set of $A$-lines corresponding to vertices belonging to the chain $\mathbf C_{j}$ (see Fig. \ref{fig:chains}), contained in the cluster labelled by $v$. Recall that $T_{v}$ is the spanning tree of the cluster $v$.
\begin{itemize}
\item[a)] Suppose that $v$ has no sub-clusters, \emph{i.e.} it is followed only by leaves. If $|T_{v}\cap\mathbf C_{j}| < 3$, then $B^{j}_{v} = T_{v}\cap\mathbf C_{j}$. Otherwise, $B^{j}_{v} = \{\ell,\ell',\ell''\}$ for three arbitrary lines in $T_{v}\cap\mathbf C_{j}$.
\item[b)] Suppose that $v$ has at least one sub-cluster. The definition is given inductively, as follows:
\begin{itemize}
\item if $|\bigcup_{\bar v \succ v} B^{j}_{\bar v}| \geq 3$, then $B^{j}_{v} = \varnothing$;
\item if $|\bigcup_{\bar v \succ v} B^{j}_{\bar v}| = 3-k$ and $|T_{v}\cap\mathbf C_{j}| < k$, then $B^{j}_{v} = T_{v}\cap\mathbf C_{j}$. Otherwise, if $|T_{v}\cap\mathbf C_{j}| \geq k$ then $B^{j}_{v} = \{\ell_{1}, \ldots, \ell_{k}\}$ for $k$ arbitrary lines in $T_{v}\cap\mathbf C_{j}$.
\end{itemize}
\end{itemize}
We define the set of bad lines on the chain $\mathbf C_{j}$ as:
\begin{equation}
B^{j} := \bigcup_{v\notin V_{f}(\tau)} B^{j}_{v}\;,
\end{equation}
and the set of \emph{good lines} in the cluster $v$ and on the whole chain as:
\begin{equation}
G_{v}^{j}:= (T_{v}\cap\mathbf C_{j}) \setminus B_{v}^{j}\;,\qquad G^{j} := \bigcup_{v\notin V_{f}(\tau)} G^{j}_{v}\;.
\end{equation}
\end{definition}
\begin{remark}\label{rem:gb} Observe that if $|G_{v}^{j}| \neq 0$ then $|\bigcup_{\bar v\succeq v}B_{\bar v}^{j}| \geq 3$. %Observe that if $|\Gamma^{A,j}_{v}|\leq 4$, then $B^{j}_{v} = T_{v}\cap\mathbf C_{j}$. Thus, good lines can only arise for clusters with $|\Gamma^{A,j}_{v}| \geq 5$. Good lines do appear if $|\Gamma^{A,j}_{v}| \geq 6$.
\end{remark}
For the moment, let us ignore the action of the $\mathfrak{R}$ operator. We shall apply an estimate similar to (\ref{eq:chainL1naive}), but only on the good lines. Consider a chain $\mathbf C_{j}$ such that $m_{j}\geq 5$, and let $s_{j}\in\{1,\dots, m_{j}-1\}$ be a scale minimiser on the chain, \emph{i.e.} $h_{s_{j}}=h^{\mathrm{min}}_{j}$. Then:
\begin{equation} \label{eq:chainL1}
\prod_{\ell\in \mathbf C_{j}} \|g^{(h_{\ell})}_{\ell}\|_{L^{1}}\leq C^{m_{j}} \prod_{\ell\in \mathbf C_{j}} 2^{-h_{\ell}} \leq C^{m_{j}} \left(\frac{C}{\eta}\right)^{|G^{j}|-1} \frac{2^{-h_{s_{j}}}}{(|G^{j}|-1)!}\prod_{\ell\notin G^{j}} 2^{-h_{\ell}}\;;
\end{equation}
we can further rewrite the right-hand side of (\ref{eq:chainL1}) as:
\begin{equation} \label{eq:chainL1gain}
\prod_{\ell\in \mathbf C_{j}} \|g^{(h_{\ell})}_{\ell}\|_{L^{1}}\leq C^{m_{j}} \Big[ \prod_{\ell\in \mathbf C_{j}} 2^{-h_{\ell}}\Big] \left(\frac{C}{\eta}\right)^{|G^{j}|-1} \frac{1}{(|G^{j}| - 1)!} \prod_{\substack{\ell\in G^{j}\\ \ell\neq \ell_{s_{j}}}} 2^{h_{\ell}} \;.
\end{equation}
Observe that the line $\ell_{s_{j}}$ can be chosen to belong to $G^{j}$. This follows from the fact that, since we are assuming that $m_{j} \geq 5$, the largest cluster ({\it i.e.} the cluster on the smallest scale) containing the chain $\mathbf C_{j}$ must contain at least one good line. The action of $\mathfrak R$ on propagators of a chain is dimensionally equivalent to taking derivatives in real space of the propagators and multiplying them by weights, and it does not hinder the application of Lemma \ref{lemma:factorial}. Its only effect is to change the first factor product in (\ref{eq:chainL1gain}), ensuring that all cluster scaling dimensions are eventually negative.

The rewriting (\ref{eq:chainL1gain}) allows us to recover the standard estimate, up to an extra multiplicative factor. Indeed, we have:
\begin{equation} \label{eq:dimfac}
\frac{1}{\beta L}|\widehat W^{(h)}_{\bf \Gamma}[\tau]| \leq |L[\tau]| C^{m}\sum_{\mathbf T} \left[2^{h (2 - m)}\prod_{v\notin V_{f}(\tau)} \frac{2^{(h_{v}-h_{v'})D_{v}}}{s_{v}!} \right]\prod_{\substack{j=1\\ m_{j}\geq 5}}^{N_{c}}\frac{(C/\eta)^{|G^{j}|-1}}{(|G^{j}|-1)!} \prod_{\substack{\ell\in G^{j}\\ \ell\neq \ell_{s_{j}}}} 2^{h_{\ell}}\;,
\end{equation}
where, as usual, $D_{v}$ indicates the internal dimension of the cluster $v$, namely
\begin{equation}
\begin{split}
D_{v}&=2-\frac{|\Gamma_{v}^{\psi}|}{2}-|\Gamma_{v}^{A}|-z_{v} + c|\lambda|\qquad \text{if $h_{v}\leq 0$} \\
D_{v} &= 1 - n_{v} - |\Gamma_{v}^{A}|\qquad \text{if $h_{v} > 0$.}
\end{split}
\end{equation}
The last two factors in equation (\ref{eq:dimfac}) are the main difference with respect to the previous estimate (\ref{eq:treebd}). Let us focus on the factorial gain. By the multinomial theorem, we have that
\begin{equation} \label{eq:boundfact}
\prod_{\substack{j=1\\ m_{j}\geq 5}}^{N_{c}}\frac{1}{(|G^{j}|-1)!} \leq \frac{(N_{c}-N_{c}^{\leq 4})^{\sum_{j}^{*}(|G^{j}|-1)}}{\big[\sum_{j}^{*}(|G^{j}|-1)\big]!}\leq \left(\frac{eN_{c}}{\sum_{j}^{*}(|G^{j}|-1)}\right)^{\sum_{j}^{*}(|G^{j}|-1)}
\end{equation}
where $N_{c}^{\leq 4}$ is the number of chains with $m_{j}\leq 4$ and the star in the sum enforces the constraint $m_{j}\geq 5$. We shall derive a lower bound for $|G^{j}|$ in terms of $m_{j}$; we start by noticing that $|G^{j}| = m_{j} -1 - |B^{j}|$, so it is enough to find an upper bound for $|B^{j}|$. The worst possible case is when all $m_{j}$ leaves are distributed in clusters with $|\Gamma^{A,j}|=4$, each containing $3$ bad lines. Since the number of such clusters is bounded by $\lfloor m_{j}/4\rfloor$, we have
\begin{equation}
|B^{j}| \leq 3\lfloor m_{j} / 4 \rfloor + 3
\end{equation}
where the $+3$ takes into account the case in which $m_{j}$ is not divisible by $4$. Thus,%So, this upper bound yields
\[
|G^{j}|-1 \geq m_{j}  - 3\lfloor m_{j} / 4 \rfloor - 4\;.
\]
Hence, $\sum_{j}^{*}(|G^{j}|-1)$ is bounded below as:
\begin{equation} \label{eq:lowercomb}
\begin{split}
\sum_{\substack{j=1\\m_{j}\geq 5}}^{N_{c}} \Big(m_{j}  - 3\lfloor m_{j} / 4 \rfloor - 4\Big)
&\geq \sum_{\substack{j=1\\m_{j}\geq 5}}^{N_{c}} \Big(\frac{m_{j}}{4} - 4\Big)\\
&= \frac{m}{4} - \sum_{\substack{j=1\\m_{j}\leq 4}}^{N_{c}} \frac{m_{j}}{4}- 4 N_{c}^{\geq 5}\\
&\geq \frac{m}{4} -  N_{c}^{\leq 4} - 4 N_{c}^{\geq 5}\\
&\geq \frac{m}{4} - 4 N_{c}\;,
\end{split}
\end{equation}
where $N_{c}^{\leq b}$ (resp. $N_{c}^{\geq b}$) indicates the number of chains with length at most (resp. at least) $b$. Insertion of (\ref{eq:lowercomb}) into (\ref{eq:boundfact}) yields then
\begin{equation}
\prod_{\substack{j=1\\ m_{j}\geq 5}}^{N_{c}}\frac{1}{(|G^{j}|-1)!} \leq \Bigg( \frac{4eN_{c}}{m - 16N_{c}} \Bigg)^{\sum_{j}^{*} (|G^{j}| - 1)}\;.
\end{equation}
Observe that the number of chains is bounded as $N_{c} \leq 2n$, and we are assuming that $n\leq \lfloor\alpha m\rfloor$. Thus, for $\alpha$ small enough, the argument of the bracket is less than $1$. Hence, we get an upper bound if we replace the exponent by the lower bound in (\ref{eq:lowercomb}):
\begin{equation}
\prod_{\substack{j=1\\ m_{j}\geq 5}}^{N_{c}}\frac{1}{(|G^{j}|-1)!} \leq (4e)^{m}\Bigg( \frac{N_{c}}{m - 16N_{c}} \Bigg)^{\frac{m}{4} - 4 N_{c}}\;.
\end{equation}
We then observe that for $N_{c}<m/17$ the right-hand side is an increasing function of $N_{c}$. Thus, since $N_{c}\leq 2\lfloor\alpha m\rfloor$, for $\alpha$ small enough we get:
\begin{equation}
\prod_{\substack{j=1\\ m_{j}\geq 5}}^{N_{c}}\frac{1}{(|G^{j}|-1)!}\leq f(n;m):=\begin{cases}\displaystyle \frac{1}{\lfloor m/4 - 4\rfloor!} & n=0\\
\displaystyle \left(\frac{2n}{m-32n}\right)^{\frac{m}{4}-8\alpha m} & \mathrm{else\;,}
\end{cases}
\end{equation}
reducing  (\ref{eq:dimfac}) to
\begin{equation} \label{eq:finexp}
\frac{1}{\beta L}|\widehat W^{(h)}_{\bf \Gamma}[\tau]| \leq \tilde C^{m}|L[\tau]| f(n;m)\sum_{\mathbf T} \left[2^{h (2 - m)}\prod_{v\notin V_{f}(\tau)} \frac{2^{(h_{v}-h_{v'})D_{v}}}{s_{v}!} \right] \prod_{\substack{j=1\\ m_{j}\geq 5}}^{N_{c}} \prod_{\substack{\ell\in G^{j}\\ \ell\neq \ell_{s_{j}}}} 2^{h_{\ell}-h_{\eta}}\;,
\end{equation}
with $\tilde C=4eC$. Next, let us consider the last factor in the right-hand side. This factor is a priori dangerous, due to the fact that $h_{\ell} \geq h_{\eta}$. We will take care of this bad factor using the negativity of the scaling dimensions. Let us define the set:
\begin{equation}
\frak{G}^{j}_{v}:=\Big\{\ell\in\bigcup_{v\preceq \bar v} T_{\bar v}\ \Big|\,\ \ell\in\mathbf C_{j}\setminus(B^{j}\cup\{\ell_{s_{j}}\})\Big\}\;,\qquad \frak{G}_{v} = \bigcup_{j} \frak{G}^{j}_{v}\;.
\end{equation}
Observe that $\frak{G}^{j}_{v} = \bigcup_{v\preceq \bar v} (G^{j}_{\bar v} \setminus \{\ell_{s_{j}}\})$. Also, observe that, as a consequence of Definition \ref{def:badlines}, $|\frak{G}^{j}_{v}| = 0$ if $|\Gamma^{A,j}_{v}| \leq 4$, and in particular $|\frak{G}^{j}_{v}| = 0$ if $m_{j} \leq 4$. In fact, if $|\frak{G}^{j}_{v}|\neq 0$ then $|B^{j}_{v}| \geq 3$, which means that $|\Gamma^{A,j}_{v}|$ must be at least $5$. Also, observe that $\frak{G}_{v_{0}}=\bigcup_{j : m_{j}\geq 5} (G^{j}\setminus\{\ell_{s_{j}}\})$. 

By a telescopic argument, we rewrite the last factor in (\ref{eq:finexp}) as:
\begin{equation}
\begin{split}
\prod_{\substack{j=1\\ m_{j}\geq 5}}^{N_{c}} \prod_{\ell\in G^{j}\setminus\{\ell_{s_{j}}\}} 2^{h_{\ell}} &= \prod_{\substack{j=1\\ m_{j}\geq 5}}^{N_{c}} \prod_{v\notin V_{f}(\tau)} 2^{h_{v}| G^{j}_{v} \setminus \{\ell_{s_{j}}\} |} \\
&= 2^{h |\frak{G}_{v_{0}}|} \prod_{v \notin V_{f}(\tau)} 2^{(h_{v}-h_{v'})|\frak{G}_{v}|}\;,
\end{split}
\end{equation}
which allows to rewrite (\ref{eq:finexp}) as
\begin{equation}
\frac{1}{\beta L}|\widehat W^{(h)}_{\bf \Gamma}[\tau]| \leq \tilde C^{m}|L[\tau]| f(n;m)\sum_{\mathbf T} \frac{2^{h \tilde D_{v_{0}}}}{\eta^{|\mathfrak G_{v_{0}}|}}\prod_{v\notin V_{f}(\tau)} \frac{2^{(h_{v}-h_{v'})\tilde D_{v}}}{s_{v}!}\;,
\end{equation}
with modified scaling dimension
\begin{equation}
\begin{split}
\tilde D_{v} &= 2 - \frac{|\Gamma^{\psi}_{v}|}{2} - |\Gamma^{A}_{v}|+|\frak{G}_{v}| -z_{v} + c|\lambda|\qquad \text{if $h_{v} \leq 0$}\\
\tilde D_{v} &= 1 - n_{v} - |\Gamma^{A}_{v}| + |\frak{G}_{v}|\qquad \text{if $h_{v} >0$,}
\end{split}
\end{equation}
and modified external dimension
\begin{equation}
\begin{split}
\tilde D_{v_{0}} &= 2 - m + |\frak{G}_{v_{0}}| \qquad \text{if $h \leq 0$}\\
\tilde D_{v_{0}} &= 1 - n - m + |\frak{G}_{v_{0}}| \qquad \text{if $h > 0$}\;.
\end{split}
\end{equation}
The key feature that prompted us to extract a factorial gain only from the good lines becomes now clear: since by construction we have that
\begin{equation}
|\frak{G}_{v}| \leq \Theta(|\Gamma^{A}_{v}| - 5) \times \Big(|\Gamma^{A}_{v}| - 4\Big)\;,
\end{equation}
with $\Theta$ the Heaviside function, it follows that $\tilde D_{v}$ is always strictly negative, hence we have preserved summability over the internal scales and labellings of the GN tree. Had we taken as sets of good lines the whole chains, this might not have been the case.

In particular, the external scaling dimension $\tilde D_{v_{0}}$ is also strictly negative: therefore, proceeding as in the proof of Proposition \ref{prop:CorrDimBound}, but with modified scaling dimensions, one is able to show that summation over choices of $(h,\tau,\mathbf \Gamma, \mathbf T)$ yields
\begin{equation}\label{eq:modif}
\begin{split}
&\sum_{h=h_{\beta}}^{0} \sum_{\tau\in \widetilde{\mathcal T}_{h,n}^{m}} \sum_{\mathbf\Gamma\in \mathcal L_{\tau}}|L[\tau]|\sum_{\mathbf T}  \frac{2^{h(2-m+|\frak{G}_{v_{0}}|)}}{\eta^{|\frak{G}_{v_{0}}|}}\prod_{v\notin V_{f}(\tau)} \frac{2^{(h_{v}-h_{v'})\tilde D_{v}}}{s_{v}!} \\
&\qquad \leq C^{n+m} |\lambda|^{n} \,\eta^{2-m} (1 + \mathbbm 1_{m\geq m(\theta)} |\theta|^{\varepsilon/2 + c|\lambda|} \eta^{-c|\lambda|})\;.
\end{split}
\end{equation}
The bound (\ref{eq:modif}) is proved as (\ref{eq:estcorr}), replacing $|\Gamma^{A}_{v}|$ in the proof of (\ref{eq:estcorr}) with $|\Gamma^{A}_{v}| - |\frak{G}_{v}|$. This concludes the proof of Proposition \ref{prop:SumBoundLow}.
\end{proof}
This improved bound is enough to show a summable estimate for the $m$-point cumulants.
\begin{proposition}[Improved estimate for density correlators] \label{thm:small}
For $m\geq 16$, there exists $C>0$, such that for $\lambda$ sufficiently small (but fixed), and for $\alpha \in(0,1/34)$, we have that
\begin{equation}\label{eq:small}
\begin{split}
&\frac{1}{\beta L}|\langle\mathbf T \hat n_{\ul p_{1}};\cdots; \hat n_{\ul p_{m-1}}; \hat \jmath_{\mu,\ul p_{m}}\rangle_{\beta,L}| \\
&\quad \leq m!\, \eta^{2-m}\,C^{m}\left[\frac{1}{\lfloor m/4-4\rfloor!}+|\lambda|\left(\frac{2\alpha}{1-32\alpha}\right)^{m(1/4-8\alpha)} + |\lambda|^{\alpha m}\right] R_{m}(\eta,\theta)
\end{split}
\end{equation}
with
\begin{equation}
R_{m}(\eta,\theta) = (1 + \mathbbm 1_{m\geq m(\theta)} |\theta|^{\varepsilon/2 + c|\lambda|} \eta^{-c|\lambda|})\;.
\end{equation}
\end{proposition}
\begin{proof} The starting point is Eq. (\ref{eq:GNsplit}). As already commented, the first term in the right-hand side is bounded by $C^{m} m!$, which is much smaller than the right-hand side of (\ref{eq:small}). Then, to estimate the first sum in the right-hand side of (\ref{eq:GNsplit}) we use Proposition \ref{prop:SumBoundLow}. The contribution with $n=0$ reproduces the first term in the right-hand side of (\ref{eq:small}). The other terms in the sum are bounded as: 
\begin{equation}
\frac{1}{\beta L}\left|\sum_{n=1}^{n_{*}} \widehat W_{m;n}(\ul p_{1},\dots,\ul p_{m})\right|\leq m!\, C^{m}\, \left(\frac{2\alpha}{1-32\alpha}\right)^{m(1/4-8\alpha)}\, \eta^{2-m} R_{m}(\eta,\theta)\sum_{n=1}^{n_{*}} C^{n}\left|\lambda\right|^{n}\;
\end{equation}
and the thesis is obtained by estimating the sum with $C|\lambda|/(1-C|\lambda|)$. Finally, consider the second sum in the right-hand side of (\ref{eq:GNsplit}). Here it is sufficient to proceed as in the proof of Proposition \ref{prop:CorrDimBound}:
\begin{equation}
\begin{split}
\frac{1}{\beta L}\left|\sum_{n>n_{*}} \widehat W_{m;n}(\ul p_{1},\dots,\ul p_{m})\right| &\leq m!\, C^{m}\,\eta^{2-m} R_{m}(\eta,\theta) \sum_{n>n_{*}} C^{n}|\lambda|^{n} \\&\leq m!\, C^{m(1+\alpha)}|\lambda|^{\alpha m}\,\eta^{2-m}R_{m}(\eta,\theta)\;.
\end{split}
\end{equation}
This concludes the proof of the proposition.
\end{proof}
\begin{remark}
\begin{itemize}
\item[(i)] Recall that the presence of the non-trivial factor $R_{m}(\eta,\theta)$ is due to the estimate for the GN trees with one off-diagonal $A$ endpoint. If we only restrict to the contribution of these trees, we have:
\begin{equation}\label{eq:estodimpro}
\begin{split}
&\frac{1}{\beta L}|\langle\mathbf T \hat n_{\ul p_{1}};\cdots; \hat n_{\ul p_{m-1}}; \hat \jmath_{\mu,\ul p_{m}}\rangle^{(\text{i.r.}), (\text{od})}_{\beta,L}|  \leq C^{m} m!\, \eta^{2-m} \\&\quad \cdot \left[\frac{1}{\lfloor m/4-4\rfloor!}+|\lambda|\left(\frac{2\alpha}{1-32\alpha}\right)^{m(1/4-8\alpha)} + |\lambda|^{\alpha m}\right] \mathbbm 1_{m\geq m(\theta)} |\theta|^{\varepsilon/2 + c|\lambda|} \eta^{-c|\lambda|}\;.
\end{split}
\end{equation}
Plugging this estimate in (\ref{eq:fullrespint}), we see that the corresponding contribution is summable in $m$ and it is vanishing as $\eta \to 0^{+}$, for $|\theta| \leq K\eta$ and for $|\lambda|, \eta$ small enough, as assumed in Theorem \ref{thm:main}.
\item[(ii)] Proposition \ref{thm:small} allows to prove that the contribution of infrared trees with diagonal $A$ end-points is summable, but so far the bounds are not strong enough to show that the linear response dominates. Actually, without a smallness condition on the perturbation, the sum of the higher order terms is of the same order of the linear response. In order to prove that the higher order terms are indeed subleading as $\eta \to 0$, we will prove a cancellation for the scaling limit of the density-density correlation functions, which is reminiscent of bosonization.
\end{itemize}
%Notice that Proposition \ref{thm:small}  allows to prove that higher order corrections are summable in $m$, as long as $\alpha$ and $\lambda$ are chosen small enough. Nevertheless, it is not sufficient to prove \emph{smallness} in $\eta$ of these corrections, thus not allowing to prove absence of corrections to Kubo's formula: in order to prove such smallness, we will prove a cancellation for the scaling limit of the density-density correlation functions.
\end{remark}
\section{The reference model}\label{Sec:RefMod}
Here we discuss a $1+1$ dimensional quantum field theory, called the reference model, that describes the infrared scaling limit of the correlation functions of the lattice model. We will follow the work \cite{MPmulti}, to which we refer for further details. The key feature of the model is that its correlation functions can be explicitly determined, combining Schwinger-Dyson equations and Ward identities, associated with the (anomalous) chiral gauge symmetry of the QFT. Also, the low-energy behavior of the lattice correlation functions can be quantitatively captured by the reference model, for a suitable choice of parameters.
\begin{remark}
All bare parameters of this reference model QFT should carry a `ref' superscript to distinguish them from the corresponding lattice model ones. In order to avoid cluttering an already heavy notation, we shall omit such superscript for the rest of this section (unless otherwise stated). We will reintroduce it in Section \ref{sec:final}, when comparing lattice and reference model correlations.
\end{remark}
\subsection{Sets and cutoffs}
%s 
Let $\beta>0$, and let $L, \mathfrak N_{0}, \mathfrak N_{1} \in \mathbb{N}$ such that $\mathfrak a:=\beta/\mathfrak N_{0}=L/\mathfrak N_{1}$. We define the space-time lattice of sides $\beta, L$ and mesh $\Ma$ as:
\begin{equation}\label{LambdaRM}
\La:=\left\{\underline x=(n_0\Ma,n_1\Ma) \;|\;0\le n_\mu\le \mathfrak N_{\mu}-1\right\}.
\end{equation}	 
Also, we define the set of momenta compatible with antiperiodic boundary conditions as:
\begin{equation}\label{RM-FS}
\D := \left\{\underline k=\left(\frac{2\pi}{\beta}\left(m_0+\frac{1}{2}\right),\frac{2\pi}{L}\left(m_1+\frac{1}{2}\right)\right)\;\Bigg|\;0\le m_\mu\le \mathfrak N_{\mu} -1\right\}.
\end{equation}
We will consider Grassmann fields defined on $\La$, satisfying antiperiodic boundary conditions. Let $\chi$ be the smooth cutoff function defined in equation (\ref{eq:cutoff}). We define its $\varepsilon$-deformation as:
\begin{equation}
\label{chi-def}
\chi^{\varepsilon}(t) = C_\varepsilon\int_0^{\infty}ds e^{-|t-s|^2/\varepsilon}\chi(s)\;,	
\end{equation}	
where $C_{\varepsilon}$ is chosen such that $C_{\varepsilon}\int dt \, e^{-|t-s|^2/\varepsilon} = 1$. We have that $\lim_{\varepsilon\to 0} \chi^{\varepsilon}(t) = \chi(t)$ and $\chi^{\varepsilon}(t) > 0$ for all $t$. Furthermore, given $v_\omega\in\mathbb R\setminus\{0\}$ for $\omega=1,\dots,N_{f}$, we define the norm:
\begin{equation}
|\underline k|_\omega:=\sqrt{k_0^2+v_\omega^2k_1^2}
\end{equation}
where the collection $\{v_\omega\}_{\omega}$ will play the role of bare velocities for the reference model. The set $\D$ is naturally equipped with periodic boundary conditions. Thus, it is convenient to adapt this norm to $\D$ by introducing:
\begin{equation}
\|\underline k\|_\omega:=\inf_{a_1,a_2\in\mathbb Z}|\underline k-a_1\underline G_1-a_2\underline G_2|_\omega\;,	
\end{equation}	 
with $\underline{G}_{1} = (2\pi / \Ma, 0)$, $\underline{G}_{2} = (0, 2\pi / \Ma)$. Finally, we shall introduce a function encoding an ultraviolet cutoff, as follows:
\begin{equation}\label{CutOff_hN}
\chi^{\omega,\varepsilon}_{N}(\underline k):=\chi^{\varepsilon}(2^{-N}\|\underline k\|_{\omega})\;.
\end{equation}
Observe that the function $\chi^{\omega}_{N}(\underline k) = \lim_{\varepsilon \to 0} \chi^{\omega,\varepsilon}_{N}(\underline k)$ is supported for momenta $\underline k$ such that $\|\underline k\|_\omega\le \frac{3}{2}2^{N}$. The limit $N\to \infty$ will be referred to as the ultraviolet limit.
\subsection{Fields and propagators} 
To each $\underline k \in\D$, and for each $\omega=1,\ldots, N_{f}$, we associate Grassmann variables $\hat \psi^{\pm}_{\underline k,\omega}$. They are extended periodically to the whole $\frac{2\pi}{\beta}(\Z+1/2)\times\frac{2\pi}{L} (\mathbb{Z} + 1/2)$ as:
\begin{equation}\label{hatPhi}
\hat{\psi}^{\pm}_{\underline k+a_{1}\underline G_1+a_{2}\underline G_2,\omega}:=\hat{\psi}^{\pm}_{\underline k,\omega}\qquad \forall (a_{1},a_{2})\in\Z^{2}\;.	
\end{equation}
The configuration-space Grassmann field is then:
\begin{equation}\label{eq:xspace}
{\psi}^{\pm}_{\underline x,\omega}:=\frac{1}{\beta L}\sum_{\underline k\in\D}e^{\mp i\underline k\cdot\underline x}\hat{\psi}^{\pm}_{\underline k,\omega}\;,
\end{equation}
and we notice that it satisfies $(\beta, L)$-antiperiodicity in space-time. The relation (\ref{eq:xspace}) can be inverted as:
\begin{equation}
\hat{\psi}^{\pm}_{\underline k,\omega}=\Ma^2\sum_{\underline x\in\La}e^{\pm i\underline k\cdot \underline x}\psi^{\pm}_{\underline x,\omega}\;.	
\end{equation}	
The momentum-space propagator is set to be, for $\ul k\in\D$,
\begin{equation}\label{propagator_hN}
\hProp(\underline k) :=\frac{1}{Z_\omega}\frac{\chi^{\omega,\varepsilon}_{N}(\underline k)}{\mathcal D_{\omega,\Ma}(\underline k)}\;,\quad \Den(\underline k) :=\frac{i}{\Ma}\sin(\Ma k_0)+\frac{v_\omega}{\Ma}\sin (\Ma k_1)\;,		
\end{equation}		 
where $Z_\omega,v_\omega$ are free parameters, to be suitably chosen later on. Observe that the denominator in \eqref{propagator_hN} vanishes for $\ul k=(0,0),(\pi/\Ma,0),(0,\pi/\Ma),(\pi/\Ma,\pi/\Ma)$ (and their translations by $\underline{G}_{i}$). However, for $\frak{a}>0$ these points do not belong $\D$, which means that the propagator is well-defined. Let us define: 
\begin{equation}
\hat g^{(\leq N)}_{\omega}(\underline k) := \lim_{\varepsilon \to 0}\lim_{\Ma\to 0}\hProp (\ul k)=\frac{1}{Z_{\omega}}\frac{\chi^{\omega}_{N}(\underline k)}{ik_0+v_\omega k_1}\;.
\end{equation}
This limiting propagator is only singular at $(0,0)$ (and its translations). For non-zero $\varepsilon, \frak{a}$, the propagator is arbitrarily small at $(\pi/\Ma,0),(0,\pi/\Ma),(\pi/\Ma,\pi/\Ma)$, provided $\frak{a}$ is small enough. Thus, in the regime $\frak{a} \ll \varepsilon$, the main singular contribution is due to the point $(0,0)$. Finally, the real-space propagator is:
\begin{equation}
\Prop(\underline x):=\frac{1}{\beta L}\sum_{\underline k\in\D}e^{i\underline k\cdot \underline x}\hProp(\underline k)\;,\qquad \forall \underline x\in\Ma\mathbb Z^2\;.	
\end{equation}	
\subsection{Generating functional of correlations}
Let $P_N[d\psi]$ be the Grassmann Gaussian integration with the propagators $\hat g^{(\leq N)}_{\omega, \varepsilon, \Ma}(\underline{k})$:
\begin{equation}
\int P_N[d\psi]\hat\psi^-_{\underline k,\omega}\hat\psi^+_{\underline q,\omega'}=\beta L\delta_{\underline k,\underline q}\delta_{\omega,\omega'}\,\hProp(\underline k)\;;
\end{equation}	
equivalently, 
\begin{equation}
	\int P_N[d\psi]\psi^-_{\underline x,\omega}\psi^+_{\underline y,\omega'}=\delta_{\omega,\omega'}\,\Prop(\underline x-\underline y)\;.	
\end{equation}	
The interaction of the Grassmann QFT is defined as:
\begin{equation}\label{eq:intgra}
V(\psi) := \frac{\Ma^4}{2}\sum_{\underline x,\underline y\in \La}\sum_{\omega,\omega '}\lambda_{\omega\omega '}Z_{\omega}Z_{\omega '} v(\underline x-\underline y) n_{\underline x,\omega}n_{\underline y,\omega '}	
\end{equation}	
with $\lambda_{\omega\omega'}\in\mathbb R,\;\lambda_{\omega\omega'}=\lambda_{\omega'\omega},\;\lambda_{\omega\omega}=0$. The function $v$ satisfies the periodicity condition:
\begin{equation}\label{PC}
v(\underline x+a_0\beta\underline e_0+a_1 L\underline e_1)=v(\underline x)\qquad \qquad \forall (a_{0},a_{1})\in \Z^{2},	
\end{equation}	
and also $\hat v(0)=1$. In equation (\ref{eq:intgra}), $n_{\underline x,\omega}=\psi^+_{\underline x,\omega}\psi^-_{\underline x,\omega}$ is the Grassmann counterpart of the density operator, as usual. Let $\Dp$ be the set of momenta compatible with periodic boundary conditions:
\begin{equation}
\Dp:=\left\{\underline p=\left(\frac{2\pi}{\beta}n_0,\frac{2\pi}{L}n_1\right)\;\Bigg|\;(n_{0},n_{1})\in\mathbb Z^{2}\right\}.
\end{equation}
Then, we can write:
\begin{equation}
v(\underline x)=\frac{1}{\beta L}\sum_{\underline p\in \Dp}e^{i\underline p\cdot\underline x}\,\hat v(\underline p)\;,\qquad \hat v(\underline p)=\Ma^2\sum_{\underline x\in\Lambda_{\beta,L,\Ma}}e^{-i\underline p\cdot\underline x}\,v(\underline x)\;.	
\end{equation}	
We will suppose the function $\hat v$ to be the restriction to $\Dp$ of an even, smooth and compactly supported function on $\mathbb R^2$, so that 
\begin{equation}		
\|\underline x\|^n|v(\underline x)|\le C_n\qquad \forall n\in\mathbb N,		
\end{equation}
where $\|\underline x\|:=\inf_{a_0,a_1\in\mathbb Z}|\underline x-a_0\beta\underline e_0-a_1L\underline e_1|$. In Fourier space, equation (\ref{eq:intgra}) reads:
\begin{equation}\label{V}
V(\psi)=\frac{1}{2\beta L}\sum_{\underline p\in \Dp}\sum_{\omega,\omega'}\lambda_{\omega\omega'}Z_{\omega}Z_{\omega'}\hat v(\underline p)\hat n_{-\underline p,\omega}\hat n_{\underline p,\omega'}	
\end{equation}	
where
\begin{equation}\label{n}
\hat n_{\underline p,\omega}=\Ma^2\sum_{\underline x\in\La}e^{-i\underline p\cdot\underline x}n_{\underline x,\omega}=\frac{1}{\beta L}\sum_{\underline k\in\D}\hat\psi^+_{\underline k-\underline p,\omega}\hat \psi^-_{\underline k,\omega}.	
\end{equation}			
Let us define the source terms as:
\begin{equation}\label{Gamma}\begin{split}
B(\psi;\phi)&:=\sum_{\omega=1}^{N_{f}}\Ma^2\sum_{\underline x\in \La}\left[\psi^+_{\underline x,\omega}\mathsf Q_{\omega}\phi^-_{\underline x,\omega}+\phi^+_{\underline x,\omega}\bar{\mathsf Q}_{\omega}\psi^-_{\underline x,\omega}\right]\\
C(\psi;A)&:=\sum_{\omega=1}^{N_{f}}\sum_{\mu=0}^{1}\Ma^2\sum_{\underline x\in \La}\mathsf Z^{\mu}_{\omega}A^{\mu}_{\underline x,\omega}n_{\underline x,\omega}=\sum_{\omega=1}^{N_{f}}\sum_{\mu=0}^{1}\frac{1}{\beta L}\sum_{\underline p\in \Dp}\mathsf Z^{\mu}_{\omega}\hat A^{\mu}_{-\underline p,\omega}\hat n_{\underline p,\omega}
\end{split}	
\end{equation}
where the external fields $A_{\underline x,\omega},\phi^{\pm}_{\underline x,\omega}$ take the form
\begin{equation}\label{A}
A^{\mu}_{\underline x,\omega}=\frac{1}{\beta L}\sum_{\underline p\in\Dp}e^{i\underline p\cdot \underline x}\hat A^{\mu}_{\underline p,\omega}\;,\qquad \phi^{\pm}_{\underline x,\omega}:=\frac{1}{\beta L}\sum_{\underline k\in \D}e^{\mp i\underline k\cdot\underline x}\hat \phi^{\pm}_{\underline k,\omega}\;.
\end{equation}	
\begin{remark}
    As far as correlations of the reference model are concerned, the constants $\mathsf Q^\pm_\omega$ and $\mathsf Z^{\mu}_{\omega}$ can be set to 1 without loss of generality. They will become relevant, and will be chosen specifically, when we will perform the comparison of reference model correlations with the corresponding ones of the lattice model, see Propositions \ref{prop:CorrComparison}-\ref{prop:PropVertexComparison} below.
\end{remark}
\begin{remark} Later, we shall use the reference model to approximate the correlation functions of the lattice model. To do this, it will be enough to consider an external field that does not depend on the chirality label: $A^{\mu}_{\underline x,\omega}\equiv A^{\mu}_{\underline x}$. This is consistent with the fact that the external field of the lattice model does not depend on the chirality label. The reason why we introduce more structure in the source term at this stage, is that the chirality-dependent external field will be needed to study the (chiral) correlation functions of the reference model. These correlation function satisfy suitable anomalous Ward identities, which we will use to infer key relations between the running coupling constants of the reference model; these relations will allow us to control their renormalization group flow. The construction of the RG fixed point of the reference model will be the key starting point to quantitatively analyze the original lattice model, through a rigorous comparison with the reference model.
\end{remark}

Finally, we define the partition function with source terms as
\begin{equation}\label{PartFunct}
\mathcal Z_{\beta,L,\Ma,\varepsilon,N}(\phi,A):=\int P_N[d\psi]e^{-V(\psi)+B(\psi;\phi)+C(\psi;A)}\;,	
\end{equation}
and also the generating functional of correlations as
\begin{equation}\label{Gen}
\mathcal{W}_{\beta,L,\Ma,\varepsilon,N}(\phi,A):=\log\frac{\mathcal Z_{\beta,L,\Ma,\varepsilon,N}(\phi,A)}{\mathcal Z_{\beta,L,\Ma,\varepsilon,N}(0,0)}\;. 	
\end{equation}
Then, the $m$-point density-current correlation functions are set to be
\begin{equation}\label{m-corr}
\prod_{i=1}^{m}[\mathsf Z^{\mu_{i}}_{\omega_{i}}]\langle \hat n_{\underline p_1,\omega_1};\cdots;\hat n_{\underline p_{m},\omega_m}\rangle_{\beta,L,\mathfrak a,\varepsilon,N}:=(\beta L)^{m}\frac{\partial^m\; \mathcal{W}_{\beta,L,\Ma,\varepsilon,N}(\phi,A)}{\partial \hat A^{\mu_{1}}_{-\underline p_1,\omega_1}\cdots\partial\hat A^{\mu_{m}}_{-\underline p_{m},\omega_m}}\Bigg|_{A=\phi=0}	\;,
\end{equation}
with $\underline p_{m} := -\sum_{i=1}^{m-1}\underline p_i$. At fixed cut-offs, all expressions are well-defined, for $|\lambda_{\omega\omega'}|$ small enough, thanks to the convergence of fermionic cluster expansion. Renormalization group methods are then used to prove a result that is uniform in all cut-offs.

\subsection{Sketch of RG analysis of the reference model}
The generating functional admits a multi-scale expansion similar to the one presented for the lattice model in Section \ref{sec:RGLatticeModel}, but now the scales will run from $h_{\beta}$ to $N$; we shall distinguish two different regimes, according to whether the scale is positive or negative.
\paragraph{Ultraviolet regime.} We rewrite the total propagator as its infre-red part plus the sum of single (positive) scale propagators $\hat g^{(h)}_{\omega,\varepsilon,\mathfrak a}(\underline k)$, that is 
\begin{equation} \label{eq:propUV}
\hat g^{(\leq N)}_{\omega,\varepsilon,\mathfrak a}(\underline k) =	\hat g^{[h_{\beta},0]}_{\omega,\varepsilon,\mathfrak a}(\underline k)+\sum_{h=1}^N\hat g^{(h)}_{\omega,\varepsilon,\mathfrak a}(\underline k)\;,\qquad \hat g^{(h)}_{\omega,\varepsilon,\mathfrak a}(\underline k) =\frac{1}{Z_{\omega}} \frac{\chi_{h}^{\omega,\varepsilon}(\underline k)-\chi_{h-1}^{\omega,\varepsilon}(\underline k)}{\Den(\ul k)}\;,
\end{equation}	 
and we iteratively integrate them out, obtaining
\begin{equation}\label{eq:leqw}
\mathcal{W}_{\beta,L,\Ma, \varepsilon,N}(\phi,A) = \mathcal{W}^{(h)}_{\beta,L,\Ma, \varepsilon,N}(\phi,A) + \log \int P_{h}[d\psi^{(\leq h)}] e^{V^{(h)}(\sqrt{Z}\psi^{(\leq h)}; \phi,A)}
\end{equation}
with $P_{h}$ being the Grassmann Gaussian measure associated to $\{\hat g^{(\leq h)}_{\omega,\varepsilon,\Ma}\}_{\omega}$ and $\sqrt{Z}\psi^{(\leq h)}$ a short-hand for $\sqrt{Z_{\omega}}\psi^{(\leq h)}_{\omega}$. The effective interaction $V^{(h)}$ takes the form
\begin{equation} \label{eq:VhUV}
V^{(h)}(\psi; \phi,A) = \sum_{\Gamma=(\Gamma_{\psi},\Gamma_{\phi},\Gamma_{A})} \int_{\beta, L} d \underline{X}\,\psi_{\Gamma_{\psi}}(\underline{X}_{\psi}) \phi_{\Gamma_{\phi}}({\ul X_{\phi}}) A_{\Gamma_{A}}({\ul X_{A}})  W_{\Gamma}^{(h)}(\ul X)
\end{equation}
and $\mathcal W^{(h)}_{\beta,L,\Ma,\varepsilon, N}$ admits an analogous expansion.
For $k\in \mathbb N$, if we define the intensive weighted $k$-norms
\begin{equation}\label{Norm}
\|W_\Gamma^{(h)}\|_{1,k}:=\sum_{\substack{\{m_{ij}\}_{i,j=1}^{|\Gamma|}\\\sum_{ij}m_{i,j}=k}}\frac{\mathfrak a^{2|\Gamma|}}{|\Lambda_{\beta,L,\mathfrak a}|}\sum_{\underline X} |W^{(h)}_\Gamma(\underline X)|	\prod_{i<j}\|\underline x_i-\underline x_j\|^{m_{ij}}_{\beta,L}\;,	
\end{equation}
then, for the case of relevant and marginal kernels, these norms satisfy bounds which are better than the standard ones given in (\ref{eq:boundeffpot}).
\begin{proposition}[Improved estimates for UV kernels, Proposition 8.1 of \cite{MPmulti}]
Let us call $\lambda_{\mathrm{max}}:=\max_{\omega,\omega'}|\lambda_{\omega\omega'}|$. Then, for $\lambda_{\mathrm{max}}$ small enough uniformly in all cut-off parameters, for all $h\in[0,N]$, there exists $C>0$ independent of $\beta$, $L$, $\Ma$, $\varepsilon$, $N$, such that:
\begin{equation}\label{eq:UVest}
\|W^{(h)}_\Gamma-W^{(N)}_\Gamma\|_{1,k}\le C\lambda_{\mathrm{max}}2^{-kh}2^{h(D_{\Gamma}-\vartheta_{\Gamma})}\bigg[\prod_{f\in\Gamma_A}\big|\mathsf Z^{\mu(f)}_{\omega(f)}\big|\bigg]\bigg[\prod_{f\in\Gamma_{\phi}}\big|\mathsf Q^{\varepsilon(f)}_{\omega(f)}\big|\bigg]\;,
\end{equation}	  
where $D_{\Gamma}$ is the scaling dimension of $W^{(j)}_{\Gamma}$, namely
\[
D_{\Gamma}:= 2 - \frac{|\Gamma_{\psi}|}{2} -\frac{3|\Gamma_{\phi}|}{2} - |\Gamma_{A}|
\]
and $\vartheta_{\Gamma}$ is given by
\begin{equation}
\vartheta_{\Gamma}=\begin{cases}
2 & \text{if } (|\Gamma_{\psi}|,|\Gamma_{\phi}|,|\Gamma_{A}|) = (2,0,0)\\
1 & \text{if } (|\Gamma_{\psi}|,|\Gamma_{\phi}|,|\Gamma_{A}|) = (4,0,0), (1,1,0), (2,0,1)\\
0 & \text{otherwise.}
\end{cases}
\end{equation}
\end{proposition}
These improved estimates suggest, following \cite{BFM}, that it is convenient to localize the relevant and marginal kernels wholly: indeed, for any $h\geq 0$, we set
\begin{equation} \label{eq:locUV}
\mathfrak L V^{(h)}_{n,m,l}(\psi,\phi,A) := \begin{cases}
V^{(h)}_{n,m,l} & \text{if } (n,m,l)=(2,0,0), (4,0,0), (1,1,0), (2,0,1)\\
0 & \text{else.}
\end{cases}
\end{equation}
Reintroducing temporarily the `ref' superscript, we see that (\ref{eq:propUV}) implies that the bare couplings $(Z^{\mathrm{ref}}_{\omega},v^{\mathrm{ref}}_{\omega})$ do not need to be updated for positive scales. The relevant and marginal kernels are controlled by (\ref{eq:UVest}), and all scaling dimensions in the GN tree expansion are negative.
\paragraph{Infrared regime.} For negative scales $h\in[h_{\beta},-1]$, one proceeds similarly to the infrared leg of the lattice model described in Section \ref{sec:RGLatticeModel}. The quadratic and quartic contributions to the effective potential are localized according to Definition \ref{def:loc}, thus obtaining effective couplings $(Z^{\mathrm{ref}}_{h,\omega}, v^{\mathrm{ref}}_{h,\omega}, \lambda^{\mathrm{ref}}_{h,\omega\omega'})$; no coupling $\nu^{\mathrm{ref}}_{h,\omega}$ is created thanks to the reference model propagator being odd. The marginal terms dependent on the external field are also localized similarly to the prescription of Definition \ref{def:loc}; however, due to the fact that the external field depends on the chirality of the fermions, the general structure of the $A$-dependent marginal term is:
\begin{equation}\label{eq:marginalA}
\mathcal{L} V_{2,0,1}^{(h)}(\sqrt{Z_{h}} \psi) = \sum_{\omega, \omega'=1}^{N_{f}}\sum_{\mu=0}^{1}\frac{1}{\beta L}\sum_{\underline p\in \Dp}\mathsf Z^{\mu,\omega}_{h,\omega'} Z_{h,\omega'}\hat A^{\mu}_{-\underline p,\omega}\hat n_{\underline p,\omega'}.
\end{equation}
The running coupling constants $\mathsf Z^{\mu,\omega}_{h,\omega'}$ have exactly the same flow: they only differ from the choice of the initial datum, $\mathsf Z^{0}_{\omega'}$ and $\mathsf Z^{1}_{\omega'}$.

We shall give an informal overview of the control of these couplings in Appendix \ref{app:RGref}. They are all marginal, and the control of their RG flows requires identifying nontrivial cancellations. The final result is summarized in the following proposition, which is a consequence of Theorem 8.5 of \cite{MPmulti}, in turn based on the analysis for the case of two chiralities \cite{BMWI, BMchiral}.
\begin{proposition}[Bounds for the reference model running coupling constants]\label{prp:flowref} There exists $\bar\lambda>0$ such that for all $\lambda_{\mathrm{max}}:= \max_{\omega,\omega'}|\lambda^{\mathrm{ref}}_{\omega\omega'}| < \bar \lambda$
\begin{equation} \label{eq:BdsEffCoupref}
\begin{split}
\left|\frac{Z^{\mathrm{ref}}_{h-1,\omega}}{Z^{\mathrm{ref}}_{h,\omega}}\right|\leq e^{c\lambda_{\mathrm{max}}}\;,&\qquad |v^{\mathrm{ref}}_{h,\omega}-v^{\mathrm{ref}}_{\omega}|\leq C\lambda_{\mathrm{max}}\;,\\
|\mathsf Z^{\mu,\omega}_{h,\omega'}|\leq C\;,&\qquad |\lambda^{\mathrm{ref}}_{h,\omega\omega'}-\lambda^{\mathrm{ref}}_{\omega\omega'}|\leq C\lambda_{\mathrm{max}}^{2}\;,
\end{split}
\end{equation}
for some $c,C>0$ independent of $\beta,L,\Ma,\varepsilon,h, N$, for all $\omega,\omega'=1,\dots, N_{f}$, $\mu=0,1$, and for all $h_{\beta}\leq h<0$.
\end{proposition}
\begin{remark}
Notice that the independence of these bounds from $\beta$ allows to control the dynamical system also in the zero temperature limit $h_{\beta}\to-\infty$.
\end{remark}

All in all, we may summarize these finding for the ultraviolet and infrared regimes in a GN tree expansion for the generating functional of the reference model: we have,
\begin{equation}
\begin{split}
\mathcal W_{\beta,L,\Ma,\varepsilon,N}(\phi,A) &= \sum_{h=h_{\beta}}^{N} \sum_{n,m\geq 0}\sum_{\tau\in\mathcal T^{m,\mathrm{ref}}_{h,n}} \sum_{\substack{\mathbf\Gamma\in\mathcal L_{\tau}\\ \Gamma^{\psi}_{v_{0}}=\varnothing}} \mathcal W^{(h)}_{\mathbf\Gamma}[\phi,A;\tau]\\
\mathcal W^{(h)}_{\mathbf\Gamma}[\phi,A;\tau] &= \int_{\beta,L,\Ma} d\ul X_{I_{v_{0}}}\, \phi_{\Gamma^{\phi}_{v_{0}}} A_{\Gamma^{A}_{v_{0}}} W^{(h)}_{\mathbf\Gamma}[\ul X_{\Gamma_{v_{0}}};\tau]\;,
\end{split}
\end{equation}
with
\[
 \int_{\beta,L,\Ma} d\ul X_{\Gamma} := \prod_{f\in\Gamma} \Ma^{2}\sum_{\ul x(f)\in\Lambda_{\beta,L,\Ma}}\;,
\]
and where $\mathcal T^{m,\mathrm{ref}}_{h,n}$ denotes the set of GN trees with root scale $h$, $m$ external field leaves and $n$ leaves of type $\lambda$ that can attach on any scale between $h$ and $N$. The set of labellings $\mathcal L_{\tau}$ respects the two different localization protocols for positive and negative scales.
The tree-dependent kernels $W^{(h)}_{\mathbf \Gamma}$ satisfy the GN-type estimate
\begin{equation}
\|W^{(h)}_{\mathbf\Gamma}[\tau]\|_{1,0}\leq C^{n+m} \lambda_{\mathrm{max}}^{n}\, 2^{h\big(2-\frac{3}{2}|\Gamma^{\phi}_{v_{0}}|-|\Gamma^{A}_{v_{0}}|\big)} \prod_{v\notin V_{f}(\tau)} 2^{(h_{v}-h_{v'})D_{v}}\;,
\end{equation}
with internal dimension
\[
D_{v}=\begin{cases}
\displaystyle 2-\frac{|\Gamma^{\psi}_{v}|}{2}-\frac{3|\Gamma^{\phi}_{v}|}{2}-|\Gamma^{A}_{v}| & \text{if } h_{v}\geq 0\\
\displaystyle 2-\frac{|\Gamma^{\psi}_{v}|}{2}-\frac{3|\Gamma^{\phi}_{v}|}{2}-|\Gamma^{A}_{v}| -z_{v} + c\lambda_{\mathrm{max}} & \text{if } h_{v}\leq -1\;.
\end{cases}
\]
We stress that $D_{v}$ is always strictly negative for all vertices $v\notin V_{f}(\tau)$: while this is manifest for vertices at negative scale, it is true also for the case $h_{v}\geq0$, since thanks to (\ref{eq:locUV}) we have localized all the relevant and marginal kernels.
\subsection{Ward identities for the $m$-point correlation functions}
In this section we will prove identities for the $m$-point correlation functions of the reference model, that allow to prove the vanishing of the density-density correlations with $m\geq 3$ after removing all cutoffs. This will follow from the anomalous chiral gauge symmetry of the QFT, to be discussed here. In what follows, we shall use the following short-hand notation:
\begin{equation}
\begin{split}
\langle\, \cdot\,\rangle_N &\equiv \langle\, \cdot\,\rangle_{\beta,L,\mathfrak a, \varepsilon, N} \\
\mathcal{Z}_N(\phi,A) &\equiv \mathcal Z_{\beta,L,\mathfrak a, \varepsilon, N}(\phi,A)\\
\mathcal{W}_N(\phi,A) &\equiv \mathcal{W}_{\beta,L,\mathfrak a, \varepsilon, N}(\phi,A). 	
\end{split}
\end{equation}	   
We notice that, thanks to the presence of the antiperiodic boundary conditions, the propagator is not singular for $\Ma>0$, and the Grassmann algebra is finite dimensional.

Let us consider the chiral local gauge transformation:
\begin{equation}
	\label{transf}
	\psi^{\pm}_{\underline x,\omega}\mapsto e^{\pm i\alpha_\omega(\underline x)}	\psi^{\pm}_{\underline x,\omega}
\end{equation}	
with $\{\alpha_\omega(\underline x)\}$ a collection of functions on $(\Ma\mathbb Z)^2$ such that $\alpha_\omega(\underline x+n_0\beta\underline e_0+n_1L\underline e_1)=\alpha_\omega(\underline x)$ for all chiralities $\omega$.
Then, it is possible to check that, given any monomial $Q(\psi^+,\psi^-)$, denoting with $Q_\alpha (\psi^+,\psi^-)$ the same monomial under the transformation \eqref{transf}:
\begin{equation}
\int d\psi\, Q(\psi^+,\psi^-)=\int d \psi\, Q_\alpha (\psi^+,\psi^-).
\end{equation}	
Thus, by linearity and by finiteness of the Grassmann algebra, for any function $f(\psi)$ we have:
\begin{equation}\label{eq:inv}
\int d\psi\, f(\psi)=\int d\psi\, f_\alpha(\psi)\;.	
\end{equation}	
Let:
\begin{equation}\begin{split}\label{Dp}
		\Delta_{\underline p,\omega}(\psi)&:=\frac{1}{\beta L}\sum_{\underline k\in \mathbb D^a_{\beta,L,\mathfrak a}}\hat\psi^+_{\underline k-\underline p,\omega}\Delta_{\omega}(\underline k,\underline p)\hat \psi^-_{\underline k, \omega}		\\
		\Delta_{\omega}(\underline k,\underline p)&:= -Z_\omega\frac{\Den(\underline k -\underline p)}{\chi_{N}^{\omega,\varepsilon}(\underline k-\underline p)}+Z_\omega\frac{\Den(\underline k)}{\chi_{N}^{\omega,\varepsilon}(\underline k)}-Z_\omega\Den(\underline p)\;,
\end{split}
\end{equation}	
with $\Den(\underline k)=i\mathfrak a^{-1}\sin(\mathfrak a k_0)+v_\omega\mathfrak a^{-1} \sin(\mathfrak a k_1)$. Applying the identity (\ref{eq:inv}) to the partition function, and differentiating with respect to $\hat \alpha_{\omega}(\underline{p})$, we get:
\begin{equation}
	\label{GenWId}
	0=\int P_{N}[d\psi] e^{-V(\psi)+C(\psi;A)}\left[Z_\omega\Den(\underline p)\hat n_{\underline p,\omega}-\Delta_{\underline p,\omega}(\psi)\right],
\end{equation}
where $P_{N}$ is the Grassmann integration given by:
\begin{equation} \label{eq:refcov}
\begin{split}
P_{N}[d\psi] &\propto d\psi\, e^{-K_{N}(\psi)} \\
K_{N}(\psi) &= \frac{1}{\beta L}\sum_{\omega}\sum_{\underline k\in\mathbb D^{\mathrm a}_{\beta,L,\mathfrak a}}\hat\psi^+_{\underline k,\omega}\left[\hat g^{(\leq N)}_{\omega,\varepsilon,\Ma}(\underline k)\right]^{-1}\hat\psi^-_{\underline k,\omega}\;,
\end{split}
\end{equation}
with
\begin{equation}\label{PropN}
\hProp(\underline k) =\frac{1}{Z_\omega}\frac{\chi^{\omega,\varepsilon}_N(\underline k)}{\mathcal D_{\omega,\mathfrak a}(\underline k)}\;,\qquad \chi^{\omega,\varepsilon}_N(\underline k) =\chi^{\varepsilon}(2^{-N}\|\underline k\|_\omega)\;.
\end{equation}	
From now on, we shall use the short-hand notation $\hat g^{(\le N)}_{\omega}(\underline k) \equiv \hat g^{(\leq N)}_{\omega,\varepsilon,\Ma}(\underline k)$.
\begin{proposition}[Ward identity for the density-density correlators]\label{prp:ref} The following identity holds true:
\begin{equation}\label{eq:WIm}
\begin{split}
&Z_{\omega_1}\DenI(\underline p_1) \langle \hat n_{\underline{p}_{1},\omega_{1}};\cdots;\hat n_{\underline{p}_{m},\omega_{m}}\rangle_N \\
&\qquad =-\DenI(\underline p_1) \mathfrak B^N_{\omega_1}(\underline p_1)\sum_{\bar\omega}\lambda_{\omega_{1}\bar\omega}Z_{\bar\omega}\hat v(\underline p_{1})\langle \hat n_{\underline p_1,\bar\omega};\cdots;\hat n_{\underline{p}_{m},\omega_{m}}\rangle_N\\
&\quad\qquad +\frac{1}{\beta L}\sum_{\underline k \in \D}\Delta_{\omega_1}(\underline k,\underline p_1)\hat g^{(\le N)}_{\omega_1}(\underline k-\underline p_1)\hat g^{(\le N)}_{\omega_1}(\underline k) \mathcal R^{N}_{\ul\omega}(\underline k;{\bf{\underline p}}_m)\;,
\end{split}
\end{equation}
where $\mathfrak B^N_{\omega}(\underline p)$ is the anomalous bubble diagram,
\begin{equation}\label{eq:Bolla}
\mathfrak B^N_{\omega}(\underline p):=-\frac{1}{\beta L}\frac{Z_{\omega}}{\Den(\underline p)}\sum_{\underline k\in\D}\Delta_{\omega}(\underline k,\underline p)\hat g^{(\le N)}_{\omega}(\underline k-\underline p)\hat g^{(\le N)}_{\omega}(\underline k)\;,
\end{equation}
$\underline{{\bf p}}_{m} = (\underline{p}_{1}, \ldots, \underline{p}_{m})$, and the error term $\mathcal R^{N}_{\ul\omega}$ satisfies the estimate, for $m\geq 3$:
\begin{equation}\label{eq:Rest}
\frac{1}{\beta L}\sup_{\underline{k}:\, c_{1} 2^{N} \leq \|\underline{k}\|_{\omega_{1}} \leq c_{2}2^{N}} \Big|\mathcal R^{N}_{\ul\omega}(\underline k;{\bf{\underline p}}_m)\Big| \leq C_{\beta,m}2^{-\gamma N}\;,
\end{equation}
where $\gamma>0$ and the constant $C_{\beta,m}$ is independent of $\Ma, \varepsilon, N$, provided $\frak{a}, \varepsilon$ and $\frak{a} / \varepsilon$ are small enough.
\end{proposition}
The proof of Proposition \ref{prp:ref} is postponed to Appendix \ref{app:ref}.
\begin{corollary}[Vanishing of the density-density correlators] \label{cor:vanishingcorr} Let $m\geq 3$. Then, there exists $\gamma>0$ such that density-density correlators satisfy the following estimate:
\begin{equation}\label{eq:bdmdensity}
\lim_{\varepsilon\to0}\lim_{\Ma\to0}\frac{1}{\beta L}\big| \langle \hat n_{\underline{p}_{1},\omega_{1}};\cdots;\hat n_{\underline{p}_{m},\omega_{m}}\rangle_N \big| \leq C_{\beta, m} N  2^{-\gamma N}
\end{equation}
with a positive constant $C_{\beta,m}$ independent of $N$.
\end{corollary}
\begin{remark} Notice that the bounds (\ref{eq:Rest}), (\ref{eq:bdmdensity}) are not optimal in their $\beta$-dependence. With quite some extra effort, one could establish an estimate uniform in $\beta$, and not uniform in the external momenta. This improvement will however not be needed in our application, where the external momenta are bounded, and where the $N\to \infty$ limit is performed at $\beta$ fixed.
\end{remark}
\begin{proof}[Proof of Corollary \ref{cor:vanishingcorr}] Let us start by estimating the last term in the right-hand side of (\ref{eq:WIm}). Since $\mathcal R^{N}_{\ul\omega}$ is bounded uniformly in $\Ma$, $\varepsilon$ small (see Remark \ref{rem:bdR}), we may take the limit $\Ma\to 0$ and then $\varepsilon\to 0$ as well. In this limit, the function $\Delta_{\omega_{1}}\hat g^{(\leq N)}_{\omega_{1}}\hat g^{(\leq N)}_{\omega_{1}}$ becomes, by \eqref{Dp} and \eqref{PropN}:
\begin{equation}\label{Rewr}	
\frac{\chi^{\omega_{1}}_{N}(\underline{k} + \underline{p}_{1}) - \chi^{\omega_{1}}_{N}(\underline{k} - \underline{p}_{1}) }{Z_{\omega_{1}}^{2}D_{\omega_{1}}(\underline{k})}\chi^{\omega_{1}}_{N}(\underline{k})
\end{equation}
with $D_{\omega_1}(\underline k):=ik_0+v_{\omega_1}k_1$ for $\ul k\in\mathbb D^{\mathrm a}_{\beta,L}=\frac{2\pi}{\beta}\big(\Z+\frac{1}{2}\big)\times \frac{2\pi}{L}\big(\Z+\frac{1}{2}\big)$, so that 
\begin{equation}\label{Rewr2}
\begin{split}
\frac{1}{\beta L}\sum_{\underline k \in \D}\Delta_{\omega_1}&(\underline k,\underline p_1)\hat g^{(\le N)}_{\omega_1}(\underline k-\underline p_1)\hat g^{(\le N)}_{\omega_1}(\underline k) \mathcal R^{N}_{\ul\omega}(\underline k;{\bf{\underline p}}_m)\\
&\underset{\Ma,\varepsilon\to0}{\longrightarrow}\frac{1}{\beta L}\sum_{\ul k\in\mathbb D^{\mathrm{a}}_{\beta,L}}\frac{\chi^{\omega_{1}}_{N}(\underline{k} + \underline{p}_{1}) - \chi^{\omega_{1}}_{N}(\underline{k} - \underline{p}_{1}) }{Z_{\omega_{1}}^{2}D_{\omega_{1}}(\underline{k})}\chi^{\omega_{1}}_{N}(\underline{k}) \mathcal R^{N}_{\ul \omega}(\ul k;\ul{\mathbf p}_{m})\;,
\end{split}
\end{equation}
where with a slight abuse of notation we renamed $\mathcal R^{N}_{\ul \omega}$ its $\Ma,\varepsilon\to0$ limit.

Thanks to the properties of the cut-off function $\chi^{\omega_{1}}_{N}$, we notice that, for fixed $\ul p_{1}\neq\ul 0$, the function \eqref{Rewr} is supported on very high momenta, \emph{i.e.} for $c_{1}2^{N}\leq \|\ul k\|_{\omega_{1}}\leq c_{2}2^{N}$.
Hence, we have:
\begin{equation}
\begin{split}
&\Bigg| \frac{1}{\beta L}\sum_{\ul k\in\mathbb D^{\mathrm{a}}_{\beta,L}}\frac{\chi^{\omega_{1}}_{N}(\underline{k} + \underline{p}_{1}) - \chi^{\omega_{1}}_{N}(\underline{k} - \underline{p}_{1}) }{Z_{\omega_{1}}^{2}D_{\omega_{1}}(\underline{k})}\chi^{\omega_{1}}_{N}(\underline{k}) \mathcal R^{N}_{\ul \omega}(\ul k;\ul{\mathbf p}_{m}) \Bigg| \\
&\quad \leq C_{\beta,m}2^{-\gamma N} \frac{1}{\beta L}\sum_{\ul k\in\mathbb D^{\mathrm{a}}_{\beta,L}} \Bigg|\frac{\chi^{\omega_{1}}_{N}(\underline{k} + \underline{p}_{1}) - \chi^{\omega_{1}}_{N}(\underline{k} - \underline{p}_{1}) }{Z_{\omega_{1}}^{2}D_{\omega_{1}}(\underline{k})}\chi^{\omega_{1}}_{N}(\underline{k}) \Bigg| \\
&\quad \leq \tilde C_{\beta,m} 2^{-\gamma N} N\;,
\end{split}
\end{equation}
where in the first step we employed (\ref{eq:Rest}); the last factor $N$ comes from the logarithmic divergence of the sum. Therefore, dividing left-hand side and right-hand side of (\ref{eq:WIm}) by $Z_{\omega_1}D_{\omega_{1}}(\underline p_1)$, we have:
\begin{equation}
\lim_{\varepsilon\to0}\lim_{\Ma\to0}\Big|\sum_{\bar \omega} (T^{-1}_{N}(\underline{p}_{1}))_{\omega_{1}\bar \omega} \langle \hat n_{\underline{p}_{1},\bar \omega};\cdots;\hat n_{\underline{p}_{m},\omega_{m}} \rangle_{N}\Big| \leq  \frac{\tilde C_{\beta, m}}{\|\underline{p}_{1}\|_{\omega_{1}}} 2^{-\gamma N} N,
\end{equation}
where $T^{-1}_{N}(\underline{p})$ is the matrix:
\begin{equation}\label{eq:Tdef}
(T^{-1}_{N}(\underline{p}))_{\omega\omega'} = \delta_{\omega,\omega'} + \mathfrak B^N_{\omega}(\underline p) Z_{\omega}^{-1}\lambda_{\omega\omega'}Z_{\omega'}\hat v(\underline p)\;.
\end{equation}
The final bound (\ref{eq:bdmdensity}) follows from the invertibility of the matrix $T^{-1}_{N}(\underline{p})$ defined in (\ref{eq:Tdef}), for $\lambda_{\mathrm{max}}:=\max_{\omega,\omega'}|\lambda_{\omega\omega'}|$ small enough.
\end{proof}
\begin{remark}
\begin{itemize}
\item[(i)] In the non-interacting case, the key estimate (\ref{eq:Rest}) can be understood easily. For $\lambda_{\omega\omega'}\equiv0$, the first term on the right-hand side of \eqref{eq:WIm} vanishes, while one can show that the second term is given by a sum of 1-loop diagrams obtained contracting $m-1$ source terms together with a modified vertex, forcing the loop momentum $\ul k$ to be of order $2^{N}$. If $m\geq 3$, one can estimate each diagram proportionally to $2^{(2-m)N}$, producing an estimate which is vanishing in the limit $N\to\infty$. For the interacting case, in Appendix \ref{app:ref} we show that $\mathcal R^{N}_{\ul\omega}$ is given by a sum of correlators with strictly negative scaling dimension, with two external fermions with momenta of the order of the UV scale $2^{N}$; the bound \eqref{eq:Rest} follows from the vanishing of such correlators at high external momenta.
\item[(ii)] The identity (\ref{eq:bdmdensity}) will play a key role in the proof of exactness of linear response for the lattice model. In the bosonization approach to the chiral Luttinger liquid, the identity (\ref{eq:bdmdensity}) as $N\to \infty$ is a consequence of the fact that the density operator is described by (the derivative of) a bosonic Gaussian free field. In our setting, this identity is proved rigorously starting from a regularized QFT. The advantage of our approach is that it is robust enough to be used to derive consequences about the original lattice model, at  low energies, which cannot be studied rigorously via bosonization. This will be discussed in the next section.
\end{itemize}
\end{remark}

To conclude this section, we report briefly the structure of the density-density and vertex correlation functions, when all cut-offs are removed and when $\beta,L$ are sent to infinity.
\begin{proposition}[Density-density correlation and vertex function; Proposition 9.11 of \cite{MPmulti}] \label{prop:2ptref}
Let us set
\[
\langle\,\cdot\,\rangle=\lim_{\beta,L\to\infty}\lim_{N\to\infty}\lim_{\varepsilon\to0}\lim_{\Ma\to0}\langle\,\cdot\,\rangle_{N}\;;
\]
then, for any $\ul p\neq 0$, the density two-point function takes the form
\begin{equation}
\langle\hat n_{\ul p,\omega};\hat n_{-\ul p,\omega'}\rangle = T_{\omega\omega'}(\ul p) \frac{\mathfrak B_{\omega'}(\ul p)}{Z_{\omega'}^{2}}\;,
\end{equation}
and the vertex function is given by
\begin{equation} \label{eq:refVertexWI}
\langle\hat n_{\ul p,\omega};\hat\psi^{-}_{\ul k,\omega'};\hat\psi^{+}_{\ul k +\ul p,\omega'}\rangle = \frac{T_{\omega\omega'}(\underline{p})}{Z_{\omega} D_{\omega}(\underline{p})} \big[\langle\hat\psi^{-}_{\ul k,\omega'}\hat\psi^{+}_{\ul k,\omega'}\rangle  - \langle\hat\psi^{-}_{\ul k+\ul p,\omega'}\hat\psi^{+}_{\ul k+\ul p,\omega'}\rangle]\;.
\end{equation}
Here
\begin{equation} \label{eq:relbubble}
\mathfrak B_{\omega}(\ul p) = \lim_{\beta,L\to\infty}\lim_{N\to\infty}\lim_{\varepsilon\to0}\lim_{\Ma\to0} \mathfrak B^{N}_{\omega}(\ul p) = \frac{1}{4\pi |v_{\omega}|}\frac{-ip_{0}+ v_{\omega}p_{1}}{ip_{0}+v_{\omega}p_{1}}\;,
\end{equation}
and 
\begin{equation}
T(\ul p) := \lim_{\beta,L\to\infty}\lim_{N\to\infty}\lim_{\varepsilon\to0}\lim_{\Ma\to0} T_{N}(\ul p) = \frac{1}{1+ \mathfrak B(\ul p)\Lambda_{Z}\hat v(\ul p)}\;,
\end{equation}
where
\[
\Lambda_{Z}:= Z^{-1}\Lambda Z\;,
\]
with $Z=\mathrm{diag}\, Z_{\omega}$, and $\Lambda_{\omega\omega'}=\lambda_{\omega\omega'}$.
\begin{proof}
We refer to Section 9 of \cite{MPmulti} for the proof. It should be noted that the expression (\ref{eq:relbubble}) of the relativistic bubble $\mathfrak B$ takes a (slightly) different form with respect to the one in \cite{MPmulti}, due to different Fourier conventions adopted in this paper.
\end{proof}

\end{proposition}

\section{Proof of Main Theorem} \label{sec:final}
\subsection{Matching the lattice and reference models}
In this section we shall prove that lattice correlators are well-approximated by the corresponding ones in the reference model, with summable bounds on the deviations. To this end, we shall consider a reference model in presence of an external field that does not depend on the chirality labels, as the external field of the lattice model: in Eq. (\ref{Gamma}), we choose $\hat A^{\mu}_{\underline{p},\omega} \equiv \hat A^{\mu}_{\underline{p}}$. The source-dependent marginal term has the form (\ref{eq:marginalA}), with now $\mathsf Z^{\mu,\omega}_{h,\omega'}$ replaced by $\mathsf Z^{\mu, \mathrm{ref}}_{h,\omega'}$. The boundedness of the flow of $\mathsf Z^{\mu,\omega}_{h,\omega'}$ also implies that the flow of $\mathsf Z^{\mu, \mathrm{ref}}_{h,\omega'}$ is bounded, since $\mathsf Z^{\mu, \mathrm{ref}}_{h,\omega'}$ can be expressed as a linear combination of $\mathsf Z^{\mu,\omega}_{h,\omega'}$.
\begin{remark}
In this section, we shall denote the reference model expectation, having already performed the $\Ma\to0$ and $\varepsilon\to0$ limits, with $\langle\,\cdot\,\rangle^{\mathrm{ref}}_{\beta,L,N}$.
\end{remark}
\begin{proposition} [Comparison of correlations] \label{prop:CorrComparison}
Let $\beta, L, N$ be large enough. Calling $\lambda$ the bare coupling of the lattice model, there exists $\bar\lambda>0$ independent of $\beta,L$ such that, for all $|\lambda|<\bar\lambda$, one may find a choice of the bare couplings of the reference model, satisfying
\begin{equation}
\begin{split}
    |v^{\mathrm{ref}}_{\omega}- v_{\omega}|\leq C|\lambda|\;, \qquad 
|\lambda^{\mathrm{ref}}_{\omega\omega'}-\lambda| \leq C|\lambda|^{2}\;,\\
|Z^{\mathrm{ref}}_{\omega}-1|\leq C|\lambda|\;,\qquad |\mathsf Z^{\mu,\mathrm{ref}}_{\omega}-1| \leq C|\lambda|\;,
\end{split}
\end{equation}
for all $\omega,\omega'$, such that the following points hold.

Let $m\geq 3$, $\ul p_{i}=(\eta_{\beta},p_{i})$ for $i=1,\dots, m-1$, and set $\ul p_{m}=-\sum_{i=1}^{m-1} \ul p_{i}$. Assume that $|p_{i}| \leq B|\theta|$, and that $|\theta| \leq K \eta$, as in Theorem \ref{thm:main}. Then, for $\eta$ small enough, the lattice density-current $m$-point correlations can be written as
\begin{equation} \label{eq:CorrComparison}
\begin{split}
\langle \mathbf T \hat n_{\ul p_{1}};\cdots; \hat n_{\ul p_{m-1}}; \hat \jmath_{\mu,\ul p_{m}}\rangle_{\beta,L} = \sum_{\ul \omega} \Big[\prod_{i=1}^{m-1} \mathsf Z^{0,\mathrm{ref}}_{\omega_{i}}\Big] \mathsf Z^{\mu,\mathrm{ref}}_{\omega_{m}}\,\langle \hat n_{\ul p_{1},\omega_{1}};\cdots;\hat n_{\ul p_{m},\omega_{m}}\rangle^{\mathrm{ref}}_{\beta, L,N} \\
 + R^{\beta,L, N}_{m;\mu} (\ul p_{1},\dots, \ul p_{m})\;;
\end{split}
\end{equation}
the error term admits the following estimate:
\begin{equation} \label{eq:EstRem}
\frac{1}{\beta L}|R^{\beta,L, N}_{m;\mu} (\ul p_{1},\dots, \ul p_{m})| \leq m! C^{m}\, \eta^{2-m+\gamma} F(m;\lambda,\alpha)\;,
\end{equation}
where
\begin{equation} \label{eq:improv}
F(m;\lambda,\alpha):=\begin{cases}
1 & \text{if } m<16\\
\displaystyle\frac{1}{\lfloor m/4-4\rfloor!} + |\lambda|\Big(\frac{2\alpha}{1-32\alpha}\Big)^{(1/4-8\alpha)m} + |\lambda|^{\alpha m} & \text{if } m\geq 16
\end{cases}
\end{equation}
with $\gamma\in(0,1)$ and any $\alpha\in(0,1/34)$, and $C$ uniform in $\beta, L, N$.
\end{proposition}
\begin{proof} As in the proof of Proposition \ref{prop:CorrDimBound}, we start by writing:
\begin{equation}
\begin{split}
\langle \mathbf T\hat n_{\ul p_{1}};\cdots; \hat \jmath_{\mu,\ul p_{m}}\rangle_{\beta,L} &= \langle \mathbf T\hat n_{\ul p_{1}};\cdots; \hat \jmath_{\mu,\ul p_{m}}\rangle_{\beta,L}^{(\text{u.v.})} \\
&\quad +\langle \mathbf T\hat n_{\ul p_{1}};\cdots; \hat \jmath_{\mu,\ul p_{m}}\rangle_{\beta,L}^{(\text{i.r.}), (\text{d})} + \langle \mathbf T\hat n_{\ul p_{1}};\cdots; \hat \jmath_{\mu,\ul p_{m}}\rangle_{\beta,L}^{(\text{i.r.}), (\text{od})}\;;
\end{split}
\end{equation}
as noticed already in (\ref{eq:uvcorr}), (\ref{eq:estodimpro}), the first and the last term already satisfy the bound (\ref{eq:EstRem}), for $\eta$ small enough. Let us focus on the second term. We begin by recalling the GN expansion of the density correlations, Eq. (\ref{eq:GNsplit}):
\begin{equation}\label{eq:GNExpLatt}
\langle \mathbf T\hat n_{\ul p_{1}};\cdots; \hat \jmath_{\mu,\ul p_{m}}\rangle_{\beta,L}^{(\text{i.r.}), (\text{d})} = \sum_{h=h_{\beta}}^{0}\sum_{n=0}^{\infty}\sum'_{\tau\in \widetilde{\mathcal T}_{h,n}^{m}} \sum_{\mathbf\Gamma\in \mathcal L_{\tau}} \widehat W^{(h)}_{\bf \Gamma}[\tau]
\end{equation}
with $\widetilde{\mathcal T}^{m}_{h,n}$ the set of scale-labelled GN trees extending up to the scale of the UV cutoff, with root-scale $h$; $m-1$ leaves of $A^{0}$-type, and one $A^{\mu}$-leaf; $\mathcal L_{\tau}$ the set of labellings of such trees selecting momenta $\ul p_{i}$ on such leaves, as described in equations (\ref{eq:GNexp})-(\ref{eq:fourierkernel}). The prime in the tree sum enforces the constraint that all $A$-leaves on scale $<0$ are diagonal, that is they involve the same fermionic chiralities.

%The ultraviolet contribution to the correlation function is much smaller in $\eta$ than the estimate (\ref{eq:CorrComparison})-(\ref{eq:EstRem}), since it is bounded as
%\begin{equation}
%\big|\langle \mathbf T\hat n_{\ul p_{1}};\cdots; \hat \jmath_{\mu,\ul p_{m}}\rangle_{\beta,L}^{(\text{u.v.})}\big| \leq m!\, C^{m}\;.
%\end{equation}

Splitting $\widetilde{\mathcal T}^{m}_{h,n}$ into $\mathcal T^{m,(\leq -1)}_{h,n}$, the set of labelled trees $\tau$ with all leaves at scales $\leq -1$, and $\widetilde{\mathcal T}^{m,[0,N_{0}]}_{h,n}$, accounting for trees having at least one endpoint on a scale between $0$ and $N_{0}$, we may rewrite (\ref{eq:GNExpLatt}) as
\begin{equation} \label{eq:split1}
\begin{split}
\langle \mathbf T\hat n_{\ul p_{1}};\cdots; \hat \jmath_{\mu,\ul p_{m}}\rangle_{\beta,L}^{(\text{i.r.}), (\text{d})}&=\sum_{h=h_{\beta}}^{-1}\sum_{n=0}^{\infty}\sum'_{\tau\in \mathcal T_{h,n}^{m,(\leq -1)}} \sum_{\mathbf\Gamma\in \mathcal L_{\tau}} \widehat W^{(h)}_{\bf \Gamma}[\tau] + E^{\beta,L}_{m;\mu,1}(\ul p_{1},\dots,\ul p_{m})\\
E^{\beta,L}_{m;\mu,1}(\ul p_{1},\dots,\ul p_{m}) &= \sum_{h=h_{\beta}}^{0}\sum_{n=0}^{\infty}\sum'_{\tau\in \widetilde{\mathcal T}_{h,n}^{m,[0,N_{0}]}} \sum_{\mathbf\Gamma\in \mathcal L_{\tau}} \widehat W^{(h)}_{\bf \Gamma}[\tau]\;.
\end{split}
\end{equation}
Now, we substitute the first term on right-hand side with the reference model correlation, obtaining
\begin{equation} \label{eq:lattrefsplit}
\begin{split}
\langle \mathbf T\hat n_{\ul p_{1}};\cdots; \hat \jmath_{\mu,\ul p_{m}}\rangle^{(\text{i.r.}), (\text{d})}_{\beta,L} &= \sum_{\ul \omega} \Big[\prod_{i=1}^{m-1} \mathsf Z^{0,\mathrm{ref}}_{\omega_{i}}\Big] \mathsf Z^{\mu,\mathrm{ref}}_{\omega_{m}}\,\langle \hat n_{\ul p_{1},\omega_{1}};\cdots;\hat n_{\ul p_{m},\omega_{m}}\rangle^{\mathrm{ref}}_{\beta, L,N}\\ &\quad + \sum_{i=1}^{3} E^{\beta,L,N}_{m;\mu,i}(\ul p_{1},\dots,\ul p_{m})
\end{split}
\end{equation}
where (observe that the sum over the GN trees of the reference model only involves diagonal $A$-leaves):
\begin{equation}
\begin{split}
E^{\beta,L,N}_{m;\mu,2}(\ul p_{1},\dots,\ul p_{m}) &:= \sum_{h=h_{\beta}}^{-1}\sum_{n=0}^{\infty}\sum'_{\tau\in \mathcal T_{h,n}^{m,(\leq -1)}} \sum_{\mathbf\Gamma\in \mathcal L_{\tau}} \Big(\widehat W^{(h)}_{\bf \Gamma}[\tau] - \widehat W^{(h),\mathrm{ref}}_{\bf \Gamma}[\tau]\Big)\\
E^{\beta,L,N}_{m;\mu,3}(\ul p_{1},\dots,\ul p_{m}) &:= -\sum_{h=h_{\beta}}^{N}\sum_{n=0}^{\infty}\sum_{\tau\in \mathcal T_{h,n}^{m,[0,N]}} \sum_{\mathbf\Gamma\in \mathcal L_{\tau}} \widehat W^{(h),\mathrm{ref}}_{\bf \Gamma}[\tau]\;.
\end{split}
\end{equation}
Notice that $E^{\beta,L,N}_{m;\mu,1}\equiv E^{\beta,L}_{m;\mu,1}$ does not depend on the reference model, hence it is independent of the reference model cut-off $N$; also it will satisfy estimates which are uniform in the lattice model ultraviolet cut-off $N_{0}$.
Furthermore, here $\mathcal T^{m,[0,N]}_{h,n}$ indicates the set of reference model GN trees which have at least one leaf attaching on a scale between $0$ and $N$. Finally, since no counterterm is needed for the reference model, \emph{i.e.} $\nu^{\mathrm{ref}}_{h,\omega}=0$, we use the convention $\widehat W^{(h),\mathrm{ref}}_{\mathbf \Gamma}[\tau]=0$ whenever there is at least one leaf of type $\nu$.

Thus, our goal is to estimate:
\begin{equation}
R^{\beta,L,N}_{m;\mu} := \sum_{i=1}^{3} E^{\beta,L,N}_{m;\mu,i}\;. 
\end{equation}
We follow the reasoning in (\ref{eq:GNsplit}) by splitting this remainder into a term with few leaves of $\lambda$ or $\nu$ type, and one with many:
\begin{equation}
\begin{split}
R^{\beta,L,N}_{m;\mu} &= \sum_{n=0}^{n_{*}} R^{\beta,L,N}_{m,n;\mu} + \sum_{n=n_{*}+1}^{\infty} R^{\beta,L,N}_{m,n;\mu}\equiv R^{\beta,L,N}_{m,\leq n_{*};\mu} + R^{\beta,L,N}_{m,> n_{*};\mu}\\
R^{\beta,L,N}_{m,n;\mu} &= \sum_{i=1}^{3} E^{\beta,L,N}_{m,n;\mu,i}\;,
\end{split}
\end{equation}
with $E^{\beta,L,N}_{m,n;\mu,i}$ collecting the contribution to $E^{\beta,L,N}_{m;\mu,i}$ coming from trees with $n$ leaves of type $\lambda,\nu$, and with $n_{*}=\lfloor\alpha m\rfloor$, for $\alpha\in(0,1/34)$ to be chosen in the proof of the main Theorem.

Consider the $i=1$ error term. Since trees $\tau$ contributing to $E^{\beta,L,N}_{m,n;\mu,1}$ have at least one vertex on scale $0$, by the short-memory property of GN trees the associated kernels will have an extra factor $2^{\gamma(h-0)}$ at the root, for $\gamma\in(0,1)$; combining this scaling gain with Proposition \ref{prop:SumBoundLow}, we are led to
\begin{equation} \label{eq:E1bound}
\frac{1}{\beta L}|E^{\beta,L,N}_{m,n;\mu,1}(\ul p_{1},\dots,\ul p_{m})| \leq m!\, \eta^{2-m+\gamma}\, C^{m}_{\gamma}
\begin{cases}
\displaystyle \frac{1}{\lfloor m/4-4\rfloor!} & n=0\\
\displaystyle C_{\gamma}^{n}|\lambda|^{n} \left(\frac{2n}{m-32n}\right)^{m/4-8n} & 0< n\leq  n_{*}\;,
\end{cases}
\end{equation}
which implies (for $|\lambda|$ small enough)
\begin{equation} \label{eq:E1few}
\frac{1}{\beta L}|E^{\beta,L,N}_{m,\leq n_{*};\mu,1}(\ul p_{1},\dots,\ul p_{m})| \leq m!\, \eta^{2-m+\gamma} C_{\gamma}^{m} \left[\frac{1}{\lfloor m/4-4\rfloor!} + |\lambda|\left(\frac{2\alpha}{1-32\alpha}\right)^{(1/4-8\alpha)m}\right]\;.
\end{equation}
Analogously, one obtains
\begin{equation} \label{eq:E1many}
\frac{1}{\beta L}|E^{\beta,L,N}_{m,> n_{*};\mu,1}(\ul p_{1},\dots,\ul p_{m})| \leq m!\, \eta^{2-m+\gamma} C_{\gamma}^{m} |\lambda|^{\alpha m}\;,
\end{equation}
which proves the claim for $E^{\beta,L,N}_{m;\mu,1}$.

One may reason analogously for the case $i=3$: indeed, we know that trees contributing to $E^{\beta,L,N}_{m,n;\mu,3}$ have at least one leaf at a non-negative scale, so let us call such scale $h_{*}\geq0$; by the short-memory property of GN trees we have a gain $2^{\gamma (h-h_{*})}$ at the root. Repeating the procedure\footnote{To be very precise, in the reference model, for $A$-leaves attaching at positive scales, the corresponding propagator $g^{(h_{i})}_{\mathrm{ref}}$ appearing on a chain, see (\ref{eq:chain}), is also multiplied by a kernel $W^{(h_{i})}_{2,0,1}$; this is due to the localization operator $\mathfrak L$ acting as the identity on marginal kernels at positive scales. Nevertheless, this does not play any role in the proof of the bounds we need.} outlined in (\ref{eq:E1bound}) we may obtain an estimate analogous to (\ref{eq:E1few}) (respectively, (\ref{eq:E1many})) for $E^{\beta,L,N}_{m,\leq n_{*};\mu,3}$ (respectively, $E^{\beta,L,N}_{m,> n_{*};\mu,3}$ ).

Finally, we move on to the $i=2$ error term. We start by noticing that the splitting (\ref{eq:lattrefsplit}) obtained above does not require a specific choice of the bare parameters of the reference model, but in order to interpret $E^{\beta,L,N}_{m;\mu,2}$ as an error term we need to tune such bare parameters in the infrared region: indeed, we choose $(Z^{\mathrm{ref}}_{\omega}, v^{\mathrm{ref}}_{\omega}, \lambda^{\mathrm{ref}}_{\omega\omega'}, \mathsf Q^{\mathrm{ref}}_{\omega}, \mathsf Z^{\mu, \mathrm{ref}}_{\omega})$ so that
\begin{equation} \label{eq:CouplingMatching}
Z^{\mathrm{ref}}_{h_{\beta},\omega} = Z_{h_{\beta},\omega}\qquad v^{\mathrm{ref}}_{h_{\beta},\omega}= v_{h_{\beta},\omega}\qquad \lambda^{\mathrm{ref}}_{h_{\beta},\omega\omega'} =  \lambda_{h_{\beta},\omega\omega'}\;,
\end{equation}
and
\begin{equation} \label{eq:vertexCoupMatching}
\mathsf Q^{\pm,\mathrm{ref}}_{h_{\beta},\omega} = \mathsf Q^{\pm}_{h_{\beta},\omega}\qquad \mathsf Z^{\mu,\mathrm{ref}}_{h_{\beta},\omega} = \mathsf Z^{\mu}_{h_{\beta},\omega}\;.
\end{equation}
To do this, one can use the implicit function theorem, since all effective couplings are non-constant, differentiable functions of their bare counterparts. Therefore, thanks to Propositions \ref{prp:flow} and \ref{prp:flowref}, we obtain that
\begin{equation} \label{eq:boundcoupdiff}
\begin{split}
\left|\frac{Z^{\mathrm{ref}}_{h,\omega}}{Z^{\mathrm{ref}}_{h-1,\omega}} - \frac{Z_{h,\omega}}{Z_{h-1,\omega}}\right|\leq C |\lambda|^{2} 2^{\gamma h}\qquad |v^{\mathrm{ref}}_{h,\omega}- v_{h,\omega}|\leq C |\lambda| 2^{\gamma h}\\
|\lambda^{\mathrm{ref}}_{h,\omega\omega'} -  \lambda_{h,\omega\omega'}|\leq C|\lambda|^{2} 2^{\gamma h}\qquad |\mathsf Z^{\mu,\mathrm{ref}}_{h,\omega}-\mathsf Z^{\mu}_{h,\omega}|\leq C |\lambda| 2^{\gamma h}
\end{split}
\end{equation}
for all $h=h_{\beta},\dots, -1$, with $\gamma\in(0,1)$, and $C$ uniform in $\beta, L, N$. Recall also that
\begin{equation} \label{eq:boundnu}
|\nu_{h,\omega}|\leq C 2^{\gamma h}\;;
\end{equation}
furthermore, comparing the lattice and reference model single-scale propagators, we can write $\hat g^{(h)}_{\omega}= \hat g^{(h),\mathrm{ref}}_{\omega} + \delta \hat g^{(h)}_{\omega}$, and one can check that the deviation satisfies the bound
\begin{equation} \label{eq:propdiffbd}
|\delta \hat g^{(h)}_{\omega}(\ul k')|\leq C 2^{\gamma h} 2^{-h}\;.
\end{equation}
Since kernels contributing to $E^{\beta,L,N}_{m,n;\mu,2}$ are given by differences of the lattice and reference model ones, an estimate for them boils down to the control of bounds of the kind (\ref{eq:boundcoupdiff})-(\ref{eq:propdiffbd}). Hence, by the continuity property of GN trees, estimates for the contributions to $E^{\beta,L,N}_{m,n;\mu,2}$ carry an extra $2^{\gamma h}$ at the root, so that, proceeding as for the case of $E^{\beta,L,N}_{m;\mu,1}$, one obtains
\begin{equation}
\begin{split}
\frac{1}{\beta L}|E^{\beta,L,N}_{m,\leq n_{*};\mu,2}(\ul p_{1},\dots,\ul p_{m})| &\leq m!\, \eta^{2-m+\gamma} C_{\gamma}^{m} \left[\frac{1}{\lfloor m/4-4\rfloor!} + |\lambda|\left(\frac{2\alpha}{1-32\alpha}\right)^{(1/4-8\alpha)m}\right]\\
\frac{1}{\beta L}|E^{\beta,L,N}_{m,> n_{*};\mu,2}(\ul p_{1},\dots,\ul p_{m})|&\leq m!\, \eta^{2-m+\gamma} C_{\gamma}^{m} |\lambda|^{\alpha m}\;.
\end{split}
\end{equation}
Adding together all the bounds above produces the claimed estimate for $|R^{\beta,L,N}_{m;\mu}|$.
\end{proof}
Next, we report for later use the comparisons of the 2-point and vertex functions.
\begin{proposition}[Comparison of 2-point and vertex functions] \label{prop:PropVertexComparison}
Under the same hypotheses of Proposition \ref{prop:CorrComparison}, we may compare the 2-point and vertex functions of the lattice model to the corresponding reference model quantities.\\

\noindent{\underline{\emph{Lattice $2$-point function.}}} There exist $\mathsf Q^{\pm,\mathrm{ref}}$ satisfying
\[
\mathsf Q^{-,\mathrm{ref}}_{\omega}= \overline{\mathsf Q^{+,\mathrm{ref}}_{\omega}}\;,\qquad \qquad \|\mathsf Q^{\mathrm{+,ref}}_{\omega}\|^{2} := \sum_{\rho\in S_{M}} |\mathsf Q^{+,\mathrm{ref}}_{\omega,\rho}|^{2}=1+O(\lambda)\;,
\]
such that, for $\ul k':=\ul k - \ul k^{\omega,\lambda}_{F,L}$ small enough we have:
\begin{equation}\label{eq:2pt2ref}
\langle \mathbf T \hat a_{\ul k,\rho} \hat a^{*}_{\ul k,\rho'}\rangle_{\beta,L} = \mathsf Q^{+,\mathrm{ref}}_{\omega,\rho} \mathsf Q^{-,\mathrm{ref}}_{\omega,\rho'}\,\langle \hat \psi^{-}_{\ul k',\omega} \hat \psi^{+}_{\ul k',\omega}\rangle^{\mathrm{ref}}_{\beta,L,N} + R^{\beta, L,N}_{2;\omega, \rho\rho'}(\underline{k}')\;.
\end{equation}
with the remainder term being more regular for small $\ul k'$, namely
\begin{equation} \label{eq:2pterr}
\frac{1}{\beta L}|R^{\beta, L,N}_{2;\omega, \rho \rho'}(\underline{k}')|\leq C_{\gamma}\|\ul k'\|^{\gamma-1}
\end{equation}
with $\gamma\in(0,1)$, and $C_{\gamma}$ independent of $\beta, L, N$.\\

\noindent{\underline{\emph{Lattice vertex function.}}}  Furthermore, for $\|\underline{k}'\|$ small enough, and $\|\underline{k}' + \underline{p}\| = \|\underline{k}'\|$, we have that:
\begin{equation}\label{eq:vertex2ref}
\begin{split}
\langle \mathbf T \hat \jmath_{\mu,\ul p}; \hat a_{\ul k,\rho}; \hat a^{*}_{\ul k+\ul p,\rho'}\rangle_{\beta,L}
= \sum_{\omega'} \mathsf Z^{\mu,\mathrm{ref}}_{\omega'}\mathsf Q^{+,\mathrm{ref}}_{\omega,\rho} \mathsf Q^{-,\mathrm{ref}}_{\omega,\rho'}\, \langle \hat n_{\ul p,\omega'};\hat \psi^{-}_{\ul k',\omega}; \hat \psi^{+}_{\ul k'+\ul p,\omega}\rangle^{\mathrm{ref}}_{\beta,L,N} \\+ R^{\beta, L,N}_{\mu;\omega,\rho\rho'}(\underline{p}, \underline{k}')\;,
\end{split}
\end{equation}
with the error term again being more regular, that is
\begin{equation} \label{eq:Zdec}
\frac{1}{\beta L}|R^{\beta, L,N}_{\mu;\omega,\rho\rho'}(\underline{p}, \underline{k}')|\leq C_{\gamma} \|\ul k'\|^{\gamma-2}\;,
\end{equation}
with $\gamma\in (0,1)$ and $C_{\gamma}$ independent of $\beta, L, N$.\\

\noindent{\underline{\emph{Density-current correlation function.}}} Finally, we have that
\begin{equation} \label{eq:Match2ptDens}
\langle\mathbf T \hat n_{\ul p}; \hat\jmath_{\mu,-\ul p}\rangle_{\beta,L} = \sum_{\omega,\omega'}\mathsf Z^{0,\mathrm{ref}}_{\omega} \mathsf Z^{\mu,\mathrm{ref}}_{\omega'} \langle\hat n_{\ul p,\omega};\hat n_{-\ul p,\omega'}\rangle^{\mathrm{ref}}_{\beta,L,N} + R^{\beta,L,N}_{2;\mu}(\ul p)
\end{equation}
where $R^{\beta,L,N}_{2;\mu}$ is H\"older-continuous at $\ul 0$, for an exponent $\gamma>0$; that is,
\begin{equation} \label{eq:Match2ptDensBd}
\frac{1}{\beta L}|R^{\beta,L,N}_{2;\mu}(\ul p)-R^{\beta,L,N}_{2;\mu}(\ul 0)|\leq C \|\ul p\|^{\gamma}\;,
\end{equation}
with $C$ independent of $\beta, L, N$.
\end{proposition}
\begin{proof}
We refer to Proposition 6.1 of \cite{AMP}, where a similar statement was proven for the case of a single, spin-degenerate edge mode on a single boundary of a 2$d$ topological insulator, and Proposition 10.2 of \cite{MPmulti}, which generalized the former to the case of arbitrarily many edge modes. In any case, the proof strategy is similar to Proposition \ref{prop:CorrComparison}, yet simpler: since the number of points is at most three, one does not have to distinguish between the $n\leq n_{*}$ and $n> n_{*}$ cases, and dimensional bounds analogous to Proposition \ref{prop:CorrDimBound} (combined with the short-memory and continuity properties of GN trees) are enough.
\end{proof}

\subsection{Evaluation of the linear response}
In this section, we shall evaluate explicitly the linear term ($m=2$) in the expansion for $\chi^{\beta,L}_{\nu}$, for $\beta,L$ large. A similar computation has been performed in \cite{MPmulti}, for the edge conductance of $2d$ topological insulators. We have
\begin{equation} \label{eq:lowestDuhamel}
\chi^{\beta,L,\mathrm{lin}}_{\nu} (x;\eta,\theta) = -\frac{1}{L}\sum_{p\in B_{L}} \hat\mu_{\theta}(-p) e^{-ipx} \frac{1}{\beta L}\langle\mathbf T \hat n_{\ul p}; \hat\jmath_{\nu,-\ul p}\rangle_{\beta,L}\;;
\end{equation}
noticing that equations (\ref{eq:Match2ptDens})-(\ref{eq:Match2ptDensBd}) prove boundedness of the density-current 2-point correlation uniformly in $\beta,L$, and $N$, and recalling that the perturbation $\mu$ is given by the $L$-periodization of a compactly supported smooth function $\mu_{\infty}$, we may rewrite (\ref{eq:lowestDuhamel}) as
\begin{equation}
\chi^{\beta,L,\mathrm{lin}}_{\nu} (x;\eta,\theta) = \chi^{\mathrm{lin}}_{\nu} (x;\eta,\theta) + O(\beta^{-1}) + O(L^{-1})\;,
\end{equation}
with $\chi^{\mathrm{lin}}_{\nu}$ encoding the zero-temperature, infinite-volume response:
\begin{equation} \label{eq:chilin2}
\chi^{\mathrm{lin}}_{\nu} (x;\eta,\theta)=-\int_{\R} \hat\mu_{\infty}(-p/\theta) e^{-ipx} \langle\mathbf T \hat n_{\ul p}; \hat\jmath_{\nu,-\ul p}\rangle\frac{dp}{2\pi\theta}
\end{equation}
with $\ul p =(\eta,p)$, and with $\langle\mathbf T \hat n_{\ul p}; \hat\jmath_{\nu,-\ul p}\rangle =\lim_{\beta\to\infty}\lim_{L\to\infty}\langle\mathbf T \hat n_{\ul p}; \hat\jmath_{\nu,-\ul p}\rangle_{\beta,L}/(\beta L)$ being the zero-temperature, infinite volume correlation function.
\begin{proposition}[Final form of linear response] \label{prop:linrespint}
Let $\theta=a\eta$. For $\nu=0,1$, the linear response $\chi^{\mathrm{lin}}_{\nu}$ of the system, \emph{i.e.} the $m=2$ contribution in the Duhamel expansion for $\chi^{\beta,L}_{\nu}$, is given by
\begin{equation}
\chi^{\mathrm{lin}}_{\nu}(x;\eta,\theta)= -\int_{\R} \hat\mu_{\infty}(q) e^{i\theta qx} (\vec 1, \mathcal K^{\nu}(q)\vec 1) \frac{dq}{2\pi} + O(\eta^{\gamma}) 
\end{equation}
for some $\gamma\in(0,1)$, $(\cdot,\cdot)$ being the scalar product on $\R^{N_{f}}$ and with $\vec 1$ being the vector with all components equal to 1. The matrix $\mathcal K^{\nu}$ given by
\begin{equation}
\mathcal K^{\nu}(q) = \widetilde{\mathcal K}(q)\frac{v^{\mathrm{ref}}_{\nu}}{2\pi|v^{\mathrm{ref}}|}\frac{v^{\mathrm{ref}}q}{-i/a+v^{\mathrm{ref}}q}\mathfrak K^{\nu}\;,
\end{equation}
with $v^{\mathrm{ref}}_{\nu}=\delta_{\nu,0} \mathbbm 1+ \delta_{\nu,1} v^{\mathrm{ref}}$, and $\widetilde{\mathcal K}$, $\mathfrak K^{\nu}$ defined to be
\begin{equation}
\begin{split}
\widetilde{\mathcal K}(q) &= \left[1-\frac{1}{4\pi|v^{\mathrm{ref}}|}\Lambda\right]\frac{1}{1+\mathfrak B(-1/a,q)\Lambda}\\
\mathfrak K^{0} &\equiv \mathbbm 1\\
\mathfrak K^{1} &\equiv \mathrm{sgn\,}v^{\mathrm{ref}}\,\frac{1+(4\pi|v^{\mathrm{ref}}|)^{-1}\Lambda}{1-(4\pi|v^{\mathrm{ref}}|)^{-1}\Lambda}\,\mathrm{sgn\,}v^{\mathrm{ref}}\;.
\end{split}
\end{equation}
\end{proposition}
\begin{proof}
Since equations (\ref{eq:Match2ptDens})-(\ref{eq:Match2ptDensBd}) hold uniformly in $\beta,L,N$, we may write
\begin{equation} \label{eq:2ptSplitting}
\begin{split}
\langle\mathbf T \hat n_{\ul p}; \hat\jmath_{\nu,-\ul p}\rangle &= \sum_{\omega,\omega'}\mathsf Z^{0,\mathrm{ref}}_{\omega} \mathsf Z^{\nu,\mathrm{ref}}_{\omega'} \langle\hat n_{\ul p,\omega};\hat n_{-\ul p,\omega'}\rangle^{\mathrm{ref}} + R_{2;\nu}(\ul 0) + \tilde R_{2;\nu}(\ul p)\\
\tilde R_{2;\nu}(\ul p)&= R_{2;\nu}(\ul p)-R_{2;\nu}(\ul 0)\;,\qquad\qquad |\tilde R_{2;\nu}(\ul p)|\leq C \|\ul p\|^{\gamma}\;.
\end{split}
\end{equation}
Next, we determine $R_{2;\nu}(\ul 0)$ in terms of the reference model correlations. The lattice continuity equation (\ref{eq:cons1d}) implies the following lattice Ward identity,
\begin{equation}
i\partial_{x_{0}}\langle\mathbf T n_{\ul x};j_{\nu,\ul y}\rangle + \mathrm d_{x} \langle\mathbf T j_{\ul x};j_{\nu,\ul y}\rangle = i\delta(x_{0}-y_{0}) \langle [n_{x},j_{\nu,y}]\rangle\;,
\end{equation}
where as usual $\mathcal O_{\ul x}:=\gamma_{x_{0}}(\mathcal O_{x})$, and where the contact term on the right-hand side arises due to the definition of time-ordering. Taking the Fourier transform on both sides, we obtain:
\begin{equation} \label{eq:latticeWardFou}
p_{0}\langle\mathbf T \hat n_{\ul p};\hat\jmath_{\nu,-\ul p}\rangle + (1-e^{-ip})\langle\mathbf T \hat\jmath_{\ul p};\hat\jmath_{\nu,-\ul p}\rangle = i\sum_{x\in\Z} e^{-ipx}\langle[n_{x},j_{\nu,0}]\rangle\;.
\end{equation}
Evaluating (\ref{eq:latticeWardFou}) at $\underline{p} = (p_{0}, 0)$, and using that $[\sum_{x\in\Z} n_{x},j_{\nu,0}]=[\mathcal N,j_{\nu,0}]=0$, we are then led to
\[
\langle\mathbf T \hat n_{(p_{0},0)};\hat\jmath_{\nu,(-p_{0},0)}\rangle = 0\;.
\]
Hence, by (\ref{eq:2ptSplitting}), we obtain
\begin{equation}
R_{2;\nu}(\ul 0) = -\sum_{\omega,\omega'} \mathsf Z^{0,\mathrm{ref}}_{\omega} \mathsf Z^{\nu,\mathrm{ref}}_{\omega'} \lim_{p_{0}\to 0}\langle\hat n_{(p_{0},0),\omega};\hat n_{(-p_{0},0),\omega'}\rangle^{\mathrm{ref}}\;.
\end{equation}
We set $S^{\mathrm{ref}}_{2}(p_{0},p)\in\mathcal M_{N_{f}}(\C)$ to be the matrix of the reference model density correlations, \emph{i.e.}
\begin{equation}
S^{\mathrm{ref}}_{2;\omega\omega'}(p_{0},p) := \langle\hat n_{\ul p,\omega};\hat n_{-\ul p,\omega'}\rangle^{\mathrm{ref}}\;.
\end{equation}
By Proposition \ref{prop:2ptref}, it takes the form
\begin{equation}
S^{\mathrm{ref}}_{2}(p_{0},p) = \frac{1}{1+\hat v(\ul p)\mathfrak B(\ul p)\Lambda_{Z}} \frac{\mathfrak B(\ul p)}{Z^{2}}= \frac{1}{1+\mathfrak B(\ul p)\Lambda_{Z}} \frac{\mathfrak B(\ul p)}{Z^{2}} + O(\|\ul p\|^{\gamma})
\end{equation}
with
\begin{equation}
\mathfrak B(\ul p) = \frac{1}{4\pi|v^{\mathrm{ref}}|}\frac{-ip_{0}+v^{\mathrm{ref}}p}{ip_{0}+v^{\mathrm{ref}}p}\;,\qquad\qquad \Lambda_{Z}=Z^{-1}\Lambda Z
\end{equation}
and $v^{\mathrm{ref}} = \mathrm{diag\;}{v^{\mathrm{ref}}_{\omega}}$, $\Lambda=(\lambda^{\mathrm{ref}}_{\omega\omega'})_{\omega\omega'}$, and $Z=\mathrm{diag}\,Z^{\mathrm{ref}}_{\omega}$.
In the second step we approximated the reference model interaction $\hat v(\ul p)$ with $\hat v(\ul 0)\equiv1$; also, we used the convenient short-hand of writing $AB^{-1}\equiv \frac{A}{B}$ whenever $A,B$ commute.

Next, realising that $\mathfrak B(\varepsilon\ul p)=\mathfrak B(\ul p)$ for all $\varepsilon\in \R^{\times}$, recalling that $\theta=a\eta$, and reparameterizing $\ul p=\theta\ul q$, with $\ul q=(a^{-1},q)$, we get
\begin{equation} \label{eq:scalcorr}
\langle\mathbf T \hat n_{\theta\ul q}; \hat\jmath_{\nu,-\theta\ul q}\rangle = (\vec{\mathsf Z}^{0,\mathrm{ref}},\mathcal Q(q;a)\vec{\mathsf Z}^{\nu,\mathrm{ref}}) + O(\eta^{\gamma}(1+a^{2} q^{2})^{\gamma/2})
\end{equation}
with $\vec{\mathsf Z}^{\nu,\mathrm{ref}}$ the vector in $\R^{N_{f}}$ with components $\mathsf Z^{\nu,\mathrm{ref}}_{\omega}$, and $(\cdot\,,\cdot)$ the standard scalar product on $\R^{N_{f}}$, and with the matrix-valued function $\mathcal Q$ being
\begin{equation}
\begin{split}
\mathcal Q(q;a) &= \frac{1}{1+\mathfrak B(1,aq)\Lambda_{Z}}\frac{\mathfrak B(1,aq)}{Z^{2}} - \frac{1}{1+\mathfrak B(1,0)\Lambda_{Z}}\frac{\mathfrak B(1,0)}{Z^{2}}\\
&= \frac{1}{1+\mathfrak B(1,aq)\Lambda_{Z}} \frac{v^{\mathrm{ref}}aq}{i+v^{\mathrm{ref}}aq} \frac{1}{1-(4\pi|v^{\mathrm{ref}}|)^{-1}\Lambda_{Z}}\frac{1}{2\pi|v^{\mathrm{ref}}|Z^{2}}\;.
\end{split}
\end{equation}
Then, inserting of (\ref{eq:scalcorr}) into (\ref{eq:chilin2}), and using the fact that $\hat\mu_{\infty}$ is of Schwartz class, we are led to
\begin{equation} \label{eq:linK}
\chi^{\mathrm{lin}}_{\nu} (x;\eta,\theta) = -\int_{\R} \hat\mu_{\infty}(q) e^{i\theta qx} (\vec{\mathsf Z}^{0,\mathrm{ref}},\mathcal Q(-q;a)\vec{\mathsf Z}^{\nu,\mathrm{ref}})\frac{dq}{2\pi} + O(\eta^{\gamma})\;.
\end{equation}
In the case in which $a(\eta):=\theta/\eta=o(1)$, we see that $\|\mathcal Q(q;a)\|=O(a)$, hence the response is trivial, namely
\[
\chi^{\mathrm{lin}}_{\nu} (x;\eta,\theta) = O(a) + O(\eta^{\gamma})\;;
\]
therefore we may focus on the Euler scaling case, \emph{i.e.} we shall assume $a$ to be constant.
 
We shall now evaluate $\vec{\mathsf Z}^{\nu,\mathrm{ref}}$ in terms of the other bare parameters of the reference model, namely $Z$, $v^{\mathrm{ref}}$, and $\Lambda$.

\paragraph{Computing the vertex renormalizations.}
To avoid cluttering of the notation, we set, for $\mu=0,1$:
\begin{equation}
\begin{split}
\Sigma_{\rho\rho'}(\ul k) &:= \langle\mathbf T \hat a_{\ul k,\rho}; \hat a^{*}_{\ul k,\rho'}\rangle\qquad\qquad S_{1,2;\mu,\rho\rho'}(\ul p,\ul k) := \langle\mathbf T \hat \jmath_{\mu,\ul p}; a_{\ul k,\rho}; a^{*}_{\ul k+\ul p,\rho'}\rangle\\
\Sigma^{\mathrm{ref}}_{\omega}(\ul k') &:= \langle \hat \psi^{-}_{\ul k',\omega} \hat \psi^{+}_{\ul k',\omega}\rangle^{\mathrm{ref}} \qquad\qquad S^{\mathrm{ref}}_{1,2;\omega\omega'}(\ul p,\ul k'):= \langle\hat n_{\ul p,\omega}; \hat \psi^{-}_{\ul k',\omega'}; \hat \psi^{+}_{\ul k'+\ul p,\omega'}\rangle^{\mathrm{ref}}
\end{split}
\end{equation}
The lattice vertex function satisfies the following lattice vertex Ward identity:
\begin{equation}\label{eq:latvertWI}
\sum_{\mu=0}^{1}p_{\mu}Y_{\mu}(p) S_{1,2;\mu,\rho\rho'}(\ul p,\ul k) = \Sigma_{\rho\rho'}(\ul k) -\Sigma_{\rho\rho'}(\ul k+\ul p)\;
\end{equation}
with $Y_{0}(p)\equiv i$, $p_{1}=p$, and $Y_{1}(p)=(1-e^{-ip})/ip=1+O(p)$. By (\ref{eq:vertex2ref}), for $\underline{k} = \underline{k}^{\omega}_{F}(\lambda) + \underline{k}'$:
\begin{eqnarray}
S_{1,2;\mu,\rho\rho'}(\ul p,\ul k)  = \sum_{\omega'} \mathsf Z^{\mu,\text{ref}}_{\omega'} \mathsf Q^{+,\text{ref}}_{\omega,\rho} \mathsf Q^{-,\text{ref}}_{\omega,\rho'}S^{\mathrm{ref}}_{1,2;\omega'\omega}(\ul p,\ul k') + R_{\mu,\omega,\rho\rho'}(\underline{p}, \underline{k}')\;,
\end{eqnarray}
and, recalling (\ref{eq:2pt2ref}):
\begin{equation}
\Sigma_{\rho\rho'}(\ul k) = \mathsf Q^{+,\mathrm{ref}}_{\omega,\rho} \mathsf Q^{-,\mathrm{ref}}_{\omega,\rho'}\,\Sigma^{\mathrm{ref}}_{\omega}(\ul k') + R_{2;\omega, \rho\rho'}(\underline{k}')\;.
\end{equation}
Since $\|\mathsf Q^{+,\mathrm{ref}}_{\omega}\|=1 + O(\lambda)$, we may multiply both sides of (\ref{eq:latvertWI}) by $\mathsf Q^{-,\text{ref}}_{\omega,\rho} \mathsf Q^{+,\text{ref}}_{\omega,\rho'}/\|\mathsf Q^{+,\mathrm{ref}}_{\omega}\|^{4}$, and we may sum over all $\rho,\rho'$, obtaining:
\begin{equation}\label{eq:WIlat2}
\begin{split}
\sum_{\mu = 0}^{1} p_{\mu} Y_{\mu}(p) \Big[ \sum_{\omega'} \mathsf Z^{\mu,\mathrm{ref}}_{\omega'} &S^{\mathrm{ref}}_{1,2;\omega'\omega}(\ul p,\ul k') + R_{\mu,\omega}(\underline{p}, \underline{k}') \Big] =\\
& \Sigma^{\mathrm{ref}}_{\omega}(\ul k')  - \Sigma^{\mathrm{ref}}_{\omega}(\ul k'+\ul p) 
 + R_{2;\omega}(\underline{k}') - R_{2;\omega}(\underline{k}' + \underline{p})\;,
\end{split}
\end{equation}
where we introduced the notations:
\begin{eqnarray}
R_{\mu,\omega}(\underline{p}, \underline{k}') &:=& \frac{1}{\|\mathsf Q^{+,\mathrm{ref}}_{\omega}\|^{4}}\sum_{\rho,\rho'} R_{\mu,\omega,\rho\rho'}(\underline{p}, \underline{k}') \mathsf Q^{-,\text{ref}}_{\omega,\rho} \mathsf Q^{+,\text{ref}}_{\omega,\rho'} \nonumber\\
R_{2;\omega}(\underline{k}') &:=& \frac{1}{\|\mathsf Q^{+,\mathrm{ref}}_{\omega}\|^{4}}\sum_{\rho,\rho'} R_{2;\omega, \rho\rho'}(\underline{k}') \mathsf Q^{-,\text{ref}}_{\omega,\rho} \mathsf Q^{+,\text{ref}}_{\omega,\rho'}\;.
\end{eqnarray}
Next, recall the explicit expression for the vertex function of the reference model, equation (\ref{eq:refVertexWI}):
\begin{equation}
S^{\mathrm{ref}}_{1,2;\omega'\omega}(\ul p,\ul k') = \frac{T_{\omega'\omega}(\underline{p})}{Z^{\text{ref}}_{\omega} D_{\omega}(\underline{p})} \big[\Sigma^{\mathrm{ref}}_{\omega}(\ul k')  - \Sigma^{\mathrm{ref}}_{\omega}(\ul k'+\ul p)]\;.
\end{equation}
Recalling the estimates (\ref{eq:2pterr}), (\ref{eq:Zdec}) for the error terms, and choosing $\|\underline{k}'\| = \kappa$, $\|\underline{k}' + \underline{p}\| = \kappa$, $\|\underline{p}\| = O(\kappa)$, we see that in (\ref{eq:WIlat2}) the $R$-terms give contributions that are suppressed by a factor $\kappa^{\gamma}$ with respect to the other quantities. Hence, we get:
\begin{equation}\label{eq:divide}
\sum_{\omega'} D^{\mathsf Z}_{\omega'}(\underline{p}) \frac{T_{\omega'\omega}(\underline{p})}{Z^{\text{ref}}_{\omega} D_{\omega}(\underline{p})}\big[ \Sigma^{\mathrm{ref}}_{\omega}(\ul k')  - \Sigma^{\mathrm{ref}}_{\omega}(\ul k'+\ul p) \big] = \Sigma^{\mathrm{ref}}_{\omega}(\ul k')  - \Sigma^{\mathrm{ref}}_{\omega}(\ul k'+\ul p) + O(\kappa^{\gamma - 1})
\end{equation}
with 
\begin{equation}
D^{\mathsf Z}_{\omega'}(\underline{p}) := ip_{0}\mathsf Z^{0,\text{ref}}_{\omega'} + p \mathsf Z^{1,\text{ref}}_{\omega'}\;.
\end{equation}
Dividing both sides of (\ref{eq:divide}) by $\big[ \Sigma^{\mathrm{ref}}_{\omega}(\ul k')  - \Sigma^{\mathrm{ref}}_{\omega}(\ul k'+\ul p) \big]/ [Z^{\text{ref}}_{\omega} D_{\omega}(\underline{p})] = O(\kappa^{-2})$, we find:
\begin{equation}\label{eq:DZD}
\sum_{\omega'}  D^{\mathsf Z}_{\omega'}(\underline{p}) T_{\omega'\omega}(\underline{p}) = Z^{\text{ref}}_{\omega} D_{\omega}(\underline{p}) + O(\kappa^{1+\gamma})\;.
\end{equation}
Recalling the expression (\ref{eq:Tdef}) for $T(\underline{p})$, we have that
\begin{equation}
\lim_{p_{0}\to 0} \lim_{p\to 0} T(\underline{p}) = \frac{1}{1-(4\pi|v^{\text{ref}}|)^{-1}\Lambda_{Z}}\;,\qquad \lim_{p\to 0} \lim_{p_{0}\to 0} T(\underline{p})  = \frac{1}{1+(4\pi|v^{\text{ref}}|)^{-1}\Lambda_{Z}}\;.
\end{equation}
Therefore, evaluation of equation (\ref{eq:DZD}) by taking $\ul p\to0$ horizontally and vertically implies the following two equalities:
\begin{equation}
\left(\frac{1}{1-(4\pi|v^{\text{ref}}|)^{-1}\Lambda_{Z}}\right)^{T} \vec{\mathsf Z}^{0,\text{ref}}  = \vec Z^{\text{ref}}\;,\qquad \left(\frac{1}{1+(4\pi|v^{\text{ref}}|)^{-1}\Lambda_{Z}}\right)^{T} \vec{\mathsf Z}^{1,\text{ref}} = v^{\text{ref}}\vec Z^{\text{ref}}\;,
\end{equation}
with $\vec Z^{\mathrm{ref}}$ the vector of components $Z^{\mathrm{ref}}_{\omega}$. Inverting, we obtain, for $\mu=0,1$,
\begin{equation}\label{eq:Z0Z1}
\vec{\mathsf Z}^{\mu,\text{ref}} = \big[ 1 - (-1)^{\mu} \Lambda_{Z}^{T} (4\pi|v^{\text{ref}}|)^{-1} \big] v^{\mathrm{ref}}_{\mu}\vec Z^{\text{ref}}\;,
\end{equation}
where $v^{\mathrm{ref}}_{0}:=\mathbbm 1_{N_{f}}$ and $v^{\mathrm{ref}}_{1}:=v^{\mathrm{ref}}$.
\paragraph{Final form of the response.}
Substituting equation (\ref{eq:Z0Z1}) into (\ref{eq:linK}), we obtain, for $\nu=0,1$,
\begin{equation} \label{eq:linsub}
\chi^{\mathrm{lin}}_{\nu} (x;\eta,\theta) = -\int_{\R} \hat\mu_{\infty}(q) e^{i\theta qx} (\vec{Z}^{\mathrm{ref}},\mathcal K^{\nu}_{Z}(q)\vec{Z}^{\mathrm{ref}})\frac{dq}{2\pi} + O(\eta^{\gamma})\;,
\end{equation}
with $\mathcal K^{\nu}_{Z}$ being
\begin{equation} \label{eq:KZ}
\begin{split}
\mathcal K^{\nu}_{Z}(q) &= \left[1-\frac{1}{4\pi|v^{\mathrm{ref}}|}\Lambda_{Z}\right]\mathcal Q(-q;a)  \left[1-(-1)^{\nu}\Lambda^{T}_{Z}\frac{1}{4\pi|v^{\mathrm{ref}}|}\right] v^{\mathrm{ref}}_{\nu}\\
&= \frac{1}{Z}\left[1-\frac{1}{4\pi|v^{\mathrm{ref}}|}\Lambda\right] Z\mathcal Q(-q;a)Z\left[1-(-1)^{\nu}\Lambda\frac{1}{4\pi|v^{\mathrm{ref}}|}\right] \frac{v^{\mathrm{ref}}_{\nu}}{Z}\;.
\end{split}
\end{equation}
Here, we employed the symmetry of $\Lambda$ and the fact that $Z$ commutes with $v^{\mathrm{ref}}$ (since both are diagonal) to bring the conjugation by $Z$ from just $\Lambda$ to the whole expression. Now we will show that $Z\mathcal Q Z$ is independent of $Z$: indeed, since $\mathfrak B$ is diagonal, we have
\begin{equation} \label{eq:ZQZ}
\begin{split}
Z\mathcal Q(-q) Z&= Z\frac{1}{1+\mathfrak B(-1,aq)\Lambda_{Z}} \frac{v^{\mathrm{ref}}q}{-i/a+v^{\mathrm{ref}}q} \frac{1}{1-(4\pi|v^{\mathrm{ref}}|)^{-1}\Lambda_{Z}}\frac{1}{2\pi|v^{\mathrm{ref}}|Z^{2}}Z\\
&=\frac{1}{1+\mathfrak B(-1,aq)\Lambda} \frac{v^{\mathrm{ref}}q}{-i/a+v^{\mathrm{ref}}q} \frac{1}{1-(4\pi|v^{\mathrm{ref}}|)^{-1}\Lambda}\frac{1}{2\pi|v^{\mathrm{ref}}|}\;.
\end{split}
\end{equation}
Inserting (\ref{eq:KZ}) and (\ref{eq:ZQZ}) into (\ref{eq:linsub}), we obtain at last
\begin{equation}
\chi^{\mathrm{lin}}_{\nu} (x) = -\int_{\R} \hat\mu_{\infty}(q) e^{i\theta qx} (\vec 1,\mathcal K^{\nu}(q)\vec 1\,)\frac{dq}{2\pi} + O(\eta^{\gamma})\;,
\end{equation}
with $\vec 1$ being the vector with all components equal to 1, and $\mathcal K^{\nu}$ a matrix-valued function depending only on $v^{\mathrm{ref}}$ and $\Lambda$, given as
\begin{equation}
\begin{split}
\mathcal K^{0}(q) &= \left[1-\frac{1}{4\pi|v^{\mathrm{ref}}|}\Lambda\right]\frac{1}{1+\mathfrak B(-1/a,q)\Lambda} \frac{1}{2\pi|v^{\mathrm{ref}}|}\frac{v^{\mathrm{ref}}q}{-i/a+v^{\mathrm{ref}}q}\\
\mathcal K^{1}(q) &= \left[1-\frac{1}{4\pi|v^{\mathrm{ref}}|}\Lambda\right]\frac{1}{1+\mathfrak B(-1/a,q)\Lambda}\frac{\mathrm{sgn\,}v^{\mathrm{ref}}}{2\pi}\frac{v^{\mathrm{ref}}q}{-i/a+v^{\mathrm{ref}}q} \mathfrak K \\
\mathfrak K &\equiv \mathrm{sgn\,}v^{\mathrm{ref}}\,\frac{1+(4\pi|v^{\mathrm{ref}}|)^{-1}\Lambda}{1-(4\pi|v^{\mathrm{ref}}|)^{-1}\Lambda}\,\mathrm{sgn\,}v^{\mathrm{ref}}\;.
\end{split}
\end{equation}
The proof is completed by recognising $\widetilde{\mathcal K}$, $\mathfrak K^{\nu}$ and $v^{\mathrm{ref}}_{\nu}$ in the expression above.
\end{proof}

\subsection{Proof of Theorem \ref{thm:main}}
\begin{proof}[Proof of Theorem \ref{thm:main}]
We begin by recalling the expansion for $\chi^{\beta,L}_{\nu}$, namely
\begin{equation}
\begin{split}
\chi^{\beta,L}_{\nu}(x;\eta,\theta)=-\sum_{m\geq 2}\frac{(-\theta)^{m-2}}{(m-1)!}\frac{1}{L^{m-1}}&\sum_{p\in B^{m-1}_{L}} \prod_{i=1}^{m-1}[\hat\mu_{\theta}(-p_{i})] e^{i p_{m}x}\\
&\cdot\frac{\langle\mathbf T\hat n_{\ul p_{1}};\cdots;\hat n_{\ul p_{m-1}};\hat\jmath_{\nu,\ul p_{m}}\rangle_{\beta,L}}{\beta L} + E^{\beta,L}_{\nu}(x;\eta,\theta)
\end{split}
\end{equation}
with $\ul p_{i}=(\eta_{\beta},p_{i})$, $\ul p_{m}=-\sum_{i=1}^{m-1} \ul p_{i}$, and with $\hat\mu_{\theta}(\cdot):=\hat\mu(\cdot/\theta)/\theta$. Thanks to (\ref{eq:errnuint}), we know that $E^{\beta,L}_{\nu}$ is bounded as
\[
|E^{\beta,L}_{\nu}(x;\eta,\theta)|\leq \frac{C}{\eta^{3}\beta}\;.
\]
Let us focus on the corrections to the linear order $m=2$. We estimate them as
\begin{equation} \label{eq:highord}
\sum_{m\geq 3}\frac{\theta^{m-2}}{(m-1)!}\frac{1}{L^{m-1}}\sum_{p\in B^{m-1}_{L}} \prod_{i=1}^{m-1}\big[|\hat\mu_{\theta}(-p_{i})|\big]\frac{\big|\langle\mathbf T\hat n_{\ul p_{1}};\cdots;\hat\jmath_{\nu,\ul p_{m}}\rangle_{\beta,L}\big|}{\beta L}\;.
\end{equation}
Thanks to Proposition \ref{prop:CorrComparison}, we rewrite the $m$-point correlation functions in terms of the corresponding ones for the reference model:
\begin{equation}\label{eq:match}
\begin{split}
\langle \mathbf T \hat n_{\ul p_{1}};\cdots; \hat \jmath_{\nu,\ul p_{m}}\rangle_{\beta,L} = \sum_{\ul \omega} \Big[\prod_{i=1}^{m-1} \mathsf Z^{0,\mathrm{ref}}_{\omega_{i}}\Big] \mathsf Z^{\nu,\mathrm{ref}}_{\omega_{m}}\,\langle \hat n_{\ul p_{1},\omega_{1}};\cdots;\hat n_{\ul p_{m},\omega_{m}}\rangle^{\mathrm{ref}}_{\beta, L,N} \\
 + R^{\beta,L, N}_{m;\nu} (\ul p_{1},\dots, \ul p_{m})\;.
\end{split}
\end{equation}
By Corollary \ref{cor:vanishingcorr}, the reference model density correlations are vanishing when the UV cut-off is removed, that is, uniformly in $N$ large,
\begin{equation} \label{eq:vanishing}
\frac{1}{\beta L}\left|\Big[\prod_{i=1}^{m-1} \mathsf Z^{0,\mathrm{ref}}_{\omega_{i}}\Big] \mathsf Z^{\nu,\mathrm{ref}}_{\omega_{m}}\,\langle \hat n_{\ul p_{1},\omega_{1}};\cdots;\hat n_{\ul p_{m},\omega_{m}}\rangle^{\mathrm{ref}}_{\beta, L,N}\right|\leq C_{\beta,m}\, N2^{-\gamma N}\;,
\end{equation}
with $\gamma>0$ and $C_{\beta,m}$ depending on $\beta,\eta$, but independent of $\ul\omega,L,N$; also, since the bounds in Proposition \ref{prop:CorrComparison} for $R^{\beta,L,N}_{m;\nu}$ are uniform in $N$ large, by combining these latter with (\ref{eq:vanishing}), in the $N\to\infty$ limit equation (\ref{eq:match}) yields
\begin{equation} \label{eq:latticebdgain}
\frac{1}{\beta L}\big|\langle\mathbf T\hat n_{\ul p_{1}};\cdots;\hat\jmath_{\nu,\ul p_{m}}\rangle_{\beta,L}\big| \leq m!\, C^{m} \eta^{2-m+\gamma} F(m;\lambda,\alpha)
\end{equation}
for some $\gamma\in(0,1-c|\lambda|)$, and for any $\alpha\in(0,1/34)$, with $C$ uniform in $\beta,L$ large. Here, $F$ is given by (\ref{eq:improv}), namely
\begin{equation}
F(m;\lambda,\alpha):=\begin{cases}
1 & \text{if } m<16\\
\displaystyle \frac{1}{\lfloor m/4-4\rfloor!} + |\lambda|\left(\frac{2\alpha}{1-32\alpha}\right)^{(1/4-8\alpha)m} + |\lambda|^{\alpha m} & \text{if } m\geq 16\;.
\end{cases}
\end{equation}
By employing (\ref{eq:latticebdgain}), we therefore get the following estimate for (\ref{eq:highord}):
\begin{equation}
\begin{split}
\sum_{m\geq 3}\frac{\theta^{m-2}}{(m-1)!}\frac{1}{L^{m-1}}&\sum_{\{p_{i}\}\in B^{m-1}_{L}} \prod_{i=1}^{m-1}\big[|\hat\mu_{\theta}(-p_{i})|\big]\frac{\big|\langle\mathbf T\hat n_{\ul p_{1}};\cdots;\hat\jmath_{\nu,\ul p_{m}}\rangle_{\beta,L}\big|}{\beta L}\\
& \leq \sum_{m\geq 3} m C^{m} \eta^{\gamma} F(m;\lambda,\alpha) a^{m-2} \frac{1}{L^{m-1}} \sum_{\{p_{i}\}\in B^{m-1}_{L}} \prod_{i=1}^{m-1} \frac{|\hat\mu(-p_{i}/\theta)|}{\theta}\\
&\leq \eta^{\gamma}\sum_{m\geq 3} mC^{m} F(m;\lambda,\alpha) a^{m-2} \|\hat\mu_{\infty}\|^{m-1}_{L^{1}}
\end{split}
\end{equation}
where we recall that $a:=\theta/\eta$, and where in the first step we employed (\ref{eq:latticebdgain}).

Now, if $a$ is small enough, we do not need the improvement guaranteed by $F(\cdot;\lambda,\alpha)$ for large inputs, since the sum over $m$ is convergent. The smallness of the sum is guaranteed by the $\eta^{\gamma}$ factor. Instead, for more general values of $a$, the aforementioned improvement is crucial to perform the sum. For given $a$, we choose $\alpha\in(0,1/34)$ small enough and $|\lambda|$ small enough so that:
\begin{equation}
2C\|\hat\mu_{\infty}\|_{L^{1}}a \left(\frac{2\alpha}{1-32\alpha}\right)^{1/4-8\alpha}<1\;,\qquad\qquad 2C\|\hat\mu_{\infty}\|_{L^{1}}a |\lambda|^{\alpha}<1\;.
\end{equation}
With this choice, we obtain:
\begin{equation}
\sum_{m\geq 3}\frac{\theta^{m-2}}{(m-1)!}\frac{1}{L^{m-1}}\sum_{\{p_{i}\}\in B^{m-1}_{L}} \prod_{i=1}^{m-1}\big[|\hat\mu_{\theta}(-p_{i})|\big]\frac{\big|\langle\mathbf T\hat n_{\ul p_{1}};\cdots;\hat\jmath_{\nu,\ul p_{m}}\rangle_{\beta,L}\big|}{\beta L}\leq \tilde C_{\lambda,\alpha}\eta^{\gamma}
\end{equation}
for some $\tilde C_{\lambda,\alpha}$ independent of $\beta,L$. This estimate for the higher order terms in the Duhamel expansion, together with Proposition \ref{prop:linrespint}, concludes the proof of Theorem \ref{thm:main}.
\end{proof}

\appendix

\section{Analysis of the Ward identities}

\subsection{Proof of Proposition \ref{prp:ref}}\label{app:ref}

Recall Eq. \eqref{GenWId},% tells us that taking a density insertion $\hat n_{\omega}$ is equivalent (in the sense of expectations) to an insertion of $\Delta_{\omega}(\psi)$; indeed, we have 
\begin{equation}
	0=\frac{1}{\mathcal Z_N(A)}	\int P_N[d\psi]\, e^{-V(\psi)+\tilde C(\psi;A)}\left[Z_\omega\Den(\underline p)\hat n_{\underline p,\omega}-\Delta_{\underline p,\omega}(\psi)\right]\;,
\end{equation}
where the source term is:
\[
\tilde C(\psi;A) := \frac{1}{\beta L}\sum_{\omega=1}^{N_{f}}\sum_{\ul p\in\Dp} \hat A_{-\ul p,\omega} \hat n_{\ul p,\omega}\;.
\]
Differentiating with respect to $\hat A_{-\underline{p}_i,\omega_i},\;i=2,\dots,m$ and setting $A=0$ we get:
\begin{equation}\label{WId}
	Z_{\omega_1}\DenI(\underline p_1)\langle \hat n_{\ul p_{1},\omega_{1}};\cdots;\hat n_{\ul p_{m},\omega_{m}}\rangle_N=\langle\Delta_{\underline p_1,\omega_1}(\psi); \hat n_{\ul p_{2},\omega_{2}};\cdots;\hat n_{\ul p_{m},\omega_{m}}\rangle_N\;.
\end{equation}	
The next lemma provides a useful rewriting of the right-hand side of (\ref{WId}).

\begin{lemma}[Schwinger-Dyson equation for density correlation functions]\label{TecLem1}
Let $m\geq 3$. For any collection of momenta $\{\ul p_{i}\neq \ul 0\}_{i=1}^{m}$ such that $p_{i,0} = p_{1,0} \neq 0$ for $i=1,\ldots, m-1$, $\underline{p}_{m}$ fixed by momentum conservation, and for any choice of chiralities $\{\omega_{i}\}_{i=1}^{m}$, we have:
\begin{equation}\label{eq:deltaWard}\begin{split}
\langle \Delta_{\underline p_1,\omega_1}(\psi);\hat n_{\ul p_{2},\omega_{2}};&\cdots;\hat n_{\ul p_{m},\omega_{m}}\rangle_N\\
=&-\DenI(\ul p_{1})\mathfrak B^N_{\omega_1}(\underline p_1)\hat v(\underline p_1)\sum_{\bar\omega}\lambda_{\omega_1\bar\omega}Z_{\bar\omega}\langle \hat n_{\underline p_1,\bar\omega};\hat n_{\ul p_{2},\omega_{2}};\cdots;\hat n_{\ul p_{m},\omega_{m}}\rangle_N\\
&+\frac{1}{\beta L}\sum_{\underline k\in \D}\Delta_{\omega_1}(\underline k,\underline p_1)\hat g^{(\le N)}_{\omega_1}(\underline k-\underline p_1)\hat g^{(\le N)}_{\omega_1}(\underline k)\,\mathcal R_{\ul\omega}^{N}(\underline k;{\bf{\underline p}}_m)\;,
\end{split}
\end{equation}	
where:
\begin{equation}\label{Resto}\begin{split}
\mathcal R_{\ul\omega}^{N}(\underline k;{\bf\underline p}_m):=&\sum_{\bar\omega,\bar\omega'} \lambda_{\omega_{1}\bar\omega}\lambda_{\omega_{1}\bar\omega'}Z_{\bar\omega}Z_{\bar\omega'}\frac{Z^{2}_{\omega_{1}}}{(\beta L)^{2}}\sum_{\ul q,\ul q'}\hat v_{\mathrm s}(\ul q)\hat v_{\mathrm s}(\ul q')\mathcal F^{2,m+1}_{\ul \omega}(\ul k,\ul q,\ul q';\ul{\mathbf p}_{m})\\
&-\sum_{i=2}^{m}\delta_{\omega_{1},\omega_{i}}\sum_{\bar\omega}\lambda_{\omega_{1}\bar\omega}Z_{\bar\omega}\frac{Z^{2}_{\omega_{1}}}{\beta L}\sum_{\ul q}\hat v_{\mathrm s}(\ul q)\mathcal F^{2,m-1}_{\ul \omega;i}(\ul k,\ul q;\ul{\mathbf p}_{m})\\
&+\sum_{\substack{i,j=2\\ i\neq j}}^{m} \delta_{\omega_{1},\omega_{i}}\delta_{\omega_{1},\omega_{j}}Z_{\omega_{1}}^{2}\mathcal F^{2,m-3}_{\ul \omega;i,j} (\ul k;\ul{\mathbf p}_{m})\;,
\end{split}
\end{equation}
with ${\bf{\underline p}}_m=(\underline p_1,\dots,\underline p_{m})$ and $\hat v_{\mathrm s}(\ul q):=[\hat v(q)+\hat v(-\ul q)]/2$ the even part of the interaction potential. The functions $\mathcal F$ are given by:
\begin{equation}\label{Resto1}
	\begin{split}
		\mathcal F^{2,m+1}_{\ul \omega}(\ul k,\ul q,\ul q';\ul{\mathbf p}_{m}):=&\langle \hat\psi^+_{\underline k-\underline p_1+\underline q,\omega_1}\hat\psi^-_{\underline k-\underline q',\omega_1}\hat n_{\underline q,\bar\omega}\hat n_{\underline q',\bar\omega'};\hat n_{\ul p_{2},\omega_{2}};\cdots;\hat n_{\ul p_{m},\omega_{m}}\rangle_N\;,\\
		\mathcal F^{2,m-1}_{\ul \omega;i}(\ul k,\ul q;\ul{\mathbf p}_{m}):=&\langle \hat\psi^+_{\underline k-\underline p_1+\underline q,\omega_1}\hat \psi^-_{\underline k+\underline p_i,\omega_1}\hat n_{\underline q,\bar\omega};\hat n_{\ul p_{2},\omega_{2}};\cdots;\widehat{\hat n_{\ul p_{i},\omega_{i}}};\cdots;\hat n_{\ul p_{m},\omega_{m}}\rangle_N\\
		&+\langle\hat \psi^-_{\underline k-\underline q,\omega_1}\hat\psi^+_{\underline k-\underline p_1-\underline p_{i},\omega_1}\hat n_{\underline q,\bar\omega};\hat n_{\ul p_{2},\omega_{2}};\cdots;\widehat{\hat n_{\ul p_{i},\omega_{i}}};\cdots;\hat n_{\ul p_{m},\omega_{m}}\rangle_N\;,\\
		\mathcal F^{2,m-3}_{\ul \omega;i,j} (\ul k;\ul{\mathbf p}_{m}):=&\langle \hat \psi^+_{\underline k-\underline p_1-\underline p_i,\omega_1}\hat \psi^-_{\underline k+\underline p_j,\omega_1};\hat n_{\ul p_{2},\omega_{2}};\cdots;\widehat{\hat n_{\ul p_{i},\omega_{i}}} ;\cdots;\widehat{\hat n_{\ul p_{j},\omega_{j}}} ;\cdots;\hat n_{\ul p_{m},\omega_{m}}\rangle_N\;,
	\end{split}	
\end{equation}
with the wide hat indicating omission of a density.
\end{lemma}	

\begin{proof} We can write $\langle \Delta_{\underline p_1,\omega_1}(\psi);\hat n_{\underline{p}_{2},\omega_{2}};\dots;\hat n_{\underline{p}_{m},\omega_{m}}\rangle_N$ as follows:
	\begin{equation}\label{1}\begin{split}
			&\langle \Delta_{\underline p_1,\omega_1}(\psi);\hat n_{\ul p_{2},\omega_{2}};\dots;\hat n_{\ul p_{m},\omega_{m}}\rangle_N\\
			&=\frac{(\beta L)^{m}\partial^m}{\partial \hat B_{-\underline{p}_1,\omega_1}\partial \hat A_{-\underline{p}_{2},\omega_{2}}\cdots\partial \hat A_{-\underline p_{m},\omega_m}}\frac{1}{\mathcal Z_{N}(A)}\int P_{N}[d\psi] e^{-V(\psi)+\tilde C(\psi;A)+\Delta(\psi;B)}\bigg|_{A,B=0}\\
			&= \frac{1}{\beta L}\sum_{\underline k\in\D}\Delta_{\omega_1}(\underline k,\underline p_1)\frac{(\beta L)^{m-1}\partial^{m-1}}{\prod_{i=2}^{m}\partial \hat A_{-\underline{p}_i,\omega_i}}\frac{1}{\mathcal Z_{N}(A)}\int P_{N}[d\psi]\, e^{-V(\psi)+\tilde C(\psi;A)}\hat\psi^+_{\underline k-\underline p_1,\omega_1}\hat \psi^-_{\underline k,\omega_1}\bigg|_{A=0}
		\end{split}
	\end{equation}		
	where we used \eqref{Dp} and where
	\begin{equation}\label{extfieldB}
		\Delta(\psi;B):=\frac{1}{\beta L}\sum_{\omega=1}^{N_{f}}\sum_{\underline p\in\Dp}\hat B_{-\underline p,\omega}\Delta_{\underline p,\omega}(\psi)\;,	
	\end{equation}	
	with $\hat B_{\underline p,\omega}$ external fields defined as $\hat A_{\underline p,\omega}$ (see \eqref{A}).
	
	From the definition of the Grassmann measure given by (\ref{eq:refcov}), we have:
	\begin{equation}\label{intbypart}
		\frac{1}{(\beta L)^{2}}e^{-K_N(\psi)}\hat \psi^+_{\underline k-\underline p_1,\omega_1}\hat\psi^-_{\underline k,\omega_1}=\hat g^{(\le N)}_{\omega_1}(\underline k-\underline p_1)\hat g^{(\le N)}_{\omega_1}(\underline k)	\frac{\partial^2}{\partial \hat \psi^+_{\underline k,\omega_1}\partial\hat\psi^-_{\underline k-\underline p_1,\omega_1}}e^{-K_N(\psi)}
	\end{equation}	
	so that, integrating by parts the argument of the derivatives in \eqref{1}, we obtain:
	\begin{equation}\label{2}\begin{split}
			&\frac{1}{\mathcal Z_N(A)}\int P_N[d\psi]\, e^{-V(\psi)+\tilde C(\psi;A)}\hat\psi^+_{\underline k-\underline p_1,\omega_1}\hat \psi^-_{\underline k,\omega_1}	\\
			&\qquad=\hat g^{(\le N)}_{\omega_1}(\underline k-\underline p_1)\hat g^{(\le N)}_{\omega_1}(\underline k)\frac{(\beta L)^{2}}{\mathcal Z_N(A)}\int P_{N}[d\psi]\frac{\partial^2}{\partial\hat\psi^-_{\underline k-\underline p_1,\omega_1}\partial \hat \psi^+_{\underline k,\omega_1}} e^{-V(\psi)+\tilde C(\psi;A)}\\
			&\qquad=\hat g^{(\le N)}_{\omega_1}(\underline k-\underline p_1)\hat g^{(\le N)}_{\omega_1}(\underline k)\frac{(\beta L)^{2}}{\mathcal Z_N(A)}\int P_{N}[d\psi]\, e^{-V(\psi)+\tilde C(\psi;A)}\times\\
			&\qquad\times\Bigg[\frac{\partial^2[V(\psi)-\tilde C(\psi;A)]}{\partial\hat\psi^-_{\underline k-\underline p_1,\omega_1}\partial \hat \psi^+_{\underline k,\omega_1}}+\frac{\partial[V(\psi)-\tilde C(\psi;A)]}{\partial\hat\psi^-_{\underline k-\underline p_1,\omega_1}}\frac{\partial[V(\psi)-\tilde C(\psi;A)]}{\partial \hat \psi^+_{\underline k,\omega_1}}\Bigg].
		\end{split}
	\end{equation}	
	
	Let us consider the term in \eqref{2} with the second derivative of $V(\psi)-\tilde C(\psi;A)$. From \eqref{V}, \eqref{n}, \eqref{Gamma}, and employing the fact that $\lambda_{\omega\omega}=0$, we find:
	\begin{equation}\label{eq:ap1}
		\frac{\partial^2[V(\psi)-\tilde C(\psi;A)]}{\partial\hat\psi^-_{\underline k-\underline p_1,\omega_1}\partial \hat \psi^+_{\underline k,\omega_1}}=\frac{Z_{\omega_1}}{(\beta L)^{2}}\bigg[\sum_{\bar\omega}\lambda_{\omega_1\bar\omega}Z_{\bar\omega}\hat v(\underline p_1)\hat n_{\underline p_1,\bar\omega}-\hat A_{\underline p_1,\omega_1}\bigg]\;.	
	\end{equation}		
Using this identity in \eqref{2} and differentiating with respect to $\{\hat A_{- \underline p_i,\omega_i}\}_{i=2}^{m}$, the contribution of the first term in the right-hand side of (\ref{2}) to (\ref{1}) is:
	\begin{equation}\label{MainTerm}
		Z_{\omega_1}\hat g^{(\le N)}_{\omega_1}(\underline k-\underline p_1)\hat g^{(\le N)}_{\omega_1}(\underline k)\hat v(\underline p_1)\sum_{\bar \omega}\lambda_{\omega_1\bar \omega}Z_{\bar\omega}\langle \hat n_{\underline p_1,\bar \omega};\hat n_{\ul p_{2},\omega_{2}};\cdots;\hat n_{\ul p_{m},\omega_{m}}\rangle_N\;.
	\end{equation}	
Observe that, since the momenta $\underline{p}_{i}$ are all nonzero, by momentum conservation we see that the term corresponding to $\hat A_{\underline{p}_{1},\omega_{1}}$ in (\ref{eq:ap1}) does not contribute to the result. The identity (\ref{eq:deltaWard}) follows after inserting (\ref{MainTerm}) in \eqref{1}, recalling the definition (\ref{eq:Bolla}) for the anomalous bubble diagram $\mathfrak B^{N}_{\omega_{1}}$, and setting:
	\begin{equation} \label{eq:remderiv}
	\begin{split}
	\mathcal R_{\ul\omega}^{N}(\ul k;\ul{\mathbf p}_{m}):= \frac{(\beta L)^{m-1}\partial^{m-1}}{\prod_{i=2}^{m}\partial \hat A_{-\underline{p}_i,\omega_i}}\frac{(\beta L)^{2}}{\mathcal Z_{N}(A)}\int &P_{N}[d\psi]\, e^{-V(\psi)+\tilde C(\psi;A)}\,\times\\ &\times\frac{\partial[V(\psi)-\tilde C(\psi;A)]}{\partial\hat\psi^-_{\underline k-\underline p_1,\omega_1}}\frac{\partial[V(\psi)-\tilde C(\psi;A)]}{\partial \hat \psi^+_{\underline k,\omega_1}}\Bigg|_{A=0}\;.
	\end{split}
	\end{equation}
	
	The remaining part of the proof is the explicit evaluation of (\ref{eq:remderiv}). We proceed as follows: first, an explicit computation yields
	\begin{equation} \label{eq:der1}
		\begin{split}
			&\frac{\partial[V(\psi)-\tilde C(\psi;A)]}{\partial\hat\psi^-_{\underline k-\underline p_1,\omega_1}}\\
			&=\frac{Z_{\omega_1}}{(\beta L)^{2}}\sum_{\underline q\in\Dp}\bigg[\sum_{\bar\omega}\lambda_{\omega_1\bar\omega}Z_{\bar\omega}\frac{\hat v(\underline q)}{2}\Big(\hat\psi^+_{\underline k-\underline p_1-\underline q,\omega_1}\hat n_{-\underline q,\bar\omega}+\hat n_{\underline q,\bar\omega}\hat \psi^+_{\underline k-\underline p_1+\underline q,\omega_1}\Big)-\hat A_{-\underline q,\omega_1}\hat \psi^+_{\underline k-\underline p_1-\underline q,\omega_1}\bigg]\\
			&=\frac{Z_{\omega_1}}{(\beta L)^{2}}\sum_{\underline q\in\Dp}\bigg[\sum_{\bar\omega}\lambda_{\omega_1\bar\omega}Z_{\bar\omega}\hat v_{\mathrm s}(\underline q)\hat\psi^+_{\underline k-\underline p_1+\underline q,\omega_1}\hat n_{\underline q,\bar\omega}-\hat A_{-\underline q,\omega_1}\hat\psi^+_{\underline k-\underline p_1-\underline q,\omega_1}\bigg]\;,
		\end{split}	
	\end{equation}	
	where $\hat v_{\mathrm s}(\ul q):=[\hat v(\ul q) + \hat v(-\ul q)]/2$, and
	\begin{equation} \label{eq:der2}
		\begin{split}
			&\frac{\partial[V(\psi)-\tilde C(\psi;A)]}{\partial \hat \psi^+_{\underline k,\omega_1}}\\
			&=\frac{Z_{\omega_1}}{(\beta L)^{2}}\sum_{\underline q\in\Dp}\bigg[\sum_{\bar\omega}\lambda_{\omega_1\bar\omega}Z_{\bar\omega}\frac{\hat v(\underline q)}{2}\left(\hat\psi^-_{\underline k+\underline q,\omega_1}\hat n_{-\underline q,\bar\omega}+\hat n_{\underline q,\bar\omega}\hat \psi^-_{\underline k-\underline q,\omega_1}\right)-\hat A_{-\underline q,\omega_1}\hat \psi^-_{\underline k+\underline q,\omega_1}\bigg]\\
			&= \frac{Z_{\omega_1}}{(\beta L)^{2}}\sum_{\underline q\in\Dp}\bigg[\sum_{\bar\omega}\lambda_{\omega_1\bar\omega}Z_{\bar\omega}\hat v_{\mathrm s}(\underline q)\hat\psi^-_{\underline k-\underline q,\omega_1}\hat n_{\underline q,\bar\omega} -\hat A_{-\underline q,\omega_1}\hat\psi^-_{\underline k+\underline q,\omega_1}\bigg]\;.
		\end{split}	
	\end{equation}	
	Second, multiplying (\ref{eq:der1}) with (\ref{eq:der2}), inserting the result in (\ref{eq:remderiv}), deriving with respect to the external fields $\{\hat A_{-\ul p_{i},\omega_{i}}\}_{i=2}^{m}$ and setting $A=0$, we are led to a host of different terms contributing to $\mathcal R_{\ul\omega}^{N}$:
	\begin{enumerate}
		\item \label{en:11} the first term of (\ref{eq:der1}) with the first one of (\ref{eq:der2}) gives
		\begin{equation}
		\sum_{\bar\omega,\bar\omega'} \lambda_{\omega_{1}\bar\omega}\lambda_{\omega_{1}\bar\omega'}Z_{\bar\omega}Z_{\bar\omega'}\frac{Z^{2}_{\omega_{1}}}{(\beta L)^{2}}\sum_{\ul q,\ul q'}\hat v_{\mathrm s}(\ul q)\hat v_{\mathrm s}(\ul q')\mathcal F^{2,m+1}_{\ul \omega}(\ul k,\ul q,\ul q';\ul{\mathbf p}_{m})
		\end{equation}
		with
		\begin{equation*}
		\mathcal F^{2,m+1}_{\ul \omega}(\ul k,\ul q,\ul q';\ul{\mathbf p}_{m}):=\langle \hat\psi^+_{\underline k-\underline p_1+\underline q,\omega_1}\hat\psi^-_{\underline k-\underline q',\omega_1}\hat n_{\underline q,\bar\omega}\hat n_{\underline q',\bar\omega'};\hat n_{\ul p_{2},\omega_{2}};\cdots;\hat n_{\ul p_{m},\omega_{m}}\rangle_N\;;
		\end{equation*}
		\item \label{en:1221} the first (resp. second) term  of (\ref{eq:der1}) with the second (resp. first) one of (\ref{eq:der2}) produces
		\begin{equation}
		-\sum_{i=2}^{m}\delta_{\omega_{1},\omega_{i}}\sum_{\bar\omega}\lambda_{\omega_{1}\bar\omega}Z_{\bar\omega}\frac{Z^{2}_{\omega_{1}}}{\beta L}\sum_{\ul q}\hat v_{\mathrm s}(\ul q)\mathcal F^{2,m-1}_{\ul \omega;i}(\ul k,\ul q;\ul{\mathbf p}_{m})
		\end{equation}
		with
		\begin{equation*}
		\begin{split}
		\mathcal F^{2,m-1}_{\ul \omega;i}(\ul k,\ul q;\ul{\mathbf p}_{m}):=&\langle \hat\psi^+_{\underline k-\underline p_1+\underline q,\omega_1}\hat \psi^-_{\underline k+\underline p_i,\omega_1}\hat n_{\underline q,\bar\omega};\hat n_{\ul p_{2},\omega_{2}};\cdots;\widehat{\hat n_{\ul p_{i},\omega_{i}}};\cdots;\hat n_{\ul p_{m},\omega_{m}}\rangle_N\\
		&+\langle\hat \psi^-_{\underline k-\underline q,\omega_1}\hat\psi^+_{\underline k-\underline p_1-\underline p_{i},\omega_1}\hat n_{\underline q,\bar\omega};\hat n_{\ul p_{2},\omega_{2}};\cdots;\widehat{\hat n_{\ul p_{i},\omega_{i}}};\cdots;\hat n_{\ul p_{m},\omega_{m}}\rangle_N
		\end{split}
		\end{equation*} 
		where the wide hat indicates omission of the corresponding density;
		\item \label{en:22} finally, the second term of (\ref{eq:der1}) with the second one of (\ref{eq:der2}) yields
		\begin{equation}
		Z_{\omega_{1}}^{2}\sum_{\substack{i,j=2\\ i\neq j}}^{m} \delta_{\omega_{1},\omega_{i}}\delta_{\omega_{1},\omega_{j}}\mathcal F^{2,m-3}_{\ul \omega;i,j} (\ul k;\ul{\mathbf p}_{m})
		\end{equation}
		with
		\begin{equation*}
		\mathcal F^{2,m-3}_{\ul \omega;i,j} (\ul k;\ul{\mathbf p}_{m}):=\langle \hat \psi^+_{\underline k-\underline p_1-\underline p_i,\omega_1}\hat \psi^-_{\underline k+\underline p_j,\omega_1};\hat n_{\ul p_{2},\omega_{2}};\cdots;\widehat{\hat n_{\ul p_{i},\omega_{i}}} ;\cdots;\widehat{\hat n_{\ul p_{j},\omega_{j}}} ;\cdots;\hat n_{\ul p_{m},\omega_{m}}\rangle_N\;.
		\end{equation*} 
	\end{enumerate}
	This concludes the proof of Lemma \ref{TecLem1}.
\end{proof}
	The next lemma establishes the bound for the remainder $\mathcal R^{N}_{\ul \omega}$.
	\begin{lemma}\label{TecLem2}
	There exist $\gamma>0$, $C>0$ such that for all $\ul\omega$, and for $\underline{p}_{i}$ as in Lemma \ref{TecLem1}, with the supplementary requirement that $|p_{i}| \leq C|\theta|$ and $p_{i,0} = \eta$ for all $i\leq m-1$, we have:	
	\begin{equation}\label{eq:remest}
		\frac{1}{\beta L}\sup_{c_{1}2^{N}\leq\|\underline k\|_{\omega_{1}}\leq c_{2}2^{N}}\big|\mathcal R^{N}_{\ul \omega}(\ul k;\ul{\mathbf p}_{m})\big|\leq C^{m} m! \beta^{m+2} 2^{-\gamma N}\;.
\end{equation}
	\end{lemma}	
Note that, from Lemma \ref{TecLem2} and from definition of $\mathcal R^{N}_{\ul\omega}$ (see \eqref{Resto}), the bound (\ref{eq:Rest}) immediately follows.
\begin{proof}
In order to use the GN tree estimates introduced earlier, we rewrite all of the $\mathcal F$ in terms of \emph{connected} correlation functions, keeping in mind the following three rules:
	\begin{enumerate}
	\item momentum conservation must hold in each connected correlation function separately; 
	\item $\langle \hat n_{\ul q}\rangle_N=0$ for all $\ul q\in\Dp$;
	\item $\langle \hat \psi^{\pm};\hat n;\cdots ;\hat n\rangle_N=0$ by the fact that expectations of odd fermionic monomials vanish.
	\end{enumerate} 
Thus, we obtain
\begin{equation} \label{eq:Fexp1}
\begin{split}
\mathcal F^{2,m+1}_{\ul \omega}(\ul k,\ul q,\ul q';\ul{\mathbf p}_{m}) =& \langle \hat\psi^+_{\underline k-\underline p_1+\underline q,\omega_1};\hat\psi^-_{\underline k-\underline q',\omega_1};\hat n_{\underline q,\bar\omega};\hat n_{\underline q',\bar\omega'};\hat n_{\ul p_{2},\omega_{2}};\cdots;\hat n_{\ul p_{m},\omega_{m}}\rangle_N\\
&+\delta_{\ul q,\ul p_{1}-\ul q'}\langle \hat\psi^+_{\underline k-\ul q',\omega_1};\hat\psi^-_{\underline k-\underline q',\omega_1}\rangle_{N}\langle\hat n_{\underline p_{1}-\ul q',\bar\omega};\hat n_{\underline q',\bar\omega'};\hat n_{\ul p_{2},\omega_{2}};\cdots;\hat n_{\ul p_{m},\omega_{m}}\rangle_N\\
&+ \delta_{\ul q',\ul p_{1}}\langle \hat\psi^+_{\underline k-\ul p_{1}+\ul q,\omega_1};\hat\psi^-_{\underline k-\underline p_{1},\omega_1};\hat n_{\ul q,\bar\omega}\rangle_{N}\langle\hat n_{\underline p_{1},\bar\omega'};\hat n_{\ul p_{2},\omega_{2}};\cdots;\hat n_{\ul p_{m},\omega_{m}}\rangle_N\\
&+ \delta_{\ul q,\ul p_{1}}\langle \hat\psi^+_{\underline k,\omega_1};\hat\psi^-_{\underline k-\underline q',\omega_1};\hat n_{\ul q',\bar\omega'}\rangle_{N}\langle\hat n_{\underline p_{1},\bar\omega};\hat n_{\ul p_{2},\omega_{2}};\cdots;\hat n_{\ul p_{m},\omega_{m}}\rangle_N\\
&+ \delta_{\ul q,-\ul q'}\langle \hat n_{\ul q,\bar\omega};\hat n_{-\ul q,\bar\omega'}\rangle_{N}\langle\hat\psi^+_{\underline k-\underline p_{1}+\ul q,\omega_1};\hat\psi^-_{\underline k+\ul q,\omega_1};\hat n_{\ul p_{2},\omega_{2}};\cdots;\hat n_{\ul p_{m},\omega_{m}}\rangle_N
\end{split}
\end{equation}
and
\begin{equation} \label{eq:Fexp2}
\begin{split}
\mathcal F^{2,m-1}_{\ul \omega;i}&(\ul k,\ul q;\ul{\mathbf p}_{m}) = \langle \hat\psi^+_{\underline k-\underline p_1+\underline q,\omega_1};\hat \psi^-_{\underline k+\underline p_i,\omega_1};\hat n_{\underline q,\bar\omega};\hat n_{\ul p_{2},\omega_{2}};\cdots;\widehat{\hat n_{\ul p_{i},\omega_{i}}};\cdots;\hat n_{\ul p_{m},\omega_{m}}\rangle_N\\
&+\delta_{\ul q,\ul p_{1}+\ul p_{i}}\langle \hat\psi^+_{\underline k+\underline p_i,\omega_1};\hat \psi^-_{\underline k+\underline p_i,\omega_1}\rangle_{N}\langle\hat n_{\underline p_{1}+\ul p_{i},\bar\omega};\hat n_{\ul p_{2},\omega_{2}};\cdots;\widehat{\hat n_{\ul p_{i},\omega_{i}}};\cdots;\hat n_{\ul p_{m},\omega_{m}}\rangle_N\;.
\end{split}
\end{equation}
Instead, we see that the terms $\mathcal F^{2,m-3}_{\ul \omega;i,j}$ are already connected, namely
\begin{equation} \label{eq:Fexp3}
		\mathcal F^{2,m-3}_{\ul \omega;i,j} (\ul k;\ul{\mathbf p}_{m}) =\langle \hat \psi^+_{\underline k-\underline p_1-\underline p_i,\omega_1} ; \hat \psi^-_{\underline k+\underline p_j,\omega_1};\hat n_{\ul p_{2},\omega_{2}};\cdots;\widehat{\hat n_{\ul p_{i},\omega_{i}}} ;\cdots;\widehat{\hat n_{\ul p_{j},\omega_{j}}} ;\cdots;\hat n_{\ul p_{m},\omega_{m}}\rangle_N\;.
		\end{equation}
		
First, consider all of the correlations occurring in (\ref{eq:Fexp1})-(\ref{eq:Fexp2}) with no $\hat\psi$ field: these involve $m+1$, $m$, $m-1$ or 2 densities. For the sake of concreteness, let us consider a token example of $m$-point correlation functions, namely $\langle\hat n_{\underline p_{1},\bar\omega};\hat n_{\ul p_{2},\omega_{2}};\cdots;\hat n_{\ul p_{m},\omega_{m}}\rangle_N$, whose GN-type estimate is
\begin{equation}
\begin{split}
\frac{1}{\beta L}\big|\langle\hat n_{\underline p_{1},\bar\omega};\hat n_{\ul p_{2},\omega_{2}};\cdots;\hat n_{\ul p_{m},\omega_{m}}\rangle_N\big|\leq\sum_{h=h_{\beta}}^{N}\sum_{n=0}^{\infty}\sum_{\tau,\mathbf \Gamma}\lambda_{\mathrm{max}}^{n}\,2^{h(2-m)}\prod_{v\notin V_{f}(\tau)} 2^{(h_{v}-h_{v'})D_{v}}\;;
\end{split}
\end{equation}
summing over $\mathbf\Gamma\in\mathcal L_{\tau}$, $\tau\in\mathcal T^{m,\mathrm{ref}}_{h,n}$, and $n\in \mathbb N$, for $\lambda_{\mathrm{max}}$ small enough and for $m\geq 3$, we obtain
\begin{equation}
\frac{1}{\beta L}\big|\langle\hat n_{\underline p_{1},\bar\omega};\hat n_{\ul p_{2},\omega_{2}};\cdots;\hat n_{\ul p_{m},\omega_{m}}\rangle_N\big|\leq m!\,C^{m} \sum_{h=h_{\beta}}^{N} 2^{h(2-m)} \leq C_{\beta,m}\;,
\end{equation}
with $C_{\beta,m}\propto m!\, \beta^{m-2}$. One may proceed in the same way for all the other correlators with $m$ or $m+1$ densities. This bound is of course not optimal, since it does not take into account the dependence on the external momenta, but it will be enough for our purposes.

For the density correlator with $m-1$ points appearing in the second row of (\ref{eq:Fexp2}), performing the same estimate we obtain
\begin{equation}
\begin{split}
&\frac{1}{\beta L}\big|\langle\hat n_{\underline p_{1}+\ul p_{i},\bar\omega};\hat n_{\ul p_{2},\omega_{2}};\cdots;\widehat{\hat n_{\ul p_{i},\omega_{i}}};\cdots;\hat n_{\ul p_{m},\omega_{m}}\rangle_N\big| \\&\qquad\qquad \leq C^{m} m! \sum_{h=h_{\beta}}^{N} 2^{h(3-m)} \leq C^{m} m! \cdot  \begin{cases}
\log \beta + N & m=3\\
\beta^{m-3} & m>3
\end{cases}
\end{split}
\end{equation}
and analogously for the two-point density correlation function in the last row of (\ref{eq:Fexp1}). Therefore, all correlation functions in (\ref{eq:Fexp1})-(\ref{eq:Fexp2}) involving only densities may be bounded by $C^{m} m!( \beta^{m-1} + N)$. Again, this is of course not optimal, but it will be enough for our purposes.

Consider now the correlation functions appearing in (\ref{eq:Fexp1})-(\ref{eq:Fexp3}) and having two $\hat\psi$ fields: we notice that, since $c_{1}2^{N}\leq\|\ul k\|_{\omega_{1}}\leq c_{2} 2^{N}$, for $N$ large enough the momenta of these $\hat\psi$ fields will be also of order $2^{N}$, since $\|\ul q\|_{\omega_{1}}, \|\ul q'\|_{\omega_{1}}, \|\ul p_{i}\|_{\omega_{1}}$ are of order 1 in $N$: recall that the momenta $\ul q$, $\ul q'$ are in the support of $\hat v$, that we are assuming $|p_{i}| \leq C|\theta|$, $p_{i,0} = \eta$ for all $i=1,\ldots, m-1$, and that, by momentum conservation, $|\underline{p}_{m}| \leq C(m-1) (\eta + |\theta|)$. For the sake of concreteness let us consider the first row of (\ref{eq:Fexp1}). Its GN expansion is:
\begin{equation}
\langle \hat\psi^+_{\underline k-\underline p_1+\underline q,\omega_1};\hat\psi^-_{\underline k-\underline q',\omega_1};\hat n_{\underline q,\bar\omega};\hat n_{\underline q',\bar\omega'};\hat n_{\ul p_{2},\omega_{2}};\cdots;\hat n_{\ul p_{m},\omega_{m}}\rangle_N= \sum_{h=h_{\beta}}^{N}\sum_{n=0}^{\infty}\sum_{\tau,\mathbf \Gamma} \widehat W^{(h)}_{\mathbf \Gamma}[\tau]\;.
\end{equation}
Since $\phi$-leaves attach only on the (couple of) scales that contain the momentum by which they are labelled, by the discussion above we see that the two $\phi$-leaves will necessarily be on scale $N$ (or $N-1$). Therefore, we may extract a short-memory factor at the root, namely
\begin{equation}
\frac{1}{\beta L}\sum_{n\geq0}\sum_{\tau\in \mathcal T^{m+4,\mathrm{ref}}_{h,n}}\sum_{\mathbf\Gamma\in\mathcal L_{\tau}}\big|\widehat W^{(h)}_{\mathbf\Gamma}[\tau]\big|\leq m!\,C^{m}\, 2^{h(-m-2)} 2^{\tilde\gamma(h-N)}\;;
\end{equation}
summation over the root scale then yields
\begin{equation}
\frac{1}{\beta L}\big|\langle \hat\psi^+_{\underline k-\underline p_1+\underline q,\omega_1};\hat\psi^-_{\underline k-\underline q',\omega_1};\hat n_{\underline q,\bar\omega};\hat n_{\underline q',\bar\omega'};\hat n_{\ul p_{2},\omega_{2}};\cdots;\hat n_{\ul p_{m},\omega_{m}}\rangle_N\big|\leq C^{m} m! \beta^{m+2} 2^{-\tilde \gamma N}\;.
\end{equation}
It is enough to proceed in the same way for all other correlators having two $\hat\psi$ fields. Combining all of the estimates that we have obtained, we may bound the $\mathcal F$ terms as
\begin{equation}\nonumber
\begin{split}
\big|\mathcal F^{2,m+1}_{\ul \omega}(\ul k,\ul q,\ul q';\ul{\mathbf p}_{m})\big|&\leq (\beta L) C^{m} m! (\beta^{m+2} + N) 2^{-\tilde\gamma N}\big[1+\beta L(\delta_{\ul q,\ul p_{1}}+\delta_{\ul q',\ul p_{1}}+\delta_{\ul q,\ul p_{1}}+\delta_{\ul q,-\ul q'})\big]\\
\big|\mathcal F^{2,m-1}_{\ul \omega;i}(\ul k,\ul q;\ul{\mathbf p}_{m})\big|&\leq (\beta L) C^{m} m! (\beta^{m+2} + N) 2^{-\tilde\gamma N} \big[1+\beta L\delta_{\ul q,\ul p_{1}+\ul p_{i}}\big]\\
\big|\mathcal F^{2,m-3}_{\ul \omega;i,j}(\ul k;\ul{\mathbf p}_{m})\big|&\leq (\beta L) C^{m} m! (\beta^{m+2} + N) 2^{-\tilde\gamma N}
\end{split}
\end{equation}
Insertion of these bounds in equation (\ref{Resto}) yields the claim (\ref{eq:remest}), for some $\gamma<\tilde\gamma$.
\end{proof}
\begin{remark} \label{rem:bdR}
Actually, the bounds for the correlators of the reference model imply that:
%repeating the proof of Lemma \ref{TecLem2}, but without requiring the extraction of a short-memory factor to ensure (exponential) smallness in $N$, we also have:
\begin{equation}
|\mathcal R^{N}_{\ul\omega}(\ul k;\ul{\mathbf p}_{m})|\leq \beta L C^{m} m! \beta^{m+2}
\end{equation}
for all $\ul k$, uniformly in $\Ma$, $\varepsilon$ small (with $\Ma/\varepsilon$ small enough). This bound is used to take the limits in (\ref{Rewr}), (\ref{Rewr2}).
\end{remark}

\subsection{Flow of the running coupling constants of the reference model}\label{app:RGref}
The proof of Proposition \ref{prp:flowref} is based on the derivation of anomalous Ward identities for the reference model in presence of an infrared cut-off on scale $h<0$, which is the scale of the couplings we want to control. The derivation of these Ward identities is carried out in a way similar to what has been discussed in Appendix \ref{app:ref}, except that now we also need to introduce external fields linearly coupled to the fermionic fields $\psi^{\pm}$; we shall refer to \cite[Section 9]{MPmulti} for further details. Here we shall summarize the main steps in the analysis, referring to the specific points of \cite{MPmulti}. In what follows, we shall use the notation $O(\lambda_{\geq h}^{n})$ to say that a certain expression is given by a convergent series in the couplings $\lambda^{\text{ref}}_{k,\omega\omega'}$ on scales $k$ greater or equal than $h$, starting from order $n$. Also, we shall use the notation $O(\lambda^{n})$ to denote a convergent series in $\lambda^{\text{ref}}_{\omega\omega'}$, starting from order $n$.

The central object is the flow of the quartic coupling $\lambda^{\text{ref}}_{h,\omega\omega'}$; let us assume that $|\lambda^{\text{ref}}_{h,\omega\omega'}| \leq C \lambda_{\text{max}}$ and that the Fermi velocities are bounded, $|v^{\text{ref}}_{h,\omega} - v^{\mathrm{ref}}_{\omega}| \leq C \lambda_{\text{max}}$, to be proven inductively. The flow of the wave function renormalization and of the vertex renormalizations are multiplicative; they have the form:
\begin{equation}\label{eq:flowZ}
Z^{\text{ref}}_{h,\omega} = Z^{\text{ref}}_{h+1,\omega} (1 + \beta^{z,\text{ref}}_{h+1,\omega})\;,\qquad \mathsf Z^{\mu,\omega}_{h,\omega'} = \mathsf Z^{\mu,\omega}_{h+1,\omega'} (1 + \beta^{\mathsf z,\mu,\omega}_{h+1,\omega'})\;.
\end{equation}
Concerning the first recursion, renormalized second order perturbation theory shows that the solution has a power law behavior $Z^{\text{ref}}_{h,\omega} \sim 2^{-\eta_{\omega} h}$, with $\eta_{\omega} = O(\lambda_{\geq h}^{2})$. The sign of the anomalous exponent $\eta_{\omega}$ will not be needed in the following; in the case of two opposite chiralities it is positive, and we expect the same to be true in general. The second flow equation in (\ref{eq:flowZ}) seems to suggest a similar behavior for $\mathsf Z^{\mu,\omega}_{h,\omega'}$; nevertheless, a more precise control is provided using the Ward identities of the reference model. These are analogous to the identities derived in Proposition \ref{prop:2ptref}, however this time in presence of an infrared momentum cutoff on scale $h$.

In the following, we omit the `ref' superscript from all running couplings, to avoid cluttering of the equations.

We denote by $\langle\, \cdot\, \rangle_{h}$ the state of the reference model, after removing all cutoffs and taking $\beta, L\to \infty$, while keeping fixed the infrared cutoff. Proceeding as in \cite[Section 9]{MPmulti}, we obtain the anomalous vertex Ward identity:
\begin{equation}\label{eq:verteidref}
\langle\hat n_{\ul p,\omega};\hat\psi^{-}_{\ul k,\omega'};\hat\psi^{+}_{\ul k +\ul p,\omega'}\rangle_{h} = \frac{T_{\omega\omega'}(\underline{p})}{Z_{\omega} D_{\omega}(\underline{p})} \big[ \langle\hat\psi^{-}_{\ul k,\omega'}\hat\psi^{+}_{\ul k,\omega'}\rangle_{h}  - \langle\hat\psi^{-}_{\ul k+\ul p,\omega'}\hat\psi^{+}_{\ul k+\ul p,\omega'}\rangle_{h} + H^{(h)}_{1,2; \omega,\omega'}(\ul p, \ul k)]\;,
\end{equation}
where the matrix $T(\ul p)$ is as in (\ref{eq:Tdef}), taking $N\to \infty$, and $H^{(h)}_{1,2}$ is an error term due to the presence of the infrared cutoff, satisfying the estimate, \cite[Eq. (9.61)]{MPmulti}:
\begin{equation}\label{eq:bdH}
\bigg| \frac{1}{Z_{\omega} D_{\omega}(\ul p)} H^{(h)}_{1,2;\omega,\omega'}(\ul p, \ul k)  \bigg| \leq C\lambda_{\geq h} \frac{2^{-2h}}{Z_{\omega} Z_{h-1,\omega'}}
\end{equation}
with $\lambda_{\geq h} = \max_{k\geq h} \max_{\omega,\omega'} |\lambda^{\text{ref}}_{k,\omega\omega'}|$. As we shall see, the identity (\ref{eq:verteidref}) can be used to control the flow of the vertex renormalization $\mathsf Z^{\mu,\omega}_{h,\omega'}$ and of the effective Fermi velocity $v_{h,\omega}$. This is done by expressing the vertex function and the two-point function in terms of GN trees, and extracting the leading infrared behavior. We have, for $\|\underline{k}\| = O(2^{h})$:
\begin{equation}\label{eq:2ptref}
 \langle\hat\psi^{-}_{\ul k,\omega'}\hat\psi^{+}_{\ul k,\omega'}\rangle_{h} = \frac{1}{Z_{h-1,\omega'} D_{h-1,\omega'}(\ul k)}(1 + r_{\omega'}(\ul k))\;, 
\end{equation}
where $r_{\omega'}(\ul k)$ denotes a sum of error terms bounded by $C2^{\theta h}$ or by $C \lambda_{\geq h}$ and $D_{h-1,\omega'}(\ul p) = ip_{0} + v_{h-1,\omega'} p_{1}$. Also, for $\| \ul k \| \sim 2^{h}$, $\| \ul k + \ul p\| \sim 2^{h}$:
\begin{equation}\label{eq:21ptref}
\mathsf Z^{\mu}_{\omega} \langle\hat n_{\ul p,\omega};\hat\psi^{-}_{\ul k,\omega'};\hat\psi^{+}_{\ul k +\ul p,\omega'}\rangle_{h} = \mathsf Z^{\mu,\omega}_{h,\omega'}  Z_{h-1,\omega'}  \langle\hat\psi^{-}_{\ul k,\omega'}\hat\psi^{+}_{\ul k,\omega'}\rangle_{h}  \langle\hat\psi^{-}_{\ul k + \ul p,\omega'}\hat\psi^{+}_{\ul k + \ul p,\omega'}\rangle_{h} (1 + r_{\omega,\omega'}(\ul k,\ul p))
\end{equation}
with $r_{\omega,\omega'}(\ul k,\ul p)$ a sum of error terms bounded by $C2^{\theta h}$ or by $C \lambda_{\geq h}$. We claim that, plugging (\ref{eq:2ptref}), (\ref{eq:21ptref}) in (\ref{eq:verteidref}), one obtains, for $\|\ul k\| \sim 2^{h}$, $\|\ul p\| \sim 2^{h}$ and $\|\ul k + \ul p\| \sim 2^{h}$:
\begin{equation}\label{eq:Zrel}
|Z_{\omega} (\mathsf Z^{\mu,\omega}_{h,\omega'} / \mathsf Z^{\mu}_{\omega}) - \delta_{\omega,\omega'}|\leq C(\lambda_{\geq h} + 2^{\theta h})\;,\qquad |v_{h-1,\omega} - v_{\omega}| \leq C(\lambda_{\geq h} + 2^{\theta h})\;.
\end{equation}
Let us check (\ref{eq:Zrel}). Consider the difference of the two-point functions in (\ref{eq:verteidref}). We have, from (\ref{eq:2ptref}):
\begin{equation}\label{eq:diff2pt}
\begin{split}
& \langle\hat\psi^{-}_{\ul k,\omega'}\hat\psi^{+}_{\ul k,\omega'}\rangle_{h} -  \langle\hat\psi^{-}_{\ul k + \ul p,\omega'}\hat\psi^{+}_{\ul k + \ul p,\omega'}\rangle_{h} \\
&\quad = \frac{1}{Z_{h-1,\omega'} D_{h-1,\omega'}(\ul k)}(1 + r_{\omega'}(\ul k)) - \frac{1}{Z_{h-1,\omega'} D_{h-1,\omega'}(\ul k + \ul p)}(1 + r_{\omega'}(\ul k + \ul p)) \\
&\quad = \bigg(\frac{1}{Z_{h-1,\omega'} D_{h-1,\omega'}(\ul k)} - \frac{1}{Z_{h-1,\omega'} D_{h-1,\omega'}(\ul k + \ul p)}\bigg)(1 + r_{\omega'}(\ul k)) \\
&\qquad + \frac{1}{Z_{h-1,\omega'} D_{h-1,\omega'}(\ul k + \ul p)} (  r_{\omega'}(\ul k) - r_{\omega'}(\ul k + \ul p) ) \\
&\quad = \frac{D_{h-1,\omega'}(\ul p)}{Z_{h-1,\omega'} D_{h-1,\omega'}(\ul k) D_{h-1,\omega'}(\ul k + \ul p)} (1 + r_{\omega'}(\ul k)) \\&\qquad + \frac{1}{Z_{h-1,\omega'} D_{h-1,\omega'}(\ul k + \ul p)} (  r_{\omega'}(\ul k) - r_{\omega'}(\ul k + \ul p) )\;.
 \end{split}
\end{equation}
Let:
\begin{equation}
\Gamma^{(h)}_{\omega,\omega'}(\ul k, \ul p) := D_{h-1,\omega'}(\ul k+\ul p)\langle\hat n_{\ul p,\omega};\hat\psi^{-}_{\ul k,\omega'};\hat\psi^{+}_{\ul k +\ul p,\omega'}\rangle_{h} D_{h-1,\omega'}(\ul k)\;.
\end{equation}
Observe that from (\ref{eq:21ptref}) together with (\ref{eq:2ptref}):
\begin{equation}\label{eq:Gammah}
\begin{split}
\Gamma^{(h)}_{\omega,\omega'}(\ul k, \ul p) &= \frac{\mathsf Z^{\mu,\omega}_{h,\omega'}}{\mathsf Z^{\mu}_{\omega} Z_{h-1,\omega'}} (1 + r_{\omega'}(\ul k)) (1 + r_{\omega'}(\ul k + \ul p)) (1 + r_{\omega,\omega'}(\ul k,\ul p)) \\
&\equiv \frac{\mathsf Z^{\mu,\omega}_{h,\omega'}}{ \mathsf Z^{\mu}_{\omega}  Z_{h-1,\omega'}} (1 + \tilde r_{\omega,\omega'}(\ul{k}, \ul p))\;.
\end{split}
\end{equation}
On the other hand, from (\ref{eq:diff2pt}):
\begin{equation}\label{eq:diff2pt2}
\begin{split}
&D_{h-1,\omega'}(\ul k+\ul p) \frac{T_{\omega\omega'}(\underline{p})}{Z_{\omega} D_{\omega}(\underline{p})}\Big( \langle\hat\psi^{-}_{\ul k,\omega'}\hat\psi^{+}_{\ul k,\omega'}\rangle_{h} -  \langle\hat\psi^{-}_{\ul k + \ul p,\omega'}\hat\psi^{+}_{\ul k + \ul p,\omega'}\rangle_{h}\Big) D_{h-1,\omega'}(\ul k) \\
&\quad = \frac{T_{\omega\omega'}(\underline{p})}{Z_{\omega} D_{\omega}(\underline{p})} \frac{D_{h-1,\omega'}(\ul p)}{Z_{h-1,\omega'}} (1 + r_{\omega'}(\ul k)) + \frac{T_{\omega\omega'}(\underline{p})}{Z_{\omega} D_{\omega}(\underline{p})} \frac{D_{h-1,\omega'}(\ul k)}{Z_{h-1,\omega'}} (  r_{\omega'}(\ul k) - r_{\omega'}(\ul k + \ul p) )\;.
\end{split}
\end{equation}
Let us study the first term in (\ref{eq:diff2pt2}) as a function of $\underline{p}$, always under the constraints $\|\ul k\| \sim 2^{h}$, $\|\ul p\| \sim 2^{h}$ and $\|\ul k + \ul p\| \sim 2^{h}$. Taking $p_{1}\to 0$:
\begin{equation}\label{eq:p10}
\lim_{p_{1}\to 0}  \frac{T_{\omega\omega'}(\underline{p})}{Z_{\omega} D_{\omega}(\underline{p})} \frac{D_{h-1,\omega'}(\ul p)}{Z_{h-1,\omega'}} (1 + r_{\omega'}(\ul k)) = \frac{T_{\omega\omega'}((p_{0}, 0))}{Z_{\omega} Z_{h-1,\omega'}}  (1 + r_{\omega'}(\ul k))\;.
\end{equation}
On the other hand, taking $p_{0}\to 0$:
\begin{equation}\label{eq:p00}
\lim_{p_{0}\to 0}  \frac{T_{\omega\omega'}(\underline{p})}{Z_{\omega} D_{\omega}(\underline{p})} \frac{D_{h-1,\omega'}(\ul p)}{Z_{h-1,\omega'}} (1 + r_{\omega'}(\ul k)) = \frac{T_{\omega\omega'}((0, p_{1})) v_{h-1,\omega'}}{Z_{\omega} Z_{h-1,\omega'} v_{\omega}}  (1 + r_{\omega'}(\ul k))\;.
\end{equation}
Instead, the last term in (\ref{eq:diff2pt2}) is bounded as, for our choice of external momenta:
\begin{equation}
\bigg|\frac{T_{\omega\omega'}(\underline{p})}{Z_{\omega} D_{\omega}(\underline{p})} \frac{D_{h-1,\omega'}(\ul k)}{Z_{h-1,\omega'}} (  r_{\omega'}(\ul k) - r_{\omega'}(\ul k + \ul p) )\bigg| \leq \frac{C(1 + v_{h-1,\omega'})( \lambda_{\geq h} + 2^{\theta h})}{Z_{\omega} Z_{h-1,\omega'}}\;.
\end{equation}
This term is clearly subleading with respect to (\ref{eq:p10}), (\ref{eq:p00}). Similarly, thanks to (\ref{eq:bdH}), also the last term in (\ref{eq:verteidref}) gives a subleading contribution:
\begin{equation}
\bigg| D_{h-1,\omega'}(\ul k+\ul p) \frac{T_{\omega\omega'}(\underline{p})}{Z_{\omega} D_{\omega}(\underline{p})} H^{(h)}_{1,2; \omega,\omega'}(\ul p, \ul k) D_{h-1,\omega'}(\ul k) \bigg| \leq \frac{C(1 + v_{h-1,\omega'}) \lambda_{\geq h}}{Z_{\omega} Z_{h-1,\omega'}}\;.
\end{equation}
Taking these considerations into account, we can rewrite the Ward identity (\ref{eq:verteidref}) as:
\begin{equation}\label{eq:WIrefvert2}
\Gamma^{(h)}_{\omega,\omega'}(\ul k, \ul p) =  \frac{T_{\omega\omega'}(\underline{p})}{Z_{\omega} D_{\omega}(\underline{p})} \frac{D_{h-1,\omega'}(\ul p)}{Z_{h-1,\omega'}} (1 + r_{\omega'}(\ul k)) + E^{(h)}_{\omega,\omega'}(\ul k, \ul p)
\end{equation}
with:
\begin{equation}
\Big| E^{(h)}_{\omega,\omega'}(\ul k, \ul p) \Big| \leq \frac{C(1 + v_{h-1,\omega'})( \lambda_{\geq h} + 2^{\theta h})}{Z_{\omega} Z_{h-1,\omega'}}\;.
\end{equation}
Recalling (\ref{eq:Gammah}), Eq. (\ref{eq:WIrefvert2}) implies, taking $p_{1}\to 0$ and using (\ref{eq:p10}):
\begin{equation}\label{eq:vert3}
\frac{\mathsf Z^{\mu,\omega}_{h,\omega'}}{ \mathsf Z^{\mu}_{\omega} Z_{h-1,\omega'}} (1 + \tilde r_{\omega,\omega'}(\ul{k}, (p_{0}, 0))) = \frac{T_{\omega\omega'}((p_{0}, 0))}{Z_{\omega} Z_{h-1,\omega'}}  (1 + r_{\omega'}(\ul k)) + E^{(h)}_{\omega,\omega'}(\ul k, (p_{0}, 0))\;;
\end{equation}
using that $T_{\omega\omega'}(\underline{p}) = \delta_{\omega,\omega'} + O(\lambda)$, Eq. (\ref{eq:vert3}) implies:
\begin{equation}\label{eq:ZZWI}
|Z_{\omega} (\mathsf Z^{\mu,\omega}_{h,\omega'} / \mathsf Z^{\mu}_{\omega}) - \delta_{\omega,\omega'}|\leq C(\lambda_{\geq h} + 2^{\theta h})\;.
\end{equation}
This proves the first of (\ref{eq:Zrel}). Let us now discuss the flow of the Fermi velocity. Taking $p_{0}\to 0$ and $\omega = \omega'$, and using (\ref{eq:p00}):
\begin{equation}
\frac{\mathsf Z^{\mu,\omega}_{h,\omega}}{ \mathsf Z^{\mu}_{\omega} Z_{h-1,\omega}} (1 + \tilde r_{\omega,\omega}(\ul{k}, (0, p_{1}))) = \frac{T_{\omega\omega}((0, p_{1})) v_{h-1,\omega}}{Z_{\omega} Z_{h-1,\omega} v_{\omega}}  (1 + r_{\omega}(\ul k)) + E^{(h)}_{\omega,\omega}(\ul k, (0, p_{1}))\;;
\end{equation}
from (\ref{eq:ZZWI}), we deduce that:
\begin{equation}
|v_{h-1,\omega} - v_{\omega}| \leq C(\lambda_{\geq h} + 2^{\theta h})\;.
\end{equation}
This concludes the proof of (\ref{eq:Zrel}). 

It remains to control the flow of the quartic coupling. This is the most technical part of the analysis, and we will not repeat it here. It was first done for the reference model with two opposite chiralities \cite{BMWI, BMchiral}, see also \cite{BGPS} for an earlier approach based on the exact solution of the Luttinger model, and then adapted to the case of several chiralities in \cite{MPmulti}. It is based on a closed equation for the four-point function, which is obtained combining a Schwinger-Dyson equation for the four-point function \cite[Proposition 9.1]{MPmulti} with Ward identities \cite[Proposition 9.5]{MPmulti}. The various correction terms entering in these Ward identities are estimated in \cite[Proposition 9.6]{MPmulti}. The flow of the quartic coupling is then controlled as in \cite[Section 9.4.3]{MPmulti}; the outcome of the analysis is that $\lambda^{\text{ref}}_{h,\omega\omega'} = \lambda^{\text{ref}}_{\omega\omega'} + O(\lambda_{\geq h}^{2})$, which gives the last of (\ref{eq:BdsEffCoupref}).


\newpage


\begin{thebibliography}{999999}
\bibitem{AMP} G. Antinucci, V. Mastropietro and M. Porta. Universal edge transport in interacting Hall systems. {\it Commun. Math. Phys.} {\bf 362}, 295-359 (2018).
\bibitem{AH} H. Araki and T. G. Ho. Asymptotic time evolution of a partitioned infinite two-sided isotropic xy-chain. {\it Proc. Steklov Inst. Math.} {\bf 228}, 191-204 (2000).
\bibitem{AP} W. Aschbacher and C.-A. Pillet. Non-equilibrium steady states of the xy chain. {\it J. Stat. Phys.} {\bf 112}, 1153-1175 (2003).
\bibitem{BDRF} S. Bachmann, W. de Roeck and M. Fraas. The adiabatic theorem and linear response theory for extended quantum systems. {\it Comm. Math. Phys.} {\bf 361}, 997-1027 (2018)
\bibitem{BDRFL} S. Bachmann, W. De Roeck, M. Fraas and M. Lange. Exactness of Linear Response in the Quantum Hall Effect. {\it Ann. Henri Poincaré} {\bf 22}, 1113-1132 (2021).
\bibitem{BFM} G. Benfatto, P. Falco and V. Mastropietro. Extended scaling relations for planar lattice models. {\it Commun. Math. Phys.} {\bf 292}, 569-605 (2009).
\bibitem{BFM1} G. Benfatto, P. Falco, V. Mastropietro. Massless Sine-Gordon and Massive Thirring Models: Proof of Coleman’s Equivalence. {\it Comm. Math. Phys.} {\bf 285}, 713-762 (2009).
\bibitem{BFM2}  G. Benfatto, P. Falco and V. Mastropietro. Universality of One-Dimensional Fermi Systems, I. Response Functions and Critical Exponents. {\it Commun. Math. Phys.} {\bf 330}, 153-215 (2014).
\bibitem{BFM3}  G. Benfatto, P. Falco and V. Mastropietro. Universality of One-Dimensional Fermi Systems, II. The Luttinger Liquid Structure. {\it Commun. Math. Phys.} {\bf 330}, 217-282 (2014).
\bibitem{BM} G. Benfatto and V. Mastropietro. Renormalization group, hidden symmetries and approximate Ward identities in the XYZ model. {\it Rev. Math. Phys.} {\bf 13}, 1323-1435 (2001).
\bibitem{BMWI} G. Benfatto and V. Mastropietro. Ward Identities and Vanishing of the Beta Function for $d=1$ Interacting Fermi Systems. {\it J. Stat. Phys.} {\bf 115}, 143-184 (2004).
\bibitem{BMdensity} G. Benfatto and V. Mastropietro. On the Density-Density Critical Indices in Interacting Fermi Systems. {\it Commun. Math. Phys.} {\bf 231}, 97-134 (2002).
\bibitem{BMchiral} G. Benfatto and V. Mastropietro. Ward Identities and Chiral Anomaly in the Luttinger Liquid. {\it Commun. Math. Phys.} {\bf 258}, 609-655 (2005).
\bibitem{BGPS} G. Benfatto, G. Gallavotti, A. Procacci and B. Scoppola. Beta function and Schwinger functions for a many fermions system in one dimension. Anomaly of the Fermi surface. {\it Comm. Math. Phys.} {\bf 160}, 93-171 (1994).
\bibitem{BCNF} B. Bertini, M. Collura, J. De Nardis and M. Fagotti. Transport in out-of-equilibrium XXZ chains: Exact profiles of charges and currents. {\it Phys. Rev. Lett.} {\bf 117}, 207201 (2016).
\bibitem{BoM} F. Bonetto and V. Mastropietro. Beta function and anomaly of the Fermi surface for a  $d=1$ system of interacting fermions in a periodic potential. {\it Comm. Math. Phys.} {\bf 172}, 57-93 (1995).
\bibitem{Bry} D. C. Brydges. {\it A Short Course on Cluster Expansions,} in Summer School in Theoretical Physics, Session XLIII: Critical Phenomena, Random Systems, Gauge Theories. Les Houches, France, August 1-September 7, 1984, pp. 129–183. 1984.
\bibitem{CMP} F. Caragiulo, V. Mastropietro and M. Porta. Edge transport in Haldane-like models with quasi-periodic disorder. {\it Comm. Math. Phys.} (to appear), {\tt arXiv:2410.03946}.
\bibitem{CMP2} H. D. Cornean, V. Moldoveanu and C.-A. Pillet. On the Steady State Correlation Functions of Open Interacting Systems. {\it Comm. Math. Phys.} {\bf 331}, 261-295 (2014).
\bibitem{Doyrev} B. Doyon, S. Gopalakrishnan, F. Møller, J Schmiedmayer and R. Vasseur. Generalized hydrodynamics: a perspective. {\it Phys. Rev. X} {\bf 15}, 010501 (2025).
\bibitem{Essrev} F. H. L. Essler. A short introduction to Generalized Hydrodynamics. {\it Physica A} {\bf 631}, 127572 (2023).
\bibitem{FMU} J. Fr\"ohlich, M. Merkli and D. Ueltschi. Dissipative transport: thermal contacts and tunneling junctions. {\it Ann. Henri Poincaré} {\bf 4}, 897-945 (2004).
\bibitem{GM} G. Gentile and V. Mastropietro. Renormalization group for one-dimensional fermions. A review on mathematical results. {\it Phys. Rep.} {\bf 352}, 273-438 (2001).
\bibitem{GGM} A. Giuliani, R. L. Greenblatt and V. Mastropietro. The scaling limit of the energy correlations in non integrable Ising models. {\it J. Math. Phys.} {\bf 53}, 095214 (2012).
\bibitem{GiuM} A. Giuliani and V. Mastropietro. The Two-Dimensional Hubbard Model on the Honeycomb Lattice. {\it Comm. Math. Phys.} {\bf 293}, 301--346 (2010).
\bibitem{GMPcond} A. Giuliani, V. Mastropietro and M. Porta. Universality of conductivity in interacting graphene. {\it Comm. Math. Phys.} {\bf 311}, 317-355 (2012).
\bibitem{GMP} A. Giuliani, V. Mastropietro and M. Porta. Universality of the Hall conductivity in interacting electron systems. {\it Comm. Math. Phys.} {\bf 349}, 1107-1161 (2017).
\bibitem{GMPweyl} A. Giuliani, V. Mastropietro and M. Porta. Anomaly non-renormalization in interacting Weyl semimetals. {\it Comm. Math. Phys.} {\bf 384}, 997-1060 (2021).
\bibitem{GMR} A. Giuliani, V. Mastropietro, and S. Rychkov. Gentle introduction to rigorous Renormalization Group: a worked fermionic example. {\it J. High Energ. Phys.} {\bf 2021}, 26 (2021). %\url{https://doi.org/10.1007/JHEP01(2021)026}
\bibitem{GMT} A. Giuliani, V. Mastropietro and F. L. Toninelli. Height fluctuations in interacting dimers. {\it Ann. de l'Institut H. Poincaré - Probabilités et Statistiques} {\bf 53}, 98-168 (2017).
\bibitem{GLMP} R. L. Greenblatt, M. Lange, G. Marcelli and M. Porta. Adiabatic Evolution of Low-Temperature Many-Body Systems. {\it Comm. Math. Phys.} {\bf 405}, 75 (2024).
\bibitem{HT} J. Henheik and S. Teufel. Justifying Kubo's formula for gapped systems at zero temperature: a brief review and some new results. {\it Rev. Math. Phys.} {\bf 33}, 2060004 (2021).
\bibitem{JOP} V. Jaksi\'c, Y. Ogata and C.-A. Pillet. The Green-Kubo formula for locally interacting fermionic open systems. {\it Ann. Henri Poincaré} {\bf 8}, 1013-1036 (2007).
\bibitem{Les} A. Lesniewski. Effective Action for the Yukawa$_2$ Quantum Field Theory. {\it Commun. Math. Phys.} {\bf 108} 437-467 (1987).
\bibitem{ML} D. C. Mattis and E. H. Lieb. Exact Solution of a Many-Fermion System and Its Associated Boson Field. {\it J. Math. Phys.} {\bf 6}, 304-312 (1965).
\bibitem{Mabook} V. Mastropietro. {\it Non-Perturbative Renormalization.} World Scientific (2008).
\bibitem{MPmulti} V. Mastropietro and M. Porta. Multi-channel Luttinger Liquids at the Edge of Quantum Hall Systems. {\it Comm. Math. Phys.} {\bf 395}, 1097-1173 (2022).
\bibitem{MCP} V. Moldoveanu, H. D. Cornean and C.-A. Pillet. Non-equilibrium steady-states for interacting open systems: exact results. {\it Phys. Rev. B} {\bf 84}, 075464 (2011).
\bibitem{MT} D. Monaco and S. Teufel. Adiabatic currents for interacting electrons on a lattice. {\it Rev. Math. Phys.} {\bf 31}, 1950009 (2019).
\bibitem{Ne} G. Nenciu. Independent electrons model for open quantum systems: Landauer-B\"uttiker formula and strict positivity of the entropy production. {\it J. Math. Phys.} {\bf 48}, 033302 (2007).
\bibitem{CDY} Olalla A. Castro-Alvaredo, B. Doyon and T. Yoshimura. Emergent hydrodynamics in integrable quantum systems out of equilibrium. {\it Phys. Rev. X} {\bf 6}, 041065 (2016).
\bibitem{PS} M. Porta and H. P. Singh. Large Scale Response of Gapless $1d$ and Quasi-$1d$ Systems. {\it Ann. Henri Poincar\'e} {\bf 27}, 2997-3055 (2026).
\bibitem{Teu} S. Teufel. Non-equilibrium almost-stationary states and linear response for gapped quantum systems. {\it Comm. Math. Phys.} {\bf 373}, 621-653 (2020).
\bibitem{Teu2}  M. Wesle, G. Marcelli, T. Miyao, D. Monaco and S. Teufel. Near linearity of the macroscopic Hall current response in infinitely extended gapped fermion systems. {\tt arXiv:2411.06967}
\end{thebibliography}
\end{document}